\documentclass[final,10pt,times,3p]{elsarticle}
\usepackage{graphicx,epstopdf,amsmath,amssymb,braket,color,epsfig,epstopdf,geometry,amsthm}
\usepackage[colorlinks=true, citecolor=red, urlcolor=blue ]{hyperref}
\numberwithin{equation}{section}
\numberwithin{figure}{section}
\usepackage[table]{xcolor}
\usepackage{times,array,multirow}
\usepackage{bm}
\usepackage{tocloft}
\usepackage{epsf}
\usepackage{epsfig}
\usepackage{tikz}
\usepackage{graphicx}
\usepackage{float}
\usepackage{subcaption}
\usepackage{tikz}
\usepackage[title]{appendix}
\usepackage{mathtools}
\usepackage{amsmath}
\allowdisplaybreaks
\biboptions{sort&compress}
\usepackage{titlesec}
\usepackage{enumitem}
\usepackage{amssymb}

\theoremstyle{plain}
\newtheorem{prop}{Proposition}[section]

\newtheorem{theo}{Theorem}[section]
\newtheorem{lem}{Lemma}[section]

\newtheorem{assu}{Assumption}[section]
\newtheorem{rhp}{RH problem}[section]

\newtheorem{rhp-Dbar}{RH-$\bar{\partial}$ problem}[section]

\def\rd{{\rm d}}
\def\re{{\rm e}}
\def\ri{{\rm i}}

\def\rim{{\rm {Im}}\,}
\def\rre{{\rm {Re}}\,}

\newtheorem{remark}{Remark}[section]
\makeatletter
\def\ps@pprintTitle{%
  \let\@oddhead\@empty
  \let\@evenhead\@empty
  \def\@oddfoot{\footnotesize\itshape
  }%
  \let\@evenfoot\@oddfoot
}
\makeatother

\usepackage{amsmath}

\begin{document}

\title{{\bf Riemann-Hilbert problem and long-time asymptotics of the Yajima-Oikawa equation}}

\author[1]{Deng-Shan Wang}
\author[1]{Yingmin Yang \corref{cor1}}
\cortext[cor1]{Corresponding author}
\ead{ymyang011@163.com}
\author[1]{Xiaodong Zhu}

\address[1]{School of Mathematical Sciences,
	Beijing Normal University, Beijing 100875, China}

\begin{abstract}

The Yajima-Oikawa equation is an integrable long wave-short wave resonance interaction  model arising as a deformation of the Zakharov system for Langmuir waves coupled to ion-acoustic waves. In this work, a Riemann-Hilbert approach is developed for the Cauchy problem for the Yajima-Oikawa equation with rapidly decaying initial data. A main novelty is the formulation of a direct and inverse scattering theory adapted to its third-order  spectral problem, including a detailed treatment of the singular spectral point \(k=0\). The associated Riemann-Hilbert problem is expressed in terms of two reflection coefficients determined by the initial data, together with possible discrete eigenvalues and norming constants. We prove a vanishing lemma which ensures the unique solvability of the Riemann-Hilbert problem under suitable positivity assumptions, and hence obtain a rigorous reconstruction formula for the solution. We also classify the admissible discrete spectrum and derive exact pure soliton solutions from the reflectionless Riemann-Hilbert problem. In the solitonless case, we apply the Deift-Zhou nonlinear steepest descent method to obtain rigorous long-time asymptotic formulas in the different regions of the upper \((x,t)\)-plane. The leading oscillatory behavior of the short-wave component is described explicitly in terms of the reflection coefficients evaluated at the stationary phase points, while the long-wave component is shown to be of lower order away from the transition region. These results provide, to the best of our knowledge, the first Riemann-Hilbert framework for the long-time asymptotic analysis of the Yajima-Oikawa equation in the presence of continuous spectrum.
	
\end{abstract}

\maketitle

	\tableofcontents
	\section{Introduction}

The resonant interaction between long waves and short waves is a fundamental mechanism in nonlinear wave propagation. It occurs when a slowly varying, weakly dispersive long wave is coupled resonantly to the envelope of a rapidly oscillating short wave. Such a mechanism appears in a variety of physical settings, including plasma physics, stratified fluids, nonlinear optics, and Bose-Einstein condensates. Among the integrable models describing this type of interaction, the Yajima-Oikawa (YO) equation occupies a particularly important place. In its normalized form, the YO equation \cite{YO_1976,YO_1975} reads
\begin{equation}\label{YO}
		\left\lbrace 
		\begin{aligned}
			&\ri\frac{\partial\phi}{\partial t}+\ri\frac{\partial\phi}{\partial x}+\frac{1}{2}\frac{\partial^2\phi}{\partial x^2}-\phi n=0,\\
			&\frac{\partial n}{\partial t}+\frac{\partial n}{\partial x}+\frac{\partial \left|\phi \right|^2 }{\partial x}=0,
		\end{aligned}
		\right.
\end{equation}
where the real-valued function \(n(x,t)\) represents the long-wave component and the complex-valued function \(\phi(x,t)\) represents the envelope of the short-wave component.
\par
The YO equation (\ref{YO}) was introduced by Yajima and Oikawa \cite{YO_1976,YO_1975} in their study of sonic-Langmuir solitons and may be regarded as a deformation or reduction of the Zakharov model for Langmuir waves interacting with ion-acoustic waves \cite{Zakharov_1972}. Closely related long-wave-short-wave resonance models were also investigated in the classical works of Benney \cite{Benney_1977}, Newell \cite{Newell-1978}, Ma and Redekopp \cite{MaYC_1978,Ma-Redekopp-1979}. In this context, the YO equation provides a mathematically tractable and physically meaningful model for the nonlinear trapping of high-frequency oscillations by low-frequency density variations. Its complete integrability, \(3\times3\) Lax representation, soliton solutions, and inverse-scattering structure make it one of the canonical examples of integrable long wave-short wave interaction systems.
\par
Since the original works of Yajima and Oikawa, many aspects of the YO equation and its variants have been studied. The inverse scattering method and explicit \(N\)-soliton solutions were developed in the early literature \cite{YO_1976,MaYC_1978}. Wright \cite{Wright_2006} constructed Bäcklund transformations and studied homoclinic structures associated with unstable plane-wave solutions. Nistazakis and his colleagues \cite{Nistazakis} derived the YO equation as an effective model from multicomponent Gross-Pitaevskii systems arising in spinor Bose-Einstein condensates. Zabolotskii \cite{Zabolotskii_2009}  later constructed an inverse scattering transform for the YO equation with mixed vanishing and nonvanishing boundary values. More recently, long wave-short wave interaction models generalizing the YO equation, such as the Yajima-Oikawa-Newell system, have attracted renewed attention from the viewpoints of integrability, stability spectra, periodic waves, solitary waves and breathers \cite{Degasperis-2021,Degasperis-2022,Li-Geng-2019}. These developments show that the YO equation remains an active and relevant model both in nonlinear science and in the theory of integrable systems.
\par
A different but closely related question concerns the long-time asymptotic behavior of solutions. For integrable dispersive equations, the inverse scattering transform converts the nonlinear Cauchy problem into a spectral problem. The long-time asymptotics of the solution can then be extracted from the associated oscillatory Riemann-Hilbert (RH) problem. One of the earliest systematic studies in this direction was carried out by Zakharov and Manakov \cite{Zakharov-Manakov-1976}, who investigated the non-soliton long-time asymptotics of integrable nonlinear wave systems within the inverse scattering framework. A major breakthrough was later achieved by Deift and Zhou \cite{DZ_1993}, who introduced the nonlinear steepest descent method for oscillatory RH problems. Since its appearance, this method has become one of the most powerful tools for obtaining rigorous asymptotic formulas for integrable equations with rapidly decaying initial data. It has been successfully applied to the KdV equation \cite{DVZ_1994}, the nonlinear Schrödinger equation \cite{DZ_1994-Tokyo,Kamvissis_2996}, the Toda lattice \cite{Kamvissis_2993,DKKZ_1996}, the sine-Gordon equation \cite{CVZ_1999}, the derivative nonlinear Schrödinger equation \cite{KV_1999}, the Camassa-Holm equation \cite{BKST_2009}, and the Ablowitz-Ladik equation \cite{Yamane_2014}. 
\par
In recent years, the Deift-Zhou method has been extended beyond the standard \(2\times2\) AKNS-type setting to integrable equations whose Lax pairs are of third order or whose associated spectral problems naturally lead to \(3\times3\) matrix RH problems. For example, for the Degasperis-Procesi equation, Boutet de Monvel, Lenells, and Shepelsky \cite{BLS_2019} studied the long-time asymptotics on the half-line, while Fan \cite{Fan-JDE-2025} established soliton resolution and the asymptotic stability of \(N\)-solitons on the line. Geng and Liu \cite{Geng-Liu-2018} considered the long-time asymptotics of the coupled nonlinear Schrödinger equation. More recently, Charlier, Lenells and Wang \cite{Analysis_2023}, as well as Charlier and Lenells \cite{CL-JMPA_2023}, analyzed the good and bad Boussinesq equations by means of \(3\times3\) RH techniques. Wang and Zhu \cite{Wang-Zhu_2023} studied the Sawada-Kotera equations on the line, and Huang, Wang, and Zhu \cite{Huang-Wang-Zhu_2024} investigated the Tzitzéica equation from affine differential geometry. These works demonstrate that third-order Lax pairs typically require new analytic ingredients, including nonstandard analyticity domains, singular behavior at special spectral points, and more involved factorizations of jump matrices.
\par 
Despite the complete integrability of the YO equation (\ref{YO}), its rigorous long-time asymptotic analysis for rapidly decaying initial data appears to have remained open. In particular, to the best of our knowledge, a RH formulation suitable for the Deift-Zhou nonlinear steepest descent analysis of the decaying Cauchy problem for the YO equation has not been previously available. The main purpose of the present paper is to fill this gap. We consider the initial-value problem
\begin{equation}\label{YOE}
		\left\lbrace 
		\begin{aligned}
			&\ri\frac{\partial\phi}{\partial t}+\ri\frac{\partial\phi}{\partial x}
			+\frac{1}{2}\frac{\partial^2\phi}{\partial x^2}-\phi n=0,\\
			&\frac{\partial n}{\partial t}+\frac{\partial n}{\partial x}
			+\frac{\partial \left|\phi \right|^2 }{\partial x}=0,\\
			&\phi(x,0)=\phi_0(x)\in\mathcal{S}(\mathbb{R}),\quad
			n(x,0)=n_0(x)\in\mathcal{S}(\mathbb{R}),
		\end{aligned}
		\right.
	\end{equation}
and study its direct scattering transform, inverse problem, soliton solutions, and long-time asymptotic behavior by means of the RH approach and the Deift-Zhou nonlinear steepest descent method \cite{DZ_1993}. In this sense, the present work provides the first RH framework designed for the rigorous long-time asymptotic analysis of the YO equation with continuous spectrum.
\par
The starting point of our analysis is the following \(3\times3\) Lax pair for the YO equation (\ref{YO}):
\begin{equation}\label{lax_phi}
		\begin{aligned}
			&\Phi_x(x,t,k)=\tilde{L}(x,t,k)\Phi(x,t,k), \\ 
			&\Phi_t(x,t,k)=\tilde{Z}(x,t,k)\Phi(x,t,k),
		\end{aligned}
	\end{equation}
where 
\begin{equation*}
\tilde{L}=\begin{pmatrix}
				3\ri k & 2\ri n & \overline{\phi}\\
				-\ri & -\ri k & 0\\
				0 & 2\phi & \ri k\\
			\end{pmatrix},\quad \tilde{Z}=\begin{pmatrix}
				-2\ri k+2 \ri k^2/3 & -\ri \phi\overline{\phi}-2\ri n & (k-1)\overline{\phi}-\ri \overline{\phi}_x/2 \\
				\ri & 2\ri k+2 \ri k^2/3 & -\overline{\phi}/2\\
				-\phi & -2(k+1)\phi+\ri \phi_x & -4\ri k^2/3
			\end{pmatrix}.
\end{equation*}
\par
Compared with the standard \(2\times2\) scattering problems, the spectral analysis of \eqref{lax_phi} has several distinctive features. First, the diagonalizing transformation used to normalize the \(x\)-part of the Lax pair is singular at \(k=0\), and hence the behavior of the Jost solutions and scattering coefficients near the origin requires special care. Second, the continuous scattering data are naturally described by two reflection coefficients, denoted by \(r_1(k)\) and \(r_2(k)\), rather than by a single scalar reflection coefficient. These two reflection coefficients enter the jump matrix of the RH problem in an essential way and satisfy nontrivial symmetry and positivity relations. Third, the phase functions appearing in the jump matrix possess different stationary phase structures in different regions of the upper \((x,t)\)-half-plane. This leads to a division of the asymptotic plane into several regions, including two Zakharov-Manakov-type regions and a transition regime near the characteristic line \(x=t\).
\par
More precisely, we first develop the direct and inverse scattering theory for Schwartz-class initial data and formulate a row-vector RH problem whose jump matrix is expressed in terms of two reflection coefficients. The discrete spectrum is described through the zeros of the relevant spectral function, and the corresponding reflectionless RH problem yields exact pure soliton solutions. In the solitonless case, we apply the Deift-Zhou nonlinear steepest descent method to derive the leading long-time asymptotics in different regions. In particular, the short-wave component \(\phi(x,t)\) exhibits oscillatory leading terms, with amplitudes and phases determined by the reflection coefficients at the stationary phase points, whereas the long-wave component \(n(x,t)\) is of lower order away from the transition region.
\par
The remainder of this paper is organized as follows. Section~\ref{sec:main}
states the main results, including the scattering data, the RH problem, the pure soliton formulas, and the long-time asymptotic theorems. Section~\ref{sec:direct} develops the direct scattering theory and formulates the associated RH problems. Section~\ref{sec:inverse} establishes the vanishing lemma and derives the pure $N$-soliton solutions. Section~\ref{sec_LT} carries out the Deift-Zhou nonlinear steepest descent analysis in the different asymptotic regions. Technical estimates for the direct scattering problem are collected in Appendix~\ref{appendix:direct}, the determinant formulas for the \(N\)-soliton solutions are given in Appendix~\ref{sec_Nsoliton}, and the local model RH problems used in the asymptotic analysis are summarized in Appendix~\ref{AppendixA}.
\par
\subsection{Notations} \label{Notations-1.1}
     We summarize some notations that will be used throughout the paper:
	  \begin{itemize}
	  	\item Define $[X]_j$ as the $j$-th column and $X_{ij}$ as the element in the $i$-th row and $j$-th column of the $3 \times 3$ matrix $X$.
	  	\item $C$ and $c$ denote generic positive constants which may change within a computation.
        \item Define $\mathcal{I}_1=[0,\tau_{\text{max}}]$, $\mathcal{I}_4=[-\tau_{\text{max}},0]$ with $\tau_{\text{max}}\in(0,1)$, and let $\mathcal{I}_2 \subset (1,\infty)$ and $\mathcal{I}_3 \subset (-\infty,1)$ be fixed compact sets.
        \item $\mathcal{S}(\mathbb{R})$ denotes the Schwartz space of rapidly decreasing functions on $\mathbb{R}$.
	  	\item $r^*(k)$ denotes the Schwartz reflection of $r(k)$.
        \item The superscript $\dagger$ denotes the conjugate transpose of a matrix.
	  \end{itemize}

\section{Main results}\label{sec:main}

   This section presents the assumptions adopted in this work and summarizes the main results. For the initial-value problem \eqref{YOE} of the YO equation, we derive several results for equation \eqref{YO} by analyzing its associated RH problem \ref{rhp_M}.
   \subsection{Direct scattering and scattering data}\ \ \ \
   Define $\mathcal{L}(k)=\text{diag}(l_1(k),l_2(k),l_3(k))$ and $\mathcal{Z}(k)=\text{diag}(z_1(k),z_2(k),z_3(k))$, where
   \begin{align*}
       &l_1(k)=3\ri k, \qquad\qquad\,\, l_2(k)=\ri k,\qquad\quad\, l_3(k)=-\ri k,\\
       &z_1(k)=\frac{2\ri k^2}{3}-2\ri k,\quad\, z_2(k)=-\frac{4\ri k^2}{3},\quad\,\, z_3(k)=\frac{2\ri k^2}{3}+2\ri k,
   \end{align*}
   and take $P(k)$ and $L_1(x,k)$ in the following forms:
   \begin{equation}\label{diagonalize}
		P(k)=\begin{pmatrix}
			-4k & 0 & 0\\
			1 & 0 & 1\\
			0 & 1  & 0\\
		\end{pmatrix},\quad \det P(k)=4k,\quad L_1(x,k)=P^{-1}(k)
		    \begin{pmatrix}
				0 & 2\ri n_0 & \overline{\phi}_0\\
				0 & 0 & 0\\
				0 & 2\phi_0 & 0\\
			\end{pmatrix}P(k).
	\end{equation}

    Let \(X(x,k)\), \(X^A(x,k)\), \(Y(x,k)\), and \(Y^A(x,k)\) be the unique solutions defined by the following Volterra integral equations:
  \begin{equation}\label{XY_XAYA}
\begin{aligned}
X(x,k)
&=I-\int_x^\infty
\re^{(x-y)\widehat{\mathcal{L}}(k)}
(L_1X)(y,k)\,\rd y,
\qquad
X^A(x,k)
=I+\int_x^\infty
\re^{-(x-y)\widehat{\mathcal{L}}(k)}
(L_1^TX^A)(y,k)\,\rd y,
\\
Y(x,k)
&=I+\int_{-\infty}^x
\re^{(x-y)\widehat{\mathcal{L}}(k)}
(L_1Y)(y,k)\,\rd y,
\qquad
Y^A(x,k)
=I-\int_{-\infty}^x
\re^{-(x-y)\widehat{\mathcal{L}}(k)}
(L_1^TY^A)(y,k)\,\rd y,
\end{aligned}
\end{equation}
    where $\re^{\widehat{\mathcal{L}}}$ is an operator acting on a $3\times3$ matrix $A$: $\re^{\widehat{\mathcal{L}}}A=\re^{\mathcal{L}}A\re^{-\mathcal{L}}$. The scattering matrices $s(k)= (s_{ij}(k))_{3×3}$ and $s^A(k)= (s^A_{ij}(k))_{3×3}=((-1)^{i+j}m_{ij})_{3\times3}$ are defined through direct scattering analysis:
\begin{align}
s(k)&=I-\int_{\mathbb{R}}\re^{-x\widehat{\mathcal{L}}(k)}(L_1X)(x,k)\,\rd x,
\qquad
s^A(k)=I+\int_{\mathbb{R}}\re^{x\widehat{\mathcal{L}}(k)}(L_1^TX^A)(x,k)\,\rd x .
\label{sk}
\end{align}

    The scattering data are primarily expressed in terms of two reflection coefficients:
    \begin{equation}\label{r1_r2}\left\lbrace 
		\begin{aligned}
			&r_1(k):=\frac{s_{12}(k)}{s_{11}(k)},\\
			&r_2(k):=\frac{s^A_{31}(k)}{s^A_{33}(k)},
		\end{aligned}\right. \qquad k\in \mathbb{R}\setminus\{0\}.                                        
	\end{equation}
\subsubsection{Discrete Spectrum}
The zeros of the analytic spectral function $s_{11}(k)$ generate the discrete
spectrum and hence the soliton component of the solution. In view of
the symmetries of the scattering data, it is sufficient to consider
the zeros of $s_{11}(k)$ in the upper half-plane. We decompose
\[
    \mathbb{C}_+
    =
    D_{\mathrm{reg}}
    \cup D_{\mathrm{sing}}
    \cup \ri\mathbb{R}_+,
\]
where
\(
    D_{\mathrm{reg}}
    =
    \{k\in\mathbb{C}:\rre  k>0,\ \rim  k>0\},
    \,\,
    D_{\mathrm{sing}}
    =
    \{k\in\mathbb{C}:\rre  k<0,\ \rim  k>0\}.
\)

A simple zero $k_j=\xi_j+\ri\eta_j\in D_{\mathrm{reg}}$ of $s_{11}(k)$ generates a
regular soliton in both components $n$ and $\phi$, whose velocity is
$1-2\xi_j$. Consequently, the soliton is right-moving, stationary, or
left-moving according as $\xi_j<\frac12$, $\xi_j=\frac12$, or
$\xi_j>\frac12$, respectively. A simple zero
$k_j\in\ri\mathbb{R}_+$ generates a unit-speed solitary wave in the
long-wave component $n$, while $\phi$ vanishes. Such a solitary wave
is nonsingular provided that its norming constant satisfies
\(
    \ri C_{k_j}\leq 0.
\)
By contrast, zeros in $D_{\mathrm{sing}}$, as well as zeros on
$\ri\mathbb{R}_+$ with $\ri C_{k_j}>0$, generate singular solitary
waves. We therefore restrict the discrete spectrum to the simple zeros
of $s_{11}(k)$ in
$D_{\mathrm{reg}}\cup\ri\mathbb{R}_+$ that correspond to non-singular solitons. For specific classifications, see Figure \ref{fig_kj}.

  \begin{figure}[htbp]
    \centering

    \begin{subfigure}[b]{0.5\textwidth}
        \centering
        \begin{tikzpicture}[scale=1]
            \draw[very thick, black!20!blue,-latex] (-3,0) -- (5,0);
            \draw[very thick, black!20!blue,-latex] (0,0) -- (0,4);

            \draw[very thick, black!20!blue,dashed] (2.5,0) -- (2.5,4);

            \fill (0,0) circle (1.5pt);
            \fill (2.5,0) circle (1.5pt);
            \node[below] at (0,0) {$\text{0}$};
            \node[below] at (2.5,0) {$\text{1/2}$};

            \node[right] at (5,0) {$\rre k$};
            \node[above] at (0,4) {$\rim k$};

            \fill [green!70!black] (-1.5,2) circle (2pt)
            node[above right] {$k_1$};
            \fill [green!70!black]  (0,3) circle (2pt) node[above right] {$k_2$};
            \fill [red!70!black] (0,1) circle (2pt) node[above right] {$k_3$};
            \fill [red!70!black] (1.2,1.5) circle (2pt) node[above right] {$k_4$};
            \fill [red!70!black] (2.5,2.5) circle (2pt) node[above right] {$k_5$};
            \fill [red!70!black] (4,2) circle (2pt) node[above right] {$k_6$};
        \end{tikzpicture}
    \end{subfigure}
    \hfill
    \begin{subfigure}[b]{0.45\textwidth}
        \centering
        \begin{tikzpicture}[scale=0.9]
         \draw[very thick,white] (-4,0) -- (4,0) -- (4,5) -- (-4,5) -- (-4,0);
         
         \node [right,below] at (-2.35,5) {\textcolor{green!70!black}{$k_1$}: Singular solitary wave};
         
         \node [right,below] at (-1.4,4) {\textcolor{green!70!black}{$k_2$}: Singular solitary wave for $\ri C_{k_2}>0$};
         
         \node [right,below] at (0,3) {\textcolor{red!70!black}{$k_3$}: Right-moving soliton with unit velocity for $\ri C_{k_3}\leq0$};
         
         \node [right,below] at (-2.45,2) {\textcolor{red!70!black}{$k_4$}: Right-moving soliton};
         
         \node [right,below] at (-3.05,1) {\textcolor{red!70!black}{$k_5$}: Static soliton};
         
         \node [right,below] at (-2.55,0) {\textcolor{red!70!black}{$k_6$}: Left-moving soliton};
            
        \end{tikzpicture}
    \end{subfigure}

    \caption{Schematic diagram of different types of discrete spectrum corresponding to a single zero of $s_{11}(k)$ in the case of a reflectionless potential.}
    \label{fig_kj}
\end{figure}

\subsubsection{Norming Constants}
After defining the reflection coefficients $r_1(k)$ and $r_2(k)$  in equation \eqref{r1_r2} using the spectral matrices $s(k)$ and $s^A(k)$, the scattering data also include the set of zeros $\textbf{Z}=\{k_j\in D_{\text{reg}}\cup \ri \mathbb{R}_+\,|\,s_{11}(k_j)=0,s_{11}'(k_j)\neq0\}$ of $s_{11}(k)$, along with the corresponding norming constants $\{C_{k_j}\}_{k_j\in\textbf{Z}}\in\mathbb{C}$. For initial data $n_0(x)$ and $\phi_0(x)$ with compact support, the norming constants are defined as follows:
\begin{equation}\label{ckj_compact}
        C_{k_j}:=\left\lbrace\begin{aligned}
        &-\frac{s_{12}(k_j)}{s_{11}'(k_j)},\quad &&k_j\in \textbf{Z}\setminus\ri \mathbb{R}_+,\\
        &-\frac{s_{13}(k_j)}{s_{11}'(k_j)},\quad &&k_j\in \textbf{Z}\cap\ri \mathbb{R}_+.
    \end{aligned}\right.
\end{equation}

When the initial data $n_0(x)$ and $\phi_0(x)$ do not have compact support, Propositions \ref{prop_s} and \ref{prop_sA} imply that the scattering data $s_{12}(k_j)$ and $s_{13}(k_j)$ are generally not well-defined. The norming constant $C_{k_j}$ requires a more complex definition. First, define the vector-valued function $m(x,k)$:
\begin{equation}\label{mxk}
    m(x,k)=\begin{pmatrix}
        Y^A_{31}(x,k)X^A_{23}(x,k)-Y^A_{21}(x,k)X^A_{33}(x,k)\\
        Y^A_{11}(x,k)X^A_{33}(x,k)-Y^A_{31}(x,k)X^A_{13}(x,k)\\
        Y^A_{21}(x,k)X^A_{13}(x,k)-Y^A_{11}(x,k)X^A_{23}(x,k)\\
    \end{pmatrix}.
\end{equation}
If $k_j\in \textbf{Z}\setminus\ri \mathbb{R}_+$, the following relation uniquely determines the complex constant $C_{k_j}$:
\begin{equation}\label{ckj_non_Cc}
    \frac{m(x,k_j)}{s_{11}'(k_j)}=C_{k_j}\re^{x(l_1(k_j)-l_2(k_j))}[X(x,k_j)]_1,\quad \forall x\in \mathbb{R}.
\end{equation}
If $k_j\in \textbf{Z}\cap\ri \mathbb{R}_+$, the following relation uniquely determines the complex constant $C_{k_j}$:
\begin{equation}\label{ckj_non_Cir}
    \frac{Y_{3}(x,k_j)}{m_{33}'(k_j)}=C_{k_j}\re^{x(l_1(k_j)-l_3(k_j))}[X(x,k_j)]_1,\quad \forall x\in \mathbb{R}.
\end{equation}

\subsubsection{Scattering Data}
Zeros of $s_{11}(k)$ in $D_{\mathrm{sing}}$ lead to singular solitons, whereas zeros on $\ri\mathbb{R}_+$ are admissible only when $\ri C_{k_j}\notin(0,\infty)$. Hence, we consider finitely many simple zeros of $s_{11}(k)$ associated with non-singular solitons and impose the following assumptions.

\begin{assu}\label{assu_k=0}\rm{(Generic behavior at $k = 0$).}
Assume that 
\begin{align*}
    &\lim_{k\to0}k\,s_{11}(k)\neq0, \quad  \lim_{k\to0}k\,s_{12}
    (k)\neq0,\quad \lim_{k\to0}s_{21}(k)\neq0,\\
    &\lim_{k\to0}k\,s^A_{11}(k)\neq0, \quad  \lim_{k\to0}k\,s^A_{21}
    (k)\neq0,\quad \lim_{k\to0}s^A_{12}
    (k)\neq0.
\end{align*}
\end{assu}

\begin{assu}\label{assu_soliton}$\rm{(Solitons)}.$
Assume that $s_{11}(k)$ has a simple zero at each point in $\textbf{Z}$ and $s_{11}(k)\neq0$ for $k\in \overline{\mathbb{C}_+}\setminus\textbf{Z}$. 
 If $k_j \in \textbf{Z} \cap \ri\mathbb{R}_+$, then we also assume that $
      \ri C_{k_j}\notin (0,\infty).$
\end{assu}
\begin{theo}\label{theo-direct}
    Suppose $n_0(x),\phi_0(x)\in\mathcal{S}(\mathbb{R})$ such
that Assumptions \ref{assu_k=0} and \ref{assu_soliton}  hold. Then the associated scattering data
\begin{equation}\label{scattering data}
    \left\lbrace r_1(k),\, r_2(k),\,\textbf{Z},\{C_{k_j}\}_{k_j\in\textbf{Z}}\right\rbrace,
\end{equation}
defined by \eqref{r1_r2}, \eqref{ckj_non_Cc} and \eqref{ckj_non_Cir} are well-defined and satisfy the following properties:
\begin{enumerate}[label=(\roman*)]
    \item $r_1(k)$ and $r_2(k)$ admit extensions belonging to $C^\infty(\mathbb{R})$.
    \item The function $r_2(k)$ and its partial derivatives $\partial_k^jr_2(k)$ have continuous boundary values at $k=0$ for $j=0,1,2,\cdots$, and there exists an expansion:
    \begin{equation*}
        r_2(k)=r_2(0)+r_2'(0)k+\frac{1}{2}r_2''(0)k^2+\cdots,\quad k\to 0,\, k\in \mathbb{R},
    \end{equation*}
    which can be differentiated termwise any number of times. Furthermore, it can be shown that
    \begin{equation*}
        r_2(0)=-1.
    \end{equation*}
    \item $r_1(k)$ and $r_2(k)$ are rapidly decreasing as $|k|\to\infty$, i.e.,
    \begin{equation*}
        \max_{j=0,1,\cdots,N}\sup_{k\in\mathbb{R}}(1+|k|)^N|\partial_k^jr_l(k)|<\infty,
    \end{equation*}
    for each integer $N\geq 0$ and $l=1,2$.
    \item For all $k\in\mathbb{R}$, we have
    \begin{align*}
            &8kr_1(k)r_1^*(-k)-r_2(k)+r_2^*(-k)=0.
        \end{align*}
    \item $\textbf{Z}$ is a finite subset of $D_{\rm{reg}}\cup\ri\mathbb{R}_+$ and $\{C_{k_j}\}_{k_j\in\textbf{Z}}\subset\mathbb{C}$.
    \item For every $k_j\in \textbf{Z}\cap\ri\mathbb{R}_+$, the norming constant $C_{k_j}$ satisfies the condition 
    \begin{equation*}
        \ri C_{k_j}\leq 0.
    \end{equation*}

\end{enumerate}
\end{theo}

The proof of Theorem \ref{theo-direct}  is presented in Subsection \ref{subsec_proof21}.

\subsection{The inverse problem}
We now consider the inverse scattering problem, namely, the
reconstruction of the solution from the scattering data obtained in the
direct scattering transform. The inverse problem is formulated in terms
of a row-vector RH problem with jump contour
$\mathbb{R}$, whose jump matrix is determined by the reflection
coefficients $r_1(k)$ and $r_2(k)$. In the presence of discrete spectral data,
the solution of the RH problem is meromorphic and may have simple poles
at the points in the symmetry set
$$
\hat{\mathbf Z}:=\mathbf Z\cup\mathbf Z^*\cup\mathbf Z_m\cup
\mathbf Z_m^*,
$$
where
$\mathbf Z^*:=\{\overline{k_j}\,|\,k_j\in\mathbf Z\}$,
$\mathbf Z_m:=\{-k_j\,|\,k_j\in\mathbf Z\}$, and
$\mathbf Z_m^*:=\{-\overline{k_j}\,|\,k_j\in\mathbf Z\}$. Introduce
        \begin{equation}\label{sym_AB}
		\mathcal{A}(k)=\begin{pmatrix}
			1 & 0 & 0\\
			0 & 1/\left( 8k\right)  & 0\\
			0 & 0 & -1\\
		\end{pmatrix},\qquad
		\mathcal{B}=\begin{pmatrix}
			0 & 0 & 1\\
			0 & 1 & 0\\
			1 & 0 & 0\\
		\end{pmatrix}.
	\end{equation}

\begin{rhp}\label{rhp_Nj}\rm{(RH problem for $N_j$, $j=1,2$)}.
	Find a $1\times3$ matrix-valued function $N_j(x,t,k)$, $j=1,2,$ with the following properties:
	\begin{enumerate}
		\item The function $N_j(x,t,k)$ is analytic for $k\in\mathbb{C}\setminus\{\mathbb{R}\cup\hat{\textbf{Z}}\}$.
		\item As $k$ approaches $\mathbb{R}$ from the left and right, the boundary values $N_{j,+}(x,t,k)$ and $N_{j,-}(x,t,k)$ of $N_j(x,t,k)$ exist and satisfy the following jump condition:
		\begin{equation*}
			N_{j,+}(x,t,k)=N_{j,-}(x,t,k)v(x,t,k), \quad k\in \mathbb{R},
		\end{equation*}
		 where 
    \begin{equation}\label{jump_0}
				v(x,t,k)=
				\begin{pmatrix}
					1 & \! \! \! -r_1(k)\re^{t\theta_{12}} &\! \! \! r_2(k)\re^{t\theta_{13}}\\
					-8kr_1^*(k)\re^{-t\theta_{12}} & \! \! \!1+8k\left|r_1(k) \right|^2 & \! \! \!-8k\alpha(k)\re^{t\theta_{23}}\\
					-r_2^*(k)\re^{-t\theta_{13}} & \! \! \!\alpha^*(k)\re^{-t\theta_{23}} & \! \! \! 1-8k\left|r_1(-k) \right|^2- \left|r_2(k) \right|^2\\
				\end{pmatrix},
		\end{equation}
		and $\alpha(k)=r_1^*(-k)+ r_1^*(k)r_2(k)$, $\theta_{ij}=(l_i(k)-l_j(k))x/t+(z_i(k)-z_j(k))$, $1\leq i<j\leq3$.
        \item At each point in $\hat{\textbf{Z}}$, two elements of $N_j(x,t,k)$ are analytic with respect to $k$, while the remaining element has at most a simple pole. For every $p\in \hat{\textbf{Z}}$, the following residue conditions hold:
          \begin{equation*}
				\mathop{\text{Res}}\limits_{k=p }N_j(x,t,k)=\lim\limits_{k\to p } N_j(x,t,k)\re^{x\widehat{\mathcal{L}}(k)+t\widehat{\mathcal{Z}}(k)}V(p),
	\end{equation*}
    where for each $k_j\in\textbf{Z}\setminus\ri\mathbb{R}_+$,
    \begin{equation}\label{Res_2}
        \begin{aligned}
           V(k_j)=\begin{pmatrix}
				0 & C_{k_j}  & 0\\
				0 & 0 & 0\\
				0 & 0 & 0\\
			\end{pmatrix},\quad V(-k_j)=\begin{pmatrix}
				0 & 0 & 0\\
				0 & 0 & 0\\
				0 & -C_{k_j}  & 0\\
			\end{pmatrix},\quad V(\overline{k_j})=\begin{pmatrix}
				0 & 0 & 0\\
				-8\overline{k_j} \,\overline{C_{k_j} } & 0 & 0\\
				0 & 0 & 0\\
			\end{pmatrix}, \quad V(-\overline{k_j})=\begin{pmatrix}
				0 & 0 & 0\\
				0 & 0 & 8\overline{k_j}\,\overline{C_{k_j}}\\
				0 & 0 & 0\\
			\end{pmatrix},
        \end{aligned}
    \end{equation}
    and for each $k_j\in\textbf{Z}\cap \ri\mathbb{R_+}$,
     \begin{equation}\label{Res22}
            V(k_j)=\begin{pmatrix}
				0 & 0 & C_{k_j}\\
				0 & 0 & 0\\
				0 & 0 & 0\\
			\end{pmatrix},\quad V(-k_j)=\begin{pmatrix}
				0 & 0 & 0\\
				0 & 0 & 0\\
				-C_{k_j} & 0 & 0\\
			\end{pmatrix}.
    \end{equation}
    
		\item \label{rhp_Nj_kinfty} $N_1(x,t,k)=\begin{pmatrix}
		    1 & 0 & 1
		\end{pmatrix}+\mathcal{O}(k^{-1})$ and $N_2(x,t,k)=\begin{pmatrix}
		    0 & 1 & 0
		\end{pmatrix}+\mathcal{O}(k^{-1})$ as $k\to\infty$.
		\item  $N_j(x,t,k)=\mathcal{O}(1)$ as $k\to0$.

		\item  $N_j(x,t, k)$ satisfies the symmetry for $k \in \mathbb{C} \setminus \mathbb{R}$:
        \begin{align}\label{sym_Nj}
            &N_j(x,t, k)=N_j(x,t,-k) \mathcal{B}.
        \end{align}
      
	\end{enumerate}
	\end{rhp}

\subsubsection{Solvability of the Riemann-Hilbert problem and reconstruction formula}

    To guarantee the unique solvability of the RH problem \ref{rhp_Nj},
we impose the following positivity condition on the reflection
coefficients, which is essential for establishing the vanishing lemma.
    \begin{assu}\label{assu_LT}
     Assume that the reflection coefficients $r_1(k)$ and $r_2(k)$ as defined in equation \eqref{r1_r2} satisfy the relations:
          \begin{equation*}
                \begin{aligned}
                    &{1-8k|r_1(-k)|^2-|r_2(k)|^2}>0,
                \end{aligned}\qquad k\in \mathbb{R}\setminus\{0\}.
                \end{equation*} 
            \end{assu}

    Note that RH problem \ref{rhp_Nj} shares the same conditions for $j = 1, 2$, with the only exception being the condition \ref{rhp_Nj_kinfty}, where the behavior as $k \to \infty$ differs. The following lemma addresses the RH problem for the eigenfunction $N(x,t,k)$, in which $N(x,t,k)=\mathcal{O}(k^{-1},k^{-2},k^{-1}),
    \ k\to\infty $ holds, and all other conditions coincide with those of RH problem \ref{rhp_Nj}.

 \begin{lem}[Vanishing Lemma]
\label{lem_unique_solvability}
Let $(x,t)\in\mathbb{R}\times[0,T]$ and assume that the scattering
data satisfy Assumptions~\ref{assu_k=0}-\ref{assu_LT}. Then
RH problem~\ref{rhp_Nj} admits a unique solution.
Indeed, the associated homogeneous RH problem for $N(x,t,k)$ has only the trivial
solution:
\(
    N(x,t,k)\equiv0,
\)
for any solution satisfying
\(
    N(x,t,k)=\mathcal{O}(k^{-1},k^{-2},k^{-1}),
    \ k\to\infty .
\)
\end{lem}
\begin{proof}
The proof follows from the vanishing lemma established in
Section~\ref{subsec_vanishing}.
\end{proof}

The unique solvability of the RH problem is the key step in the inverse
scattering procedure. The solution of the RH problem \ref{rhp_Nj} can then be used
to recover the solution of the YO equation \eqref{YO} through the large-$k$
asymptotic behavior. We therefore obtain the following result.

\begin{theo}\label{theo_reconstruction_formula}
\rm{(Solution of the YO equation).}
Let
\(
\left\{
r_1(k),\,r_2(k),\,\mathbf{Z},
\{C_{k_j}\}_{k_j\in\mathbf{Z}}
\right\}
\)
be the scattering data satisfying the assumptions stated in
Theorem~\ref{theo-direct} and Assumptions \ref{assu_k=0}-\ref{assu_LT}. Then there exists a time
$T\in(0,\infty]$ such that the RH problems~\ref{rhp_Nj}, for
$j=1,2$, admit unique solutions
$N_1(x,t,k)$ and $N_2(x,t,k)$ for every
$(x,t)\in\mathbb{R}\times[0,T)$. Moreover, the solution of the initial-value problem for the YO equation~\eqref{YOE} is reconstructed from the
large-$k$ behavior of the RH solutions as
\begin{equation}\label{reconstruct}
\left\{
\begin{aligned}
n(x,t)
&=
2\ri\frac{\partial}{\partial x}
\lim_{k\to\infty}
k\left([N_1(x,t,k)]_1-1\right),\\
\phi(x,t)
&=
\ri\lim_{k\to\infty}
k[N_2(x,t,k)]_1 .
\end{aligned}
\right.
\end{equation}
\end{theo}
    \begin{proof}
    The proof of Theorem \ref{theo_reconstruction_formula} is given in Subsection \ref{subsec_reconst}.
    \end{proof}

  \begin{figure}[h]  
    \centering      
    \includegraphics[width=0.9\textwidth]{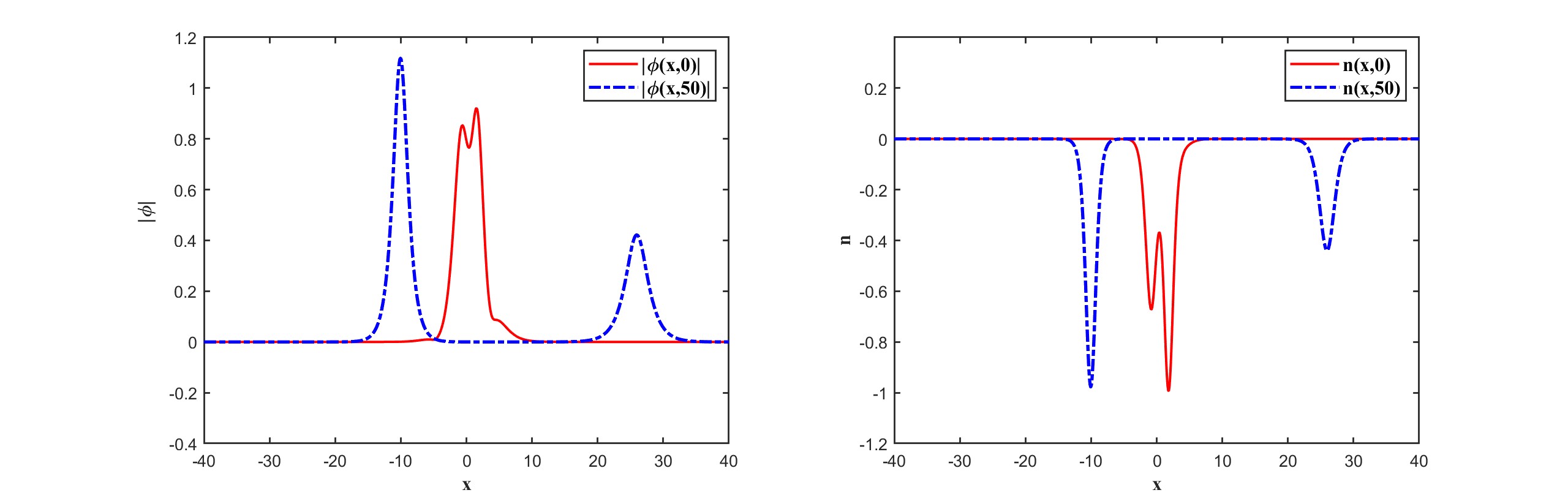}
    \caption{Two-soliton of the YO equation \eqref{YO} at $t = 0$ (red solid line) and $t = 50$ (blue dashed line), with $k_1=\frac{5}{8}+\frac{\ri}{2},\,k_2=\frac{1}{5}+\frac{\ri}{3}$, and $C_{k_1}=C_{k_2}=1$.}
   	\label{figrhp-two_soliton1}
           \end{figure}

      \subsubsection{Pure soliton solution}
      
      In the reflectionless case $r_1(k)=r_2(k)=0$, the RH problem \ref{rhp_Nj} reduces to a finite-dimensional meromorphic problem determined by the discrete spectrum and norming constants. The resulting $N$-soliton solutions are stated below.
      
      \begin{theo}\label{theo-soliton}
\rm{(Soliton solutions).}
Assume that the reflectionless condition $r_1(k)=r_2(k)=0$ for
$k\in\mathbb{R}$ holds, together with
Assumptions~\ref{assu_k=0} and \ref{assu_soliton}. Let the discrete
scattering data
$\{\mathbf{Z},\{C_{k_j}\}_{k_j\in\mathbf{Z}}\}$
be defined by \eqref{ckj_compact}, \eqref{ckj_non_Cc}, and
\eqref{ckj_non_Cir}. Then, under the conditions of
Theorem~\ref{theo_reconstruction_formula}, the reconstruction formula
\eqref{reconstruct} yields the following pure soliton solutions of the
YO equation~\eqref{YO}:
\begin{description}

\item[(i) One-soliton solution for
$k_j\in\mathbf{Z}\setminus\ri\mathbb{R}_+$:]
Suppose that $s_{11}(k)$ has exactly one simple zero
$k_1=\xi_1+\ri\eta_1\in\mathbf{Z}\setminus\ri\mathbb{R}_+$.
Then the corresponding one-soliton solution is given by
\begin{equation}\label{soliton1}
\left\{
\begin{aligned}
n(x,t)
&=-4\eta_1^2\operatorname{sech}^2
\left(2\eta_1\left(x-(1-2\xi_1)t\right)+\rho\right),\\
\phi(x,t)
&=-2\eta_1\sqrt{2\xi_1}\operatorname{sech}
\left(2\eta_1\left(x-(1-2\xi_1)t\right)+\rho\right)
\re^{2\ri(\eta_1^2-\xi_1^2)t
-2\ri\xi_1(x-t)+\ri\theta},
\end{aligned}
\right.
\end{equation}
where
\(
\rho=\log\left(
\frac{\sqrt{2\xi_1}\eta_1}
{2|C_{k_1}k_1|}
\right),
\
\re^{\ri\theta}
=\ri\frac{\overline{C_{k_1}}\,\overline{k_1}}
{|C_{k_1}k_1|}.
\)

\item[(ii) $N$-soliton solution for
$k_j\in\mathbf{Z}\setminus\ri\mathbb{R}_+$:]
Suppose that $s_{11}(k)$ has $N$ simple zeros
$k_j\in\mathbf{Z}\setminus\ri\mathbb{R}_+$,
$j=1,2,\ldots,N$. Then the corresponding $N$-soliton solution is
\begin{equation}\label{N-soliton-exact}
\left\{
\begin{aligned}
n(x,t)
&=-2\ri\frac{\partial}{\partial x}
\left[
\sum_{j=1}^N
8\overline{k_j}\,\overline{C_{k_j}}
\re^{-2\ri\overline{k_j}
((\overline{k_j}-1)t+x)}
\left(
\frac{\det\mathbb{A}^{(1)}_j+\det\mathbb{A}^{(3)}_j}
{\det\mathbb{A}}
\right)
\right],\\
\phi(x,t)
&=-\sum_{j=1}^N
8\ri\overline{k_j}\,\overline{C_{k_j}}
\re^{-2\ri\overline{k_j}
((\overline{k_j}-1)t+x)}
\frac{\det\mathbb{A}^{(2)}_j}{\det\mathbb{A}}.
\end{aligned}
\right.
\end{equation}
The matrices $\mathbb{A}$ and $\mathbb{A}^{(n)}_j$,
$n=1,2,3$, $j=1,2,\ldots,N$, are defined in
Subsection~\ref{subsec_N-soliton1}.

\item[(iii) One-soliton solution for
$k_j\in\mathbf{Z}\cap\ri\mathbb{R}_+$:]
Suppose that $s_{11}(k)$ has exactly one simple zero
$k_1=\ri\eta_1\in\mathbf{Z}\cap\ri\mathbb{R}_+$ and that
$\ri C_{k_1}\leq0$. Then
\begin{equation*}
\left\{
\begin{aligned}
n(x,t)
&=-4\eta_1^2\operatorname{sech}^2
\left(2\eta_1(x-t)+\rho_1\right),\\
\phi(x,t)&=0,
\end{aligned}
\right.
\end{equation*}
where
\(
\rho_1=\frac12\log\left(
-\frac{2\eta_1}{\ri C_{\ri\eta_1}}
\right).
\)

\item[(iv) $N$-soliton solution for
$k_j\in\mathbf{Z}\cap\ri\mathbb{R}_+$:]
Suppose that $s_{11}(k)$ has $N$ simple zeros
$k_j=\ri\eta_j\in\mathbf{Z}\cap\ri\mathbb{R}_+$,
$j=1,2,\ldots,N$. Then
\begin{equation*}\label{N-soliton_2}
\left\{
\begin{aligned}
n(x,t)
&=-2\ri\frac{\partial}{\partial x}
\left[
\sum_{j=1}^N
C_{\ri\eta_j}\re^{-4\eta_j(x-t)}
\left(
\frac{\det\mathbb{B}^{(1)}_j+\det\mathbb{B}^{(3)}_j}
{\det\mathbb{B}}
\right)
\right],\\
\phi(x,t)&=0.
\end{aligned}
\right.
\end{equation*}
The matrices $\mathbb{B}$ and $\mathbb{B}^{(n)}_j$,
$n=1,3$, $j=1,2,\ldots,N$, are defined in
Subsection~\ref{subsec_N-soliton2}.

\end{description}
\end{theo}
\begin{proof}
    The proof of this theorem  is given in Section~ \ref{sec:soliton} and Appendix~\ref{sec_Nsoliton}.
    \end{proof}
For $N=2$ in the $N$-soliton solution (\ref{N-soliton-exact}), Figure~\ref{figrhp-two_soliton1} shows the exact two-soliton solution of the YO equation~\eqref{YO} at different times. As $t$ increases, the two solitons separate, with amplitudes approaching those determined by the corresponding discrete spectral points in \eqref{soliton1}, while their directions of propagation are governed by the sign of $1-2\rre k_j$, $j=1,2$.

      \subsection{Long-time asymptotics} 
     
      Finally, we investigate the long-time asymptotic behavior of the
solution to the YO equation \eqref{YO} in the solitonless case. To exclude
contributions from the discrete spectrum and ensure that the asymptotic
behavior is governed solely by the continuous scattering data, we impose
the following assumption.

\begin{assu}\label{assu_solitonless}
\rm{(Absence of solitons).}
Assume that $s_{11}(k)$ has no zeros in $\mathbb{C}_+$.
\end{assu}

    \begin{figure}[htbp]
		\centering
	\begin{tikzpicture}[scale=1.1]

        \definecolor{mycolor1}{HTML}{D0E8FF} 
   		\definecolor{mycolor2}{HTML}{D6EAF8} 
   		\definecolor{mycolor3}{HTML}{ABC9EB} 
   		\definecolor{mycolor4}{HTML}{98A3CA}

        \fill[pink!85!red!50] (4.5,3.5) -- (0,0) -- (4.5,0.5) -- cycle;
		\fill[black!20!blue!20] (3.5,4.5) -- (0,0) -- (-4.5,0.5) -- (-4.5,4.5) -- cycle;
        \fill[green!20!white!70] (-4.5,0.5) -- (0,0) -- (-4.5,0)-- cycle;
        \fill[yellow!35!white!70] (4.5,0.5) -- (0,0) -- (4.5,0) -- cycle;

        \fill[mycolor1] (4.5,3.5) -- (0,0) -- (3.5,4.5) -- (4.5,4.5) -- cycle;

        \node[black] at (4,4.1) {$x=t$};
        \draw[thick,mycolor1,dashed] (0,0) -- (4.5,4.5) node[midway, sloped,black] {\qquad Transition region};
        \draw[thick,gray,dashed,gray] (0,0) -- (4.5,4.5);

	    \node at (0,-0.3) {0};
        \node at (3,1.3) {Region II};
        \node at (-2,2) {Region III};
        \node[above] at (-3.7,-0.1) {\small Region IV};
        \node[above] at (3.7,-0.1) {\small Region I};

        \draw[thick,gray] (3.5,4.5) -- (0.1,0.1) -- (4.5,3.5);
        \draw[thick,gray] (-4.5,0.5) -- (0,0) -- (4.5,0.5);

        \draw[very thick,-latex] (-4.7,0) -- (4.7,0) node[right] {$x$};
		\draw[very thick,-latex] (0,0) -- (0,4.7) node[above] {$t$};
        
	\end{tikzpicture}
	\caption{Division of asymptotic regions in the upper $(x,t)$-half plane.}
	\label{fig_asy}
	\end{figure}

         \begin{figure}[h]  
    \centering      
    \includegraphics[width=1\textwidth]{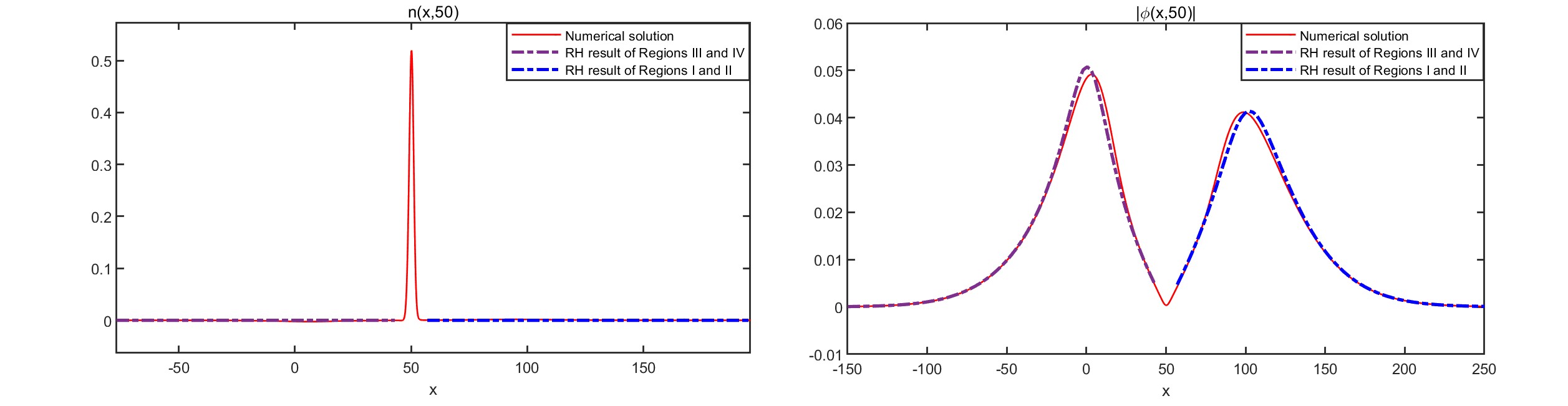}
    \caption{Comparisons of theoretical results given by Theorem \ref{them_region} and the full numerical simulations of the YO equation \eqref{YO} under the initial-value condition \eqref{initial} for $t=50$.}
    \label{figrhp-direct}
           \end{figure}

	\begin{theo}\label{them_region}
Assume that $n_0(x),\phi_0(x)\in\mathcal{S}(\mathbb{R})$ satisfy the
hypotheses of Theorem~\ref{theo-direct} and
Assumptions~\ref{assu_LT} and \ref{assu_solitonless}.
Then the solution $(n(x,t),\phi(x,t))$ of the initial-value problem
\eqref{YOE} with initial data $(n_0(x),\phi_0(x))$ admits the following
asymptotic expansions in Regions~I-IV shown in
Figure~\ref{fig_asy}:
	\begin{description}
	 \item[Region {\rm{I}}:]For $\tau := t/x\in\mathcal{I}_1$ and $k_0=\frac{1-\tau}{2\tau}$, the following asymptotic formula holds uniformly as $x \to +\infty$: 
	 \begin{equation}\label{region_1}
	 	\begin{aligned}
	 		n(x,t)&=\mathcal{O}\left(\frac{1}{x^{N}}+\frac{C_N(\tau)\ln x}{x}\right),\\
	 		\phi(x,t)&=\frac{\sqrt{2\pi}\left(2\sqrt{x\tau}\right)^{-2\ri\nu}{\rm{exp}}\left(\ri\nu \ln(2k_0)+\frac{3\pi \ri}{4}+2\ri tk_0^2-\frac{\pi\nu}{2}+s_1\right)}{2\sqrt{x\tau}r_1(-k_0)\Gamma(-\ri\nu)}+\mathcal{O}\left(\frac{1}{x^{N}}+\frac{C_N(\tau)\ln x}{x}\right).
	 	\end{aligned}
	 \end{equation}

	\item[Region {\rm{II}}:] For $\zeta := x/t\in\mathcal{I}_2$ and $k_0=\frac{x-t}{2t}$, the following asymptotic formula holds uniformly as $t \to \infty$:
            \begin{equation}\label{region_2}
                \begin{aligned}
        n(x,t)
        &=\mathcal{O}\left(\frac{\ln t}{t}\right),\\
    \phi(x,t)&=\frac{\sqrt{2\pi}\left(2\sqrt{t}\right)^{-2\ri\nu}{\rm{exp}}\left(\ri\nu \ln(2k_0)+\frac{3\pi \ri}{4}+2\ri tk_0^2-\frac{\pi\nu}{2}+s_1\right)}{2\sqrt{t}r_1(-k_0)\Gamma(-\ri\nu)}+\mathcal{O}\left(\frac{\ln t}{t}\right).
    \end{aligned}
            \end{equation}

    \item[Region {\rm{III}}:] For $\zeta\in\mathcal{I}_3$ and $k_0=\frac{x-t}{2t}$, the following asymptotic formula holds uniformly as $t \to \infty$:
    \begin{equation}\label{region_3}
          \begin{aligned}
        n(x,t)
        &=\mathcal{O}\left(\frac{\ln t}{t}\right),\\
    \phi(x,t)
&=\frac{\sqrt{2\pi}\left(2\sqrt{t}\right)^{-2\ri\hat\nu}{\rm{exp}}\left(\ri\hat\nu \ln(-2k_0)+\frac{3\pi \ri}{4}+2\ri tk_0^2-\frac{\pi\hat\nu}{2}-\pi\nu_2+ s_2\right)}{2\sqrt{t} {\alpha}^*(k_0)\Gamma(-\ri\hat\nu)}+\mathcal{O}\left(\frac{\ln t}{t}\right).
    \end{aligned}
    \end{equation}

    \item[Region {\rm{IV}}:]For $\tau \in\mathcal{I}_4$ and $k_0=\frac{1-\tau}{2\tau}$, the following asymptotic formula holds uniformly as $x \to-\infty$: 
   \begin{equation}\label{region_4}
   	\begin{aligned}
   		n(x,t)&=\mathcal{O}\left(\frac{1}{|x|^{N}}+\frac{C_N(\tau)\ln |x|}{|x|}\right),\\
   		\phi(x,t)
   		&=\frac{\sqrt{2\pi}\left(2\sqrt{x\tau}\right)^{-2\ri\hat\nu}{\rm{exp}}\left(\ri\hat\nu \ln(-2k_0)+\frac{3\pi \ri}{4}+2\ri x\tau k_0^2-\frac{\pi\hat\nu}{2}-\pi\nu_2+ s_2\right)}{2\sqrt{x\tau} {\alpha}^*(k_0)\Gamma(-\ri\hat\nu)}+\mathcal{O}\left(\frac{1}{|x|^{N}}+\frac{C_N(\tau)\ln |x|}{|x|}\right).
   \end{aligned}\end{equation}
    
		\end{description}
		Here, $\Gamma(\cdot)$ denotes the Gamma function, $C_N(\tau)$ is a smooth function of $\tau$ that vanishes to all orders as $\tau\to\infty$ for each $N$, $\alpha^*(k_0)=r_1(-k_0)+ r_1(k_0)r_2^*(k_0)$, and
        \begin{equation}\label{hatnu}
        \nu=-\frac{1}{2\pi}\ln(1-8k_0\left|r_1(-k_0) \right|^2 ),\quad \nu_2=-\frac{1}{2\pi}\ln\left({1+8k_0|r_1(k_0)|^2-|r_2(-k_0)|^2}\right),\quad \hat\nu=-\frac{1}{2\pi}\ln\left(1-8k_0|\alpha(k_0)|^2\re^{2\pi\nu_2}\right),
        \end{equation}
       moreover, $ s_1=\frac{1}{2\pi\ri}\int_{k_0}^{\infty}\ln\left({(s+k_0)}{(s-k_0)^2}\right){\rm{d}}\ln(1-8s\left|r_1(-s)\right|^2 )$, and
         \begin{align}
    s_2=&\frac{1}{2\pi\ri}\int_{-k_0}^{\infty}\ln\left|{(s+k_0)^2}{(s-k_0)}\right|{\rm{d}}\ln(1-8s \left|r_1(-s)\right|^2 ) +\frac{1}{2\pi\ri}\int^{\infty}_{-k_0}\ln\left|\frac{s-k_0}{s+k_0}\right|{\rm{d}}\ln \left(1-8s|r_1(-s)|^2-|r_2(s)|^2\right)\label{s2}\\
   	&-\frac{1}{2\pi\ri}\int_{k_0}^{-k_0}\ln\left|(s+k_0)^2(s-k_0) \right| {\rm{d}}\ln \left(1+8s|r_1(s)|^2-|r_2(-s)|^2\right).\nonumber
		\end{align}   
	\end{theo}
\begin{proof}
    The asymptotic behaviors of the solution in the four regions are derived and proved in Section \ref{sec_LT}. 
\end{proof}

        To verify the accuracy of the solution in Theorem \ref{them_region}, we consider two initial conditions composed of Gaussian wave packets:
			\begin{equation}\label{initial}
					n(x,0)=0.5\,{\rm{e}}^{-{x^2}/{2}}, \qquad\qquad
					\phi(x,0)=-0.3\,{\rm{e}}^{-{x^2}/{2}}.
			\end{equation}
Figure~\ref{figrhp-direct} illustrates the comparison between the asymptotic solution in Theorem~\ref{them_region} and the numerical simulation results obtained by using the initial condition \eqref{initial} at time $t=50$. Figure~\ref{figrhp-direct} shows good agreement between the leading-order asymptotic formulas and the direct numerical simulation.

        \begin{figure}[h]  
    \centering      
    \includegraphics[width=\textwidth]{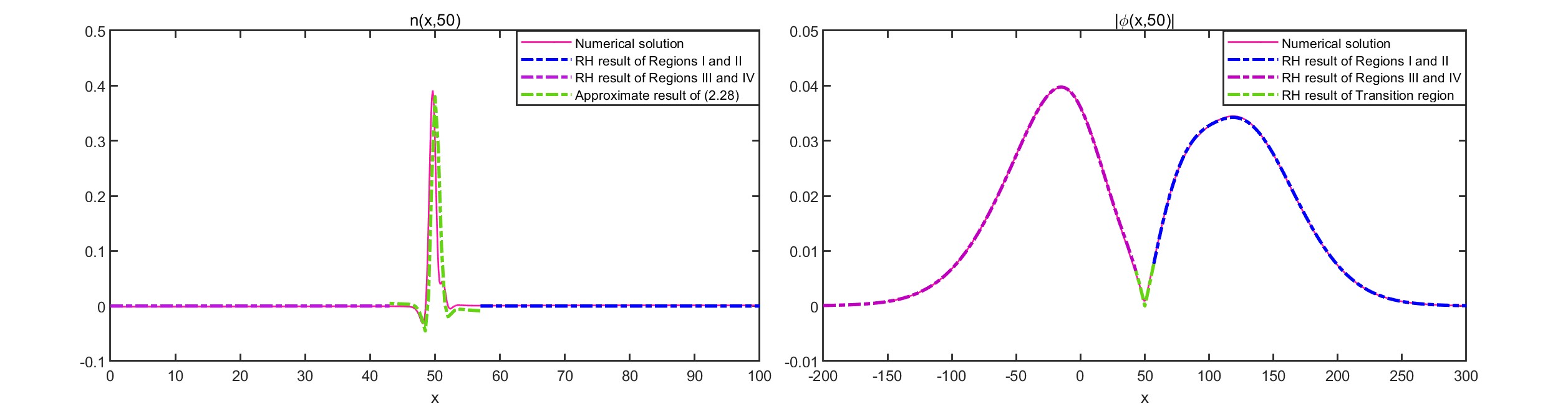}
    \caption{Comparisons of approximate results of \eqref{n_traveling} and the full numerical simulations of the YO equation \eqref{YO} under the initial-value condition \eqref{initial_middle}.}
    \label{fig_middle}
           \end{figure}

\begin{remark}\rm{(Transition region near $x=t$).}
The line $x=t$ separates Regions~II and III in Figure~\ref{fig_asy}. As
$x\to t$, the stationary point $k_0$ tends to zero. By
Theorem~\ref{theo-direct}, Assumption~\ref{assu_LT}, and
Proposition~\ref{prop_hatnu}, the asymptotic formulas
\eqref{region_2} and \eqref{region_3} for the short-wave component
$\phi(x,t)$ match in the limit $k_0\to0$. Hence, these two formulas
provide a consistent description of $\phi(x,t)$ across the transition
region.

For the long-wave component, the second equation of the YO equation \eqref{YO} yields
\begin{equation}\label{n_traveling}
n(x,t)=n_0(x-t)
-\int_0^t\partial_x
\left(|\phi(x-t+y,y)|^2\right)\rd y.
\end{equation}
Thus, the behavior of $n(x,t)$ near $x=t$ depends not only on the
large-time asymptotics of $\phi(x,t)$, but also on its evolution over the
entire interval $0\leq y\leq t$. Since the present RH analysis controls
only the large-time part, it does not provide a complete asymptotic
description of $n(x,t)$ in this region.

A numerical approximation may nevertheless be obtained by combining
the initial-data approximation for small times with the asymptotic
formulas of Theorem~\ref{them_region} for large times. For the initial
data
\begin{equation}\label{initial_middle}
n_0(x)=0.3\cos(x)\re^{-x^2/2},
\qquad
\phi_0(x)=0.6\sin(x)\re^{-x^2/2},
\end{equation}
the resulting approximation captures the principal peak near $x=t$,
although discrepancies remain because the intermediate-time behavior
of $\phi$ is not included; see Figure~\ref{fig_middle}.
\end{remark}

	\section{The direct scattering problem}\label{sec:direct}
   Given initial data $n_0(x)$ and $\phi_0(x)$ for the initial-value problem \eqref{YOE} satisfying Assumptions~\ref{assu_k=0} and~\ref{assu_soliton}, we first construct the associated Jost solutions and scattering matrix, and hence define the corresponding scattering data. Based on the Fredholm integral equations for the Jost solutions, we formulate a $3\times3$ matrix RH problem for $M(x,t,k)$. The RH problems for $N_1(x,t,k)$ and $N_2(x,t,k)$ are then obtained through the relations \eqref{def:N1N2}. As these constructions are standard in the framework of inverse scattering theory, we only state the main results here and refer the reader to Appendix~\ref{appendix:direct} and the relevant literature for the proofs of the auxiliary lemmas and technical properties.

\subsection{The Gauge Transformation}
	\ \ \ \
Before introducing the Jost solutions, we first apply the gauge transformation
\[
\Phi(x,t,k)=P(k)\Psi(x,t,k).
\]
Since $\det P(k)=4k$, the transformation is valid for $k\in\mathbb{C}\setminus\{0\}$. Under this transformation, the Lax pair \eqref{lax_phi} takes the form
	\begin{equation}\label{lax_psi}
		\begin{aligned}
			\Psi_x(x,t,k)=\check{L}(x,t,k)\Psi(x,t,k),\\
			\Psi_t(x,t,k)=\check{Z}(x,t,k)\Psi(x,t,k),\\
		\end{aligned}
	\end{equation}
	where the matrices $\check{L}$ and $\check{Z}$ are given by
		\begin{align*}
			\check{L}=P^{-1} \tilde{L} P&=\begin{pmatrix}
				3\ri k & 0 & 0\\
				0 & \ri k & 0\\
				0 & 0 & -\ri k\\
			\end{pmatrix}+\frac{1}{2k}
			\begin{pmatrix}
				-\ri n & -\overline{\phi}/2 & -\ri n\\
				4k\phi & 0 & 4k\phi\\
				\ri n & \overline{\phi}/2 & \ri n\\
			\end{pmatrix}
			:=\mathcal{L}+L_1,\\
			\check{Z}=P^{-1} \tilde{Z} P&=\begin{pmatrix}
				-2\ri k+2 \ri k^2/3 & 0 & 0\\
				0 & -4\ri k^2/3 & 0\\
				0 & 0 & 2\ri k+2 \ri k^2/3\\
			\end{pmatrix}\\
			&\quad+\frac{1}{4k}\begin{pmatrix}
				\ri\left(\phi\overline{\phi}+2 n\right) & \left(\ri \overline{\phi}_x+2(1-k)\overline{\phi} \right)/2 & \ri(\phi \overline{\phi}+2n)\\
				4k\left( \ri \phi_x+2(k-1)\phi\right)  & 0 & 4k\left( \ri \phi_x-2(k+1)\phi\right) \\
				-\ri\left(\phi\overline{\phi}+2 n\right) & -\left(\ri \overline{\phi}_x+2(1+k)\overline{\phi} \right)/2 & -\ri(\phi \overline{\phi}+2n)\\
			\end{pmatrix}:=\mathcal{Z}+Z_1.
		\end{align*}
	The matrix-valued functions $L_1$ and $Z_1$ satisfy the following symmetry relations:
	\begin{equation}\label{symmetry}
		\begin{aligned}
			&L_1(x,t,k)=-\mathcal{A}(k)^{-1} L_1^\dagger(x,t,\bar k)\mathcal{A}(k),\quad &&L_1(x,t,k)=\mathcal{B} L_1(x,t,-k)\mathcal{B},\\
			&Z_1(x,t,k)=-\mathcal{A}(k)^{-1} Z_1^\dagger(x,t,\bar k)\mathcal{A}(k),\quad &&Z_1(x,t,k)=\mathcal{B} Z_1(x,t,-k)\mathcal{B},
		\end{aligned}
	\end{equation}
    where $\mathcal{A}(k)$ and $\mathcal{B}$ are given by equation \eqref{sym_AB}.	
Since $\phi_0,n_0\in\mathcal{S}(\mathbb{R})$, we have
\(
\lim\limits_{|x|\to\infty}L_1(x,t,k)=0,\
\lim\limits_{|x|\to\infty}Z_1(x,t,k)=0.
\)
Setting
\[
\Psi(x,t,k)=X(x,t,k)e^{\mathcal{L}x+\mathcal{Z}t},
\]
the Lax pair \eqref{lax_psi} becomes
\begin{equation}\label{lax_X}
\begin{cases}
\begin{aligned}
&X_x-[\mathcal{L},X]=L_1X,\\
&X_t-[\mathcal{Z},X]=Z_1X.
\end{aligned}
\end{cases}
\end{equation}

	\subsection{The eigenfunctions $X(x, k)$ and $Y(x, k)$}
    In what follows, we focus on the $x$-part of \eqref{lax_X} by fixing $t=0$, namely
\begin{equation}\label{lax_xpart}
X_x-[\mathcal{L},X]=L_1X.
\end{equation}
The eigenfunctions $X(x,k)$ and $Y(x,k)$ associated with \eqref{lax_xpart} are defined through the following Volterra integral equations:
\begin{equation}\label{def:XY}
\begin{aligned}
&X(x,k)=I-\int_{x}^{\infty}\re^{(x-y)\widehat{\mathcal{L}}(k)}(L_1X)(y,k),\rd y,\
&Y(x,k)=I+\int_{-\infty}^{x}\re^{(x-y)\widehat{\mathcal{L}}(k)}(L_1Y)(y,k),\rd y.
\end{aligned}
\end{equation}
The following proposition collects several basic properties of the eigenfunctions $X(x,k)$ and $Y(x,k)$.
	\begin{prop}
		$($Basic properties of $X$ and $Y$$)$.
Suppose $\phi_0(x),\, n_0(x)\in\mathcal{S}(\mathbb{R})$. The functions 
$X(x,k)$ and $Y(x,k)$ uniquely determined by integral equations \eqref{def:XY} solve equation \eqref{lax_xpart} and possess the following properties:
		\begin{enumerate}
        \item The functions $[X(x,k)]_1$ and $[Y(x,k)]_3$ are defined for $x\in\mathbb{R}$ and $k\in \overline{\mathbb{C}_+}\setminus\{0\}$. $[X(x,\cdot)]_1$ and $[Y(x,\cdot)]_3$ are analytic functions for $k\in\mathbb{C}_+$ and continuous up to $k\in \overline{\mathbb{C}_+}\setminus\{0\}$, where $[A]_j$ denotes the $j$th column of the matrix $A$.
        \item The functions $[X(x,k)]_3$ and $[Y(x,k)]_1$ are defined for $x\in\mathbb{R}$ and $k\in \overline{\mathbb{C}_-}\setminus\{0\}$. $[X(x,\cdot)]_3$ and $[Y(x,\cdot)]_1$ are analytic functions for $k\in\mathbb{C}_-$ and continuous up to $k\in \overline{\mathbb{C}_-}\setminus\{0\}$.
        \item The functions $[X(x,k)]_2$ and $[Y(x,k)]_2$ are defined for $x\in\mathbb{R}$ and $k\in {\mathbb{R}}\setminus\{0\}$. $[X(x,\cdot)]_2$ and $[Y(x,\cdot)]_2$ are continuous for $k\in {\mathbb{R}}\setminus\{0\}$, but admit no analytic continuation away from the real axis.
       
            \item For each $x\in\mathbb{R}$ and $j=1,2,\cdots$, $\partial_k^j X (x,\cdot)$ is well defined for $k\in (\overline{\mathbb{C}_+},\mathbb{R},\overline{\mathbb{C}_-})\setminus\{0\}$, and $\partial_k^j Y(x,\cdot)$ is well defined for $k\in (\overline{\mathbb{C}_-},\mathbb{R},\overline{\mathbb{C}_+})\setminus\{0\}$.
            \item For each $n \geq 1$ and $\epsilon > 0$, there are bounded smooth positive functions $f_+(x)$ and $f_-(x)$ of $x \in \mathbb{R}$ with rapid decay as $x \to +\infty$ and $x \to -\infty$, respectively, such that the following estimates hold for $x \in \mathbb{R}$ and $j = 0, 1, \ldots, n$:
            \begin{equation*}
                \begin{aligned}
                    \left| \frac{\partial^j}{\partial k^j} \bigl( X(x, k) - I \bigr) \right| \leq f_+(x), \quad k\in (\overline{\mathbb{C}_+},\mathbb{R},\overline{\mathbb{C}_-}), \ |k| > \epsilon,\quad
                    \left| \frac{\partial^j}{\partial k^j} \bigl( Y(x, k) - I \bigr) \right| \leq f_-(x), \quad k\in (\overline{\mathbb{C}_-},\mathbb{R},\overline{\mathbb{C}_+}), \ |k| > \epsilon.
                \end{aligned}
            \end{equation*}

			\item  For $k\in\mathbb{R}\setminus\{0\}$, $\det X=\det Y=1$, $X(x, k)$ and $Y(x, k)$ satisfy the following symmetries for $x\in\mathbb{R}$ and $k\in (\overline{\mathbb{C}_\pm},\mathbb{R},\overline{\mathbb{C}_\mp})\setminus\{0\}$:
			\begin{equation}\label{sym_X_A}
				X^\dagger\left(  x,\bar k \right)=\mathcal{A}(k)X^{-1}(x,k)\mathcal{A}^{-1}(k),\qquad  Y^\dagger\left(  x,\bar k \right)=\mathcal{A}(k)Y^{-1}(x,k)\mathcal{A}^{-1}(k),
			\end{equation}
			\begin{equation}\label{sym_X_B}
				X(x,k)=\mathcal{B} X(x,-k)\mathcal{B},\qquad Y(x,k)=\mathcal{B} Y(x,-k)\mathcal{B}.
			\end{equation}
			\item If $n_0(x)$, $\phi_0(x)$ have compact support, then for each $x\in\mathbb{R}$, $X(x,k)$ and $Y(x,k)$  are defined and analytic in $k\in\mathbb{C}\setminus\{0\}$.
		\end{enumerate}
	\end{prop}
	
	\begin{proof}
    The proof is standard and closely follows the arguments in \cite{Lenells-Indiana}. We therefore restrict our attention to the symmetry relations in \eqref{sym_X_A}, as the other statements can be proved similarly. Replacing $k$ by $\bar{k}$ in \eqref{lax_xpart} and taking the conjugate transpose yields
		\begin{equation*}
			X_x^\dagger(x,\bar{k})-\left[\mathcal{L}(k),X^\dagger(x,\bar{k})\right]=X^\dagger(x,\bar{k}) L_1^\dagger\left(  x,\bar k \right) .
		\end{equation*}
		By the conjugation symmetry of $L_1$ in \eqref{symmetry}, it follows that
		\begin{equation*}
			X_x^\dagger(x,\bar{k})-\left[\mathcal{L}(k),X^\dagger(x,\bar{k})\right]=-X^\dagger(x,\bar{k})\mathcal{A}(k) L_1(x,k)\mathcal{A}^{-1}(k).
		\end{equation*}
		Multiplying the above equation by $\mathcal{A}^{-1}$ on the left and by $\mathcal{A}$ on the right yields
		\begin{equation}\label{AfXHA_lax}
			\mathcal{A}^{-1}(k)X_x^\dagger(x,\bar{k})\mathcal{A}(k)-\left[\mathcal{L}(k),\mathcal{A}^{-1}(k)X^\dagger(x,\bar{k})\mathcal{A}(k)\right]=-\mathcal{A}^{-1}(k)X^\dagger(x,\bar{k})\mathcal{A}(k) L_1(x,k).
		\end{equation}
		Noting that
\(
(X^{-1})_x=-X^{-1}X_xX^{-1}
=[\mathcal{L},X^{-1}]-X^{-1}L_1,
\)
we obtain
		\begin{equation}\label{Xf_lax}
			X_x^{-1}(x,k)-\left[\mathcal{L}(k),X^{-1}(x,k)\right]=-X^{-1}(x,k)L_1(x,k).
		\end{equation}
		Since the solutions of \eqref{AfXHA_lax} and \eqref{Xf_lax} satisfy the same normalization condition,
\[
\mathcal{A}^{-1}(k)X^\dagger(x,\bar{k})\mathcal{A}(k); \,\,X^{-1}(x,k)\to I,
\qquad x\to\infty,
\]
the uniqueness of the corresponding Volterra integral equation implies the symmetry relation \eqref{sym_X_A} for $X(x,k)$. The other symmetry relations stated in \eqref{sym_X_A} and \eqref{sym_X_B} follow from analogous considerations.

	\end{proof}

	\subsubsection{Asymptotics of $X(x,k)$ and $Y(x,k)$ as $k\to\infty$}\label{sub_Xk0}
	Introduce the following formal power series solutions
	\begin{equation}\label{X_pm_formal}
    \begin{aligned}
		&X_{formal}(x,k)=X_{0}(x)+\frac{X_{1}(x)}{k}+\frac{X_{2}(x)}{k^2}+\cdots,\quad &&k\to\infty,\\
        &Y_{formal}(x,k)=Y_{0}(x)+\frac{Y_{1}(x)}{k}+\frac{Y_{2}(x)}{k^2}+\cdots,\quad &&k\to\infty,\\
    \end{aligned}
	\end{equation}
	where the coefficients $\left\lbrace X_{ j}\right\rbrace_1^\infty $ and $\left\lbrace Y_{ j}\right\rbrace_1^\infty $ satisfy $\lim\limits_{x\to\infty}X_{ j}(x)=\lim\limits_{x\to-\infty}Y_{ j}(x)=0$ for $j\geq 1$.
    Notice that the matrix $L_1$ in Lax pair \eqref{lax_xpart} can be decomposed into $L_1(x,t,k)=L_{11}(x,t)+L_{12}(x,t)/k$, and we denote $\mathcal{L}(k)=kL_0$, where
    \begin{equation*}
        L_{11}(x,t)=\begin{pmatrix}
				0 & 0 & 0\\
				2\phi & 0 & 2\phi\\
				0 & 0 & 0\\
			\end{pmatrix},\quad L_{12}(x,t)=\begin{pmatrix}
				-\ri n/2 & -\overline{\phi}/4 & -\ri n/2\\
				0 & 0 & 0\\
				\ri n/2 & \overline{\phi}/4 & \ri n/2\\
			\end{pmatrix},\quad L_0=\begin{pmatrix}
			    3\ri & 0 & 0\\
                0 & \ri & 0\\
                0 & 0 & -\ri\\
			\end{pmatrix}.
    \end{equation*}
	By substituting the above expansion into \eqref{lax_xpart} and comparing the coefficients of the same powers of the parameter $k$, we can obtain the following relations:
    \begin{equation*}
        \left\lbrace\begin{aligned}
            &\partial_x X_{ j}^{(o)}=[L_0,X_{ j+1}]+(L_{11}X_{ j})^{(o)}+(L_{12}X_{ j-1})^{(o)},\\
            &\partial_x X_{ j}^{(d)}=(L_{11}X_{ j})^{(d)}+(L_{12}X_{ j-1})^{(d)},
        \end{aligned}\right.\quad j=1,2,\cdots,
    \end{equation*}
    where $A^{(d)}$ and $A^{(o)}$ denote the diagonal and off-diagonal parts of a $3\times3$ matrix $A$. After calculation, the first few coefficients are as
    follows:
	\begin{align*}
		&X_{0}(x)=Y_0(x)=I,\\
		& X_{1}(x)=\begin{pmatrix}
			-\frac{\ri}{2}\int_{\infty}^{x}n_0(x')\rd x' & 0 & 0\\
		-\ri \phi_0 & 0 & \ri \phi_0\\
			0 & 0 & \frac{\ri}{2}\int_{\infty}^{x}n_0(x')\rd x'\\ 
		\end{pmatrix},\\
		& X_{2}(x)=\begin{pmatrix}
			0 & -\frac{\ri\overline{\phi}_0}{8} & \frac{n_0}{8}\\
			0 & 0 & 0\\
			\frac{n_0}{8} & -\frac{\ri\overline{\phi}_0}{8} & 0\\ 
		\end{pmatrix}+\left( -\frac{\phi_0}{2}\int_{\infty}^{x}n_0(x')\rd x'+\frac{\phi_{0x}}{2}\right) \begin{pmatrix}
			0 & 0 & 0\\
			1 & 0 & 1\\
			0 & 0 & 0\\
		\end{pmatrix}-\frac{\ri}{2}\int_{\infty}^{x}\phi_0(x')\overline{\phi}_0(x')\rd x'\begin{pmatrix}
		0 & 0 & 0\\
		0 & 1 & 0\\
		0 & 0 & 0\\
		\end{pmatrix}\\
		&\qquad \quad +\left( 
	\frac{\ri}{4}\int_{\infty}^{x}\phi_0(x')\overline{\phi}_0(x')\rd x'	-\frac{1}{4}\int_{\infty}^{x}n_0(x'')\int_{\infty}^{x''}n_0(x')\rd x'\rd x''\right) 
		\begin{pmatrix}
			1 & 0 & 0\\
			0 & 0 & 0\\
			0 & 0 & 1\\
		\end{pmatrix},\\
		& Y_{1}(x)=\begin{pmatrix}
			-\frac{\ri}{2}\int_{-\infty}^{x}n_0(x')\rd x' & 0 & 0\\
			-\ri \phi_0 & 0 & \ri \phi_0\\
			0 & 0 & \frac{\ri}{2}\int_{-\infty}^{x}n_0(x')\rd x'\\ 
		\end{pmatrix},\\
		& Y_{2}(x)=\begin{pmatrix}
			0 & -\frac{\ri\overline{\phi}_0}{8} & \frac{n_0}{8}\\
			0 & 0 & 0\\
			\frac{n_0}{8} & -\frac{\ri\overline{\phi}_0}{8} & 0\\ 
		\end{pmatrix}+\left( -\frac{\phi_0}{2}\int_{-\infty}^{x}n_0(x')\rd x'+\frac{\phi_{0x}}{2}\right) \begin{pmatrix}
			0 & 0 & 0\\
			1 & 0 & 1\\
			0 & 0 & 0\\
		\end{pmatrix}-\frac{\ri}{2}\int_{-\infty}^{x}\phi_0(x')\overline{\phi}_0(x')\rd x'\begin{pmatrix}
		0 & 0 & 0\\
		0 & 1 & 0\\
		0 & 0 & 0\\
		\end{pmatrix}\\
		&\qquad \quad +\left( 
		\frac{\ri}{4}\int_{-\infty}^{x}\phi_0(x')\bar{\phi}_0(x')\rd x'	-\frac{1}{4}\int_{-\infty}^{x}n_0(x'')\int_{-\infty}^{x''}n_0(x')\rd x'\rd x''\right) 
		\begin{pmatrix}
			1 & 0 & 0\\
			0 & 0 & 0\\
			0 & 0 & 1\\
		\end{pmatrix}.
	\end{align*}
	Although all remaining coefficients can, in principle, be determined explicitly, we do not present them here. We next investigate the behavior of $X(x,k)$ and $Y(x,k)$ in the limit of large $k$.
	
	\begin{prop}\label{prop_Xpm_k_infty}
		$($Asymptotics of $X$ and $Y$ as $k\to \infty$$)$.
		Suppose $\phi_0(x),\, n_0(x)\in\mathcal{S}(\mathbb{R})$. As $k\to\infty$, $X$ and $Y$ coincide to all orders with $X_{formal}$ and $Y_{formal}$, respectively. To be more precise, for any integer $p\geq 0$, there exist two well-defined functions
		\begin{equation}\label{Xp_infty}
        \begin{aligned}
			&X^{[p]}(x,k)=I+\frac{X_{1}(x)}{k}+\frac{X_{2}(x)}{k^2}+\cdots+\frac{X_{p}(x)}{k^p},\\
            &Y^{[p]}(x,k)=I+\frac{Y_{1}(x)}{k}+\frac{Y_{2}(x)}{k^2}+\cdots+\frac{Y_{p}(x)}{k^p},\\
        \end{aligned}
		\end{equation}
		 such that for any integer $j\geq 0$,
		 \begin{equation*}
         \begin{aligned}
           & \left|\frac{\partial^j}{\partial k^j}\left(X-X^{[p]}\right)  \right| \leq \frac{f_+(x)}{\left|k \right|^{p+1} },\quad x\in\mathbb{R},\,k\in (\overline{\mathbb{C}_+},\mathbb{R},\overline{\mathbb{C}_-}), \\
           & \left|\frac{\partial^j}{\partial k^j}\left(Y-Y^{[p]}\right)  \right| \leq \frac{f_-(x)}{\left|k \right|^{p+1} },\quad x\in\mathbb{R},\,k\in (\overline{\mathbb{C}_-},\mathbb{R},\overline{\mathbb{C}_+}), \\
         \end{aligned}
		 \qquad\left|k \right| \geq 2,
		 \end{equation*}
	where $f_\pm(x)$ are bounded smooth positive functions of $x\in\mathbb{R}$ with rapid decay as $x\to \pm \infty$, respectively.
	\end{prop}
	\begin{proof}
		The proof follows by considering the equation satisfied by the quotients $\left(X^{[p]}\right)^{-1}X$ and $\left(Y^{[p]}\right)^{-1}Y$. We refer the reader to the \cite{Lenells-Indiana} for similar arguments.
	\end{proof}
	
	\subsubsection{Asymptotics of $X(x,k)$ and $Y(x,k)$ as $k\to 0$}
    If $n_0(x)$, $\phi_0(x)$ are compactly supported, then $X$ and $Y$ can have at most simple poles at $k = 0$, with the Laurent expansion coefficients are given by \eqref{X-kto0}. When $n_0(x), \phi_0(x) \in \mathcal{S}(\mathbb{R})$ are not compactly supported, the situation is more delicate: the columns of $X$ and $Y$ are generally not defined near $k = 0$, and a more refined description is therefore required.
	\begin{prop}\label{prop_Xkto0}$($Asymptotics of $X$ and $Y$ as $k\to 0$$)$.
			Suppose $\phi_0,\, n_0\in\mathcal{S}(\mathbb{R})$ and $p\geq0$ is an integer. There exist $ 3 \times 3 $ matrix-valued functions $X^{(l)}$ and $Y^{(l)}$, $l=-1,0,\cdots,p$, that satisfy the following properties:
			\begin{enumerate}
				\item For $ x \in \mathbb{R} $ and $j\geq0$ is any integer, the functions $X$ and $Y$ satisfy 
				\begin{equation*}
                \begin{aligned}
                   &\left|\frac{\partial^j}{\partial k^j}\left(X(x,k)-I-\sum_{l=-1}^{p}X^{(l)}(x)k^l\right)  \right| \leq {f_+(x)}{\left|k \right|^{p+1-j} }, \quad &&  \left|k \right| \leq \frac{1}{2}, \quad k\in (\overline{\mathbb{C}_+},\mathbb{R},\overline{\mathbb{C}_-}),\\
                   &\left|\frac{\partial^j}{\partial k^j}\left(Y(x,k)-I-\sum_{l=-1}^{p}Y^{(l)}(x)k^l\right)  \right| \leq {f_-(x)}{\left|k \right|^{p+1-j} }, \quad &&  \left|k \right| \leq \frac{1}{2}, \quad k\in (\overline{\mathbb{C}_-},\mathbb{R},\overline{\mathbb{C}_+}),\\
                \end{aligned}
				\end{equation*}
				where $f_\pm(x)$ are bounded smooth positive functions of $x\in\mathbb{R}$ with rapid decay as $x\to \pm \infty$, respectively.
				\item  For any $l\geq-1$, $X^{(l)}(x)$ and $Y^{(l)}(x)$ are smooth functions of $ x \in \mathbb{R} $ and decay rapidly as $x\to \pm\infty$, respectively.
				\item The specific forms of the coefficients of the first two terms in the expansion are 
				\begin{align}
					X^{(-1)}(x)=&\alpha_{11}(x)\begin{pmatrix}
						-1 & 0 & -1\\
						0 & 0 & 0\\
						1 & 0 & 1\\
					\end{pmatrix}+\alpha_{12}(x)\begin{pmatrix}
						0 & -1 & 0\\
						0 & 0 & 0\\
						0 & 1 & 0\\
					\end{pmatrix},\nonumber\\
					X^{(0)}(x)=&-I+\alpha_{13}(x)\begin{pmatrix}
						-1 & 0 & 1\\
						0 & 0 & 0\\
						1 & 0 & -1\\
					\end{pmatrix}+\alpha_{14}(x)\begin{pmatrix}
						0 & 0 & 1\\
						0 & 0 & 0\\
						1 & 0 & 0\\
					\end{pmatrix}+\alpha_{15}(x)\begin{pmatrix}
						0 & 0 & 0\\
						1 & 0 & 1\\
						0 & 0 & 0\\
					\end{pmatrix}\nonumber\\
					&+\frac{8\alpha_{12}(x)\alpha_{14}(x)-\alpha^*_{15}(x)}{16\alpha_{11}(x)}\begin{pmatrix}
						0 & 1 & 0\\
						0 & 0 & 0\\
						0 & 1 & 0\\
					\end{pmatrix}+\frac{\alpha_{12}(x)\alpha_{15}(x)-\alpha_{11}^*(x)}{\alpha_{11}(x)}\begin{pmatrix}
						0 & 0 & 0\\
						0 & 1 & 0\\
						0 & 0 & 0\\
					\end{pmatrix},
                \label{X-kto0}\\
                    Y^{(-1)}(x)=&\alpha_{21}(x)\begin{pmatrix}
						-1 & 0 & -1\\
						0 & 0 & 0\\
						1 & 0 & 1\\
					\end{pmatrix}+\alpha_{22}(x)\begin{pmatrix}
						0 & -1 & 0\\
						0 & 0 & 0\\
						0 & 1 & 0\\
					\end{pmatrix},\nonumber\\
					Y^{(0)}(x)=&-I+\alpha_{23}(x)\begin{pmatrix}
						-1 & 0 & 1\\
						0 & 0 & 0\\
						1 & 0 & -1\\
					\end{pmatrix}+\alpha_{24}(x)\begin{pmatrix}
						0 & 0 & 1\\
						0 & 0 & 0\\
						1 & 0 & 0\\
					\end{pmatrix}+\alpha_{25}(x)\begin{pmatrix}
						0 & 0 & 0\\
						1 & 0 & 1\\
						0 & 0 & 0\\
					\end{pmatrix}\nonumber\\
					&+\frac{8\alpha_{22}(x)\alpha_{24}(x)-\alpha^*_{25}(x)}{16\alpha_{21}(x)}\begin{pmatrix}
						0 & 1 & 0\\
						0 & 0 & 0\\
						0 & 1 & 0\\
					\end{pmatrix}+\frac{\alpha_{22}(x)\alpha_{25}(x)-\alpha_{21}^*(x)}{\alpha_{21}(x)}\begin{pmatrix}
						0 & 0 & 0\\
						0 & 1 & 0\\
						0 & 0 & 0\\
					\end{pmatrix}.\nonumber
				\end{align}
				 More importantly, for bounded positive functions $f_1(x)$ and $f_2(x)$, that decay rapidly as $x$ approaches both $\infty$ and $-\infty$, the following equations hold for all $x\in\mathbb{R}$ and $i=1,2$:
        \begin{equation}\label{alpha_i}
			\begin{aligned}
				&\left|\alpha_{i1}(x) \right| +\left|\alpha_{i2}(x) \right| \leq f_i(x),\quad &&\left|\alpha_{i5}(x)\right|\leq f_i(x),\\
				&\left|\alpha_{i4}(x)\right|\leq (1+|x|)f_i(x),\quad &&\left|\alpha_{i3}(x)+1\right|\leq (1+|x|)f_i(x).
			\end{aligned}
                 \end{equation}
			\end{enumerate}
	\end{prop}
\begin{proof}
        The proof of this proposition is given in Appendix \ref{Appendix_prop_kto0}.
\end{proof}

	\subsection{The spectral function $s(k)$}
    The spectral matrix function $s(k)$ is defined by equation \eqref{sk}, and its properties are given by the following proposition.
	\begin{prop}\label{prop_s}
		$($Basic properties of $s( k )$$)$.
		Given that $\phi_0(x), n_0(x)\in \mathcal{S}(\mathbb{R})$, the matrix-valued function $s( k )$ exhibits the following  properties:
		\begin{enumerate}
			\item The entries of $s(k)$ are defined and continuous for
			\begin{equation}\label{sk_k_define}
				k\in \begin{pmatrix}
					\overline{\mathbb{C}_+} & \mathbb{R} & \mathbb{R}\\
					\mathbb{R} & \mathbb{R} & \mathbb{R}\\
					\mathbb{R} & \mathbb{R} & \overline{\mathbb{C}_-}\\
				\end{pmatrix}\setminus\{0\}.
			\end{equation}
			
			The function $s_{11}(k)$ is analytic in the upper half-plane $({\rm{Im}}\, k >0)$, and the function $s_{33}(k)$ is analytic in the lower half-plane $({\rm{Im}}\, k <0)$. The remaining elements of $s(k)$ are not analytic in any region.
			\item The elements of the derivative $\partial_k^j s(k)$, $j=1,2,\cdots,$ are well-defined and continuous for $k$ in \eqref{sk_k_define}.
            \item $s(k)$ approaches the identity matrix as $k\to\infty$.  More importantly, there exist diagonal matrices $\{s_j\}_{j=1}^\infty$, such that the following asymptotics holds for $j=1,2,\cdots,N$:
            \begin{equation*}
                \left|\partial_k^j\left(s(k)-I-\sum_{j=1}^N\frac{s_j}{k^j}\right)\right|=\mathcal{O}\left(k^{-N-1}\right),\quad k\to\infty,\,\, k\,\,\rm{in} \,\,\eqref{sk_k_define}.
            \end{equation*}
            Furthermore, the off-diagonal elements of $s(k)$ decay rapidly as $k\to\infty$.
			\item  The function $s( k )$ behaves as
			\begin{equation}\label{s_k0}
				s( k )=\frac{s^{(-1)}}{k}+s^{(0)}+s^{(1)} k +s^{(2)} k ^2+\cdots,\quad  k \rightarrow0,
			\end{equation}
			where
    \begin{equation}\label{s-1}
	\begin{aligned}
			s^{(-1)}=&\int_{\mathbb{R}}\frac{ \alpha_{14}(x) n_0(x)}{2}+\frac{\alpha_{15}(x)\overline{\phi}_0(x)}{4}\rd x\begin{pmatrix}
				1 & 0 & 1\\
				0 & 0 & 0\\
				-1 & 0 & -1\\
			\end{pmatrix}\\
			&-\int_{\mathbb{R}}\frac{\ri \left( 8\alpha_{12}(x)\alpha_{14}(x)-\alpha^*_{15}(x)\right)n_0(x) }{16\alpha_{11}(x)}+\frac{\left( \alpha_{12}(x)\alpha_{15}(x)-\alpha^*_{11}(x)\right) \overline{\phi}_0(x)}{4\alpha_{11}(x) }\rd x\begin{pmatrix}
			0 & 1 & 0\\
			0 & 0 & 0\\
			0 & -1 & 0\\
		\end{pmatrix},
		\end{aligned}
        \end{equation}
		and the expansion can be differentiated termwise any number of times.
			\item  For $k\in\mathbb{R}\setminus\{0\}$, $\det s(k)=1$, and the function $s( k )$ satisfies the symmetries
			\begin{equation}\label{sym_s}
				\begin{aligned}
				&s^{\dagger}\left(  \bar k \right) =\mathcal{A}(k)s^{-1} (k)\mathcal{A}^{-1}(k),\qquad\qquad
                s(-k)=\mathcal{B}\,s(k)\,\mathcal{B}.
 			\end{aligned}
 			\end{equation}
			\item If $\phi_0(x)$, $n_0(x)$ have compact support, then $s(k)$ is defined and analytic for $k\in\mathbb{C}\setminus\{0\}$, $\det s(k)=1$, and
            \begin{equation}\label{X_s}
		X(x, k )=Y(x, k )\re^{x\widehat{\mathcal{L}}( k )}s( k ),\quad  k \in \mathbb{C}\setminus \left\lbrace 0 \right\rbrace.
	\end{equation}
		\end{enumerate}
	\end{prop}

	\begin{proof}
	The details of the proof are given in Appendix~\ref{Appendix_prop_s}.
	\end{proof}

	\subsection{The eigenfunctions $X^A(x,k)$ and $Y^A(x,k)$}
   Notice that only the first column of $X(x,k)$ and the third column of $Y(x,k)$ are analytic in the upper half-plane. However, the $x$-part of the Lax pair \eqref{lax_xpart} is a third-order system and therefore admits three linearly independent eigenfunctions. To obtain a complete set of analytic eigenfunctions, we introduce the eigenfunctions $X^A(x,k)$ and $Y^A(x,k)$.

Let $X^A:=(X^{-1})^T$. More generally, for a $3\times3$ matrix $B$, the superscript $A$ denotes the adjugate matrix of $B$, defined by
   \begin{equation*}
       B^A:=\begin{pmatrix}
           m_{11}(B) & -m_{12}(B) & m_{13}(B)\\
           -m_{21}(B) & m_{22}(B) & -m_{23}(B)\\
           m_{31}(B) & -m_{32}(B) & m_{33}(B)\\
       \end{pmatrix}.
   \end{equation*}
    We can calculate $\left(X^A \right)_x=-X^A\left( X_x\right) ^TX^A $ from the relation $X X^{-1}=I$, and $X^A$ satisfies the following adjugate Lax pair
	\begin{equation}\label{lax_XA}
		\left(X^A \right)_x+\left[ \mathcal{L},X^A\right]=-L_1^TX^A.  
	\end{equation}
	   The above equation corresponds to the solutions $X^A(x,k)$ and $Y^A(x,k)$ of two Volterra integral equations, which are given by equation \eqref{XY_XAYA}.

	\begin{prop}$($Basic properties of $X^A$ and $Y^A$$)$.
		Suppose $\phi_0(x),\, n_0(x)\in\mathcal{S}(\mathbb{R})$, the equations \eqref{XY_XAYA} uniquely define the two solutions $X^A(x,k)$ and $Y^A(x,k)$ in equation \eqref{lax_XA}, which satisfy the following properties:
		\begin{enumerate}
         \item The functions $[X^A(x,k)]_1$ and $[Y^A(x,k)]_3$ are defined for $x\in\mathbb{R}$ and $k\in \overline{\mathbb{C}_-}\setminus\{0\}$. $[X^A(x,\cdot)]_{1}$ and $[Y^A(x,\cdot)]_{3}$ are analytic functions for $k\in\mathbb{C}_-$ and continuous up to $k\in \overline{\mathbb{C}_-}\setminus\{0\}$.
        \item The functions $[X^A(x,k)]_{3}$ and $[Y^A(x,k)]_{1}$ are defined for $x\in\mathbb{R}$ and $k\in \overline{\mathbb{C}_+}\setminus\{0\}$. $[X^A(x,\cdot)]_{3}$ and $[Y^A(x,\cdot)]_1$ are analytic functions for $k\in\mathbb{C}_+$ and continuous up to $k\in \overline{\mathbb{C}_+}\setminus\{0\}$.
        \item The functions $[X^A(x,k)]_{2}$ and $[Y^A(x,k)]_{2}$ are defined for $x\in\mathbb{R}$ and $k\in {\mathbb{R}}\setminus\{0\}$. $[X^A(x,\cdot)]_{2}$ and $[Y^A(x,\cdot)]_{2}$ are continuous for $k\in {\mathbb{R}}\setminus\{0\}$, but not analytic.
       
            \item For each $x\in\mathbb{R}$ and $j=1,2,\cdots$, $\partial_k^jX^A(x,\cdot)$ is well defined for $k\in (\overline{\mathbb{C}_-},\mathbb{R},\overline{\mathbb{C}_+})\setminus\{0\}$, and $\partial_k^jY^A(x,\cdot)$ is well defined for $k\in (\overline{\mathbb{C}_+},\mathbb{R},\overline{\mathbb{C}_-})\setminus\{0\}$.
            \item For each $n \geq 1$ and $\epsilon > 0$, there are bounded smooth positive functions $f_+(x)$ and $f_-(x)$ of $x \in \mathbb{R}$ with rapid decay as $x \to +\infty$ and $x \to -\infty$, respectively, such that the following estimates hold for $x \in \mathbb{R}$ and $j = 0, 1, \ldots, n$:
            \begin{equation*}
                \begin{aligned}
                    &\left| \frac{\partial^j}{\partial k^j} \bigl( X^A(x, k) - I \bigr) \right| \leq f_+(x), \quad k\in (\overline{\mathbb{C}_-},\mathbb{R},\overline{\mathbb{C}_+}), \ |k| > \epsilon,\\
                    &\left| \frac{\partial^j}{\partial k^j} \bigl( Y^A(x, k) - I \bigr) \right| \leq f_-(x), \quad k\in (\overline{\mathbb{C}_+},\mathbb{R},\overline{\mathbb{C}_-}), \ |k| > \epsilon.
                \end{aligned}
            \end{equation*}

			\item For $k\in\mathbb{R}\setminus\{0\}$, $\det X ^A=\det Y^A=1$. The functions $X^A$ and $Y^A$ satisfy the following symmetries for $x\in\mathbb{R}$ and $k\in (\overline{\mathbb{C}_\mp},\mathbb{R},\overline{\mathbb{C}_\pm})\setminus\{0\}$:
			\begin{align*}
				&\left( X^A\right)^\dagger\left(  x,\bar k \right)=\mathcal{A}^{-1}(k) X^T(x,k)\mathcal{A}(k),\quad && \left( Y^A\right)^\dagger\left(  x,\bar k \right)=\mathcal{A}^{-1}(k) Y^T(x,k)\mathcal{A}(k),\\
				&X^A(x,k)=\mathcal{B}X^A(x,-k)\mathcal{B},\quad && Y^A(x,k)=\mathcal{B}Y^A(x,-k)\mathcal{B}.
			\end{align*}
			\item If $n_0(x)$, $\phi_0(x)$ have compact support, then for each $x\in\mathbb{R}$, $X^A(x,k)$ and $Y^A(x,k)$  are defined and analytic in $k\in\mathbb{C}\setminus\{0\}$.
			
		\end{enumerate}
	\end{prop}

	Similar to the discussion in Subsection \ref{sub_Xk0}, the equation \eqref{lax_xpart} has the following formal power series solutions
	\begin{equation*}
    \begin{aligned}
		&X^A_{formal}(x,k)= X_{0}^A(x)+\frac{ X_{1}^A(x)}{k}+\frac{ X_{2}^A(x)}{k^2}+\cdots,\quad &&k\to\infty,\\
        &Y^A_{formal}(x,k)= Y_{0}^A(x)+\frac{ Y_{1}^A(x)}{k}+\frac{ Y_{2}^A(x)}{k^2}+\cdots,\quad &&k\to\infty,
    \end{aligned}
	\end{equation*}
	where the coefficients $\left\lbrace  X_{ j}^A\right\rbrace_{j=1}^\infty $ and $\left\lbrace  Y_{ j}^A\right\rbrace_{j=1}^\infty $ satisfy $\lim\limits_{x\to\pm\infty} X_{ j}^A(x)=\lim\limits_{x\to\pm\infty} Y_{ j}^A(x)=0$ for $j\geq 1$. More importantly, $X^A_{formal}(x,k)$ and $Y^A_{formal}(x,k)$ be the cofactor matrices of the functions $X_{formal}$ and $Y_{formal}$ in \eqref{X_pm_formal}.
	
	\begin{prop}
		$($Asymptotics of $X^A$ and $Y^A$ as $k\to \infty$$)$.
		Suppose $\phi_0(x),\, n_0(x)\in\mathcal{S}(\mathbb{R})$. As $k$ approaches infinity, $X^A$ and $Y^A$ coincide to all orders with $X^A_{formal}$and $Y^A_{formal}$, respectively. To be more precise, for any integer $p\geq 0$, there exist well-defined functions
		\begin{equation*}
        \begin{aligned}
			&\left( X^A\right) ^{[p]}(x,k)=I+\frac{ X_{1}^A(x)}{k}+\frac{ X_{2}^A(x)}{k^2}+\cdots+\frac{ X_{p}^A(x)}{k^p},\\
            &\left( Y^A\right) ^{[p]}(x,k)=I+\frac{ Y_{1}^A(x)}{k}+\frac{ Y_{2}^A(x)}{k^2}+\cdots+\frac{ Y_{p}^A(x)}{k^p},\\
        \end{aligned}
		\end{equation*}
		such that for any integer $j\geq 0$,
		\begin{equation*}
            \left|\frac{\partial^j}{\partial k^j}\left(X^A-\left( X^A\right) ^{[p]}\right)  \right| \leq \frac{f_+(x)}{\left|k \right|^{p+1} },\quad
            \left|\frac{\partial^j}{\partial k^j}\left(Y^A-\left( Y^A\right) ^{[p]}\right)  \right| \leq \frac{f_-(x)}{\left|k \right|^{p+1} },
			\quad x\in\mathbb{R},\quad \left|k \right| \geq 2,
		\end{equation*}
		where $f_\pm(x)$ are bounded smooth positive functions of $x\in\mathbb{R}$ with rapid decay as $x\to \pm \infty$, respectively.
	\end{prop}

	\begin{prop}\label{prop_XAkto0}$($Asymptotics of $X^A$ and $Y^A$ as $k\to 0$$)$.
	Suppose $\phi_0(x),\, n_0(x)\in\mathcal{S}(\mathbb{R})$ and $p\geq0$ is an integer. There exist $ 3 \times 3 $ matrix-valued functions $\left( X^A\right) ^{(l)}$ and $\left( Y^A\right) ^{(l)}$, $l=-1,0,\cdots,p$, that satisfy the following properties:
			\begin{enumerate}
				\item For $ x \in \mathbb{R} $, $\left|k \right| \leq \frac{1}{2}$, and $j\geq0$ is any integer, the functions $X^A$ and $Y^A$ satisfy 
				\begin{equation*}
                \begin{aligned}
                    &\left|\frac{\partial^j}{\partial k^j}\left(X^A(x,k)-I-\sum_{l=-1}^{p}\left( X^A\right) ^{(l)}(x)k^l\right)  \right| \leq {f_+(x)}{\left|k \right|^{p+1-j} },\\
                    &\left|\frac{\partial^j}{\partial k^j}\left(Y^A(x,k)-I-\sum_{l=-1}^{p}\left( Y^A\right) ^{(l)}(x)k^l\right)  \right| \leq {f_-(x)}{\left|k \right|^{p+1-j} },\\
                \end{aligned}
				\end{equation*}
				where $f_\pm(x)$ are bounded smooth positive functions of $x\in\mathbb{R}$ with rapid decay as $x\to \pm \infty$, respectively.
				\item  For any $l\geq-1$, $\left( X^A\right) ^{(l)}(x)$ and $\left( Y^A\right) ^{(l)}(x)$ are smooth functions of $ x \in \mathbb{R} $ and decay rapidly as  $x\to +\infty$ and $x\to -\infty$, respectively.
				\item The specific forms of the coefficients of the first two terms in the expansion are 
				\begin{align*}
				\left(X^A \right)^{(-1)}(x)=&\beta_{11}(x)\begin{pmatrix}
						-1 & 0 & 1\\
						0 & 0 & 0\\
						-1 & 0 & 1\\
					\end{pmatrix}+\beta_{12}(x)\begin{pmatrix}
						0 & 0 & 0\\
						-1 & 0 & 1\\
						0 & 0 & 0\\
					\end{pmatrix},\\
				\left(X^A \right)^{(0)}(x)=&-I+\beta_{13}(x)\begin{pmatrix}
						-1 & 0 & 1\\
						0 & 0 & 0\\
						1 & 0 & -1\\
					\end{pmatrix}+\beta_{14}(x)\begin{pmatrix}
						0 & 0 & 1\\
						0 & 0 & 0\\
						1 & 0 & 0\\
					\end{pmatrix}+\beta_{15}(x)\begin{pmatrix}
						0 & 1 & 0\\
						0 & 0 & 0\\
						0 & 1 & 0\\
					\end{pmatrix}\\
					&+\frac{8\beta_{12}(x)\beta_{14}(x)-\beta^*_{15}(x)}{16\beta_{11}(x)}\begin{pmatrix}
						0 & 0 & 0\\
						1 & 0 & 1\\
						0 & 0 & 0\\
					\end{pmatrix}+\frac{\beta_{12}(x)\beta_{15}(x)-\beta_{11}^*(x)}{\beta_{11}(x)}\begin{pmatrix}
						0 & 0 & 0\\
						0 & 1 & 0\\
						0 & 0 & 0\\
					\end{pmatrix},\\
                    \left(Y^A \right)^{(-1)}(x)=&\beta_{21}(x)\begin{pmatrix}
						-1 & 0 & 1\\
						0 & 0 & 0\\
						-1 & 0 & 1\\
					\end{pmatrix}+\beta_{22}(x)\begin{pmatrix}
						0 & 0 & 0\\
						-1 & 0 & 1\\
						0 & 0 & 0\\
					\end{pmatrix},\\
				\left(Y^A \right)^{(0)}(x)=&-I+\beta_{23}(x)\begin{pmatrix}
						-1 & 0 & 1\\
						0 & 0 & 0\\
						1 & 0 & -1\\
					\end{pmatrix}+\beta_{24}(x)\begin{pmatrix}
						0 & 0 & 1\\
						0 & 0 & 0\\
						1 & 0 & 0\\
					\end{pmatrix}+\beta_{25}(x)\begin{pmatrix}
						0 & 1 & 0\\
						0 & 0 & 0\\
						0 & 1 & 0\\
					\end{pmatrix}\\
					&+\frac{8\beta_{22}(x)\beta_{24}(x)-\beta^*_{25}(x)}{16\beta_{21}(x)}\begin{pmatrix}
						0 & 0 & 0\\
						1 & 0 & 1\\
						0 & 0 & 0\\
					\end{pmatrix}+\frac{\beta_{22}(x)\beta_{25}(x)-\beta_{21}^*(x)}{\beta_{21}(x)}\begin{pmatrix}
						0 & 0 & 0\\
						0 & 1 & 0\\
						0 & 0 & 0\\
					\end{pmatrix}.
				\end{align*}
			 More importantly, for bounded positive functions $f_i(x)$, $i=1,2$, that decay rapidly as $x$ approaches both $\infty$ and $-\infty$, the following equations hold for all $x\in\mathbb{R}$ and $i=1,2$:
				\begin{equation*}
			\begin{aligned}
				&\left|\beta_{i1}(x) \right| +\left|\beta_{i2}(x) \right| \leq f_i(x),\quad |\beta_{i3}(x)+1|\leq (1+|x|)f_i(x),\quad |\beta_{i4}(x)|\leq (1+|x|)f_i(x),\quad |\beta_{i5}(x)|\leq f_i(x).
			\end{aligned}
                 \end{equation*}
				
			\end{enumerate}
		\end{prop}
        \begin{proof}
            The proof is given in Appendix \ref{Appendix_prop_kto02}.
        \end{proof}

	\subsection{The spectral function $s^A(k)$}
   Before presenting the properties of $s^A(k)$, we clarify that, in the general case, $s^A(k)$ is defined by the integral equation \eqref{sk}. Thus, although we keep the notation $s^A(k)$, it should not be interpreted as the adjugate matrix of $s(k)$ unless $n_0(x)$ and $\phi_0(x)$ are compactly supported.

	\begin{prop}\label{prop_sA}
		$($Basic properties of $s^A( k )$$)
		$.
		Given that $\phi_0(x), n_0(x)\in \mathcal{S}(\mathbb{R})$, the matrix-valued function $s^A( k )$ exhibits the following  properties:
		\begin{enumerate}
			\item The entries of $s^A(k)$ are defined and continuous for
			\begin{equation}\label{sk_define}
				k\in \begin{pmatrix}
					\overline{\mathbb{C}_-} & \mathbb{R} & \mathbb{R}\\
					\mathbb{R} & \mathbb{R} & \mathbb{R}\\
					\mathbb{R} & \mathbb{R} & \overline{\mathbb{C}_+}\\
				\end{pmatrix}\setminus\{0\}.
			\end{equation}
			The function $s^A_{11}(k)$ is analytic on the upper half complex plane $({\rm{Im}}\, k <0)$, and the function $s^A_{33}(k)$ is analytic on the lower half complex plane $({\rm{Im}}\, k >0)$. The remaining elements of $s^A(k)$ are not analytic in any region.
			\item The elements of the derivative $\partial_k^j s^A(k)$, $j=1,2,\cdots,$ are well-defined and continuous for $k$ in \eqref{sk_define}.

            \item $s^A(k)$ approaches the identity matrix as $k\to\infty$,  More importantly, there exist diagonal matrices $\{s^A_j\}_{j=1}^\infty$, such that the following equation holds for $j=1,2,\cdots,N$:
            \begin{equation*}
                \left|\partial_k^j\left(s^A(k)-I-\sum_{j=1}^N\frac{s^A_j}{k^j}\right)\right|=\mathcal{O}\left(k^{-N-1}\right),\quad k\to\infty,\,\, k\,\,\rm{in} \,\,\eqref{sk_define}.
            \end{equation*}
            Furthermore, the off-diagonal elements of $s^A(k)$ decay rapidly as $k\to\infty$.
			\item  The function $s^A( k )$ behaves as
			\begin{equation*}
				s^A( k )=\frac{\left( s^A\right) ^{(-1)}}{k}+\left( s^A\right)^{(0)}+\left( s^A\right)^{(1)} k +\left( s^A\right)^{(2)} k ^2+\cdots,\quad  k \rightarrow0,
			\end{equation*}
            where
			\begin{align}\label{sA-1}
			\left( s^A\right) ^{(-1)}=&\int_{\mathbb{R}}\frac{\ri n_0(x)\left(\beta_{13}(x)+\beta_{14}(x)\right)}{2}-2\beta_{12}(x)\phi_0(x)\rd x\begin{pmatrix}
				1 & 0 & -1\\
				0 & 0 & 0\\
				1 & 0 & -1\\
			\end{pmatrix}+\int_{\mathbb{R}} \frac{\left(\beta_{13}(x)+\beta_{14}(x)\right)\overline{\phi}_0(x)}{4}\rd x\begin{pmatrix}
				0 & 0 & 0\\
				1 & 0 & -1\\
				0 & 0 & 0\\
			\end{pmatrix},
		\end{align}
			and the expansion can be differentiated termwise any number of times.
			\item  For $k\in\mathbb{R}\setminus\{0\}$, $\det s(k)=1$, and the function $s^A( k )$ satisfies the symmetries
			\begin{equation}\label{sym_sA}
				\begin{aligned}
				\left(s^A\right)^\dagger\left(  \bar k \right)=\mathcal{A}^{-1}(k)s^{T} (k)\mathcal{A}(k),\qquad
					 s^A(-k)=\mathcal{B}s^A(k)\mathcal{B}.
				\end{aligned}
			\end{equation}

			\item If $\phi_0(x)$, $n_0(x)$ have compact support, then $s^A(k)$ is defined and analytic for $k\in\mathbb{C}\setminus\{0\}$, $\det s^A(k)=1$, and 
        \begin{equation*}
		X^A(x, k )=Y^A(x, k )\re^{-x\widehat{\mathcal{L}}( k )}s^A( k ),\qquad  k \in \mathbb{C}\setminus \left\lbrace 0 \right\rbrace.
	\end{equation*}
		\end{enumerate}
	\end{prop}
	\begin{proof}
	    Most of the proof follows the same lines as that of Proposition~\ref{prop_s}, and we therefore omit the details.
	\end{proof}

	  In Proposition \ref{prop_sA}, we provide the symmetry relation \eqref{sym_sA} satisfied by the scattering matrix $s(k)$, which consequently implies 
    \begin{equation*}
    		\left(s^A\right)^\dagger\left(  -\bar k \right)=\mathcal{B}\mathcal{A}^{-1}(k)s^{T} (k)\mathcal{A}(k)\mathcal{B}.
    \end{equation*}
     This can be expressed in the following matrix form:
     \begin{equation}\label{sym_k-bark}
     	\begin{pmatrix}
     		\overline{m_{11}(s( -\bar k)) } & -\overline{m_{12}(s( -\bar k))} & \overline{m_{13}(s( -\bar k))} \\
     		-\overline{m_{21}(s( -\bar k))} & \overline{m_{22}(s( -\bar k)) } & -\overline{m_{23}(s( -\bar k))} \\
     		\overline{m_{31}(s( -\bar k)) } & -\overline{m_{32}(s( -\bar k))} & \overline{m_{33}(s( -\bar k)) } \\
     	\end{pmatrix}=\begin{pmatrix}
     	s_{33}(k) & -8ks_{32}(k) & -s_{31}(k)\\
     	-\frac{1}{8k}s_{23}(k) & s_{22}(k) & \frac{1}{8k}s_{21}(k)\\
     	-s_{13}(k) & 8ks_{12}(k) & s_{11}(k)\\
     	\end{pmatrix}.
     \end{equation}
     Thus, when a simple zero $k_1\in\textbf{Z}\cap \ri\mathbb{R}_+$ of $s_{11}(k)$, in other words, $k_1=-\overline{k_1}$, the following lemma holds.
     
     \begin{lem}\label{lem_Ck1_im}
     	Suppose $\phi_0(x),\, n_0(x)\in\mathcal{S}(\mathbb{R})$ are compactly supported, and $s_{11}(k)$ has a simple zero at certain $k_1\in\textbf{Z}\cap \ri\mathbb{R}_+$. Then $s_{12}(k_1)=s_{21}(k_1)=0$, and furthermore $\ri {s_{13}(k_1)}/{s'_{11}(k_1)}\in \mathbb{R}$.
     \end{lem}
     
     \begin{proof}
     	Since $n_0(x)$ and $\phi_0(x)$ have compact support, all elements of $s(k)$ are well-defined and analytic in $k\in \mathbb{C}\setminus\{0\}$. Consequently, the symmetry with respect to $k\in \mathbb{C}\setminus\{0\}$ holds everywhere.
     	
     	When $k_1\in\textbf{Z}\cap \ri\mathbb{R}_+$, the symmetry relation given by equation \eqref{sym_k-bark} holds for $k=k_1$, and we can derive from $s_{11}(k_1)=0$ that:
     	\begin{equation*}
     		m_{33}(s(k_1))=s_{11}(k_1)s_{22}(k_1)-s_{12}(k_1)s_{21}(k_1)=-s_{12}(k_1)s_{21}(k_1)=0.
     	\end{equation*}
     	Thus, the following two cases arise:
     	\begin{itemize}
     		\item If $s_{12}(k_1)=0$ holds, then applying the symmetry relation \eqref{sym_k-bark} again yields 
     		\begin{equation*}
     			m_{32}(s(k_1))=s_{12}(k_1)s_{23}(k_1)-s_{13}(k_1)s_{21}(k_1)=-s_{13}(k_1)s_{21}(k_1)=0.
     		\end{equation*}
     		However, when $s_{13}(k_1)$, the first row of the matrix $s(k)$ vanishes entirely at $k_1$, which contradicts $\det s(k) = 1$. Therefore, only the case $s_{21}(k_1)=0$ is possible.
     		\item If $s_{21}(k_1)=0$ holds, then 
     		\begin{equation*}
     			m_{23}(s(k_1))=s_{11}(k_1)s_{32}(k_1)-s_{12}(k_1)s_{31}(k_1)=-s_{12}(k_1)s_{31}(k_1)=0.
     		\end{equation*}
     		For reasons similar to those above, only $s_{12}(k_1)=0$ holds. 
     	\end{itemize}
     	Hence, both conditions $s_{12}(k_1)=0$ and $s_{21}(k_1)=0$ hold simultaneously under the assumptions of this lemma.
     	
     	By applying the symmetry relation \eqref{sym_k-bark} twice and utilizing the previously established condition $s_{12}(k_1)=s_{21}(k_1)=0$, one can derive the real variable relationship.\qedhere
     \end{proof}
	
	\subsection{Construct the function $M(x,k)$}\label{sec_M}
    In this subsection, we construct the sectionally analytic function $M(x,k)$ to formulate the RH problem corresponding to the YO equation \eqref{YO}. We first consider the construction at time $t=0$, and denote $M(x,0,k)$ by $M(x,k)$. For $k \in \mathbb{C}_\pm$, we write $M(x,k) = M_\pm(x,k)$, where $M_+(x,k)$ and $M_-(x,k)$ are matrix-valued solutions of equation \eqref{lax_X}, defined by the following Fredholm integral equations:
	\begin{equation}\label{Mn_fredholm}
			\left( M_\pm\right)_{ij}(x,k) =\delta_{ij}+\int_{\gamma _{ij}^\pm}\left( \re^{(x-y)\widehat{\mathcal{L}}(k)}\left(L_1M_\pm \right) \right)_{ij}\rd y,
		\qquad i,j=1,2,3, 
	\end{equation}
	where
	\begin{equation*}\delta_{ij}=\left\lbrace\begin{aligned}
			&1,\quad i=j,\\
			&0,\quad i\neq j,\\
		\end{aligned} \right. \qquad
		\gamma _{ij}^+=\left\lbrace
		\begin{aligned}
			&\left(-\infty,x \right),\quad i< j, \\
			&\left(+\infty,x \right),\quad i\geq j,  
		\end{aligned} \right. \qquad
		\gamma _{ij}^-=\left\lbrace
		\begin{aligned}
			&\left(-\infty,x \right),\quad i> j, \\
			&\left(+\infty,x \right),\quad i\leq j. 
		\end{aligned} \right. \qquad
	\end{equation*}
	The definition of $\gamma _{ij}^\pm$ is such that the exponents in the off-diagonal elements of the matrix $\re^{(x-y)\widehat{\mathcal{L}}(k)}\left(L_1M\right) $ in the integral equation \eqref{Mn_fredholm} are bounded for $k\in \mathbb{C}_\pm$ and $y\in \gamma _{ij}^\pm$. The definition \eqref{Mn_fredholm} of $M_\pm$ can be extended by continuity to the boundary of $\mathbb{C}_\pm$. The following proposition will show that all the elements of $M_\pm$ are well-defined for $k\in \mathbb{C}_\pm\setminus \mathcal{Q}$, where 
	  \begin{equation}\label{Q_setminus}
	  	\mathcal{Q}=\mathcal{Z}\cup \{0\},
	  \end{equation}
	and $\mathcal{Z}$ denotes the set of zeros of the Fredholm determinants associated with \eqref{Mn_fredholm}.
	
	\begin{prop}\label{prop_Mn_basic}
		$($Basic properties of $M_\pm$$)$.
	Suppose $\phi_0(x),\, n_0(x)\in\mathcal{S}(\mathbb{R})$, the equation \eqref{Mn_fredholm} uniquely define the two solutions $M_\pm(x,k)$ in equation \eqref{lax_X}, which satisfy the following properties for $ x \in \mathbb{R}$:
	\begin{enumerate}
        \item The function $M_\pm(x,k)$ is defined for $x\in \mathbb{R}$ and $k\in \overline{\mathbb{C}_\pm}\setminus \mathcal{Q}$. For each fixed $k\in \overline{\mathbb{C}_\pm}\setminus \mathcal{Q}$, $M_\pm(\cdot, k)$ is smooth and satisfies equation \eqref{lax_X}. For each fixed $x\in\mathbb{R}$, the matrix-valued function $M_\pm(x, \cdot)$ is continuous on $\overline{\mathbb{C}_\pm}\setminus \mathcal{Q}$ and analytic on $\mathbb{C}_\pm \setminus \mathcal{Q}$.
		
		\item For each $\epsilon>0$, there exists $C:=C(\epsilon)$ such that
		\begin{equation*}
			\left|M_\pm(x,k) \right| \leq C,\qquad  k\in \overline{\mathbb{C}_\pm},\,\, {\rm{dist}}(k,\mathcal{Q})\geq \epsilon. 
		\end{equation*}
		\item For each $j\in \mathbb{N}_+$, the partial derivative $\frac{\partial^j M_\pm}{\partial k^j}(x, \cdot)$ has a continuous extension to $ \overline{\mathbb{C}_\pm} \setminus \mathcal{Q}$.
		\item For $k \in \overline{\mathbb{C}_\pm} \setminus \mathcal{Q}$ and $x\in\mathbb{R}$, $\det M_\pm(x, k) = 1$.
		\item The sectionally analytic function $ M(x, k) = M_\pm(x, k) $ for $k \in \mathbb{C}_\pm$ satisfies the symmetries
        \begin{equation}\label{M-sym}
            \begin{aligned}
			&M^{-1}(x, k) = \mathcal{A}^{-1}(k) M^\dagger(x, \bar k) \mathcal{A}(k) ,\\
            &M(x, k)=\mathcal{B} M(x,-k) \mathcal{B},
		\end{aligned}  \qquad k \in \mathbb{C} \setminus \mathcal{Q}.
        \end{equation}
		\end{enumerate}
	\end{prop}
	
	\begin{lem}\label{M_kinfty}
		$($Asymptotics of $M$ as $k\to\infty$$)$.
		Suppose $\phi_0(x),\, n_0(x)\in\mathcal{S}(\mathbb{R})$ and $\phi_0(x),\, n_0(x)\not\equiv0$. For any integer $p>0$, the function $X^{[p]}$ defined in \eqref{Xp_infty} such that
		\begin{equation*}
			\left| M(x,k)-X^{[p]}(x,k)\right| \leq \frac{C}{\left| k\right| ^{p+1}},\qquad x\in\mathbb{R},\,\,k\in \mathbb{C}\setminus\mathbb{R},\,\,\left| k\right| \geq 2.
		\end{equation*}
	\end{lem}

	\begin{lem}\label{lem_M_ST}
		$($Relation between $M\pm$ and $X$ and $Y$$)$.
		Suppose $\phi_0(x),\, n_0(x)\in\mathcal{S}(\mathbb{R})$ have compact support. Then
		\begin{equation}\label{def_ST}
			\begin{aligned}
				M_\pm(x,k)=&Y(x,k)\re^{x\widehat{\mathcal{L}}(k)}S_\pm(k)=X(x,k)\re^{x\widehat{\mathcal{L}}(k)}T_\pm(k),\\
			\end{aligned}
			\qquad  x\in\mathbb{R},\,\,k\in \overline{\mathbb{C}_\pm}\setminus\mathcal{Q},
		\end{equation}
	where $S_\pm(k)$ and $T_\pm(k)$ are defined by the  entries and the $(ij)$-th minor $m_{ij}(s)$ of the matrix $s(k)$, which are expressed by
\begin{equation*}\begin{aligned}
		S_+(k) &=\begin{pmatrix}
			s_{11}	&	0	&0\\
			s_{21}	&\frac{m_{33}(s)}{s_{11}}	&0\\
			s_{31}	&\frac{m_{23}(s)}{s_{11}}&\frac{1}{m_{33}(s)}
		\end{pmatrix},&&
		S_-(k) =\begin{pmatrix}
			\frac{1}{m_{11}(s)}	&\frac{m_{21}(s)}{s_{33}}		&s_{13}\\
			0		&\frac{m_{11}(s)}{s_{33}}	&s_{23}\\
			0	&0&s_{33}
		\end{pmatrix},\\
		T_+(k) &=\begin{pmatrix}
			1	&	-\frac{s_{12}}{s_{11}}	&\frac{m_{31}(s)}{m_{33}(s)}\\
			0	&1	&-\frac{m_{32}(s)}{m_{33}(s)}\\
			0	&0&1
		\end{pmatrix},&&
		T_-(k) =\begin{pmatrix}
			1	&0		&0\\
			-\frac{m_{12}(s)}{m_{11}(s)}		&1	&0\\
			\frac{m_{13}(s)}{m_{11}(s)}	&-\frac{s_{32}}{s_{33}}&1
		\end{pmatrix}.
\end{aligned}\end{equation*}

	\end{lem}
	
	\begin{proof}
		Let $C>0$ such that $\left[-C,C \right]  \subset \mathbb{R}$  and $\phi_0(x)$, $n_0(x)$ have support in $\left[-C,C \right]$. The following definitions are given:
		\begin{equation*}
			\left\lbrace 
			\begin{aligned}
				&S_\pm(k)=\lim_{x\to-\infty}\re^{x\widehat{\mathcal{L}}(k)}M_\pm(x,k),\\
				&T_\pm(k)=\lim_{x\to\infty}\re^{x\widehat{\mathcal{L}}(k)}M_\pm(x,k),
			\end{aligned}\right.  \qquad k\in \bar{\mathbb{C}}_\pm\setminus\mathcal{Q}.
		\end{equation*}
		The above limit exists because for $x \in  \mathbb{R}\setminus \left[-C,C \right]$, $L_1 = 0$, which means that $\re^{x\widehat{\mathcal{L}}(k)}M_\pm(x,k)$ is independent of $x$ in this case. Recalling the definition of $s(k)$, we can obtain that equation \eqref{def_ST} holds and
		\begin{equation}\label{s_ST}
			s(k)=S_\pm(k)T^{-1}_\pm(k), \qquad k\in\bar{\mathbb{C}}_\pm\setminus\mathcal{Q}.
		\end{equation}
		Given $s(k)$, the above expression constitutes a matrix factorization problem, and $S_\pm(k)$ and $T_\pm(k)$ can be uniquely determined. According to the Fredholm integral equation defined at the beginning of this section, we can obtain
		\begin{equation*}
			\left\lbrace \begin{aligned}
				&  \left( S_+(k)\right) _{ij}=0, \quad i<j,     \\
				&  \left( S_-(k)\right) _{ij}=0, \quad  i>j,    
			\end{aligned}\right. \qquad
			\left\lbrace \begin{aligned}
				&  \left( T_+(k)\right) _{ij}=\delta_{ij}, \quad i\geq j,     \\
				&  \left( T_-(k)\right) _{ij}=\delta_{ij}, \quad  i\leq j.
			\end{aligned}\right.
		\end{equation*}
		Thus, the equation \eqref{s_ST} also defines nine equations concerning nine unknowns. The specific forms of solutions can be derived from the above and are demonstrated in the lemma. \qedhere
		
	\end{proof}

    Let $\eta \in C_c^\infty(\mathbb{R})$ be a cutoff function such that $\eta(x) = 1$ for $|x| \leq 1$ and $\eta(x) = 0$ for $|x| \geq 2$. For each integer $j \geq 1,$ define $\eta_j(x) := \eta(x/j) $. If  $f \in \mathcal{S}(\mathbb{R}) $, then  $\{\eta_j f\}_{j=1}^\infty$  forms a sequence of smooth, compactly supported functions that converges to $f$ in the topology of $\mathcal{S}(\mathbb{R})$  as  $j \to \infty$.
	
	\begin{lem}\label{lem_Cauchy sequence}
		 Let $\phi_0(x),\, n_0(x)\in\mathcal{S}(\mathbb{R})$. Let ${s(k), M_\pm(x,k)}$ and ${s^{(i)}(k), M_\pm^{(i)}(x,k)}$ be the spectral functions and eigenfunctions associated with $(\phi_0, n_0)$ and
		\begin{equation}\label{ni_phii}
			\left(\phi_0^{(i)}(x), n_0^{(i)}(x)\right) := (\eta_i \phi_0, \eta_i n_0) \in \mathcal{S}(\mathbb{R}) \times \mathcal{S}(\mathbb{R}), 
		\end{equation}
		respectively. Then
\begin{equation*} \begin{aligned} 
&\lim_{i \to \infty} s^{(i)}(k) = s(k),\ k \,\, {\rm in}\, \eqref{sk_k_define},
\quad \lim_{i \to \infty} \left(s^A\right)^{(i)}(k) = s^A(k),\ k \,\, {\rm in}\,\, \eqref{sk_define},
\quad \lim_{i \to \infty} M_\pm^{(i)}(x,k) = M_\pm(x,k), x \in \mathbb{R}, \,\,\ k \in \overline{\mathbb{C}_\pm}\setminus \mathcal{Q}. \\ 
&\lim_{i \to \infty} X^{(i)}(x,k) = X(x,k), x\in\mathbb{R},\,\,k\in (\overline{\mathbb{C}_+},\mathbb{R},\overline{\mathbb{C}_-})\setminus\{0\},  \quad\lim_{i \to \infty} Y^{(i)}(x,k) = Y(x,k), x\in\mathbb{R},\,\,k\in (\overline{\mathbb{C}_-},\mathbb{R},\overline{\mathbb{C}_+})\setminus\{0\}.
\end{aligned} \end{equation*}

	\end{lem}
	
	\begin{proof}
		The proof hinges on verifying that the solutions to the Volterra equations \eqref{XY_XAYA}, and the Fredholm equation \eqref{Mn_fredholm} depend continuously on the potential functions $\phi_0(x),\, n_0(x)\in\mathcal{S}(\mathbb{R})$. For details, refer to Lemma 4.5 in Ref. \cite{Lenells-Indiana}.
	\end{proof}
	
	\begin{lem}\label{lem_M_jump}
		$($Jump condition for $M$$)$.
		Suppose $\phi_0(x),\, n_0(x)\in\mathcal{S}(\mathbb{R})$. For each $x\in\mathbb{R}$, $M(x,k)$ satisfies the jump condition
		\begin{equation*}
			M_+(x,k)=M_-(x,k)v(x,0,k),\qquad k\in \mathbb{R}\setminus\mathcal{Q},
		\end{equation*}
	where $v(x,0,k)$ is the jump matrix defined in equation \eqref{jump_0} 
	 and $\mathcal{Q}$ is defined by equation \eqref{Q_setminus}.
	\end{lem}
	
	\begin{proof}
		Suppose $C>0$ such that $\left[-C,C \right]  \subset \mathbb{R}$  and $\phi_0(x)$, $n_0(x)$ have support in $\left[-C,C \right]$, then we have
		
		\begin{equation*}
			v(x,0,k)=\re^{x\widehat{\mathcal{L}}(k)}\left(S_-^{-1}(k)S_+(k) \right) :=\re^{x\widehat{\mathcal{L}}(k)}J(k),
		\end{equation*}
		where
		\begin{equation*}
			\begin{aligned}
				J(k)&=\begin{pmatrix}
					1 & 0 & 0\\
					\frac{m_{12}(s)}{m_{11}(s)} & 1 & 0\\
					\frac{s_{31}}{s_{33}} & \frac{s_{32}}{s_{33}} & 1\\
				\end{pmatrix}\begin{pmatrix}
				1 & -\frac{s_{12}}{s_{11}} & \frac{m_{31}(s)}{m_{33}(s)}\\
				0 & 1 & -\frac{m_{32}(s)}{m_{33}(s)}\\
				0 & 0 & 1\\
				\end{pmatrix}=\begin{pmatrix}
				1 & -\frac{s_{12}}{s_{11}} & \frac{m_{31}(s)}{m_{33}(s)}\\
				\frac{m_{12}(s)}{m_{11}(s)} & 1-\frac{s_{12}m_{12}(s)}{s_{11}m_{11}(s)} & -\frac{s_{23}}{m_{11}(s)m_{33}(s)}\\
				\frac{s_{31}}{s_{33}} & \frac{m_{23}(s)}{s_{11}s_{33}} & \frac{1}{s_{33}m_{33}(s)}\\
				\end{pmatrix}.
			\end{aligned}
		\end{equation*}
		According to the symmetry conditions satisfied by $s(k)$ and $s^A(k)$, denoted as equations \eqref{sym_s} and \eqref{sym_sA}, we can derive

	\begin{equation*}
			\begin{aligned}
				J(k)&=\begin{pmatrix}
					1 & -r_1(k) & r_2(k)\\
					-8kr_1^*(k) & 1+8k\left|r_1(k) \right|^2  & -8k\alpha(k) \\
					-r_2^*(k) & \alpha^*(k) & 1-8k\left|r_1(-k) \right|^2- \left|r_2(k) \right|^2\\
				\end{pmatrix},
			\end{aligned}
		\end{equation*}
        where $\alpha(k)=r_1^*(-k)+ r_1^*(k)r_2(k)$, $r_1(k)$ and $r_2(k)$ are defined by equation \eqref{r1_r2}.
        
        If $\phi_0(x),\, n_0(x)\in\mathcal{S}(\mathbb{R})$ are not compactly supported, we take a sequence of compactly supported smooth functions $\left(\phi_0^{(i)}(x), n_0^{(i)}(x)\right)$ that are defined as in equation \eqref{ni_phii} and converge to $(\phi_0,\, n_0)$. According to Lemma \ref{lem_Cauchy sequence}, taking the limit $i \to \infty$  in the corresponding relations for $\left(\phi_0^{(i)}(x), n_0^{(i)}(x)\right)$. \qedhere
		
	\end{proof}
	
	\begin{lem}\label{lem_M_pm}
		Suppose $\phi_0(x),\, n_0(x)\in\mathcal{S}(\mathbb{R})$. The functions $M_\pm$ can be expressed in terms of the entries of $X$, $Y$, $X^A$ and $Y^A$, $s$ and $s^A$ as follows:
		\begin{align*}
			&M_+=\begin{pmatrix}
				X_{11} & \frac{X^A_{23}Y^A_{31}-X^A_{33}Y^A_{21}}{s_{11}} & \frac{Y_{13}}{s^A_{33}}\\
				X_{21} & \frac{X^A_{33}Y^A_{11}-X^A_{13}Y^A_{31}}{s_{11}} & \frac{Y_{23}}{s^A_{33}}\\
				X_{31} & \frac{X^A_{13}Y^A_{21}-X^A_{23}Y^A_{11}}{s_{11}} & \frac{Y_{33}}{s^A_{33}}\\
			\end{pmatrix},\qquad  	M_-=\begin{pmatrix}
			\frac{Y_{11}}{s^A_{11}} & \frac{X^A_{31}Y^A_{23}-X^A_{21}Y^A_{33}}{s_{33}} & X_{13}\\
			\frac{Y_{21}}{s^A_{11}} & \frac{X^A_{11}Y^A_{33}-X^A_{31}Y^A_{13}}{s_{33}} & X_{23}\\
			\frac{Y_{31}}{s^A_{11}} & \frac{X^A_{21}Y^A_{13}-X^A_{11}Y^A_{23}}{s_{33}} & X_{33}\\
			\end{pmatrix},\\
           & M_+^A=\begin{pmatrix}
             \frac{Y_{11}^A}{s_{11}} & \frac{X_{31}Y_{23} - X_{21}Y_{33}}{s_{33}^A} & X_{13}^A \\
             \frac{Y_{21}^A}{s_{11}} & \frac{X_{11}Y_{33} - X_{31}Y_{13}}{s_{33}^A} & X_{23}^A \\
            \frac{Y_{31}^A}{s_{11}} & \frac{X_{21}Y_{13} - X_{11}Y_{23}}{s_{33}^A} & X_{33}^A
         \end{pmatrix},\qquad M_-^A=\begin{pmatrix}
				X^A_{11} & \frac{X_{23}Y_{31}-X_{33}Y_{21}}{s^A_{11}} & \frac{Y^A_{13}}{s_{33}}\\
				X^A_{21} & \frac{X_{33}Y_{11}-X_{13}Y_{31}}{s^A_{11}} & \frac{Y^A_{23}}{s{33}}\\
				X^A_{31} & \frac{X_{13}Y_{21}-X_{23}Y_{11}}{s^A_{11}} & \frac{Y^A_{33}}{s_{33}}\\
			\end{pmatrix},
            \end{align*}
		for  all $x\in\mathbb{R}$ and $k\in\overline{ \mathbb{C}_\pm}\setminus\mathcal{Q}$.
		
	\end{lem}
	
	\begin{proof}
Let $\left(\phi_0^{(i)}(x), n_0^{(i)}(x)\right)$ be the sequence defined in equation \eqref{ni_phii} that converges to $\left(\phi_0^{(i)}(x), n_0^{(i)}(x)\right)$. By Lemma \ref{lem_M_ST}, we have:
		\begin{equation*}
			M^{(i)}_\pm(x,k)=Y^{(i)}(x,k)\re^{x\widehat{\mathcal{L}}(k)}S^{(i)}_\pm(k)=X^{(i)}(x,k)\re^{x\widehat{\mathcal{L}}(k)}T^{(i)}_\pm(k),
			\qquad  x\in\mathbb{R},\,\,k\in \overline{\mathbb{C}_\pm}\setminus\mathcal{Q}.
		\end{equation*}
        Thus, the first and third columns of $M_+$ and $M_-$ can be expressed for $m\in\mathbb{R}$, and $i\geq 1$ as:
        \begin{align*}
            &[M_+^{(i)}(x,k)]_1=[X^{(i)}(x,k)]_1,\quad [M_+^{(i)}(x,k)]_3=\frac{[Y^{(i)}(x,k)]_3}{(s^A_{33})^{(i)}},\quad k\in \overline{\mathbb{C}_+}\setminus\mathcal{Q},\\
            &[M_-^{(i)}(x,k)]_1=\frac{[Y^{(i)}(x,k)]_1}{(s^A_{11})^{(i)}},\quad [M_-^{(i)}(x,k)]_3=[X^{(i)}(x,k)]_3,\quad k\in \overline{\mathbb{C}_-}\setminus\mathcal{Q}.
        \end{align*}
     Then, applying Lemma \ref{lem_Cauchy sequence} and taking the limit as $i \to \infty$, we have
     \begin{align*}
            &[M_+(x,k)]_1=[X(x,k)]_1,\quad [M_+(x,k)]_3=\frac{[Y(x,k)]_3}{s^A_{33}},\quad k\in \overline{\mathbb{C}_+}\setminus\mathcal{Q},\\
            &[M_-(x,k)]_1=\frac{[Y(x,k)]_1}{s^A_{11}},\quad [M_-(x,k)]_3=[X(x,k)]_3,\quad k\in \overline{\mathbb{C}_-}\setminus\mathcal{Q}.
        \end{align*}
    Similarly, we can also obtain the representations for the first and third columns of $M_{\pm}^A(x,k)$.

    On the other hand, the second column of $[M_+(x,k)]_2$ and $[M_-(x,k)]_2$ can be expressed in terms of $[M^A_\pm(x,k)]_1$ and $[M^A_\pm(x,k)]_3$ according to the identity relation $(M^A)^A=M$ as: 
    \begin{align*}
        &M_{+,12}=-m_{12}(M_+^A)=\frac{X^A_{23}Y^A_{31}-X^A_{33}Y^A_{21}}{s_{11}},\quad  M^A_{+,12}=-m_{12}(M_+)=\frac{X_{31}Y_{23} - X_{21}Y_{33}}{s_{33}^A},\\
        &M_{-,12}=-m_{12}(M_-^A)=\frac{X^A_{31}Y^A_{23}-X^A_{21}Y^A_{33}}{s_{33}},\quad  M^A_{-,12}=-m_{12}(M_-)=\frac{X_{23}Y_{31}-X_{33}Y_{21}}{s^A_{11}}.
    \end{align*}
    The remaining entries can be obtained by similar calculations.
	\end{proof}
	
	\begin{lem}
		Suppose $\phi_0(x),\, n_0(x)\in\mathcal{S}(\mathbb{R})$ are such that Assumption \ref{assu_soliton} holds. Then Proposition \ref{prop_Mn_basic} and Lemmas \ref{lem_M_ST}, \ref{lem_Cauchy sequence}, \ref{lem_M_jump} remain valid after replacing $\mathcal{Q}$ with $\hat{\textbf{Z}}\cup\{0\}$.
	\end{lem}
	
	\begin{proof}
		Under the condition that the initial data $\phi_0(x),\, n_0(x)\in\mathcal{S}(\mathbb{R})$ and satisfy Assumption \ref{assu_soliton}, Lemma \ref{lem_M_pm} shows that the function $M$ has no singularities other than $k=0$ and $k\in\hat{\textbf{Z}}$.  
		Using Symmetries \eqref{M-sym}, it follows that $M_{\pm}$ has no singularities in $\overline{\mathbb{C}_\pm}\setminus\left(\hat{\textbf{Z}}\cup\{0\}\right)$.  
		Therefore, $M_\pm(x,k)$ can be continuously extended to any $k\in \left(\overline{\mathbb{C}_\pm}\cap \mathcal{Z}\right)\setminus\left(\hat{\textbf{Z}}\cup\{0\}\right)$.\qedhere

	\end{proof}

	\begin{lem}$($Asymptotics of $M$ as $k\to 0$$)$.
		Under the condition that the initial data $\phi_0(x),\, n_0(x)\in\mathcal{S}(\mathbb{R})$ and satisfy Assumptions \ref{assu_soliton} and \ref{assu_k=0}, there exist $3 \times 3$ matrix-valued functions $\left\lbrace M_\pm^{l}(x),N_\pm^{l}(x)\right\rbrace_{l=1}^p$ for $p\geq1$ be an integer with the following properties:
		\begin{enumerate}
			\item The function $M(x,k)$ satisfies the asymptotic estimate
			\begin{equation*}
				\Bigl| M_\pm(x,k) - \sum_{l=-1}^{p} M_\pm^{(l)}(x) k^l \Bigr| \leq C |k|^{p+1}, \qquad x \in \mathbb{R},\; |k| \leq \tfrac{1}{2},\; k \in \overline{\mathbb{C}_\pm}\setminus\left(\hat{\textbf{Z}}\cup\{0\}\right),
			\end{equation*}
			where $C > 0$ is a constant and $\left\lbrace M_\pm^{l}\right\rbrace _{l=-1}^p$ are smooth functions of $x\in\mathbb{R}$. Furthermore,  some coefficients take the following form:
			\begin{align*}
				&M^{(-1)}_+(x) = \alpha_{11}(x) \begin{pmatrix}
					-1 & 0 & 0\\
					0 & 0 & 0\\
					1 & 0 & 0
				\end{pmatrix}+\delta_{-1}(x)\begin{pmatrix}
				    0 & 1 & 0\\
                    0 & 0 & 0\\
                    0 & -1 & 0\\
				\end{pmatrix}, \quad
				&&M^{(-1)}_-(x) = \alpha_{11}(x) \begin{pmatrix}
					0 & 0 & -1\\
					0 & 0 & 0\\
					0 & 0 & 1
				\end{pmatrix}+\delta_{-1}(x)\begin{pmatrix}
				    0 & -1 & 0\\
                    0 & 0 & 0\\
                    0 & 1 & 0\\
				\end{pmatrix},\\
            &M^{(0)}_+(x) = \begin{pmatrix}
                -\alpha_{13}(x) & * & -\delta_0(x)\\
                \alpha_{15}(x) & * & 0\\ 
                \alpha_{13}(x) & * & \delta_0(x)\\
            \end{pmatrix},\quad && M^{(0)}_-(x) = \begin{pmatrix}
                -\delta_0(x) & * & \alpha_{13}(x)\\
                0 & * & \alpha_{15}(x)\\ 
                \delta_0(x) & * & -\alpha_{13}(x)\\
            \end{pmatrix},
			\end{align*}
            where $*$ denotes an unspecified entry, $\delta_{-1}(x)=\frac{\beta_{11}(x)\beta_{22}(x)-\beta_{12}(x)\beta_{21}(x)}{s_{11}^{(-1)}}$, and $\delta_0(x)=-\frac{\alpha_{21}(x)}{(s^A)^{(-1)}_{11}}$, $s_{11}^{(-1)}$ and $(s^A)^{(-1)}_{11}$ are the elements at the first row and first column in equation \eqref{s-1} and \eqref{sA-1}, respectively.
			\item The function $M^{-1}(x,k)$ satisfies the asymptotic estimate
			\begin{equation*}
				\Bigl| M^{-1}_\pm(x,k) - \sum_{l=-1}^{p} N_\pm^{(l)}(x) k^l \Bigr| \leq C |k|^{p+1}, \qquad x \in \mathbb{R},\; |k| \leq \tfrac{1}{2},\; k \in \overline{\mathbb{C}_\pm}\setminus\left(\hat{\textbf{Z}}\cup\{0\}\right),
			\end{equation*}
			where $C > 0$ is a constant and $\left\lbrace N_\pm^{l}\right\rbrace _{l=-1}^p$ are smooth functions of $x\in\mathbb{R}$. Furthermore,  some coefficients take the following form
            \begin{align*}
				&N^{(-1)}_+(x) = \hat{\alpha}_{11}(x) \begin{pmatrix}
					0 & 0 & 0\\
					0 & 0 & 0\\
					1 & 0 & 1
				\end{pmatrix}+\hat{\alpha}_{15}(x)(x)\begin{pmatrix}
				    0 & 0 & 0\\
                    0 & 0 & 0\\
                    0 & 1 & 0\\
				\end{pmatrix}, \quad
				&&N^{(-1)}_-(x) = \hat{\alpha}_{11}(x) \begin{pmatrix}
					-1 & 0 & -1\\
					0 & 0 & 0\\
					0 & 0 & 0
				\end{pmatrix}+\hat{\alpha}_{15}(x)(x)\begin{pmatrix}
				    0 & -1 & 0\\
                    0 & 0 & 0\\
                    0 & 0 & 0\\
				\end{pmatrix},\\
            &N^{(0)}_+(x) = \begin{pmatrix}
                -\alpha_{13}(x) & * & -\delta_0(x)\\
                \alpha_{15}(x) & * & 0\\ 
                \alpha_{13}(x) & * & \delta_0(x)\\
            \end{pmatrix},\quad && N^{(0)}_-(x) = \begin{pmatrix}
                -\delta_0(x) & * & \alpha_{13}(x)\\
                0 & * & \alpha_{15}(x)\\ 
                \delta_0(x) & * & -\alpha_{13}(x)\\
            \end{pmatrix}.
			\end{align*}
			\item For each $x\in\mathbb{R}$, the functions $N_{1,\pm}(x,t)=\begin{pmatrix}
				1 & 0 & 1
			\end{pmatrix}M_\pm(x,k)$ and $N_{2,\pm}(x,t)=\begin{pmatrix}
			0 & 1 & 0
			\end{pmatrix}M_\pm(x,k)$ are bounded as $k\to0$, $k\in\overline{\mathbb{C}_\pm}\setminus\left(\hat{\textbf{Z}}\cup\{0\}\right)$.
		\end{enumerate}
		
	\end{lem}
	
	\begin{proof}
		Using the properties of functions $X$ and $Y$, $X^A$ and $Y^A$, $s$, and $s^A$ as $k\to0$ from Propositions \ref{prop_Xkto0}, \ref{prop_s}, \ref{prop_XAkto0} and \ref{prop_sA}, the proof can be completed.
	\end{proof}

   \subsection{Proof of Theorem \ref{theo-direct}}\label{subsec_proof21}
    Given that Assumption \ref{assu_soliton} holds, meaning that set $\textbf{Z}\cap \mathbb{R}$ is empty, it follows that $s_{11}(k)$ and $s^A_{33}(k)$ are nonzero for $k\in\mathbb{R}$. Revisiting the definitions of $r_1(k)$ and $r_2(k)$ in the equation \eqref{r1_r2}, and combining with Propositions \ref{prop_s} and \ref{prop_sA}, the functions $r_1(k)$ and $r_2(k)$ are smooth on $\mathbb{R}$. Hence, $r_1(k),\,r_2(k) \in C^\infty(\mathbb{R})$. The second and third properties can also be derived from the two propositions concerning $s(k)$ and $s^A(k)$. For $k\in\mathbb{R}$, using the definitions of reflection coefficients $r_1(k)$ and $r_2(k)$ in equation \eqref{r1_r2}, we can obtain
    \begin{align*}
        8kr_1(k)r_1^*(-k)-r_2(k)+r_2^*(-k)&=8k\frac{s_{12}(k)}{s_{11}(k)}\frac{s^*_{12}(-k)}{s^*_{11}(-k)}-\frac{s^A_{31}(k)}{s^A_{33}(k)}+\frac{(s^A_{31})^*(-k)}{(s^A_{33})^*(-k)}=\frac{s_{12}(k)}{s_{11}(k)}\frac{m_{32}(s(k))}{m_{33}(s(k))}-\frac{m_{31}(s(k))}{m_{33}(s(k))}-\frac{s_{13}(k)}{s_{11}(k)}\\
        &=\frac{s_{12}(k)m_{32}(s(k))-s_{11}(k)m_{31}(s(k))}{s_{11}(k)m_{33}(s(k))}-\frac{s_{13}(k)}{s_{11}(k)}=0.
    \end{align*}
    In the above calculation, we used the symmetric relations \eqref{sym_s} and \eqref{sym_sA}, as well as $s^T(k)s^A(k)=I$. The following will focus on proving the last two properties.

     The zeros of the denominators $s_{11}(k)$ and $s^A_{11}(k)$ in equation \eqref{r1_r2} correspond to solitons. It has been proved in Propositions \ref{prop_s} and \ref{prop_sA} that $s_{11}(k)$ and $s^A_{11}(k)$ have smooth extensions to $\overline{\mathbb{C}_+}\setminus\{0\}$ and $\overline{\mathbb{C}_-}\setminus\{0\}$. For simplicity, we will restrict ourselves to the generic case where the zeros of $s_{11}(k)$ and $s_{11}^A(k)$ do not lie on the contour $\mathbb{R}$.

    Let us prove \textbf{(v)}. First, according to Assumption \ref{assu_soliton}, $\textbf{Z}$ is a finite subset of $D_{\mathrm{reg}} \cup \ri\mathbb{R}_+$. We need to determine that the complex constants $\{C_{k_j}\}_{k_j\in\textbf{Z}}$, under general Schwartz initial conditions, are well-defined by equations \eqref{ckj_non_Cc} and \eqref{ckj_non_Cir}. Using Lemmas \ref{lem_M_ST} and \ref{lem_Cauchy sequence}, the following holds for $k\in\mathbb{C}_+$:
    \begin{equation}\label{M+i}
    	M_+^{(i)}(x,k)=Y^{(i)}(x,k)\re^{x\widehat{\mathcal{L}(k)}}S_+^{(i)}(k)=X^{(i)}(x,k)\re^{x\widehat{\mathcal{L}(k)}}T_+^{(i)}(k),
    \end{equation}
    where
    \begin{equation*}
    		S^{(i)}_+(k) =\begin{pmatrix}
    			s^{(i)}_{11}	&	0	&0\\
    			s^{(i)}_{21}	&\frac{m^{(i)}_{33}(s)}{s^{(i)}_{11}}	&0\\
    			s^{(i)}_{31}	&\frac{m^{(i)}_{23}(s)}{s^{(i)}_{11}}&\frac{1}{m^{(i)}_{33}(s)}
    		\end{pmatrix},\qquad
    		T^{(i)}_+(k) =\begin{pmatrix}
    			1	&	-\frac{s^{(i)}_{12}}{s^{(i)}_{11}}	&\frac{m^{(i)}_{31}(s)}{m^{(i)}_{33}(s)}\\
    			0	&1	&-\frac{m^{(i)}_{32}(s)}{m^{(i)}_{33}(s)}\\
    			0	& 0 & 1
    		\end{pmatrix}.
    	\end{equation*}
    	
    	\begin{itemize}
    		\item For the case  $k_j\in\textbf{Z}\setminus\ri\mathbb{R}_+$. Considering the second column in equation \eqref{M+i}, we have
    		\begin{equation*}
    			m_{33}^{(i)}(s)[Y^{(i)}]_2+\re^{x(l_3-l_2)}m_{23}^{(i)}(s)[Y^{(i)}]_3=-\re^{x(l_1-l_2)}s_{12}^{(i)}[X^{(i)}]_1+s_{11}^{(i)}[X^{(i)}]_2.
    		\end{equation*}
    		Applying operator $\det\left( \cdot, [X^{(i)}]_1,[Y^{(i)}]_3\right) $ to both sides of the above equation yields
    		\begin{equation*}
    			m_{33}^{(i)}(s)\det\left( [Y^{(i)}]_2, [X^{(i)}]_1,[Y^{(i)}]_3\right)=s_{11}^{(i)}\det\left([X^{(i)}]_2, [X^{(i)}]_1,[Y^{(i)}]_3\right).
    		\end{equation*}
    		Using the definition of the inner product, the above equation can be written as 
    		\begin{equation}\label{m33i_inner}
    			m_{33}^{(i)} (s)[X^{(i)}]_1\cdot \left[\left( Y^A\right) ^{(i)}\right]_1=s_{11}^{(i)} [Y^{(i)}]_3\cdot \left[\left( X^A\right) ^{(i)}\right]_3,
    		\end{equation}
    		 where the inner product notation $\cdot$ is defined for the column vectors 
    		$u=(u_1,u_2,u_3)$ and $v=(v_1,v_2,v_3)$ as $u\cdot v=u_1v_1+u_2v_2+u_3v_3$. For $ k \in\mathbb{C}_+ $, all quantities in the equation have well-defined limits as $ i \to \infty $, so letting  
    		$ i \to \infty $ yields
    		\begin{equation}\label{s11_sa33}
    			s_{33}^A (k)[X(x,k)]_1\cdot [Y^A(x,k)]_1=s_{11}(k)[Y(x,k)]_3[X^A(x,k)]_{3},\quad x\in\mathbb{R},\,\, k\in\mathbb{C_+}.
    		\end{equation}
    		\quad According to symmetry relation \eqref{sym_k-bark}, we have relation $s_{11}(k)=\overline{s_{33}^A( -\bar k) }$. Thus, $s^A_{33}(k_j)=\overline{s_{11}(- \overline{k_j})}\neq 0$ for $k_j\in \textbf{Z}\setminus\ri\mathbb{R}_+$ and $- \overline{k_j}\in \textbf{Z}_m^*\setminus\ri\mathbb{R}_+$ under Assumption \ref{assu_soliton}. Taking $k=k_j$ in equation \eqref{s11_sa33} yields 
    	\begin{equation*}
    		[X(x,k_j)]_1\cdot [Y^A(x,k_j)]_1=0.
    	\end{equation*}
    	On the other hand, since $\det\left( [X]_{1}, [X]_{1},[X]_{2}\right)(x,k_j)=0$, it follows that 
    	\begin{equation*}
    		[X(x,k_j)]_1\cdot [X^A(x,k_j)]_{3}=0.
    	\end{equation*}
    	Hence, $[X(x,k_j)]_1$ belongs to the kernel of the linear map $\mathbb{C}^3 \to \mathbb{C}^2$ defined by:
    	\begin{equation}\label{wxkj}
    		w(x,k_j)=\begin{pmatrix}
    			Y_{11}^A(x,k_j) & Y_{21}^A(x,k_j) & Y_{31}^A(x,k_j)\\
    			X_{13}^A(x,k_j) & X_{23}^A(x,k_j) & X_{33}^A(x,k_j)\\
    	\end{pmatrix}.
    	\end{equation}
    	\quad The vector $m(x, k_j)$ defined in equation \eqref{mxk} is obtained as the cross product of the two row vectors of $w(x,k_j)$, so it also lies in the kernel of $w(x,k_j)$. As $x \to -\infty$, using the asymptotic properties of $[Y^A]_1$ and $[X^A]_3$, it follows that $[Y^A(x,k_j)]_1$ and $[X^A(x,k_j)]_{3}$ are linearly independent for sufficiently large negative $x$. Consequently, for any $x$, $[Y^A(x,k_j)]_1$ and $[X^A(x,k_j)]_{3}$ are linearly independent. Thus, $\text{rank} \, w(x,k_j) = 2 $, and by the rank-nullity theorem, we have $\dim \ker w(x,k_j) = 1$. It follows that for any $x$, $[X(x,k_j)]_1$ and $m(x, k_j)$ are linearly dependent. Therefore, there exists a function $C_{k_j}(x)$ such that
    	\begin{equation*}
    		\frac{m(x, k_j)}{s'_{11}(k_j)} = C_{k_j}(x)e^{x(l_1(k_j) - l_2(k_j))}[X(x, k_j)]_{1}, \quad \forall x\in\mathbb{R}.
            \end{equation*}
    	Since $e^{xl_1(k_j)}[X(x, k_j)]_{1}$ and $e^{xl_2(k_j)}m(x, k_j)$ satisfy the same linear ordinary differential equation in $x$, $C_{k_j}$ is in fact independent of $x$. We have completed the proof of equation \eqref{ckj_non_Cc} for $k\in\textbf{Z}\setminus\ri\mathbb{R}_+$.

    	\item For the case  $k_j\in\textbf{Z}\cap\ri\mathbb{R}_+$. It follows that $s_{11}(k_j)=\overline{m_{33}(s(k_j))}$, that is, $m_{33}(s(k))$ also has a simple zero at $k_j$ and $m_{33}(s(k_j))'\neq 0$. In particular, since $s_{11}(k)$ has finitely many simple zeros in $\textbf{Z}$, there exists $\epsilon> 0$ such that $\overline{D_\epsilon(k_j)}\subset \mathbb{C}_+$, and $s_{11}(k) \neq 0$ for all $k\in\overline{D_\epsilon(k_j)}\setminus\{k_j\}$. Thus the functions $s_{11}^{(i)}$ converge uniformly to $s_{11}$ on $\partial D_\epsilon(k_j)$, and $s_{11}^{(i)}$ are nonzero on $\partial D_\epsilon(k_j)$ for all sufficiently large $i$, satisfying:
        \begin{equation*}
        	\frac{1}{2\pi\ri} \int_{\partial D_\epsilon(k_j)} \frac{\left( s_{11}^{(i)}\right)'(k) }{s_{11}^{(i)}(k)}\rd k\to 1,\quad \frac{1}{2\pi\ri} \int_{\partial D_\epsilon(k_j)} \frac{k\left( s_{11}^{(i)}\right)'(k) }{s_{11}^{(i)}(k)}\rd k\to k_j,\quad \text{as}\,\, i\to\infty.
        \end{equation*}
    	The reason lies in the fact that for $s_{11}(k)$, the following two equalities hold:
    	\begin{equation*}
    		\frac{1}{2\pi\ri} \int_{\partial D_\epsilon(k_j)} \frac{ s_{11}'(k) }{s_{11}(k)}\rd k=1,\quad \frac{1}{2\pi\ri} \int_{\partial D_\epsilon(k_j)} \frac{k s_{11}'(k) }{s_{11}(k)}\rd k= k_j.
    	\end{equation*}
    	For sufficiently large index $i$, the following holds:
    	\begin{equation*}
    		\frac{1}{2\pi\ri} \int_{\partial D_\epsilon(k_j)} \frac{\left( s_{11}^{(i)}\right)'(k) }{s_{11}^{(i)}(k)}\rd k= 1.
    	\end{equation*}
    	The function $s_{11}^{(i)}(k)$ has exactly one zero within $D_\epsilon(k_j)$, denoted as $k_j^{(i)}$. Moreover, from the second integral limit above, it follows that $k_j^{(i)}\to k_j$ as $i\to\infty$. The difference between the regions where $k_j^{(i)}$ lie will lead to two distinct proof processes. By selecting subsequences, it is sufficient to consider only \textbf{Case 1}: $k_j^{(i)}\in \ri\mathbb{R}_+$ for all $i$ and \textbf{Case 2}: $k_j^{(i)}\notin \ri\mathbb{R}_+$ for all $i$.
    	
    	\quad Firstly, consider \textbf{Case 1}. Using Lemma \ref{lem_Ck1_im}, the following holds:
    	\begin{equation*}
    		s_{12}^{(i)}\left( k_j^{(i)}\right)=s_{21}^{(i)}\left( k_j^{(i)}\right)=m_{33}^{(i)}\left(s( k_j^{(i)})\right)=m_{23}^{(i)}\left( s(k_j^{(i)})\right)=m_{32}^{(i)}\left( s(k_j^{(i)})\right)=0,\quad \forall \,i.
    	\end{equation*}
    	The second column in the equation \eqref{M+i} is represented as; $$[M^{(i)}_{+}(x,k)]_2=\frac{m^{(i)}_{33}(s(k))[Y^{(i)}(x,k)]_2+\re^{x(l_3(k)-l_2(k))}m^{(i)}_{23}(s(k))[Y^{(i)}(x,k)]_3}{s^{(i)}_{11}(k)}=-\frac{\re^{x(l_1(k)-l_2(k))}s^{(i)}_{1}(k)[X^{(i)}(x,k)]_1}{s^{(i)}_{11}(k)}.$$
    	 Clearly, $k_j^{(i)}$ is a simple zero for both the numerator and the denominator of this expression for $k\in D_{\epsilon}( k_j) $. In other words, 
    	\begin{equation*}
    		\mathop{\text{Res}}\limits_{k=k_j^{(i)}}[M^{(i)}_{+}(x,k)]_2=0.
    	\end{equation*}
    Now consider the transformed expression for the third column of equation \eqref{M+i}:
    	\begin{equation}\label{M+i3}
    		\begin{aligned}
    			\left[Y^{(i)}\left( x,k_j^{(i)}\right)\right]_3 &=-\re^{x(l_2(k_j^{(i)})-l_3(k_j^{(i)}))}m_{32}^{(i)}\left( s(k_j^{(i)})\right) \left[X^{(i)}\left( x,k_j^{(i)}\right)\right]_2+m_{33}^{(i)}\left( s(k_j^{(i)})\right) \left[X^{(i)}\left( x,k_j^{(i)}\right)\right]_3\\
                &\quad +\re^{x(l_1(k_j^{(i)})-l_3(k_j^{(i)}))}m_{31}^{(i)}\left( s(k_j^{(i)})\right)\left[X^{(i)}\left( x,k_j^{(i)}\right)\right]_1\\
    			&=\re^{x(l_1(k_j^{(i)})-l_3(k_j^{(i)}))}m_{31}^{(i)}\left( s(k_j^{(i)})\right)\left[X^{(i)}\left( x,k_j^{(i)}\right)\right]_1.
    		\end{aligned}
    	\end{equation}
    	Dividing both sides by $m_{33}^{(i)}\left(s( k_j^{(i)})\right) ' $, we obtain:
    	\begin{equation*}
    			\frac{\left[Y^{(i)}\left( x,k_j^{(i)}\right)\right]_3}{m_{33}^{(i)}\left(s( k_j^{(i)})\right) '} =C_{k_j^{(i)}}\re^{x(l_1(k_j^{(i)})-l_3(k_j^{(i)}))}\left[X^{(i)}\left( x,k_j^{(i)}\right)\right]_1,
    	\end{equation*}
    	where $C_{k_j^{(i)}}=\frac{m_{31}^{(i)}\left(s( k_j^{(i)})\right)
    	}{m_{33}^{(i)}\left(s( k_j^{(i)})\right) '}=-\frac{s_{13}^{(i)}\left( k_j^{(i)}\right)
    	}{\left( s_{11}^{(i)}\right)'\left( k_j^{(i)}\right) }$. Since $X_{11}(x,k_j) \to 1$ as $x \to +\infty$, there exists $x_0$ such that $X_{11}(x_0,k_j) \neq 0$. Therefore, the following limit exists and is a finite number:
    	\begin{equation*}
    		\lim_{i\to\infty}C_{k_j^{(i)}}=\lim_{i\to\infty}\frac{Y^{(i)}_{13}\left( x,k_j^{(i)}\right)\re^{-x(l_1(k_j^{(i)})-l_3(k_j^{(i)}))}}{m_{33}^{(i)}\left(s( k_j^{(i)})\right) 'X^{(i)}_{11}\left( x,k_j^{(i)}\right) } =\frac{Y_{13}\left( x,k_j\right)\re^{-x(l_1(k_j)-l_3(k_j))}}{ m_{33}(s( k_j))'X_{11}\left( x,k_j\right) }.
    	\end{equation*}
    	 Taking the limit in equation \eqref{M+i3} thus establishes formula \eqref{ckj_non_Cir} in \textbf{Case 1}.
    	 
    	 \quad We next consider \textbf{Case 2}. It is known that $k_j^{(i)}\notin \ri\mathbb{R}_+$ is the unique zero of $s_{11}^{(i)}(k)$ in $D_\epsilon(k_j)$ and that $s_{11}^{(i)}(k)=\overline{m_{33}^{(i)}(s(-\bar k))}$, we obtain that $-\overline{k_j^{(i)}}$ is also the unique zero of $m_{33}^{(i)}(s(k))$ in $D_\epsilon(k_j)$. Moreover, it holds that $s_{11}^{(i)}\left(-\overline{k_j^{(i)}} \right)\neq 0 $ and $m_{33}^{(i)}\left(s({k_j^{(i)}} )\right)\neq 0 $. Using the result from equation \eqref{m33i_inner} we have proved for $k_j^{(i)} \notin \ri\mathbb{R}_+ $, and taking values in $k_j^{(i)}$ and $-\overline{k_j^{(i)}}$, it follows that
    	 \begin{equation*}
    	 	\left[X^{(i)}\left( x,k_j^{(i)}\right)\right]_1 \cdot \left[\left( Y^A\right) ^{(i)}\left(x, k_j^{(i)}\right)\right]_1 =0,\qquad \left[Y^{(i)}\left( x,-\overline{k_j^{(i)}}\right)\right]_3 \cdot \left[\left( X^A\right) ^{(i)}\left(x,-\overline{k_j^{(i)}}\right)\right]_3 =0.
    	 \end{equation*}
    	 Thus, it follows that
    	\begin{equation*}
    		\left[X\left( x,k_j\right)\right]_1 \cdot  \left[Y^A\left( x,k_j\right)\right]_1 =0,\qquad \left[Y\left( x,k_j\right)\right]_3 \cdot  \left[X^A\left( x,k_j\right)\right]_3 =0.
    	\end{equation*}
    	On the other hand,
    	\begin{equation*}
    		\left[Y\left( x,k_j\right)\right]_3 \cdot  \left[Y^A\left( x,k_j\right)\right]_1 =0,\qquad \left[X\left( x,k_j\right)\right]_1 \cdot  \left[X^A\left( x,k_j\right)\right]_3 =0.
    	\end{equation*}
    	The above equation also indicates that both vectors $\left[X\left( x,k_j\right)\right]_1$ and $\left[Y\left( x,k_j\right)\right]_3$ belong to the kernel of matrix $w(x,k_j)$ defined by \eqref{wxkj}. A similar argument shows that $\left[X\left( x_0,k_j\right)\right]_1$ and $\left[Y\left( x_0,k_j\right)\right]_3$ are linearly dependent at some $x_0\in\mathbb{R}$. Thus it follows that they are linearly dependent for all $x$, and there exists a constant $C_{k_j} \in \mathbb{C}$ independent of $x$ such that \eqref{ckj_non_Cir} holds.

    	\end{itemize}

          In Remark \ref{remark_singular_soliton}, we will show that when zeros of $s_{11}(k)$ appear in $D_{\text{sing}}$, the solitons of YO equation \eqref{YO} are singular. Here, singularity means that the solution is not smooth as a function of $x \in \mathbb{R}$ but has poles on the real axis when viewed as a function of $x$ for any fixed $t$. Additionally, some of the solitons generated by zeros of $s_{11}(k)$ appear in $\ri \mathbb{R}_+$ are also singular. In the subsequent proof, we will derive conditions to characterize the zeros that correspond to non-singular solitons.
    	
    	The proof of \textbf{(vi)} also proceeds in two cases: when the initial data $n_0$ and $\phi_0$ have compact support, and when they do not have compact support. Using the conclusion of Lemma \ref{lem_Ck1_im}, it suffices to show that the quantity $C_{k_j}\in\ri\mathbb{R}$ defined by formulas \eqref{ckj_compact}, \eqref{ckj_non_Cc}, and \eqref{ckj_non_Cir}. This follows from a direct computation, which we omit here.  \hfill $\square$

   \subsection{Construction of $M(x,t,k)$ for $t>0$} 
   	Suppose $\{n(x,t), \phi(x,t)\}$ is a Schwartz class solution of equation \eqref{YOE} with existence time $T \in (0, \infty]$ and initial data $n_0(x), \phi_0(x) \in \mathcal{S}(\mathbb{R})$. Supposing that the Assumptions \ref{assu_soliton} and \ref{assu_k=0} hold, the reflection coefficients $r_1(k)$ and $r_2(k)$ are defined from $n_0, \phi_0$ via \eqref{r1_r2}. Then we consider the time dependence of $M$. By replacing $L_1(x, k)$ in the integral equation \eqref{Mn_fredholm} with the time-dependent matrix $L_1(x, t, k)$, we define two time-dependent eigenfunctions $M_\pm(x, t, k)$. Then, it will be shown that $M(x,t,k)$ satisfies the following RH problem.
   \begin{rhp}\label{rhp_M}\rm{(RH problem for $M(x,t,k)$)}.
	Find a $3\times3$ matrix-valued function $M(x,t,k)$ with the following properties:
	\begin{enumerate}
		\item The function $M(x,t,k)$ is analytic for $k\in\mathbb{C}\setminus\{\mathbb{R}\cup\hat{\textbf{Z}}\}$.
		\item As $k$ approaches $\mathbb{R}$ from the left and right, the boundary values $M_+(x,t,k)$ and $M_-(x,t,k)$ of $M(x,t,k)$ exist and satisfy the following relationship:
		\begin{equation*}
			M_+(x,t,k)=M_-(x,t,k)v(x,t,k), \quad k\in \mathbb{R},
		\end{equation*}
		 where $v(x,t,k)$ is given by equation \eqref{jump_0}.	
	
        \item At each point in $\hat{\textbf{Z}}$, two columns of $M$ are analytic, while the remaining column has at most a simple pole. For every $p\in \hat{\textbf{Z}}$, the following residue conditions hold:
          \begin{equation}\label{residue_M}
				\mathop{\text{Res}}\limits_{k=p }M(x,t,k)=\lim\limits_{k\to p } M(x,t,k)\re^{x\widehat{\mathcal{L}}(k)+t\widehat{\mathcal{Z}}(k)}V(p),
	\end{equation}
    where $V(P)$ is determined by \eqref{Res_2} and \eqref{Res22}.
		\item  As $k\rightarrow\infty$, $k\in\mathbb{C}\setminus\mathbb{R}$, we have
        \begin{equation*}
        M(x,t,k)=I+\frac{M_1(x,t)}{k}+\frac{M_2(x,t)}{k^2}+\mathcal{O}\left(k^{-3}\right),
        \end{equation*}
        where the matrices $M_1$ and $M_2$ are related only to the variables $x$ and $t$, independent of the spectral parameter $k$, and satisfy $M_{1,12}(x,t)=M_{1,13}(x,t)=0$.
		\item  As $k\rightarrow0$, there exist matrices $\{M^{(\pm)}_l(x,t)\}_{l=-1}^{+\infty}$ such that for any $N\geq-1$,
        \begin{equation*}
            M_\pm(x,t,k)=\sum\limits_{l=-1}^{N}M^{(l)}_\pm(x,t)k^l+\mathcal{O}\left(k^{N+1}\right),\quad {\rm{as}}\,\,k\to 0,\,k\in \mathbb{C}_\pm.
        \end{equation*}
        Furthermore, there exist scalar coefficients $\alpha_{11}$, $\delta_{-1}$, $\delta_0$ depending on $x$ and $t$, but not on $k$, such that
        \begin{equation*}
            M^{(-1)}_+(x,t) = \alpha_{11}(x,t) \begin{pmatrix}
					-1 & 0 & 0\\
					0 & 0 & 0\\
					1 & 0 & 0
				\end{pmatrix}+\delta_{-1}(x,t)\begin{pmatrix}
				    0 & 1 & 0\\
                    0 & 0 & 0\\
                    0 & -1 & 0\\
				\end{pmatrix}, \,\,
				M^{(-1)}_-(x,t) = \alpha_{11}(x,t) \begin{pmatrix}
					0 & 0 & -1\\
					0 & 0 & 0\\
					0 & 0 & 1
				\end{pmatrix}+\delta_{-1}(x,t)\begin{pmatrix}
				    0 & -1 & 0\\
                    0 & 0 & 0\\
                    0 & 1 & 0\\
				\end{pmatrix},
        \end{equation*}
        and the third column of $M^{(0)}_+(x,t)$ and the first column of $M^{(0)}_-(x,t)$ are as follows
			\begin{equation*}
			 [M^{(0)}_+(x,t)]_3 = [M^{(0)}_-(x,t)]_1 = \delta_0(x,t) (-1 \, 0 \, 1)^T.
			\end{equation*}

		\item  $M$ satisfies the symmetries for $k \in \mathbb{C} \setminus \mathbb{R}$:
        \begin{align*}
            &M^{-1}(x,t, k) = \mathcal{A}^{-1}(k) M^\dagger(x, t,\bar k) \mathcal{A}(k) ,\\
            &M(x,t, k)=\mathcal{B} M(x,t,-k) \mathcal{B},
        \end{align*}
        where $\mathcal{A}(k),\,\mathcal{B}$ defined in \eqref{sym_AB}.
	\end{enumerate}
	\end{rhp}
   
\begin{lem}\label{lem:M_lax}
For each $(x,t)\in\mathbb{R}\times[0,T)$, the functions
$M_\pm(x,t,k)$ are smooth in $(x,t)$ for
$k\in\overline{\mathbb{C}_\pm}\setminus(\hat{\textbf{Z}}\cup\{0\})$
and satisfy the Lax pair equations \eqref{lax_X}. Moreover, the
matrix-valued function $M(x,t,k)$ is analytic in
$\mathbb{C}\setminus(\mathbb{R}\cup\hat{\textbf{Z}})$ and admits
continuous boundary values on $\mathbb{R}\setminus\{0\}$. These
boundary values satisfy the jump condition \eqref{jump_0} for
$k\in\mathbb{R}\setminus\{0\}$, while $M(x,t,k)$ satisfies the residue
conditions \eqref{residue_M} at each point $p\in\hat{\textbf{Z}}$.
\end{lem}
\begin{proof}
    The proofs that $M(x,t,k)$ satisfies the Lax pair equations and the jump condition on the real axis are standard and follow the arguments in \cite{Lenells-Indiana}; we therefore omit the details. 
    
    In what follows, we focus exclusively on the residue conditions. 
    Following the approach of \cite{CL-JMPA_2023}, we establish the residue conditions associated with the discrete spectral points and the norming constants in the scattering data \eqref{scattering data}. This will complete the verification that the function $M(x,t,k)$ satisfies the $3\times 3$ matrix RH problem \ref{rhp_M}.

    First, we consider the time evolution of the scattering matrix $s(k;t)$.
Let $(\phi_0^{(i)}(\cdot,t),n_0^{(i)}(\cdot,t))$ be two sequences in
$C_c^\infty(\mathbb{R})$ which converge to
$(\phi(\cdot,t),n(\cdot,t))$, and let
$X^{(i)}(x,t,k)$, $Y^{(i)}(x,t,k)$, and $s^{(i)}(k,t)$ denote the
corresponding Jost solutions and scattering matrix, as defined above
with $t$ fixed.

Since $(\phi(\cdot,t),n(\cdot,t))$ is a solution of the YO equation
\eqref{YO}, the matrices $\check{L}(x,t,k)$ and $\check{Z}(x,t,k)$ in
\eqref{lax_psi} satisfy the compatibility condition
\[
\check{L}_t-\check{Z}_x+[\check{L},\check{Z}]=0.
\]
Therefore, one verifies that
\begin{equation}\label{times-1}
X_t^{(i)}-[\mathcal{Z},X^{(i)}]-Z_1X^{(i)},\qquad
Y_t^{(i)}-[\mathcal{Z},Y^{(i)}]-Z_1Y^{(i)},
\end{equation}
satisfy the $x$-part of the Lax pair \eqref{lax_xpart}. Since
$X^{(i)}$ and $Y^{(i)}$ tend to $I$ as $x\to+\infty$ and
$x\to-\infty$, respectively, and since $Z_1\to0$ as
$x\to\pm\infty$, both quantities in \eqref{times-1} vanish
identically. Hence $X^{(i)}$ and $Y^{(i)}$ also satisfy the $t$-part
of the Lax pair \eqref{lax_X}.

It follows that
\[
X^{(i)}(x,t,k)e^{x\mathcal{L}+t\mathcal{Z}}
=
Y^{(i)}(x,t,k)e^{x\mathcal{L}+t\mathcal{Z}}
\bigl(Y^{(i)}(0,0,k)\bigr)^{-1}X^{(i)}(0,0,k)
=
Y^{(i)}(x,t,k)e^{x\mathcal{L}+t\mathcal{Z}}s^{(i)}(k).
\]
Combining this identity with \eqref{X_s}, we obtain
\(
s^{(i)}(k,t)=e^{t\widehat{\mathcal{Z}}}s^{(i)}(k).
\)
Letting $i\to\infty$, we arrive at
\[
s(k,t)=e^{t\widehat{\mathcal{Z}}}s(k),
\]
for $k$ in the domain specified in \eqref{sk_k_define}. The same
argument gives
\(
s^A(k,t)=e^{t\widehat{\mathcal{Z}}}s^A(k).
\)
Consequently, the zeros of $s_{11}(k)$ are independent of the time
evolution. We now turn to the time evolution of the norming constants
$C_{k_j}(t)$.

Notice that, by Lemma \ref{lem_M_pm}, the time-dependent function
$M(x,t,k)$ is defined in exactly the same way as $M(x,k)$, with
$X(x,k)$, $Y(x,k)$, $X^A(x,k)$, and $Y^A(x,k)$ replaced by their
time-dependent counterparts
$X(x,t,k)$, $Y(x,t,k)$, $X^A(x,t,k)$, and $Y^A(x,t,k)$,
respectively.

Suppose that $k_j\in \textbf{Z}\setminus\ri\mathbb{R}_+$. Then
$e^{xl_1(k_j)+tz_1(k_j)}[X(x,t,k_j)]_1$ and $e^{xl_2(k_j)+tz_2(k_j)}m(x,t,k_j)$ (the numerator of the second
column of $M(x,t,k)$) both satisfy the Lax pair \eqref{lax_X}.
Repeating the argument leading from \eqref{M+i} to \eqref{wxkj}, we
find that
\[
\frac{m(x,t,k_j)}{s'_{11}(k_j)}
=C_{k_j}
e^{x(l_1(k_j)-l_2(k_j))
+t(z_1(k_j)-z_2(k_j))}
[X(x,t,k_j)]_1=
C_{k_j}(t)
e^{x(l_1(k_j)-l_2(k_j))}
[X(x,t,k_j)]_1,
\qquad x\in\mathbb{R}.
\]
Since both sides satisfy the same Lax pair, comparing their
$t$-dependence yields
\(
C_{k_j}(t)
=
C_{k_j}
e^{t\bigl(z_1(k_j)-z_2(k_j)\bigr)}.
\)

Suppose that $k_j\in\ri\mathbb{R}$. Since $[M(x,t,k)]_2$ is analytic at
$k=k_j$, and both
$e^{xl_1(k_j)+tz_1(k_j)}[X(x,t,k_j)]_1$ and
$e^{xl_3(k_j)+tz_3(k_j)}[Y(x,t,k_j)]_3$
satisfy the Lax pair \eqref{lax_X}, repeating the argument used in the
derivation of \eqref{wxkj} yields
\[
\frac{\bigl[Y(x,t,k_j)\bigr]_3}{(s^A)'(k_j)}
=
C_{k_j}
e^{x(l_1(k_j)-l_3(k_j))
+t(z_1(k_j)-z_3(k_j))}
\bigl[X(x,t,k_j)\bigr]_1.
\]
Hence, 
\(
C_{k_j}(t)
=
C_{k_j}
e^{t\bigl(z_1(k_j)-z_3(k_j)\bigr)}.
\) This completes the proof of \eqref{residue_M}.
\end{proof}
   
\begin{prop}\label{prop:M-recover} ${\rm (Reconstruction~formula)}$.
Suppose that $(\phi(x,t),n(x,t))$ is a Schwartz-class solution of
\eqref{YOE} on $\mathbb{R}\times[0,T)$ with initial data
$(\phi_0(x),n_0(x))\in\mathcal{S}(\mathbb{R})\times\mathcal{S}(\mathbb{R})$,
where $T\in(0,\infty]$, and assume that Assumptions
\ref{assu_soliton} and \ref{assu_k=0} hold.
Let $M(x,t,k)$ be the solution of the RH problem \ref{rhp_M},
with reflection coefficients $r_1(k)$ and $r_2(k)$ defined by
\eqref{r1_r2}, and norming constants $C_{k_j}$ defined by
\eqref{ckj_non_Cc} and \eqref{ckj_non_Cir} for each
$k_j\in\mathbf Z$.
Then $(\phi(x,t),n(x,t))$ can be recovered from $M(x,t,k)$ via
\begin{equation}\label{recover:M}
\begin{cases}
\begin{aligned}
\phi(x,t)&=2\ri\lim_{k\to\infty}k\,\partial_x\bigl(M_{11}(x,t,k)+M_{31}(x,t,k)-1\bigr),\\
n(x,t)&=\ri\lim_{k\to\infty} k\, M_{21}(x,t,k).
\end{aligned}
\end{cases}
\end{equation}
\end{prop}
\begin{proof}
By Lemma~\ref{lem:M_lax}, the function $M(x,t,k)$ satisfies the Lax pair
\eqref{lax_X}. The reconstruction formula \eqref{recover:M} then follows
immediately from the large-$k$ expansion of $M(x,t,k)$ given in
Lemma~\ref{M_kinfty}.
\end{proof}

\subsection{Construction of $N_1(x,t,k)$ and $N_2(x,t,k)$} 

Define the vector-valued functions $N_1(x,t,k)$ and $N_2(x,t,k)$ by
\begin{equation}\label{def:N1N2}
\begin{aligned}
&N_1(x,t,k)=
\begin{pmatrix}
1 & 0 & 1
\end{pmatrix}
M(x,t,k),\\
&N_2(x,t,k)=
\begin{pmatrix}
0 & 1 & 0
\end{pmatrix}
M(x,t,k).
\end{aligned}
\end{equation}

\begin{prop}\label{prop:N1N2_recover}
Let $T\in(0,\infty]$, and let $n(x,t)$ and $\phi(x,t)$ be Schwartz-class solutions of the YO equation \eqref{YOE} with initial data $n_0(x)$ and $\phi_0(x)$ satisfying Assumptions~\ref{assu_soliton} and \ref{assu_k=0}. Let $M(x,t,k)$ be the solution of the RH problem~\ref{rhp_M}. Then the $1\times 3$ vector-valued meromorphic functions $N_1(x,t,k)$ and $N_2(x,t,k)$, defined by \eqref{def:N1N2}, satisfy the RH problem~\ref{rhp_Nj}. In particular,
\begin{equation*}
\left\{
\begin{aligned}
    n(x,t)
    &=2\ri\frac{\partial}{\partial x}
    \lim_{k\to\infty}
    k\bigl([N_1(x,t,k)]_1-1\bigr),\\
    \phi(x,t)
    &=\ri\lim_{k\to\infty}
    k[N_2(x,t,k)]_1,
\end{aligned}
\right.
\end{equation*}
where $[\cdot]_j$ denotes the $j$-th component of a vector.
\end{prop}
\begin{proof}
The fact that $N_1(x,t,k)$ and $N_2(x,t,k)$ satisfy the RH problem~\ref{rhp_Nj} follows directly from their definitions in \eqref{def:N1N2} and the RH problem~\ref{rhp_M} for $M(x,t,k)$. The reconstruction formula for $n(x,t)$ and $\phi(x,t)$ are immediate consequences of Proposition~\ref{prop:M-recover}.
\end{proof}

\section{The inverse scattering problem}\label{sec:inverse}

In this section, we complete the proof of Theorem~\ref{theo_reconstruction_formula} and derive the pure soliton solutions of the YO equation~\eqref{YO}. We begin by proving a vanishing lemma for the RH problem~\ref{rhp_Nj}, which ensures the existence of unique solution. Subsequently, we investigate the reflectionless case $r_1(k)=r_2(k)=0$ and obtain the corresponding $N$-soliton solutions via RH problem~\ref{rhp_M}.

\subsection{Vanishing lemma}

\subsubsection{The Proof of the Lemma~\ref{lem_unique_solvability}}
    \begin{proof}\label{subsec_vanishing}
We first establish the vanishing lemma in the solitonless case and then extend the analysis to the general case.\
\par
{\bf Vanishing lemma for the solitonless case.}

Define the matrix-valued function $D(k)$ by
\begin{equation}\label{def:D}
    D(k)=\operatorname{diag}(1,\pm 8k,1),
\end{equation}
where the sign is chosen according to the rule that "$+$" corresponds to $\rre  k>0$, while "$-$" corresponds to $\rre  k<0$. 

Introduce the meromorphic function
\begin{equation*}
    H(k)=N(k)D(k)N^\dagger(\bar{k}).
\end{equation*}
Using the jump relations for $N(x,t,k)$ in \eqref{jump_0} and the absence of solitons, one verifies directly that $H(k)$ is analytic in $\mathbb{C}\setminus(\mathbb{R}\cup \ri\mathbb{R})$, satisfies the symmetry relation $H(\bar{k})^\dagger=H(k)$, and obeys the decay condition $H(k)=\mathcal{O}(k^{-2})$ as $k\to\infty$.

Let $\Gamma$ denote the oriented ray from $0$ to $\ri\infty$. The Cauchy-Goursat theorem then yields
\begin{equation}\label{eq:vanish-id-1}
    \int_{\mathbb{R}}H_+(k)\,\rd k+\int_{\Gamma}\bigl(H_+(k)-H_-(k)\bigr)\,\rd k=0,
\end{equation}
where $H_+(k)$ and $H_-(k)$ denote the boundary values of $H(k)$ on the left and right sides of $\Gamma$ and $\mathbb{R}$, respectively.
For $k\in\Gamma$, a direct computation shows that
\begin{equation*}
    H_+(k)-H_-(k)=-16k\,n_2(k)n_2^*(k)=-16k|n_2(k)|^2,
\end{equation*}
where the second equality follows from the symmetry relation \eqref{sym_Nj}. Parametrizing $\Gamma$ by $k=\ri s$, $s\in(0,\infty)$, we obtain
\begin{equation}\label{ineq:vanish-1}
    \int_{\Gamma}\bigl(H_+(k)-H_-(k)\bigr)\,\rd k
    =\int_{0}^{\infty}16s\,|n_2(\ri s)|^2\,\rd s
    \geq 0.
\end{equation}

On the other hand, using the symmetry relation \eqref{sym_Nj} together with the identity $D(k)=\mathcal{B}D(-k)\mathcal{B}$, we obtain
\begin{equation}\label{eq:H-sym-1}
 \begin{aligned}
    \int_{\mathbb{R}} H_+(k)\,\rd k
    &= \int_{0}^{\infty} N_{+}(k)D(k)N_{+}^{\dagger}(\bar{k})\,\rd k
     + \int_{-\infty}^{0} N_{+}(k)D(k)N_{+}^{\dagger}(\bar{k})\,\rd k \\
    &= \int_{0}^{\infty} N_{+}(k)D(k)N_{+}^{\dagger}(\bar{k})\,\rd k
     + \int_{0}^{\infty} N_{+}(k)\mathcal{B}D(-k)\mathcal{B}
       N_{+}^{\dagger}(\bar{k})\,\rd k \\
    &= 2\int_{0}^{\infty} N_{+}(k)D(k)N_{+}^{\dagger}(\bar{k})\,\rd k,
\end{aligned}   
\end{equation}
where $N_{+}(k)$ denotes the boundary value of $N(k)$ from the upper half-plane along $\mathbb{R}$. Moreover, it follows from the jump relation \eqref{jump_0} that
\begin{align*}
    \int_{0}^{\infty}\bigl(H_+(k)+H_+^\dagger(\bar{k})\bigr)\,\rd k
    = \int_{0}^{\infty}
    N_{-}(k)\bigl(v(k)D(k)+D(k)v^\dagger(\bar{k})\bigr)
    N_{-}^{\dagger}(\bar{k})\,\rd k .
\end{align*}
Define
\begin{equation*}
    \mathcal{V}(k):=v(k)D(k)+D(k)v^\dagger(\bar{k})
    =
    \begin{pmatrix}
        2 & -16k\,r_1(k)\re^{t\theta_{12}} & 0\\
        -16k\,r_1^*(k)\re^{-t\theta_{12}}
        & 16k\bigl(1+8k|r_1(k)|^2\bigr) & 0\\
        0 & 0 &
        2\bigl(1-8k|r_1(-k)|^2-|r_2(k)|^2\bigr)
    \end{pmatrix},
    \qquad k\in(0,\infty).
\end{equation*}
The leading principal minors $\mathcal{V}_j(k)$, $j=1,2,3$, of $\mathcal{V}(k)$ are given by
\begin{equation*}
    \mathcal{V}_1(k)=2,\qquad
    \mathcal{V}_2(k)=32k,\qquad
    \mathcal{V}_3(k)=64k\bigl(1-8k|r_1(-k)|^2-|r_2(k)|^2\bigr).
\end{equation*}
By Assumption~\ref{assu_LT}, all leading principal minors of $\mathcal{V}(k)$ are positive for $k\in(0,\infty)$. Hence, $\mathcal{V}(k)$ is positive definite for $k\in(0,\infty)$, which implies that
\begin{equation}\label{ineq:vanish-2}
     \int_{0}^{\infty}\bigl(H(k)+H^\dagger(\bar{k})\bigr)\,\rd k\geq0.
\end{equation}

Combining \eqref{eq:H-sym-1} with the symmetry relation $H(\bar{k})^\dagger=H(k)$, we can rewrite \eqref{eq:vanish-id-1} as
\begin{equation*}
    \int_{0}^{\infty}\bigl(H(k)+H^\dagger(\bar{k})\bigr)\,\rd k
    +\int_{\Gamma}\bigl(H_+(k)-H_-(k)\bigr)\,\rd k
    =0.
\end{equation*}
Since both terms are nonnegative by \eqref{ineq:vanish-1} and \eqref{ineq:vanish-2}, they must both vanish. In particular,
\[
\int_{0}^{\infty}
N_{-}(k)\mathcal{V}(k)N_{-}^{\dagger}(\bar{k})\,\rd k=0.
\]
As $\mathcal{V}(k)$ is positive definite for $k\in(0,\infty)$, we deduce that $N_{-}(k)\equiv0$ on $(0,\infty)$. By analyticity of $N(k)$ in $\mathbb{C}\setminus\mathbb{R}$, the identity theorem then yields $N(k)\equiv0$.

\begin{figure}[H]
    \centering

    \includegraphics[width=0.5\textwidth]{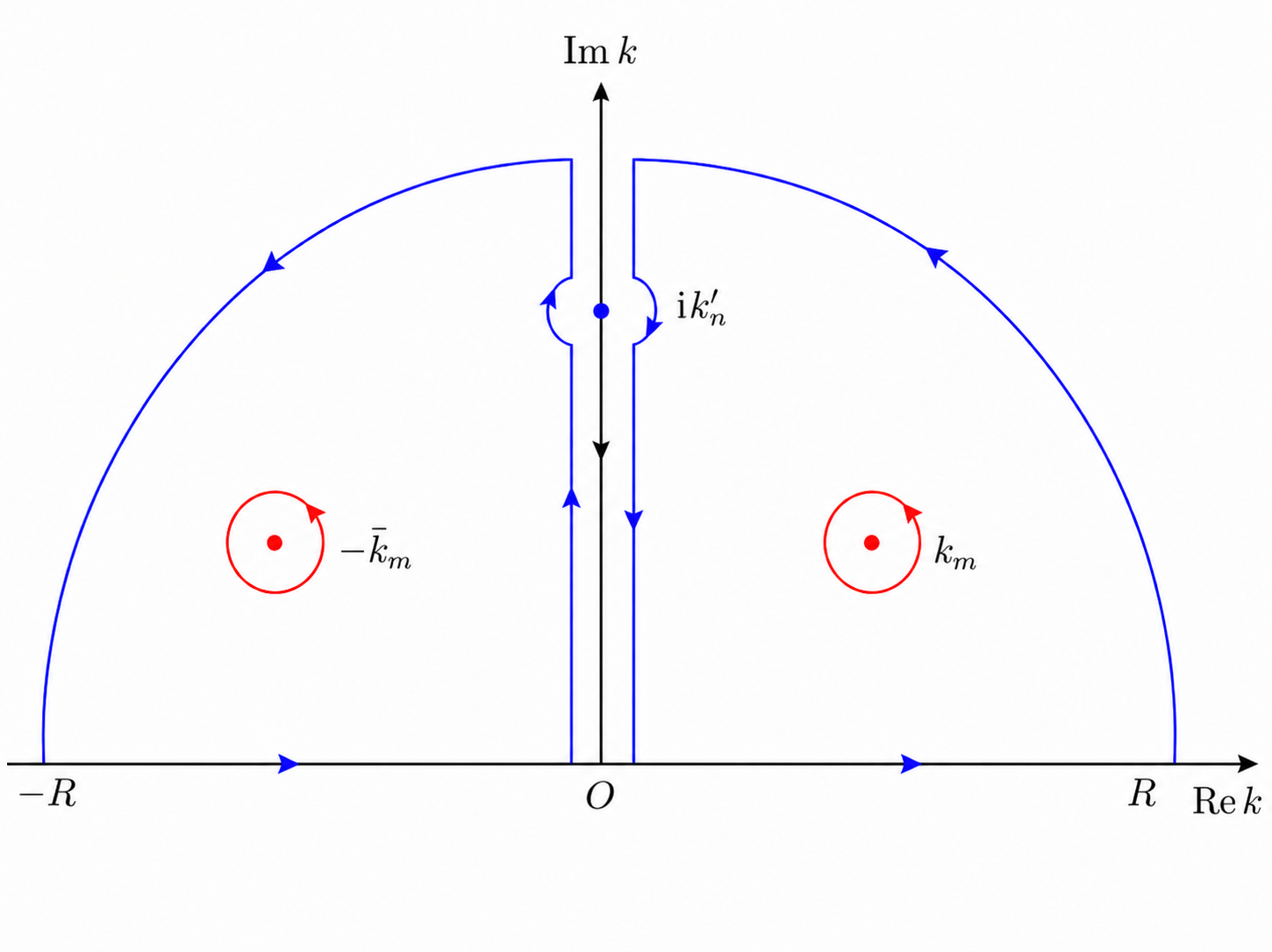}
    \caption{\footnotesize The contour deformation used in the proof of the vanishing lemma for the soliton case. The boundaries of the disks centered at the discrete spectral points off the imaginary axis are oriented counterclockwise, whereas the two semicircular boundaries surrounding each $k_n'\in \ri\mathbb{R}_{+}$ are oriented clockwise. }
    \label{fig:vanish}
\end{figure}\
\par
{\bf Vanishing lemma for the soliton case.}

Before establishing the vanishing lemma for RH problems~\ref{rhp_Nj} in the presence of solitons, we introduce a transformed function $\tilde{N}(k)$. Let $\mathbb{D}_{\epsilon}(k_0)$ denote the open disk centered at $k_0$ with radius $\epsilon$, namely, $\mathbb{D}_{\epsilon}(k_0):=\{k\in\mathbb{C}: |k-k_0|<\epsilon\}$. Furthermore, let $k_m\in\mathbf{Z}\setminus \ri\mathbb{R}_{+}$ denote the simple zeros of $s_{11}(k)$ lying in the first quadrant of the complex plane, and let $k_n'\in\ri\mathbb{R}_{+}$ denote the simple zeros of $s_{11}(k)$ located on the positive imaginary axis. To eliminate the pole conditions from the RH problem, we introduce the transformed meromorphic function $\tilde{N}(x,t,k)$ defined by
\begin{equation*}
    \tilde{N}(x,t,k):=\begin{cases}
        \begin{aligned}
            &N(x,t,k)Q_m^{(1)}(x,t,k),&&k\in\mathbb{D}_{\epsilon}(k_m),\\
            &N(x,t,k)Q_m^{(2)}(x,t,k),&&k\in\mathbb{D}_{\epsilon}(-\bar{k}_m),\\
            &N(x,t,k)Q_m^{(3)}(x,t,k),&&k\in\mathbb{D}_{\epsilon}(-k_m),\\
            &N(x,t,k)Q_m^{(4)}(x,t,k),&&k\in\mathbb{D}_{\epsilon}(\bar{k}_m),\\
            &N(x,t,k)P_n^{(1)}(x,t,k),&&k\in\mathbb{D}_{\epsilon}(k_n'),\\
            &N(x,t,k)P_n^{(2)}(x,t,k),&&k\in\mathbb{D}_{\epsilon}(\bar{k}_n'),\\
            &N(x,t,k),&&\text{otherwise},
        \end{aligned}
    \end{cases}
\end{equation*}
where
$$
Q_m^{(1)}(x,t,k):=\begin{pmatrix}
    1&-\frac{C_{k_m}\re^{t\theta_{12}}}{k-k_m}&0\\
    0&1&0\\
    0&0&1
\end{pmatrix},\quad
Q_m^{(2)}(x,t,k):=\begin{pmatrix}
    1&0&0\\
    0&1&-\frac{8k\overline{C_{k_m}}\re^{t\theta_{23}}}{k+\overline{k_m}}\\
    0&0&1
\end{pmatrix},\quad
Q_m^{(3)}(x,t,k):=\begin{pmatrix}
    1&0&0\\
    0&1&0\\
    0&\frac{C_{k_m}\re^{t\theta_{32}}}{k+k_m}&1
\end{pmatrix},
$$
and
$$
Q_m^{(4)}(x,t,k):=\begin{pmatrix}
    1&0&0\\
    \frac{8k\overline{C_{k_m}}\re^{t\theta_{21}}}{k-\overline{k_m}}&1&0\\
    0&0&1
\end{pmatrix},\quad
P_n^{(1)}(x,t,k):=\begin{pmatrix}
    1&0&-\frac{C_{k_n'}\re^{t\theta_{13}}}{k-k_n'}\\
    0&1&0\\
    0&0&1
\end{pmatrix},\quad
P_n^{(2)}(x,t,k):=\begin{pmatrix}
    1&0&0\\
    0&1&0\\
    \frac{C_{k_n'}
    \re^{t\theta_{31}}}{k-\overline{k_n'}}&0&1
\end{pmatrix}. 
$$
Thus, $\tilde{N}(k)$, which henceforth denotes $\tilde{N}(x,t,k)$, satisfies a RH problem in which the pole conditions are removed and replaced by additional jump conditions on the boundaries of the disks; see Figure~\ref{fig:vanish}. For convenience, the boundaries of the disks off the imaginary axis are oriented counterclockwise, whereas the boundaries of the disks intersecting the imaginary axis are oriented clockwise. Furthermore, the symmetry relation
\begin{equation}\label{symmetry:tildeN}
    \tilde{N}(x,t,k)=\tilde{N}(x,t,-k)\mathcal{B},
\end{equation}
follows directly from the corresponding symmetry of $N(x,t,k)$ and the definitions of the transformation matrices.

Following the same argument as in the solitonless case, we introduce the meromorphic function
\[
\tilde{H}(k)=\tilde{N}(k)D(k)\tilde{N}^{\dagger}(\bar{k}),
\]
where $D(k)$ is defined in \eqref{def:D}. It follows immediately that $\tilde{H}(k)$ is analytic in
\(
\mathbb{C}\setminus\Bigl(\mathbb{R}\bigcup\ri\mathbb{R}\bigcup\Bigl(\underset{m}{\cup}\partial\mathbb{D}_{\epsilon}(k_m)\Bigr)\bigcup\Bigl(\underset{n}{\cup}\partial\mathbb{D}_{\epsilon}(k_n')\Bigr)\Bigr).
\)

We next examine the jump conditions satisfied by $\tilde{H}(k)$ on the boundaries of the disks off the imaginary axis.
\begin{itemize}
    \item For $k\in\partial\mathbb{D}_{\epsilon}(k_m)$, a direct computation shows that
\begin{equation*}
\begin{aligned}
    \tilde{H}_+(k)-\tilde{H}_-(k)
    &=\tilde{N}_+(k)D(k)\tilde{N}_+^{\dagger}(\bar{k})
      -\tilde{N}_-(k)D(k)\tilde{N}_-^{\dagger}(\bar{k}) 
    =\tilde{N}_-(k)
      \Bigl[
      Q_m^{(1)}(k)D(k)\bigl(Q_m^{(4)}(\bar{k})\bigr)^{\dagger}
      -D(k)
      \Bigr]
      \tilde{N}_-^{\dagger}(\bar{k}) 
    =0.
\end{aligned}
\end{equation*}

    \item For $k\in\partial\mathbb{D}_{\epsilon}(-\bar{k}_m)$, a direct computation shows that
\begin{equation*}
\begin{aligned}
    \tilde{H}_+(k)-\tilde{H}_-(k)
    &=\tilde{N}_+(k)D(k)\tilde{N}_+^{\dagger}(\bar{k})
      -\tilde{N}_-(k)D(k)\tilde{N}_-^{\dagger}(\bar{k}) 
    =\tilde{N}_-(k)
      \Bigl[
      Q_m^{(2)}(k)D(k)\bigl(Q_m^{(3)}(\bar{k})\bigr)^{\dagger}
      -D(k)
      \Bigr]
      \tilde{N}_-^{\dagger}(\bar{k}) 
    =0.
\end{aligned}
\end{equation*}
\end{itemize}

By direct computation, or equivalently by the symmetry relations, one also verifies that $\tilde{H}_+(k)-\tilde{H}_-(k)=0$ for $k\in\partial\mathbb{D}_{\epsilon}(-k_m)\cup\partial\mathbb{D}_{\epsilon}(\bar{k}_m)$. Consequently, $\tilde{H}(k)$ extends analytically across the boundaries of all disks off the imaginary axis.

Regarding the disks intersecting the positive imaginary axis, we obtain the following identities.
\begin{itemize}
    \item For $k\in\mathbb{D}_{\epsilon}(k'_n)$ with $\rre  k>0$, a direct computation shows that
    \begin{equation}\label{eq:vanish-soliton-1}
    \begin{aligned}
    \tilde{H}_+(k)-\tilde{H}_-(k)&=\tilde{N}_+(k)D(k)\tilde{N}_+^{\dagger}(\bar{k})
      -\tilde{N}_-(k)D(k)\tilde{N}_-^{\dagger}(\bar{k})=\tilde{N}_-(k)
      \Bigl[
      P_n^{(1)}(k)D(k)\bigl(P_n^{(2)}(\bar{k})\bigr)^{\dagger}
      -D(k)
      \Bigr]
      \tilde{N}_-^{\dagger}(\bar{k}) \\
      &=\frac{(\overline{C_{k_n'}}-{C_{k_n'}})\re^{t\theta_{13}}}{k-k_n'}n_1(k)n_3^*(k)
      =-\frac{2C_{k_n'}\re^{t\theta_{13}}}{k-k_n'}|n_1(k)|^2.
    \end{aligned}
    \end{equation}

    \item For $k\in\mathbb{D}_{\epsilon}(k'_n)$ with $\rre  k<0$, a direct computation shows that
    \begin{equation}\label{eq:vanish-soliton-2}
    \begin{aligned}
    \tilde{H}_+(k)-\tilde{H}_-(k)&=\tilde{N}_+(k)D(k)\tilde{N}_+^{\dagger}(\bar{k})
      -\tilde{N}_-(k)D(k)\tilde{N}_-^{\dagger}(\bar{k})=\tilde{N}_-(k)
      \Bigl[
      P_n^{(1)}(k)D(k)\bigl(P_n^{(2)}(\bar{k})\bigr)^{\dagger}
      -D(k)
      \Bigr]
      \tilde{N}_-^{\dagger}(\bar{k}) \\
      &=\frac{(\overline{C_{k_n'}}-{C_{k_n'}})\re^{t\theta_{13}}}{k-k_n'}n_1(k)n_3^*(k)
      =-\frac{2C_{k_n'}\re^{t\theta_{13}}}{k-k_n'}|n_1(k)|^2,
    \end{aligned}
    \end{equation}
\end{itemize}
where the last equality follows from the fact that $C_{k_n'}$ is purely imaginary (see Lemma~\ref{lem_Ck1_im}) and the symmetry relation \eqref{symmetry:tildeN} satisfied by $\tilde{N}(k)$.

To account for the disks centered at the poles on the positive imaginary axis, we deform the contour $\Gamma$ by removing the segments $(k_n'-\ri\epsilon,k_n'+\ri\epsilon)$ and replacing them with the corresponding semicircular boundaries. Applying the Cauchy-Goursat theorem to the resulting contour and using the analyticity of $\tilde{H}(k)$ in $\mathbb{C}_{+}\setminus\Bigl(\ri\mathbb{R}_{+}\bigcup\Bigl(\underset{n}{\cup}\partial\mathbb{D}_{\epsilon}(k_n')\Bigr)\Bigr)$, we obtain
\begin{equation}\label{eq:H-upper}
\begin{aligned}
\int_{\mathbb{R}}\tilde{H}_+(k)\,\rd k
+\int_{\Gamma\setminus\underset{n}{\cup}(k_n'-\ri\epsilon,k_n'+\ri\epsilon)}
\bigl(\tilde{H}_+(k)-\tilde{H}_-(k)\bigr)\,\rd k 
+\int_{\underset{n}{\cup}\partial\mathbb{D}_{\epsilon}(k_n')\cap\{\rre  k>0\}}
\tilde{H}_+(k)\,\rd k
+\int_{\underset{n}{\cup}\partial\mathbb{D}_{\epsilon}(k_n')\cap\{\rre  k<0\}}
\tilde{H}_+(k)\,\rd k=0,
\end{aligned}
\end{equation}
and
\begin{equation}\label{eq:H-disks}
\begin{aligned}
\int_{\underset{n}{\cup}(k_n'-\ri\epsilon,k_n'+\ri\epsilon)}
\bigl(\tilde{H}_+(k)-\tilde{H}_-(k)\bigr)\,\rd k 
-\int_{\underset{n}{\cup}\partial\mathbb{D}_{\epsilon}(k_n')\cap\{\rre  k>0\}}
\tilde{H}_-(k)\,\rd k
-\int_{\underset{n}{\cup}\partial\mathbb{D}_{\epsilon}(k_n')\cap\{\rre  k<0\}}
\tilde{H}_-(k)\,\rd k=0.
\end{aligned}
\end{equation}

Adding \eqref{eq:H-upper} and \eqref{eq:H-disks}, and using \eqref{eq:vanish-soliton-1} and \eqref{eq:vanish-soliton-2}, we obtain
\begin{equation*}
    \int_{\mathbb{R}}\tilde{H}_+(k)\,\rd k
    +\int_{\Gamma}
    \bigl(\tilde{H}_+(k)-\tilde{H}_-(k)\bigr)\,\rd k
    -\sum_n 4\pi \ri C_{k_n'}\re^{t\theta_{13}(k_n')}
    |n_1(k_n')|^2
    =0.
\end{equation*}
Since $\re^{t\theta_{13}(k_n')}>0$ for $k_n'\in\ri\mathbb{R}_{+}$ and Assumption~\ref{assu_soliton} implies that $ \ri C_{k_n'}\le 0$, it follows that
\[
\sum_n 4\pi  \ri C_{k_n'}\re^{t\theta_{13}(k_n')}
|n_1(k_n')|^2\le0.
\]
Therefore, by the same argument as in the solitonless case, we conclude that $\tilde{N}(k)$ vanishes identically whenever Assumption~\ref{assu_soliton} holds. The unique solvability of RH problems~\ref{rhp_Nj} then follows from the standard Fredholm theory; see Zhou~\cite{Zhou_SIAM_1989_vanishing}.
    \end{proof}

    \subsubsection{Proof of Theorem \ref{theo_reconstruction_formula}}\label{subsec_reconst}
	\begin{proof}
Suppose that the initial data $n_0(x)$ and $\phi_0(x)$ belong to the Schwartz class, and that the associated reflection coefficients $r_1(k)$ and $r_2(k)$ satisfy Assumptions~\ref{assu_LT} and \ref{assu_soliton}. By the definition of $N_1(x,t,k)$ and $N_2(x,t,k)$ in \eqref{def:N1N2}, together with the vanishing lemma established in Lemma~\ref{lem_unique_solvability}, RH problem~\ref{rhp_Nj} admits a unique solution for all $(x,t)$. Consequently, the functions $N_1(x,t,k)$ and $N_2(x,t,k)$ are uniquely determined. It then follows from Proposition~\ref{prop:N1N2_recover} that the potentials $n(x,t)$ and $\phi(x,t)$ can be uniquely reconstructed. Therefore, the solution of the YO equation exists and is unique, which completes the proof of Theorem~\ref{theo_reconstruction_formula}.
\end{proof}

	\subsection{Pure soliton solution}\label{sec:soliton}

	This subsection considers the reflectionless case, i.e., where \( r_1(k) = r_2(k) = 0 \) holds identically for \( k \in \mathbb{R} \), and the scattering data \(\bigl\{ \mathbf{Z}, \, \{C_{k_j}\}_{k_j \in \mathbf{Z}} \bigr\}\) satisfies Assumption \ref{assu_soliton}. This setting yields the pure soliton solutions of the YO equation \eqref{YO}. Let \( k_j = \xi_j + \ri \eta_j \in \mathbf{Z} \) (\( j = 1, 2, \dots, N \)) be the \(N\) simple zeros of the spectral function \( s_{11}(k) \). We then treat the following two cases separately and present the explicit expression of the one-soliton solution for the YO equation \eqref{YO}, and the result for the $N$-soliton solution is given in Appendix \ref{sec_Nsoliton}.
	
	First, note that the function $M(x,t,k)$ in the RH problem \ref{rhp_M} exhibits a singularity at $k = 0$. We therefore introduce the following transformation:
	\begin{equation}\label{tildeM_def}
		M(x,t,k)=\left[ I+\frac{y_1(x,t)}{k}\begin{pmatrix}
			-1 & 0 & -1\\
			0 & 0 & 0\\
			1 & 0 & 1\\
		\end{pmatrix}+\frac{y_2(x,t)}{k}\begin{pmatrix}
			0 & -1 & 0\\
			0 & 0 & 0\\
			0 & 1 & 0\\
		\end{pmatrix}\right] \tilde{M}(x,t,k),
	\end{equation}
	where $y_1(x,t)=\lim\limits_{k\to \infty}k\tilde{M}_{13}$ and  $y_2(x,t)=\lim\limits_{k\to \infty}k\tilde{M}_{12}$. Thus, it can be proven that $\tilde{M}(x,t,k)$ has no singularity at $k = 0$.

     \subsubsection{The Case $k_j \in \textbf{\textbf{Z}} \setminus \ri\mathbb{R}_+$}
     In this case, we consider the spectral point $k_j \in \textbf{\textbf{Z}} \setminus \ri\mathbb{R}_+$ in RH problem \ref{rhp_M}. Using the residue conditions \eqref{Res_2} associated with the eigenfunction $M(x,t,k)$. The function $\tilde{M}(x,t,k)$ can be expressed in the following form:
	\begin{equation}\label{tildeMk}
		\begin{aligned}
			\tilde{M}(x,t,k)=& I-\frac{8\overline{k_1}\,\,\overline{C_{k_1} }\re^{-2\ri \overline{k_1}((\overline{k_1}-1)t+x)}}{k-\overline{k_1}}\begin{pmatrix}
				 \tilde{M}_{12}\left(x,t,\overline{k_1}\right) & 0 & 0\\
				 \tilde{M}_{22}\left(x,t,\overline{k_1}\right) & 0 & 0\\
				 \tilde{M}_{32}\left(x,t,\overline{k_1}\right) & 0 & 0\\
			\end{pmatrix}+\frac{C_{k_1} \re^{2\ri k_1  \left( \left(k_1 -1 \right)t+x\right)  }}{k-k_1 } \begin{pmatrix}
				0 & \tilde{M}_{11}(x,t,k_1 ) & 0\\
				0 & \tilde{M}_{21}(x,t,k_1 ) & 0\\
				0 & \tilde{M}_{31}(x,t,k_1 ) & 0\\
			\end{pmatrix}\\
            &-\frac{C_{k_1} \re^{2\ri k_1 \left( \left(k_1 -1 \right)t+x\right)  }}{k+k_1 } \begin{pmatrix}
			0 & \tilde{M}_{13}(x,t,-k_1 ) & 0\\
			0 & \tilde{M}_{23}(x,t,-k_1 ) & 0\\
			0 & \tilde{M}_{33}(x,t,-k_1 ) & 0\\
			\end{pmatrix}+\frac{8\overline{k_1}\,\,\overline{C_{k_1} }\re^{-2\ri \overline{k_1}((\overline{k_1}-1)t+x)}}{k+\overline{k_1}}\begin{pmatrix}
				0 & 0 & \tilde{M}_{12}\left(x,t,-\overline{k_1}\right)\\
				0 & 0 & \tilde{M}_{22}\left(x,t,-\overline{k_1}\right)\\
				0 & 0 & \tilde{M}_{32}\left(x,t,-\overline{k_1}\right)\\
			\end{pmatrix}.
		\end{aligned}
	\end{equation}
	Exploiting the symmetry properties, the functions $y_j(x,t)$, $j=1,2$, can be expressed as follows:
	\begin{equation*}
		\begin{aligned}
				&y_1(x,t)=8\overline{k_1}\,\,\overline{C_{k_1} }\re^{-2\ri \overline{k_1}((\overline{k_1}-1)t+x)}\tilde{M}_{32}\left(x,t,\overline{k_1}\right),\\ &y_2(x,t)=C_{k_1} \re^{2\ri k_1  \left( \left(k_1 -1 \right)t+x\right)  }\left(\tilde{M}_{11}(x,t,k_1 )- \tilde{M}_{31}(x,t,k_1 )\right) .
		\end{aligned}
	\end{equation*}
	The reconstructed formula \eqref{reconstruct} can be rewritten as
	\begin{equation*}
		\left\lbrace 
		\begin{aligned}
			&n(x,t)=-2\ri\frac{\partial}{\partial x}\left[8\overline{k_1}\,\,\overline{C_{k_1} }\re^{-2\ri \overline{k_1}((\overline{k_1}-1)t+x)} \left( \tilde{M}_{12}\left(x,t,\overline{k_1}\right)+\tilde{M}_{32}\left(x,t,\overline{k_1}\right)\right)\right]  ,\\
			&\phi(x,t)=-8\ri \overline{k_1}\,\,\overline{C_{k_1} }\re^{-2\ri \overline{k_1}((\overline{k_1}-1)t+x)}\tilde{M}_{22}\left(x,t,\overline{k_1}\right) . 
		\end{aligned}\right.
	\end{equation*}
    Furthermore, by taking different values of $k$ in relation \eqref{tildeMk}, we can compute the following:
	\begin{equation*}
		\left\lbrace 
		\begin{aligned}
			n(x,t)=&-\frac{\partial}{\partial x}\left(  \frac{32\ri\left| k_1 C_{k_1} \right| ^2\left( \overline{k_1}-k_1 \right) \re^{2\ri ( k_1 -\overline{k_1}) ( (k_1 +\overline{k_1}-1 )t+x)} }{\left(\overline{k_1}-k_1  \right)^2\left( \overline{k_1}+k_1 \right)-16 \left|k_1 C_{k_1}   \right| ^2\re^{2\ri ( k_1 -\overline{k_1}) ( (k_1 +\overline{k_1}-1 )t+x)}  }\right)   \\
            =&-4\eta_1^2 \text{sech}^2\left(2\eta_1 \left(x-\left(1-2\xi_1\right)t\right)+\rho_1\right),\\
			\phi(x,t)=&-\frac{8\ri \overline{k_1}\,\,\overline{C_{k_1} }{\left(\overline{k_1}-k_1  \right)^2 \left(\overline{k_1}+k_1  \right) }\re^{-2\ri \overline{k_1}((\overline{k_1}-1)t+x)}}{{\left(\overline{k_1}-k_1  \right)^2 \left(\overline{k_1}+k_1  \right) }-{16\left|   k_1 C_{k_1}\right| ^2\re^{2\ri \left( k_1 -\overline{k_1}\right) \left( \left(k_1 +\overline{k_1}-1 \right)t+x\right)} }}\\
            =&-2\eta_1\sqrt{2\xi_1}\text{sech}\left(2\eta_1 \left(x-\left(1-2\xi_1\right)t\right)+\rho_1\right)\re^{2\ri \left(\eta_1^2-\xi_1^2\right)t-2\ri\xi_1\left(x-t\right)+\ri \theta},
		\end{aligned}\right.
	\end{equation*}
	where $
        \rho_1=\log \frac{\sqrt{2\xi_1}\eta_1}{2\left|k_1C_{k_1} \right|}$, $\re^{\ri \theta}=\ri\,\,\frac{\overline{k_1 }\,\,\overline{C_{k_1} }}{|k_1C_{k_1}   |}.$ It can be verified that the form of the one-soliton solution given here is equivalent to that presented in Ref. \cite{YO_1976}.
        
    \begin{remark}\label{remark_singular_soliton}
    	The calculation above shows that for $\xi_1<0$, the solution is given by the \text{cosech} function; specifically, when $k_1\in D_{\text{sing}}$, the corresponding exact solution is singular.
    \end{remark}

	\subsubsection{The Case $k_j \in \textbf{Z} \cap\ri\mathbb{R}_+$}
     When $k_j \in \textbf{Z} \cap\ri\mathbb{R}_+$, $j=1,2,\cdots,N$, the relation $k_j=-\overline{k_j}=\ri\eta_j$ holds. Therefore, in the RH problem \ref{rhp_M}, the residues of $M(x,t,k)$ at $\pm k_j$ are governed by specific residue formulas. We will also show that, in this case, the two variables associated with the YO equation \eqref{YO} are such that the real variable $n(x,t)$ corresponds to the soliton defined by a $\text{sech}^2$ function, while the complex variable $\phi(x,t)$ remains identically zero. We continue to employ transformation \eqref{tildeM_def} to define the eigenfunction $\tilde{M}$ that is nonsingular at $k=0$.
    For $N=1$, we have
    \begin{equation*}
    	\begin{aligned}
    		\tilde{M}(x,t,k)&=I+\frac{C_{\ri\eta_1} \re^{-4\eta_1(x-t)}}{k-\ri\eta_1 } \begin{pmatrix}
    			0 & 0 & \tilde{M}_{11}(x,t,\ri\eta_1 )\\
    			0 & 0 & \tilde{M}_{21}(x,t,\ri\eta_1 )\\
    			0 & 0 & \tilde{M}_{31}(x,t,\ri\eta_1 )\\
    		\end{pmatrix}-\frac{C_{\ri\eta_1} \re^{-4\eta_1(x-t)}}{k+\ri\eta_1 } \begin{pmatrix}
    			\tilde{M}_{31}(x,t,\ri\eta_1 ) & 0 & 0\\
    			\tilde{M}_{21}(x,t,\ri\eta_1 ) & 0 & 0\\
    			\tilde{M}_{11}(x,t,\ri\eta_1 ) & 0 & 0\\
    		\end{pmatrix}.
    	\end{aligned}
    \end{equation*}
    A similar calculation yields:
    \begin{equation*}\label{reconstruct_soliton2}
    	\left\lbrace 
    	\begin{aligned}
    		n(x,t)&=-2\ri\frac{\partial}{\partial x}\left[C_{\ri\eta_1} \re^{-4\eta_1(x-t)} \left( \tilde{M}_{31}(x,t,\ri\eta_1 )+\tilde{M}_{11}(x,t,\ri\eta_1 )\right)\right] =-4\eta_1^2\text{sech}^2(2\eta_1(x-t)+\rho_1),\\
    		\phi(x,t)&=-\ri C_{\ri\eta_1} \re^{-4\eta_1(x-t)} \tilde{M}_{21}(x,t,\ri\eta_1 )=0, 
    	\end{aligned}\right.
    \end{equation*}
    where $\rho_1=\frac{1}{2}\log\left(-\frac{2 \eta_1}{\ri C_{\ri\eta_1}} \right) $.

    \begin{remark}
    	 Note that in Assumption \ref{assu_soliton}, we have already stipulated the condition $\ri C_{k_j}\notin (0,\infty)$ for $k_j\in \textbf{Z}\cap\ri\mathbb{R}_+$. In fact, we have previously mentioned that when $k_1\in \textbf{Z}\cap\ri\mathbb{R}_+$, the soliton solution of the YO equation \eqref{YO} is non-singular if and only if the norming constant $C_{k_1}$ satisfies condition $\ri C_{k_1}<0$.
    \end{remark}

	\section{The long-time asymptotic analysis}\label{sec_LT}
	Based on the RH problem \ref{rhp_M} constructed in Subsection \ref{sec_M}, this section analyzes the long-time asymptotic behavior of the solutions to the YO problem \eqref{YOE} in different regions of the $(x,t)$-space (see Figure \ref{fig_asy}), and provides the leading-order asymptotic terms and error terms of the solutions. This section is developed under the solitonless Assumption \ref{assu_solitonless}. Specifically, the condition \textbf{3} in the RH problem \ref{rhp_M} is not considered here. The long-time asymptotic analysis with solitons included will be presented in our follow-up work. 
	 
	In the following analysis, we will first focus on the results in Regions \text{II} and \text{III}, and then use the analytical procedures from these two regions to briefly establish the conclusions for Regions \text{I} and \text{IV}. Define two parameters $\xi:=x/t$ and $\tau:=t/x$. The four regions shown in Figure \ref{fig_asy} can be determined by the ranges of these parameters as follows: Region \text{I} corresponds to $\tau\in\mathcal{I}_1$, Region \text{II} corresponds to $\xi\in\mathcal{I}_2$, Region \text{III} corresponds to $\xi\in\mathcal{I}_3$, and Region \text{IV} corresponds to $\tau\in\mathcal{I}_4$, where $\mathcal{I}_j~(j=1,2,3,4)$ have been defined in Subsection \ref{Notations-1.1}.

    \subsection{Long-time asymptotics in Region {\rm{II}}}\label{subsec_region1}
        In Subsection \ref{sec_M}, we construct the RH problem for the initial problem \eqref{YOE} of the YO equation and provide the jump relations satisfied by the eigenfunction $M(x,t,k)$. To analyze the long-time asymptotic behavior of the solutions of \eqref{YOE}, we first consider the dispersion relations:
      \begin{align}\label{theta}
            &\theta_{12}(\zeta,k):=2\ri k\left(k+\zeta-1\right),\quad\theta_{13}(\zeta,k):=4\ri k\left(\zeta-1\right),\quad \theta_{23}(\zeta,k):=-2\ri k\left(k+1-\zeta\right),
        \end{align}
        where $\zeta\in\mathcal{I}_2$ and $t>0$. 
    After some calculations, $\theta_{23}(\zeta,k)$ has a single zero $k_0=(\zeta-1)/2$ and $\theta_{12}(\zeta,k)$ has a single zero $-k_0=(1-\zeta)/2$. The signature tables for $\theta_{12}$, $\theta_{13}$ and $\theta_{23}$ are shown in Figure \ref{fig_theta_sign1} for $(x,t)$ in Region \rm{II}.

    \begin{figure}[htbp]
    \centering
    \begin{subfigure}[t]{0.28\textwidth}
        \centering
     	\begin{tikzpicture} [scale=0.7]
        
         \fill[black!20!blue!20] (-3,0) -- (-1.5,0) -- (-1.5,3) -- (-3,3) -- cycle;
        \fill[black!20!blue!20] (-1.5,0) -- (3,0) -- (3,-3) -- (-1.5,-3) -- cycle;
	 		
	  \draw [very thick,black!20!blue](-3,0) -- (3,0);
        \draw [very thick,black!20!blue](-1.5,-3) -- (-1.5,3);

        \fill (-1.5,0) circle (1.5pt);
        \fill (0,0) circle (1.5pt);
	  \fill (1.5,0) circle (1.5pt);

        \node[below] at (-2,0) {$-k_0$};
	  \node[below] at (0,0) {$\text{0}$};
	  \node[below] at (1.5,0) {$k_0$};

      \node at (-2.4,2) {$U_2^{(12)}$};
      \node at (1,2) {$U_1^{(12)}$};
      \node at (1,-2) {$U_4^{(12)}$};
      \node at (-2.4,-2) {$U_3^{(12)}$};

        \node[right] at (3,0) {$\mathbb{R}$};
     	\end{tikzpicture}
     		\caption{The signature of $\theta_{12}$.}
    \end{subfigure}
    \quad
    \begin{subfigure}[t]{0.28\textwidth}
        \centering
     	\begin{tikzpicture} [scale=0.7]
        
         \fill[black!20!blue!20] (-3,0) -- (3,0) -- (3,-3) -- (-3,-3) -- cycle;
	 		
	  \draw [very thick,black!20!blue](-3,0) -- (3,0);

        \fill (-1.5,0) circle (1.5pt);
        \fill (0,0) circle (1.5pt);
	  \fill (1.5,0) circle (1.5pt);

        \node[above] at (-1.5,0) {$-k_0$};
	  \node[above] at (0,0) {$\text{0}$};
	  \node[above] at (1.5,0) {$k_0$};
      
      \node at (0,2) {$U_1^{(13)}$};
      \node at (0,-2) {$U_2^{(13)}$};

        \node[right] at (3,0) {$\mathbb{R}$};
     	\end{tikzpicture}
     		\caption{The signature of $\theta_{13}$.\label{subfig_theta13_1}}
    \end{subfigure}
    \quad
    \begin{subfigure}[t]{0.28\textwidth}
        \centering
     	\begin{tikzpicture} [scale=0.7]
        
         \fill[black!20!blue!20] (3,0) -- (1.5,0) -- (1.5,3) -- (3,3) -- cycle;
        \fill[black!20!blue!20] (1.5,0) -- (-3,0) -- (-3,-3) -- (1.5,-3) -- cycle;
	 		
	  \draw [very thick,black!20!blue](-3,0) -- (3,0);
        \draw [very thick,black!20!blue](1.5,-3) -- (1.5,3);

        \fill (-1.5,0) circle (1.5pt);
        \fill (0,0) circle (1.5pt);
	  \fill (1.5,0) circle (1.5pt);

        \node[below] at (-1.5,0) {$-k_0$};
	  \node[below] at (0,0) {$\text{0}$};
	  \node[below] at (1.8,0) {$k_0$};

      \node at (2.4,2) {$U_1^{(23)}$};
      \node at (-1,2) {$U_2^{(23)}$};
      \node at (-1,-2) {$U_3^{(23)}$};
      \node at (2.4,-2) {$U_4^{(23)}$};

        \node[right] at (3,0) {$\mathbb{R}$};
     	\end{tikzpicture}
     		\caption{The signature of $\theta_{23}$.}
             \end{subfigure}
             \caption{ Open sets in the complex $k$-plane for Region \rm{II}: $\rre \theta_{ij}>0$ (shaded) and $\rre \theta_{ij}<0$ (white).}
             \label{fig_theta_sign1}
\end{figure}

\subsubsection{First transformation}\label{subsec_611}
In Figure \ref{fig_theta_sign1}, we introduce several open sets $U_k^{(ij)}$, $k=1,2,3,4$ to define the regions where $\rre \theta_{ij}<0$, $1\leq i<j\leq3$. To eliminate the jump on $\mathbb{R}$, we need to perform the first transformation, and prior to that, first need to perform analytical approximations for the reflection coefficients $r_1( k)$, $r_2(k)$, and $\hat{r}_1(k)$, where $\hat{r}_1(k)$ is defined as
\begin{equation*}
    \hat{r}_1( k)=\frac{r_1(k)}{1+8k|r_1(k)|^2}.
\end{equation*}

\begin{lem}\label{lem_r_decomposition}
For each $\zeta\in\mathcal{I}_2$ and $t>0$, the following decompositions hold:
    \begin{align*}
        &r_1(k)=r_{1,a}(x,t,k)+r_{1,r}(x,t,k),\quad && k\in \left[-k_0,\infty\right),\\
        &\hat{r}_1(k)=\hat{r}_{1,a}(x,t,k)+\hat{r}_{1,r}(x,t,k),\quad && k\in \left(-\infty,-k_0\right],\\
        &r_2(k)=r_{2,a}(x,t,k)+r_{2,r}(x,t,k),\quad && k\in \mathbb{R},
    \end{align*}
    where the above functions satisfy the following properties:
    \begin{itemize}
     \item $r_{1,a}(x,t,k)$ is defined and continuous for $k\in \overline{U_1^{(12)}}$, and is analytic for $k\in U_1^{(12)}$. $\hat{r}_{1,a}(x,t,k)$ is defined and continuous for $k\in \overline{U_3^{(12)}}$, and is analytic for $k\in U_3^{(12)}$. $r_{2,a}(x,t,k)$ is defined and continuous for $k\in \overline{\mathbb{C}_+}$, and is analytic for $k\in \mathbb{C}_+$.
    \item The functions $r_{1,a}$, $\hat{r}_{1,a}$ and $r_{2,a}$ satisfy
        \begin{align*}
            &\left|\partial_x^l\left(r_{1,a}(x,t,k)-\sum_{j=0}^{N}\frac{r^{(j)}_{1,a}(-k_0)}{j!}(k+k_0)^j\right)\right|\leq C \left|k+k_0\right|^{N+1}\re^{\frac{t}{4}\left|\rre\theta_{12}(\zeta,k)\right|},\quad &&k\in \overline{U_1^{(12)}},\\
            &\left|\partial_x^l \left(\hat{r}_{1,a}(x,t,k)-\sum_{j=0}^{N}\frac{\hat{r}^{(j)}_{1,a}(-k_0)}{j!}(k+k_0)^j \right) \right|\leq C \left|k+k_0\right|^{N+1}\re^{\frac{t}{4}\left|\rre\theta_{12}(\zeta,k)\right|},\quad  &&k\in \overline{U_3^{(12)}},\\
            &\left|\partial_x^l\left(r_{2,a}(x,t,k)-\sum_{j=0}^{N}\frac{r^{(j)}_{2,a}(0)}{j!}k^j\right)\right|\leq C \left|k\right|^{N+1}\re^{\frac{t}{4}\left|\rre\theta_{13}(\zeta,k)\right|},\quad  &&k\in \overline{\mathbb{C}_+},
            \end{align*}
            and
        \begin{align*}    
            &\left|\partial_x^l\left(r_{1,a}(x,t,k)\right)\right|\leq  \frac{C}{1+\left|k\right|}\re^{\frac{t}{4}\left|\rre\theta_{12}(\zeta,k)\right|},\quad  &&k\in \overline{U_1^{(12)}},\\
            &\left|\partial_x^l\left(\hat{r}_{1,a}(x,t,k)\right)\right|\leq  \frac{C}{1+\left|k\right|}\re^{\frac{t}{4}\left|\rre\theta_{12}(\zeta,k)\right|},\quad  &&k\in \overline{U_3^{(12)}},\\
            &\left|\partial_x^l\left(r_{2,a}(x,t,k)\right)\right|\leq  \frac{C}{1+\left|k\right|}\re^{\frac{t}{4}\left|\rre\theta_{13}(\zeta,k)\right|},\quad  &&k\in \overline{\mathbb{C}_+},
        \end{align*}
        where $l=0,1$, and the constant $C$ is independent of $\zeta,t,k$.
    \item For each $1\leq p\leq \infty$ and $l=0,1$, the following formulas hold
    \begin{align*}
        &\left\|(1+|\cdot|)\partial_x^lr_{1,r}(x,t,\cdot)\right\|_{L^p\left(-k_0,\infty\right)}=\left\|\frac{r_{1,r}(x,t,\cdot)}{\cdot+k_0}\right\|_{L^p\left(-k_0,\infty\right)}=\mathcal{O}\left(t^{-N}\right),\\
        &\left\|(1+|\cdot|)\partial_x^l\hat{r}_{1,r}(x,t,\cdot)\right\|_{L^p\left(-\infty,-k_0\right)}=\left\|\frac{\hat{r}_{1,r}(x,t,\cdot)}{\cdot+k_0}\right\|_{L^p\left(-\infty,-k_0\right)}=\mathcal{O}\left(t^{-N}\right),\\
        &\left\|(1+|\cdot|)\partial_x^l{r}_{2,r}(x,t,\cdot)\right\|_{L^p(\mathbb{R})}=\mathcal{O}\left(t^{-N}\right),
    \end{align*}
    uniformly for $\zeta\in\mathcal{I}_2$ as $t\to\infty$.
        \end{itemize}
\end{lem}

\begin{proof}
    The proof can be found in Ref. \cite{DZ_1993} and \cite{Lenells_2017}.
\end{proof}

In what follows, for notational convenience, we denote $r_{j,a}(x,t,k)$ and $r_{j,r}(x,t,k)$ by $r_{j,a}(k)$ and $r_{j,r}(k)$, $j=1,2$, respectively.

Define $\Gamma^{(1)}$ as shown in Figure \ref{fig_Gamma1}, and construct a function $\delta(\zeta,k)$ that is analytic everywhere except for jump across $\Gamma^{(1)}_1$:
     \begin{equation}\label{delta}
    	\left\{
    	\begin{aligned}
    		&\begin{aligned}\delta_{+}(k)
    			=&\delta_{-}(k)(1+8k\left|r_1(k)\right|^2),\quad  &&k\in \Gamma^{(1)}_1,\\
    			=&\delta_{-}(k), &&k\in\mathbb{C}\setminus\Gamma^{(1)}_1,
    		\end{aligned}\\
            &\delta(k)\rightarrow1,\qquad\qquad\qquad\qquad\quad\quad\!\!\!\!  k\rightarrow\infty.
    	\end{aligned}
    	\right.
    \end{equation}
    The function $\delta(k)$ defined in equation \eqref{delta} can be expressed as the following integral equation by Plemelj formulas:
    \begin{equation}\label{vol_delta}
        \delta(\zeta,k)={\rm{exp}}\left\{\frac{1}{2 \pi\ri}\int_{-\infty}^{-k_0}\frac{\ln \left(1+8s\,\left|r_1(s)\right|^2\right)}{s-k}{\rm{d}}s\right\},\quad k\in \mathbb{C}\backslash \Gamma^{(1)}_1.
    \end{equation}
    Then we can naturally obtain
    \begin{align*}
        \delta(\zeta,-k)&={\rm{exp}}\left\{\frac{1}{2 \pi\ri}\int_{-\infty}^{-k_0}\frac{\ln \left(1+8s\,\left|r_1(s)\right|^2\right)}{s+k}{\rm{d}}s\right\}={\rm{exp}}\left\{-\frac{1}{2 \pi\ri}\int^{\infty}_{k_0}\frac{\ln \left(1-8s\,\left|r_1(-s)\right|^2\right)}{s-k}{\rm{d}}s\right\},\quad k\in \mathbb{C}\backslash \Gamma^{(1)}_3.
    \end{align*}

    \begin{figure}[htbp]
    	\centering
    	\begin{tikzpicture}[scale=1]
        
        \draw[very thick, black!20!blue] (-4,0) -- (4,0);
    		
    		\draw[very thick, black!20!blue, -latex] (-4,0) -- (-2.5,0);
            \draw[very thick, black!20!blue, -latex] (-4,0) -- (0.2,0);
            \draw[very thick, black!20!blue, -latex] (-4,0) -- (3,0);

    		\fill (-1.5,0) circle (1.5pt);
            \fill (0,0) circle (1.5pt);
	      \fill (1.5,0) circle (1.5pt);

           \node[below] at (-1.5,0) {$-k_0$};
	     \node[below] at (0,-0.2) {$\text{0}$};
	     \node[below] at (1.5,0) {$k_0$};
         
           \node[red!70!black,above] at (-3,0.2) {$\Gamma^{(1)}_1$};
	     \node[red!70!black,above] at (0,0.2) {$\Gamma^{(1)}_2$};
          \node[red!70!black,above] at (3,0.2) {$\Gamma^{(1)}_3$};

          \node[right] at (4,0) {$\rre k$};
    		
    	\end{tikzpicture}
    	\caption{The jump contour $\Gamma^{(1)}$ in the complex $k$-plane.}
    	\label{fig_Gamma1}
    \end{figure}
    
Let $\ln_0(k)=\ln|k|+\ri \arg_0k$ and  $\ln_\pi(k)=\ln|k|+\ri \arg_\pi k$ with $\arg_0 k\in (0,2\pi)$ and $\arg_\pi k\in (-\pi,\pi)$. Present the following properties of the function \eqref{delta} mentioned above.
    \begin{lem}\label{lem_delta}
    	The function $\delta(\zeta,k)$ satisfies the following properties for $\zeta\in\mathcal{I}_2$:
    	\begin{enumerate}
     		\item  The functions 
            $\delta(\zeta,k)$ and $\delta^{-1}(\zeta,k)$ are analytic on $\mathbb{C}\setminus\left(-\infty,-k_0\right] $, and $\delta$ can be written as
			\begin{equation*}
				\delta(\zeta,k)={\rm{e}}^{\ri\nu\ln_\pi(k+k_0)}{\rm{e}}^{-\chi(\zeta,k)},\quad \zeta\in\mathcal{I}_2,
			\end{equation*}
			where 
			\begin{equation*}
				\begin{aligned}
					&\nu=-\frac{1}{2\pi}\ln(1-8k_0\left|r_1(-k_0) \right|^2 ),\qquad  \chi(\zeta,k)=\frac{1}{2\pi\ri}\int_{-\infty}^{-k_0}\ln_\pi(k-s){\rm{d}}\ln(1+8s\left|r_1(s)\right|^2 ).
				\end{aligned}
			\end{equation*}
    	\item For each $\zeta\in\mathcal{I}_2$, $\delta(k)$ is bounded in $\mathbb{C}\setminus\left(-\infty,-k_0\right] $ and satisfies the following symmetry property:
    	\begin{equation*}
    		\delta^{-1}(k)= \overline{\delta(\overline{k})},\qquad k\in \mathbb{C}\backslash \Gamma^{(1)}_1.
    	\end{equation*}
        \item As $k\to -k_0$ along the non-tangential direction of $(-\infty,-k_0)$, we have
        \begin{align*}
            &\left|\chi(\zeta,k)-\chi({\zeta,-k_0})\right|\leq C|k+k_0|\left(1+|\ln|k+k_0||\right),\\
            &\left|\partial_x\left(\chi(\zeta,k)-\chi({\zeta,-k_0})\right)\right|\leq \frac{C}{t}\left(1+|\ln|k+k_0||\right),
             \end{align*}
             and
     \begin{align*}
            &\left|\partial_x\chi(\zeta,-k_0)\right|=\frac{1}{2t} \left|\partial_\zeta\chi(\zeta,-k_0)+\partial_{-k_0}\chi(\zeta,-k_0)\right|\leq\frac{C}{t},\\
 &\partial_x\left(\delta(\zeta,k)^{\pm1}\right)=\frac{\pm\ri \nu}{2t(k+k_0)}\delta(\zeta,k)^{\pm1},
        \end{align*}
        where $C$ is independent of $\zeta$.
            
    	\end{enumerate}
    \end{lem}

    \begin{proof}
        The lemma is derived from equation \eqref{vol_delta}, and can be proved by direct estimation.
    \end{proof}

    \begin{remark}\label{remark_delta_-k}
        $\delta(-k)$ is analytic on $\mathbb{C}\setminus\left[k_0,\infty\right) $ and can be written as
			\begin{equation*}
				\delta(\zeta,-k)={\rm{e}}^{\ri\nu\ln_0(k-k_0)}{\rm{e}}^{-\chi(\zeta,-k)},
			\end{equation*}
            and $\chi(\zeta,-k)=-\frac{1}{2\pi\ri}\int^{\infty}_{k_0}\ln_0(k-s){\rm{d}}\ln(1-8s\left|r_1(-s)\right|^2 ).$
    \end{remark}

    Based on the function \eqref{delta}, we can construct the following transformation for the eigenfunction $M(x,t,k)$:
    \begin{equation}\label{trans_Delta}
        M^{(1)}(x,t,k)=M(x,t,k)\Delta(k),
    \end{equation}
    where the $3\times3$ matrix-valued function $\Delta(k)$ is defined by 
    \begin{equation}\label{Delta}
        \Delta(k)=\begin{pmatrix}
            \delta(k) & 0 & 0\\
            0 & \frac{1}{{\delta(k)}{\delta(-k)}} & 0\\
            0 & 0 & {\delta(-k)}\\
        \end{pmatrix}.
    \end{equation}
    According to the definition in equation \eqref{delta} and the boundedness indicated in Lemma \ref{lem_delta} of $\delta(\zeta,\pm k)$, we conclude that $\Delta^{\pm1}(\zeta,k)$ are uniformly bounded with respect to $\zeta\in\mathcal{I}_2$ and $k\in \mathbb{C}\setminus\left(\Gamma^{(1)}_1\cup \Gamma^{(1)}_3\right)$, and
    \begin{equation*}\label{Delta_bound}
        \Delta(\zeta,k)=I+\mathcal{O}\left(k^{-1}\right),\qquad k\to\infty.
    \end{equation*}
    Direct computation can also verify that $\Delta(\zeta,k)$ satisfies the symmetries 
    \begin{equation*}
        \Delta^\dagger(\zeta,\bar k)=\mathcal{A}(k)\Delta^{-1}(\zeta,k)\mathcal{A}^{-1}(k),\qquad \Delta(\zeta,k)=\mathcal{B}\Delta(\zeta,-k)\mathcal{B}.
    \end{equation*}

    Under above transformation \eqref{trans_Delta}, the jump matrix  
$v^{(1)}(x,t,k)$ corresponding to $M^{(1)}(x,t,k)$ can be divided into three parts (see Figure \ref{fig_Gamma1}), where the expressions for 
 $v^{(1)}(k)=\Delta_-^{-1}(k)v(x,t,k)\Delta_{+}(k)=v_j^{(1)}(k)$, $k\in \Gamma_j^{(1)}$, $j=1,2,3$, are as follows:
        \begin{align}
           & v^{(1)}_1(x,t,k)
            =\begin{pmatrix}
					1 & -\frac{\hat{r}_1(k)\re^{t\theta_{12}(\zeta,k)}}{\delta_{-}(k)^2\delta(-k)} & 0\\
					0 & 1 & 0\\
					-\frac{\delta_{-}(k)}{\delta(-k)}r_2(-k)\re^{-t\theta_{13}(\zeta,k)} & \frac{1}{\delta_{-}(k)\delta(-k)^2}\frac{\alpha^*(k)\re^{-t\theta_{23}(\zeta,k)} }{1+8k\left|r_1(k) \right|^2} & 1\\
				\end{pmatrix}\begin{pmatrix}
					1 & 0 & \frac{\delta(-k)}{\delta_{+}(k)}r_2^*(-k)\re^{t\theta_{13}(\zeta,k)}\\
					-\delta_{+}(k)^2{\delta(-k)}8k\hat{r}_1^*(k)\re^{-t\theta_{12}(\zeta,k)} & 1 & -\frac{\delta_{+}(k){\delta(-k)^2}8k\alpha(k)\re^{t\theta_{23}(\zeta,k)}}{1+8k\left|r_1(k) \right|^2
					}\\
					0 & 0 & 1\\
				\end{pmatrix},\label{v_1_1}\\  
                &v^{(1)}_2(x,t,k)
                =\begin{pmatrix}
					1 & 0 & 0\\
					-{\delta(k)^2\delta(-k)}8kr_1^*(k)\re^{-t\theta_{12}(\zeta,k)} & 1 & 0\\
					-\frac{\delta(k)r_2^*(k)\re^{-t\theta_{13}(\zeta,k)}}{\delta(-k)} & \frac{r_1(-k)\re^{-t\theta_{23}(\zeta,k)}}{\delta(k)\delta(-k)^2} & 1\\
				\end{pmatrix}\begin{pmatrix}
					1 & -\frac{r_1(k)\re^{t\theta_{12}(\zeta,k)}}{\delta(k)^2\delta(-k)} & \frac{\delta(-k)r_2(k)\re^{t\theta_{13}(\zeta,k)}}{\delta(k)}\\
					0 & 1 & -{\delta(k)}{\delta(-k)^2}8kr_1^*(-k)\re^{t\theta_{23}(\zeta,k)}\\
					0 & 0 & 1\\
				\end{pmatrix}, \label{v_1_2}\\
               & v^{(1)}_3(x,t,k)
                =\begin{pmatrix}
					1 & 0 & 0\\
					-{\delta(k)^2}{\delta_{-}(-k)}8kr_1^*(k)\re^{-t\theta_{12}(\zeta,k)} & 1 & -{\delta(k)}{\delta_{-}(-k)^2}8k\hat{r}_1^*(-k)\re^{t\theta_{23}(\zeta,k)}\\
					-\frac{\delta(k)r_2^*(k)\re^{-t\theta_{13}(\zeta,k)}}{\delta_{-}(-k)} & 0 & 1\\
				\end{pmatrix}\begin{pmatrix}
					1 & -\frac{r_1(k)\re^{t\theta_{12}(\zeta,k)}}{\delta(k)^2\delta_{+}(-k)} & \frac{\delta_{+}(-k)r_2(k)\re^{t\theta_{13}(\zeta,k)}}{\delta(k)}\\
					0 & 1 & 0\\
					0 & \frac{\hat{r}_1(-k)\re^{-t\theta_{23}(\zeta,k)}}{\delta(k)\delta_{+}(-k)^2} & 1\\
				\end{pmatrix}.\label{v_1_3}
            \end{align}
  
 \subsubsection{Second transformation}\label{subsec_right_trans2}
  This subsection employs transformations based on the dispersion relations $\theta_{ij}$ ($1\leq i<j\leq3$) displayed in Figure \ref{fig_theta_sign1}. Using the product structure of the jump matrices $v^{(1)}_j$ ($j=1,2,3$) from equations \eqref{v_1_1}, \eqref{v_1_2} and \eqref{v_1_3}, we introduce the function $T^{(1)}$ associated with $\theta_{13}$. Specifically, for each region $D_j^{(1)}$ ($j=1,2,3,4$) shown in Figure \ref{fig_RegionD_1}, we define $T^{(1)}(\zeta,k)=T^{(1)}_j(\zeta,k)$ as follows:

   \begin{figure}[htbp]
    	\centering
    	\begin{tikzpicture}[scale=0.6]

        \fill[fill=white!90!blue] (0,-5) -- (-5,-5) -- (-5,5) -- (0,5) -- cycle;
        \fill[fill=white!90!purple] (0,-5) -- (5,-5) -- (5,5) -- (0,5) -- cycle;

        \draw[very thick, black!20!blue] (-5,0) -- (5,0);
        \draw[very thick, black!20!blue] (0,-5) -- (0,5);

    		\fill (-2,0) circle (1.5pt);
            \fill (0,0) circle (1.5pt);
	      \fill (2,0) circle (1.5pt);

           \node[below] at (-2,-0.1) {$-k_0$};
	     \node[below] at (0.2,0) {$\text{0}$};
	     \node[below] at (2,-0.1) {$k_0$};

         \node at (2,1.5) {$D^{(1)}_1$};
         \node at (-2,1.5) {$D^{(1)}_2$};
         \node at (-2,-1.8) {$D^{(1)}_3$};
         \node at (2,-1.8) {$D^{(1)}_4$};

          \node[right] at (5,0) {$\rre k$};

    	\end{tikzpicture}
    	\caption{The definition of region $D^{(1)}$ involved in transformation $T^{(1)}$ in the complex $k$-plane.}
    	\label{fig_RegionD_1}
    \end{figure}

    \begin{align}\label{trans_T1}
        &T^{(1)}_1(\zeta,k)=\setlength{\arraycolsep}{3pt}\begin{pmatrix}
            1 & 0 & -\frac{\delta(-k)}{\delta(k)}r_{2,a}(k)\re^{t\theta_{13}(\zeta,k)}\\
            0 & 1 & 0\\
            0 & 0 & 1\\
        \end{pmatrix},\qquad
        T^{(1)}_2(\zeta,k)=\setlength{\arraycolsep}{3pt}\begin{pmatrix}
            1 & 0 & -\frac{\delta(-k)}{\delta(k)}{\small \left(r_{2,a}(k)-8kr_{1,a}(k)r_{1,a}^*(-k)\right)}\re^{t\theta_{13}(\zeta,k)}\\
            0 & 1 & 0\\
            0 & 0 & 1\\
        \end{pmatrix},\\
        &T^{(1)}_3(\zeta,k)=\setlength{\arraycolsep}{3pt}\begin{pmatrix}
            1 & 0 & 0\\
            0 & 1 & 0\\
            -\frac{\delta(k)}{\delta(-k)}{\small \left(r_{2,a}^*(k)-8kr_{1,a}^*(k)r_{1,a}(-k)\right)}\re^{-t\theta_{13}(\zeta,k)} & 0 & 1\\
        \end{pmatrix},\qquad
        T^{(1)}_4(\zeta,k)=\setlength{\arraycolsep}{3pt}\begin{pmatrix}
            1 & 0 & 0\\
            0 & 1 & 0\\
            -\frac{\delta(k)}{\delta(-k)}r_{2,a}^*(k)\re^{-t\theta_{13}(\zeta,k)} & 0 & 1\\
        \end{pmatrix}.\nonumber
    \end{align}
    
    For the other two functions $\theta_{12}$ and $\theta_{23}$, the product of the transformation matrices corresponding to $\rre k>0$ and $\rre k<0$ cannot be interchanged when dealing with them. Therefore, we can only break down the treatment of these two functions into two parts, first considering the region definitions in Figure \ref{fig_RegionD_1} and the function $T^{(2)}(\zeta,k)=T^{(2)}_j(\zeta,k)$ for $k\in D^{(1)}_j $, $j=1,2,\cdots,4$, are as follows:

     \begin{equation}\label{trans_T2}
     	\begin{aligned}
     		&T^{(2)}_1(\zeta,k)=\begin{pmatrix}
     			1 & \frac{1}{\delta(k)^2\delta(-k)}r_{1,a}(k)\re^{t\theta_{12}(\zeta,k)} & 0\\
     			0 & 1 & 0\\
     			0 & 0 & 1\\
     		\end{pmatrix},\quad &&T^{(2)}_2(\zeta,k)=\begin{pmatrix}
     			1 & 0 & 0\\
     			0 & 1 & {\delta(k)}{\delta(-k)^2}8kr_{1,a}^*(-k)\re^{t\theta_{23}(\zeta,k)}\\
     			0 & 0 & 1\\
     		\end{pmatrix},\nonumber\\
     		&T^{(2)}_3(\zeta,k)=\begin{pmatrix}
     			1 & 0 & 0\\
     			0 & 1 & 0\\
     			0 & \frac{1}{\delta(k)\delta(-k)^2}r_{1,a}(-k)\re^{-t\theta_{23}(\zeta,k)}  & 1\\
     		\end{pmatrix},\quad &&T^{(2)}_4(\zeta,k)=\begin{pmatrix}
     			1 & 0 & 0\\
     			-{\delta(k)^2}{\delta(-k)}8kr_{1,a}^*(k)\re^{-t\theta_{12}(\zeta,k)} & 1 & 0\\
     			0 & 0 & 1\\
     		\end{pmatrix}.
     	\end{aligned}
     \end{equation}

     The final transformation for the jump on $\mathbb{R}$ is defined by $T^{(3)}(\zeta,k)=T^{(3)}_j(\zeta,k)$ for $k\in D_j^{(2)}$, $j=1,2,\cdots,8,$ which is applied in several regions of Figure \ref{fig_RegionD_3}:

    \begin{equation}\label{trans_T3}
    	\begin{aligned}
    		&T^{(3)}_1(\zeta,k)=\begin{pmatrix}
    			1 & 0 & 0\\
    			0 & 1 & 0\\
    			0 & -\frac{1}{\delta(k)\delta(-k)^2}\hat{r}_{1,a}(-k)\re^{-t\theta_{23}(\zeta,k)} & 1\\
    		\end{pmatrix},\quad &&T^{(3)}_2(\zeta,k)=\begin{pmatrix}
    			1 & 0 & 0\\
    			0 & 1 & {\delta(k)}{\delta(-k)^2}8kr_{1,a}^*(-k)\re^{t\theta_{23}(\zeta,k)}\\
    			0 & 0 & 1\\
    		\end{pmatrix},\\
    		&T^{(3)}_3(\zeta,k)=\begin{pmatrix}
    			1 & 0 & 0\\
    			0 & 1 & 0\\
    			0 & \frac{1}{\delta(k)\delta(-k)^2}r_{1,a}(-k)\re^{-t\theta_{23}(\zeta,k)}  & 1\\
    		\end{pmatrix},\quad &&T^{(3)}_4(\zeta,k)=\begin{pmatrix}
    			1 & 0 & 0\\
    			0 & 1 & -{\delta(k)}{\delta(-k)^2}8k\hat{r}_{1,a}^*(-k)\re^{t\theta_{23}(\zeta,k)}\\
    			0 & 0 & 1\\
    		\end{pmatrix},\\
    		&T^{(3)}_5(\zeta,k)=\begin{pmatrix}
    			1 & \frac{1}{\delta(k)^2\delta(-k)}r_{1,a}(k)\re^{t\theta_{12}(\zeta,k)} & 0\\
    			0 & 1 & 0\\
    			0 & 0 & 1\\
    		\end{pmatrix},\quad &&T^{(3)}_6(\zeta,k)=\begin{pmatrix}
    			1 & 0 & 0\\
    			{\delta(k)^2}{\delta(-k)}8k\hat{r}_{1,a}^*(k)\re^{-t\theta_{12}(\zeta,k)} & 1 & 0\\
    			0 & 0 & 1\\
    		\end{pmatrix},\\
    		&T^{(3)}_7(\zeta,k)=\begin{pmatrix}
    			1 & -\frac{1}{\delta(k)^2\delta(-k)}\hat{r}_{1,a}(k)\re^{t\theta_{12}(\zeta,k)} & 0\\
    			0 & 1 & 0\\
    			0 & 0  & 1\\
    		\end{pmatrix},\quad &&T^{(3)}_8(\zeta,k)=\begin{pmatrix}
    			1 & 0 & 0\\
    			-{\delta(k)^2}{\delta(-k)}8kr_{1,a}^*(k)\re^{-t\theta_{12}(\zeta,k)} & 1 & 0\\
    			0 & 0 & 1\\
    		\end{pmatrix},\\
    		&T^{(3)}_j(\zeta,k)=I, \qquad  k\in \mathbb{C}\setminus\bigcup_{j=1}^8D^{(2)}_j.\nonumber
    	\end{aligned}
    \end{equation}

     \begin{figure}[htbp]
     	\centering
     	\begin{tikzpicture}[scale=1]
     		
     		\fill[fill=white!90!blue] (0,2) -- (2,0) -- (0,-2) -- cycle;
     		\fill[fill=white!90!blue] (2,0) -- (2+2/1.414,2/1.414) -- (4,0) --  (2+2/1.414,-2/1.414)  -- cycle;
     		\fill[fill=white!90!purple] (0,2) -- (-2,0) -- (0,-2) -- cycle;
     		\fill[fill=white!90!purple] (-2,0) -- (-2-2/1.414,2/1.414) -- (-4,0) --  (-2-2/1.414,-2/1.414)  -- cycle;

     		\draw[very thick, black!20!blue] (-4,0) -- (4,0);
     		\draw[very thick, black!20!blue] (0,-2) -- (0,2);

     		\draw[very thick, black!20!blue] (2+2/1.414,2/1.414) -- (-4/1.414+2,-4/1.414);
     		\draw[very thick, black!20!blue] (2+2/1.414,-2/1.414) -- (-4/1.414+2,4/1.414);
     		
     		\draw[very thick, black!20!blue] (-2-2/1.414,2/1.414) -- (4/1.414-2,-4/1.414);
     		\draw[very thick, black!20!blue] (-2-2/1.414,-2/1.414) -- (4/1.414-2,4/1.414);

     		\fill (-2,0) circle (1.5pt);
     		\fill (0,0) circle (1.5pt);
     		\fill (2,0) circle (1.5pt);
     		
     		\node[below] at (-2,-0.1) {$-k_0$};
     		\node[below] at (0.15,0) {$\text{0}$};
     		\node[below] at (2,-0.1) {$k_0$};
     		
     		\node at (3.2,0.5) {$D^{(2)}_1$};
     		\node at (0.7,0.5) {$D^{(2)}_2$};
     		\node at (3.2,-0.5) {$D^{(2)}_4$};
     		\node at (0.7,-0.5) {$D^{(2)}_3$};
     		
     		\node at (-3.2,0.5) {$D^{(2)}_6$};
     		\node at (-0.7,0.5) {$D^{(2)}_5$};
     		\node at (-3.2,-0.5) {$D^{(2)}_7$};
     		\node at (-0.7,-0.5) {$D^{(2)}_8$};

     		\node[right] at (4,0) {$\rre k$};

     	\end{tikzpicture}
     	\caption{The definition of region $D^{(2)}$ involved in transformation $T^{(3)}$ in the complex $k$-plane.}
     	\label{fig_RegionD_3}
     \end{figure}

    According to three transformations \eqref{trans_T1},  \eqref{trans_T2} and \eqref{trans_T3} defined above, we construct 
    \begin{equation*}\label{trans_T}
        M^{(2)}(x,t,k)=M^{(1)}(x,t,k)T(x,t,k),
    \end{equation*}
    to resolve the jump on the real axis $\mathbb{R}$, 
     where 
     \begin{equation}\label{T}
         T(x,t,k)=T^{(1)}(x,t,k)T^{(2)}(x,t,k)T^{(3)}(x,t,k).
     \end{equation}
     And $M_+^{(2)}(\zeta,k)=M_-^{(2)}(\zeta,k)v^{(2)}(\zeta,k)$ for $k\in \Gamma^{(2)}$,  as given in Figure \ref{fig_Gamma2}.
     \begin{lem}\label{lem_T_bound}
        $T(x,t,k)$ is uniformly bounded for $k\in \mathbb{C}\setminus\Gamma^{(2)}$, $\zeta\in\mathcal{I}_2$, and $t>0$. More importantly,
        \begin{equation*}\label{T_k_infty}
            T(k)=I+\mathcal{O}(k^{-1}),\quad k\to\infty.
        \end{equation*}
     \end{lem}
     \begin{proof}
         For the three sub-transformations $T^{(j)}(k)$, $j=1,2,3$ belonging to $T(k)$, we only prove one of the regions, and the proofs for the other regions are similar. For the transformation $T^{(1)}(\zeta,k)$, where $k\in D_2^{(1)}$, we have 
         \begin{equation*}
             \rre \theta_{13}(\zeta,k)<0,
         \end{equation*}
         and according to Lemmas \ref{lem_r_decomposition}, \ref{lem_delta} and the relationship $\theta_{12}(k)+\theta_{23}(k)=\theta_{13}(k)$, the following holds
         \begin{equation*}
             \left|-\left(T^{(1)}_2(\zeta,k)\right)_{13}\right|=\left|\frac{\delta(-k)}{\delta(k)}\left(r_{2,a}(k)-8kr_{1,a}(k)r_{1,a}^*(-k)\right)\re^{t\theta_{13}(\zeta,k)}\right|\leq \left(\frac{C}{1+|k|}+\frac{C|k|}{1+|k|^2}\right)\re^{-\frac{3}{4}t|\rre\theta_{13}|}.
         \end{equation*}
         This proves the claim for $k\in D_2^{(1)}$. 
    For the function $T^{(2)}_1(k)$, $k\in D_1^{(2)}$, calculate similarly 
     \begin{equation*}
             \left|\left(T^{(2)}_1(\zeta,k)\right)_{12}\right|=\left|\frac{1}{\delta(k)^2\delta(-k)}r_{1,a}(k)\re^{t\theta_{12}(\zeta,k)}\right|\leq \frac{C}{1+|k|}\re^{-\frac{3}{4}t|\rre\theta_{12}|}.
         \end{equation*}
         Finally, for $k\in D_1^{(2)}(k)$,
         \begin{equation*}
             \left|\left(T^{(3)}_1(\zeta,k)\right)_{12}\right|=\left|-\frac{1}{\delta(k)\delta(-k)^2}\hat{r}_{1,a}(-k)\re^{-t\theta_{23}(\zeta,k)}\right|\leq \frac{C}{1+|k|}\re^{-\frac{3}{4}t|\rre\theta_{23}|}.
         \end{equation*}
     \end{proof}
     
      \begin{figure}[htbp]
     	\centering
     	\begin{tikzpicture}[scale=1]
     		
     		\draw[very thick, black!40!green] (-4,0) -- (4,0);
     		\draw[very thick, black!40!green] (0,2) -- (0,3.5);
     		\draw[very thick, black!40!green] (0,-2) -- (0,-3.5);
     		
     		\draw[very thick, dashed, black!40!green] (0,-2) -- (0,2);
     		
     		\draw[very thick, black!40!green, -latex] (-4,0) -- (-2.8,0);
     		\draw[very thick, black!40!green, -latex] (-4,0) -- (-0.8,0);
     		\draw[very thick, black!40!green, -latex] (-4,0) -- (1.2,0);
     		\draw[very thick, black!40!green, -latex] (-4,0) -- (3.2,0);
     		
     		\draw[very thick, black!40!green, -latex] (0,-2) -- (0,-3);
     		\draw[very thick, black!40!green, -latex] (0,2) -- (0,3);
     		
     		\draw[very thick, black!20!blue] (2+2/1.414,2/1.414) -- (0,-2);
     		\draw[very thick, black!20!blue] (2+2/1.414,-2/1.414) -- (0,2);
     		
     		\draw[very thick, blue!20!purple] (-2-2/1.414,2/1.414) -- (0,-2);
     		\draw[very thick, blue!20!purple] (-2-2/1.414,-2/1.414) -- (0,2);
     		
     		\draw[very thick, black!20!blue,-latex] (2,0) --(2+1/1.2,1/1.2);
     		\draw[very thick, black!20!blue,-latex] (2,0) -- (-1/1.2+2,-1/1.2);
     		\draw[very thick, black!20!blue,-latex] (2,0) -- (-1/1.2+2,1/1.2);
     		\draw[very thick, black!20!blue,-latex] (2,0) -- (2+1/1.2,-1/1.2);
     		
     		\draw[very thick, blue!20!purple,-latex] (-2,0) --(-2+1/1.2,1/1.2);
     		\draw[very thick, blue!20!purple,-latex] (-2,0) -- (-1/1.2-2,-1/1.2);
     		\draw[very thick, blue!20!purple,-latex] (-2,0) -- (-1/1.2-2,1/1.2);
     		\draw[very thick, blue!20!purple,-latex] (-2,0) -- (-2+1/1.2,-1/1.2);
     		
     		\fill (-2,0) circle (1.5pt);
     		\fill (0,0) circle (1.5pt);
     		\fill (2,0) circle (1.5pt);
     		
     		\node[below] at (-2,-0.1) {$-k_0$};
     		\node[below] at (0.3,0) {$\text{0}$};
     		\node[below] at (2,-0.1) {$k_0$};
     		
     		\node[red!70!black,left] at (1.9+1/1.2,1/1.2) {$1$};
     		\node[red!70!black,right] at (2.1-1/1.2,1/1.2) {$2$};
     		\node[red!70!black,right] at (2.1-1/1.2,-1/1.2) {$3$};
     		\node[red!70!black,left] at (1.9+1/1.2,-1/1.2) {$4$};
     		
     		\node[red!70!black,left] at (-2.1+1/1.2,1/1.2) {$5$};
     		\node[red!70!black,right] at (-1.9-1/1.2,1/1.2) {$6$};
     		\node[red!70!black,right] at (-1.9-1/1.2,-1/1.2) {$7$};
     		\node[red!70!black,left] at (-2.1+1/1.2,-1/1.2) {$8$};
     		
     		\node[red!70!black,above] at (-3,0.1) {$9$};
     		\node[red!70!black,above] at (-1,0.1) {$10$};
     		\node[red!70!black,above] at (1,0.1) {$11$};
     		\node[red!70!black,above] at (3,0.1) {$12$};
     		
     		\node[red!70!black,right] at (0,2.8) {$13$};
     		\node[red!70!black,right] at (0,-2.8) {$14$};
     		\node[red!70!black,right] at (0,0.7) {$15$};

     		\node[right] at (4,0) {$\rre k$};
     		
     	\end{tikzpicture}
     	\caption{The jump contour $\Gamma^{(2)}$ in the complex $k$-plane.}
     	\label{fig_Gamma2}
     \end{figure}

    The eigenfunction $ M^{(2)}(x,t,k)$ is analytic on $\mathbb{C}\setminus\Gamma^{(2)}$ (see Figure \ref{fig_Gamma2}), and the jump matrix $v^{(2)}_j(x,t,k)$ for $k\in \Gamma_j^{(2)}$, $j=1,2,\cdots,8,$ are

     \begin{align*}
         &v^{(2)}_1(x,t,k)=\begin{pmatrix}
            1 & 0 & 0\\
            0 & 1 & 0\\
            0 & \frac{1}{\delta(k)\delta(-k)^2}\hat{r}_{1,a}(-k)\re^{-t\theta_{23}(\zeta,k)} & 1\\
        \end{pmatrix},\quad &&v^{(2)}_2(x,t,k)=\begin{pmatrix}
            1 & 0 & 0\\
            0 & 1 & {\delta(k)}{\delta(-k)^2}8kr_{1,a}^*(-k)\re^{t\theta_{23}(\zeta,k)}\\
            0 & 0 & 1\\
        \end{pmatrix},\\
        &v^{(2)}_3(x,t,k)=\begin{pmatrix}
            1 & 0 & 0\\
            0 & 1 & 0\\
            0 & -\frac{1}{\delta(k)\delta(-k)^2}r_{1,a}(-k)\re^{-t\theta_{23}(\zeta,k)}  & 1\\
        \end{pmatrix},\quad &&v^{(2)}_4(x,t,k)=\begin{pmatrix}
            1 & 0 & 0\\
            0 & 1 & -{\delta(k)}{\delta(-k)^2}8k\hat{r}_{1,a}^*(-k)\re^{t\theta_{23}(\zeta,k)}\\
            0 & 0 & 1\\
        \end{pmatrix},\\
        &v^{(2)}_5(x,t,k)=\begin{pmatrix}
            1 & -\frac{1}{\delta(k)^2\delta(-k)}r_{1,a}(k)\re^{t\theta_{12}(\zeta,k)} & 0\\
            0 & 1 & 0\\
            0 & 0 & 1\\
        \end{pmatrix},\quad &&v^{(2)}_6(x,t,k)=\begin{pmatrix}
            1 & 0 & 0\\
            {\delta(k)^2}{\delta(-k)}8k\hat{r}_{1,a}^*(k)\re^{-t\theta_{12}(\zeta,k)} & 1 & 0\\
            0 & 0 & 1\\
        \end{pmatrix},\\
        &v^{(2)}_7(x,t,k)=\begin{pmatrix}
            1 & \frac{1}{\delta(k)^2\delta(-k)}\hat{r}_{1,a}(k)\re^{t\theta_{12}(\zeta,k)} & 0\\
            0 & 1 & 0\\
            0 & 0  & 1\\
        \end{pmatrix},\quad &&v^{(2)}_8(x,t,k)=\begin{pmatrix}
            1 & 0 & 0\\
            -{\delta(k)^2}{\delta(-k)}8kr_{1,a}^*(k)\re^{-t\theta_{12}(\zeta,k)} & 1 & 0\\
            0 & 0 & 1\\
        \end{pmatrix}.
     \end{align*}

     \begin{remark}\label{remark_v2_9-15}
         In the jump contour depicted in Figure \ref{fig_Gamma2}, it can be observed that we have employed three different colors. The red lines intersect at the critical point $-k_0$, while the blue lines intersect at the critical point $k_0$. The corresponding jump matrices for these jump lines have already been provided above. 
         
         Additionally, due to the transformation $T^{(1)}(\zeta,k)$ in \eqref{trans_T1}, the eigenfunction $M^{(2)}(x,t,k)$ also experiences jumps along the purely imaginary axis. After calculation, the corresponding jump matrix $v^{(15)}(k)=I$ for $k\in \Gamma^{(2)}_{15}$, which corresponds to the part of the green dashed line in Figure \ref{fig_Gamma2}. The two lines above and below it correspond to $\Gamma^{(2)}_{13}$ and $\Gamma^{(2)}_{14}$, respectively, and their jump matrices are given as follows:
          \begin{align*}
         &v^{(2)}_{13}(x,t,k)=\begin{pmatrix}
            1 & -\frac{1}{\delta(k)^2\delta(-k)}r_{1,a}(k)\re^{t\theta_{12}(\zeta,k)} & 0\\
            0 & 1 & {\delta(k)}{\delta(-k)^2}8kr_{1,a}^*(-k)\re^{t\theta_{23}(\zeta,k)}\\
            0 & 0 & 1\\
        \end{pmatrix},\quad &&k\in \Gamma^{(2)}_{13},\\
        &v^{(2)}_{14}(x,t,k)=\begin{pmatrix}
            1 & 0 & 0\\
            -{\delta(k)^2}{\delta(-k)}8kr_{1,a}^*(k)\re^{-t\theta_{12}(\zeta,k)} & 1 & 0\\
            0 & -\frac{1}{\delta(k)\delta(-k)^2}r_{1,a}(-k)\re^{-t\theta_{23}(\zeta,k)} & 1\\
        \end{pmatrix},\quad &&k\in \Gamma^{(2)}_{14}.
     \end{align*}
         It is evident that $\rre \theta_{12}(k)<0$, $\rre \theta_{23}(k)<0$ for $k\in \Gamma^{(2)}_{13}$, and $\rre \theta_{12}(k)>0$, $\rre \theta_{23}(k)>0$ for $k\in \Gamma^{(2)}_{14}$, which means that the two jump matrices decay exponentially to the identity matrix.

          Finally, the initial jump contour $\mathbb{R}$ has been resolved after this transformation $T(\zeta,k)$. Since the forms of $v^{(2)}_j(k)$, $k\in \Gamma^{(2)}_j$, $j=9,10,11,12$, are too complex to be presented here, they can be obtained through direct calculation:
          \begin{align*}
         &v^{(2)}_9(x,t,k)=\left(T_3^{(1)}(x,t,k)T_3^{(2)}(x,t,k)T_3^{(3)}(x,t,k)\right)^{-1}v^{(1)}_1(x,t,k)T_2^{(1)}(x,t,k)T_2^{(2)}(x,t,k)T_2^{(3)}(x,t,k),\\
        &v^{(2)}_{10}(x,t,k)=\left(T_3^{(1)}(x,t,k)T_3^{(2)}(x,t,k)T_4^{(3)}(x,t,k)\right)^{-1}v^{(1)}_2(x,t,k)T_2^{(1)}(x,t,k)T_2^{(2)}(x,t,k)T_4^{(3)}(x,t,k),\\
        &v^{(2)}_{11}(x,t,k)=\left(T_4^{(1)}(x,t,k)T_3^{(2)}(x,t,k)T_4^{(3)}(x,t,k)\right)^{-1}v^{(1)}_2(x,t,k)T_1^{(1)}(x,t,k)T_2^{(2)}(x,t,k)T_1^{(3)}(x,t,k),\\
        &v^{(2)}_{12}(x,t,k)=\left(T_4^{(1)}(x,t,k)T_4^{(2)}(x,t,k)T_4^{(3)}(x,t,k)\right)^{-1}v^{(1)}_3(x,t,k)T_1^{(1)}(x,t,k)T_1^{(2)}(x,t,k)T_1^{(3)}(x,t,k).
     \end{align*}
          It can be shown that every non-zero term in $v^{(2)}_j-I$ can be fully controlled by $r_{1,r}(k)$, $r_{2,r}(k)$ and $\hat r_{1,r}(k)$. 
     \end{remark}

    \begin{lem}\label{lem_v2_error}
    As $t\to\infty$, the jump matrix $v^{(2)}(x,t,k)$ converges to the identity matrix $I$  and $\partial_xv^{(2)}(k)$ converges to the zero matrix $O$ uniformly for $\zeta\in\mathcal{I}_2$ and $k\in \Gamma^{(2)}$ except near the two critical points $\pm k_0$. More importantly, the jump matrices $v^{(2)}_j(k)$, $j=9,10,\cdots,14$, satisfy:
       \begin{align}
            &\left\|(1+|\cdot|)\partial_x^l\left(v^{(2)}_j(k)-I\right)\right\|_{(L^1\cap L^\infty)\left(\Gamma^{(2)}_j\right)}\leq Ct^{-N}, \quad  &&j=9,10,11,12,\label{Error_9-12}\\
            &\left\|(1+|\cdot|)\partial_x^l\left(v^{(2)}_j(k)-I\right)\right\|_{(L^1\cap L^\infty)\left(\Gamma^{(2)}_j\right)}\leq C\re^{-ct}, \quad &&j=13,14.\label{Error_1314}
        \end{align}
    \end{lem}
    \begin{proof}
        First, consider the jump matrix $v^{(2)}_1(x,t,k)$. Since $\rre\theta_{23}(k)\geq0$ for $k\in\Gamma^{(2)}_1$, $v^{(2)}_1(x,t,k)$ (resp. $\partial_xv^{(2)}_1(x,t,k)$) converges to $I$ (resp. $O$) as $t\to\infty$. Moreover, its (2,1)-element does not converge uniformly to $0$ as $k$ near $k_0$, because $\rre\theta_{23}(k_0)=0$. A similar proof applies to the cases of $v^{(2)}_j(x,t,k)$, $j=2,3,\cdots,8$.

        For $k\in \mathbb{R}$, we have $\rre\theta_{12}(k)=\rre\theta_{13}(k)=\rre\theta_{23}(k)=0$. Moreover, it is calculated that each non-zero element in $v^{(2)}_j-I$ involves both $r_{1,r}(k)$, $r_{2,r}(k)$ and $\hat r_{1,r}(k)$. Using Lemma \ref{lem_delta}, we find that
       \begin{equation*}
           \left|v^{(2)}_j-I\right|\leq C\max\left\{\left\|r_{1,r}(k)\right\|_{(-k_0,\infty)},\left\|r_{2,r}(k)\right\|_{\mathbb{R}},\left\|\hat r_{1,r}(k)\right\|_{(-\infty ,-k_0)}\right\}\leq Ct^{-N},\quad j=9,10,11,12.
       \end{equation*}
       For $\partial_x \left(v^{(2)}_j-I\right)$,  the result follows similarly from the third statements in Lemma \ref{lem_r_decomposition} and Lemma \ref{lem_delta}. We complete the proof of equation \eqref{Error_9-12}.
        
   For $k\in \Gamma^{(2)}_{13}$, we have $\rre\theta_{13}(k)<0$, and equation \eqref{Error_1314} holds, according to the following formula:
   \begin{equation*}
       \left|\left(v^{(2)}_{13}-I\right)_{12}\right|=\left|-\frac{1}{\delta(k)^2\delta(-k)}r_{1,a}(k)\re^{t\theta_{12}(\zeta,k)}\right|\leq\frac{C|k|}{1+|k|}\re^{-\frac{3}{4}t\rre\theta_{12}(\zeta,k)}\leq C \re^{-ct}.
   \end{equation*}
   The remaining cases are proved similarly.
    \end{proof}

    \subsubsection{Local parametrix at $\pm k_0$}
    The RH problem for $M^{(2)}(x,t,k)$ satisfies the property that matrix $v^{(2)}-I$ decays to the zero matrix everywhere as $t\to\infty$, except near two critical points $\pm k_0$. This implies that we only need to consider the neighborhoods of these two points when calculating the long-time asymptotics of $M^{(2)}$. The basic idea for solving the problem next is to reduce the RH problem for $M^{(2)}(x,t,k)$ near $k_0$ and $-k_0$ to the RH problems on the cross, and then use parabolic cylinder functions to obtain the exact solution as $t\to\infty$.

        Let $\epsilon=k_0/2$, $B_\epsilon(\pm k_0)=\{ k|\left|k\mp k_0\right|<\epsilon \}$, $B_\epsilon=B_\epsilon(k_0)\cup B_\epsilon(-k_0)$, $X^{ k_0}_j=\left( k_0+X_j\right)\cap B_\epsilon( k_0)$ and $X^{- k_0}_{j+4}=\left(-k_0+X_j\right)\cap B_\epsilon(-k_0)$, $j=1,2,3,4$, where $X=\cup_{j=1}^4X_j$ is the contour defined in the Appendix (see Figure \ref{fig_modelX}) and $X^\epsilon=X^{k_0}\cup X^{-k_0} $. To obtain the relationship between $M^{(2)}$ and $M^{X_j}$, $j=1,2$ for the RH problem \ref{pcmodel_k1} and \ref{pcmodel_-k1} in the Appendix \ref{AppendixA}, we perform the following scaling transformations on the variable $k$:
    \begin{align}
        & z_1=2\sqrt{t}(k-k_0),\qquad  k\in B_\epsilon(k_0),\label{z_1}\\
        & z_2=2\sqrt{t}(k+k_0),\qquad  k\in B_\epsilon(-k_0).\label{z_2}
    \end{align}
    Based on the aforementioned scaling transformation, the functions can be written as
    \begin{align*}
    &t\theta_{23}(\zeta,k)=t\theta_{23}(\zeta,k_0)-2\ri t (k-k_0)^2=t\theta_{23}(\zeta,k_0)-\frac{\ri z_1^2}{2},\\
        &t\theta_{12}(\zeta,k)=t\theta_{12}(\zeta,-k_0)+2\ri t (k+k_0)^2=t\theta_{12}(\zeta,-k_0)+\frac{\ri z_2^2}{2}.
    \end{align*}

    First, consider the case where $k$ is located near $k_0$. For $k\in B_\epsilon(k_0)\setminus\left[k_0,3k_0/2\right) $,
    \begin{align*}
        {\delta(k)}{\delta(-k)^2}
        &=\delta(k)(2\sqrt{t})^{-2\ri \nu}\re ^{2\ri\nu\ln_0(z_1)}\re^{-2\chi(\zeta,-k)}:=\re ^{2\ri\nu\ln_0(z_1)}d^{(k_0)}_{0}(\zeta)d^{(k_0)}_{1}(\zeta,k),
    \end{align*}
    where the functions $d^{(k_0)}_{0}$ and $d^{(k_0)}_{1}$ are defined for $k\in B_\epsilon(k_0)\setminus\left[k_0,3k_0/2\right)$ by
    \begin{align*}
        &d^{(k_0)}_{0}(\zeta)=(2\sqrt{t})^{-2\ri \nu}\re^{-2\chi(\zeta,-k_0)}\delta(k_0),\quad d^{(k_0)}_{1}(\zeta,k)=\re^{-2\chi(\zeta,-k)+2\chi(\zeta,-k_0)}\frac{\delta(k)}{\delta(k_0)}.
    \end{align*}
    Define $\tilde M^{(k_0)}$ for $k$ near $k_0$ by
    \begin{equation*}\label{trans_M_k1}
        \tilde{M}^{(k_0)}(x,t,k)=M^{(2)}(x,t,k)Y_1(\zeta), \quad k\in B_\epsilon(k_0),
    \end{equation*}
    where
    \begin{equation*}
        Y_1(\zeta)=\begin{pmatrix}
            1 & 0 & 0\\
            0 & \left(d^{(k_0)}_{0}(\zeta)\right)^{1/2}\re^{\frac{t}{2}\theta_{23}(\zeta,k_0)} & 0\\
            0 & 0 & \left(d^{(k_0)}_{0}(\zeta)\right)^{-1/2}\re^{-\frac{t}{2}\theta_{23}(\zeta,k_0)}
        \end{pmatrix}.
    \end{equation*}
    Then the jump matrices $\tilde{v}^{(k_0)}_j(k)$, $k\in X^{k_0}_j$, $j=1,2,3,4$, of $\tilde{M}^{(k_0)}$ across $X^{k_0}$ are
    \begin{align*}
         &\tilde{v}^{(k_0)}_1(\zeta,z_1(k))=\begin{pmatrix}
            1 & 0 & 0\\
            0 & 1 & 0\\
            0 & \re ^{-2\ri\nu\ln_0(z_1)}\left(d^{(k_0)}_{1}\right)^{-1}\hat{r}_{1,a}(-k)\re^{\frac{\ri z_1^2}{2}} & 1\\
        \end{pmatrix},\quad&& \tilde{v}^{(k_0)}_2(\zeta,z_1(k))=\begin{pmatrix}
            1 & 0 & 0\\
            0 & 1 & \re ^{2\ri\nu\ln_0(z_1)}d^{(k_0)}_{1}8kr_{1,a}^*(-k)\re^{-\frac{\ri z_1^2}{2}}\\
            0 & 0 & 1\\
        \end{pmatrix},\\
        &\tilde{v}^{(k_0)}_3(\zeta,z_1(k))=\begin{pmatrix}
            1 & 0 & 0\\
            0 & 1 & 0\\
            0 & -\re ^{-2\ri\nu\ln_0(z_1)}\left(d^{(k_0)}_{1}\right)^{-1}r_{1,a}(-k)\re^{\frac{\ri z_1^2}{2}}  & 1\\
        \end{pmatrix},\quad&&\tilde{v}^{(k_0)}_4(\zeta,z_1(k))=\begin{pmatrix}
            1 & 0 & 0\\
            0 & 1 & -\re ^{2\ri\nu\ln_0(z_1)}d^{(k_0)}_{1}8k\hat{r}^*_{1,a}(-k)\re^{-\frac{\ri z_1^2}{2}}\\
            0 & 0 & 1\\
        \end{pmatrix}.
    \end{align*}

    For the case where $k$ is close to $-k_0$, we define
    \begin{align*}
        \frac{1}{\delta(k)^2\delta(-k)}&=\frac{1}{\delta(-k)}(2\sqrt{t})^{2\ri \nu}\re ^{-2\ri\nu\ln_\pi(z_2)}\re^{2\chi(\zeta,k)}:=\re ^{-2\ri\nu\ln_\pi(z_2)}d_{0}^{(-k_0)}(\zeta)d_{1}^{(-k_0)}(\zeta,k), \quad k\in B_\epsilon(-k_0)\setminus\left(-3k_0/2,-k_0\right],
    \end{align*}
    where the functions $d_{0}^{(-k_0)}$ and $d_{1}^{(-k_0)}$ are defined by
    \begin{align*}
        &d_{0}^{(-k_0)}(\zeta)=(2\sqrt{t})^{2\ri \nu}\re^{2\chi(\zeta,-k_0)}\frac{1}{\delta(k_0)},\quad d_{1}^{(-k_0)}(\zeta,k)=\re^{2\chi(\zeta,k)-2\chi(\zeta,-k_0)}\frac{\delta(k_0)}{\delta(-k)}.
    \end{align*}
    Define $\tilde{M}^{(-k_0)}$ for $k$ near $-k_0$ by
    \begin{equation*}\label{trans_M_-k1}
        \tilde{M}^{(-k_0)}(x,t,k)=M^{(2)}(x,t,k)Y_2(\zeta), \quad k\in B_\epsilon(-k_0),
    \end{equation*}
    where
    \begin{equation*}
        Y_2(\zeta)=\begin{pmatrix}
            \left(d_{0}^{(-k_0)}(\zeta)\right)^{1/2}\re^{\frac{t}{2}\theta_{12}(\zeta,-k_0)}  & 0 & 0\\
            0 & \left(d_{0}^{(-k_0)}(\zeta)\right)^{-1/2}\re^{-\frac{t}{2}\theta_{12}(\zeta,-k_0)} & 0\\
            0 & 0 & 1
        \end{pmatrix}.
    \end{equation*}
    Then the jump matrices $\tilde{v}^{(-k_0)}_j(k)$, $k\in X^{-k_0}_j$, $j=5,6,7,8$, of $\tilde{M}^{(k_0)}$ across $X^{-k_0}$ are
    \begin{align*}
         &\tilde{v}^{(-k_0)}_5(\zeta,z_2(k))=\begin{pmatrix}
            1 & -\re ^{-2\ri\nu\ln_\pi(z_2)}d_{1}^{(-k_0)}r_{1,a}(k)\re^{\frac{\ri z_2^2}{2}} & 0\\
            0 & 1 & 0\\
            0 & 0 & 1\\
        \end{pmatrix},\quad&&\tilde{v}^{(-k_0)}_6(\zeta,z_2(k))=\begin{pmatrix}
            1 & 0 & 0\\
            \re ^{2\ri\nu\ln_\pi(z_2)}\left(d_{1}^{(-k_0)}\right)^{-1}8k\hat{r}^*_{1,a}(k)\re^{-\frac{\ri z_2^2}{2}} & 1 & 0\\
            0 & 0 & 1\\
        \end{pmatrix},\\
        &\tilde{v}^{(-k_0)}_7(\zeta,z_2(k))=\begin{pmatrix}
            1 & \re ^{-2\ri\nu\ln_\pi(z_2)}d_{1}^{(-k_0)}\hat{r}_{1,a}(k)\re^{\frac{\ri z_2^2}{2}} & 0\\
            0 & 1 & 0\\
            0 & 0 & 1\\
        \end{pmatrix},\quad&&\tilde{v}^{(-k_0)}_8(\zeta,z_2(k))=\begin{pmatrix}
            1 & 0 & 0\\
            -\re ^{2\ri\nu\ln_\pi(z_2)}\left(d_{1}^{(-k_0)}\right)^{-1}8kr^*_{1,a}(k)\re^{-\frac{\ri z_2^2}{2}} & 1 & 0\\
            0 & 0 & 1\\
        \end{pmatrix}.
    \end{align*}

   Denote $q=r_1(-k_0)$ and $\hat{q}=q/(1-8k_0|q|^2)$. For any fixed $z_1$, $r_{1,a}(-k)\to q$, $\hat{r}_{1,a}(-k)\to\hat{q}$ and $d_1^{(k_0)}\to 1$ as $t\to\infty$. And for any fixed $z_2$, $r_{1,a}(k)\to q$, $\hat{r}_{1,a}(k)\to\hat{q}$ and $d_1^{(-k_0)}\to 1$ as $t\to\infty$. Consequently, we obtain that $\tilde{v}^{(k_0)}_j(x,t,k)$ tends to the jump matrix $v^{X_1}_j$, $j=1,2,3,4,$ defined in the RH problem \ref{pcmodel_k1} in the Appendix \ref{AppendixA}, and $\tilde{v}^{(-k_0)}_j(x,t,k)$ tends to the jump matrix $v^{X_2}_j$, $j=5,6,7,8,$ defined in the RH problem \ref{pcmodel_-k1}. This also indicates that as $t\to\infty$, the jumps of $M^{(2)}$ approach the function $M^{X_1}Y_1^{-1}$ as $k$ tends to $k_0$, and approach the function $M^{X_2}Y_2^{-1}$ as $k$ tends to $-k_0$. In other words, we can use the function $M^{(\pm k_0)}$, defined as follows, to approximate the eigenfunction $M^{(2)}(x,t,k)$ in the neighborhood $B_\epsilon(\pm k_0)$ of $\pm k_0$:
  \begin{align}
      &M^{(k_0)}(x,t,k)=Y_1(\zeta)M^{X_1}(q(\zeta),z_1(\zeta,k))Y^{-1}_1(\zeta), \quad &&k\in B_\epsilon(k_0),\label{Mk_0}\\
      &M^{(-k_0)}(x,t,k)=Y_2(\zeta)M^{X_2}(q(\zeta),z_2(\zeta,k))Y^{-1}_2(\zeta), \quad &&k\in B_\epsilon(-k_0).\label{M-k_0}
  \end{align}
  And $M^{(\pm k_0)}(k)\to I$ on $\partial B_\epsilon(\pm k_0)$ as $t\to \infty$, this ensures $M^{(\pm k_0)}(\zeta,k)$ is a good approximation of $M^{(2)}(\zeta,k)$ in $B_\epsilon(\pm k_0)$ for large $t$.

    \begin{lem}\label{lem_Y12_bound}
        For $\zeta\in\mathcal{I}_2$ and $t\geq2$, the functions $Y_1(\zeta)$ and $Y_2(\zeta)$ are uniformly bounded, i.e.,
        \begin{equation*}
            \sup_{\zeta\in\mathcal{I}_2}\sup_{t\geq2}\left|\partial_x^lY_j^{\pm 1}(\zeta)\right|<C,\quad l=0,1,\,\,j=1,2.
        \end{equation*}
        The functions $d_0^{(\pm k_0)}(\zeta)$ and $d_1^{(\pm k_0)}(\zeta,k)$ satisfy
        \begin{align*}\label{Error_d0}
            &\left|d_0^{(k_0)}(\zeta)\right|=\re^{2\pi\nu},  \quad&&\left|\partial_xd_0^{(k_0)}(\zeta)\right|\leq\frac{C\ln t}{t},\\
            &\left|d_0^{(-k_0)}(\zeta)\right|=1, \quad&&\left|\partial_xd_0^{(-k_0)}(\zeta)\right|\leq\frac{C\ln t}{t},
        \end{align*}
        and 
        \begin{equation}\label{Error_d1-1}
             \begin{aligned}
            &\left|d_1^{(k_0)}(\zeta,k)-1\right|\leq C|k-k_0|(1+|\ln|k-k_0||), \quad\,\,\,\left|\partial_xd_1^{(k_0)}(\zeta,k)\right|\leq\frac{C|\ln |k-k_0||}{t},\\
            &\left|d_1^{(-k_0)}(\zeta,k)-1\right|\leq C|k+k_0|(1+|\ln|k+k_0||),\quad\left|\partial_xd_1^{(-k_0)}(\zeta,k)\right|\leq\frac{C|\ln |k+k_0||}{t}.
        \end{aligned}
        \end{equation}
    
    \end{lem}
    \begin{proof}
        Direct calculation yields
        \begin{equation*}
        	\begin{aligned}
        	 &\rre\chi(\zeta,-k)|_{k=k_0}=-\frac{1}{2\pi}\int_{k_0}^{\infty}\pi \rd\ln (1-8s|r_1(-s)|^2)=\frac{1}{2} \ln (1-8k_0|r_1(-k_0)|^2)=-\pi \nu,\\
        	&\rre\chi(\zeta,k)|_{k=-k_0}=-\frac{1}{2\pi}\int_{k_0}^{\infty}0 \rd\ln (1-8s|r_1(-s)|^2)=0,
        	\end{aligned}
        \end{equation*}
        and consequently $\left|d_0^{(k_0)}(\zeta)\right|=\re^{2\pi\nu}$ and $\left|d_0^{(-k_0)}(\zeta)\right|=1$. According to the conclusion in Lemma \ref{lem_delta},
        \begin{align*}
            \left|\partial_x d_0^{(k_0)}(\zeta)\right|&=\left|d_0^{(k_0)}(\zeta)\partial_x\ln d_0^{(k_0)}(\zeta)\right|=\re^{2\pi\nu}\left|\partial_x\ln d_0^{(k_0)}(\zeta)\right|\leq C\left(|\ln t||\partial_x\nu|+|\partial_x\chi(\zeta,-k_0)|+|\partial_x\ln \delta(k_0)|\right)\\
            &\leq \frac{C}{t}\left(|\ln t||\partial_\zeta\nu|+|\partial_\zeta\chi(\zeta,-k_0)|+|\partial_\zeta\ln \delta(k_0)|\right)\leq \frac{C\ln t}{t}.
        \end{align*}
        The case for $d_0^{(-k_0)}$ is proved similarly. The two expressions on the left-hand side of equation \eqref{Error_d1-1} can also be directly obtained from the conclusions in Lemma \ref{lem_delta} that the function $\delta(\zeta,k)$ satisfies. Finally, $\partial_\zeta \log \frac{\delta(k)}{\delta(k_0)}$ is analytic, and $\left|d_1^{(k_0)}(\zeta,k)\right|\leq C$ for $k\in X^{k_0}$, subsequently
        \begin{align*}
            \left|\partial_x d_1^{(k_0)}(\zeta,k)\right|&=\left|d_1^{(k_0)}(\zeta,k)\partial_x\log d_1^{(k_0)}(\zeta,k)\right|\leq C\left(|\partial_x\left(\chi(\zeta,-k)-\chi(\zeta,-k_0)\right)|+\frac{1}{t}\left|\partial_\zeta\log \frac{\delta(k)}{\delta(k_0)}\right|\right)\leq \frac{C|\ln |k-k_0||}{t}.
        \end{align*}
    \end{proof}

    \begin{lem}\label{lem_Error_v2-vk1}
    For each $\zeta\in\mathcal{I}_2$ and $t\geq 2$, the eigenfunction $M^{(k_0)}$ defined in equation \eqref{Mk_0} is analytic and bounded function for $k\in B_\epsilon(k_0)\setminus X^{ k_0}$, and satisfies the jump relation $M^{(k_0)}_+(x,t,k)=M^{(k_0)}_-(x,t,k)v_j^{(k_0)}(x,t,k)$ across $k\in X_j^{k_0}$, $j=1,2,3,4$. The function $M^{(-k_0)}$ defined in equation \eqref{M-k_0} is analytic and bounded function for $k\in B_\epsilon(-k_0)\setminus X^{- k_0}$, and satisfies the jump relation $M^{(-k_0)}_+(x,t,k)=M^{(-k_0)}_-(x,t,k)v_j^{(-k_0)}(x,t,k)$ across $k\in X_j^{-k_0}$, $j=5,6,7,8$. Moreover, the jump matrices $v^{(\pm k_0)}$ satisfy the following estimates:
        \begin{align}
            &\left\|\partial_x^l\left(v^{(2)}(x,t,\cdot)-v^{(\pm k_0)}(x,t,\cdot)\right)\right\|_{L^1(X^{\pm k_0})}\leq \frac{C \ln t}{t},\label{Error_v2-vk1}\\
            &\left\|\partial_x^l\left(v^{(2)}(x,t,\cdot)-v^{(\pm k_0)}(x,t,\cdot)\right)\right\|_{L^\infty(X^{\pm k_0})}\leq \frac{C \ln t}{\sqrt{t}}.\label{Error_v2-v-k1}
        \end{align}
        In addition,
    \begin{equation}\label{Mk1-I}
        \left\|\partial_x^l\left(\left(M^{(\pm k_0)}(x,t,\cdot)\right)^{-1}-I\right)\right\|_{L^\infty(\partial B_\epsilon(\pm k_0))}=\mathcal{O}\left(t^{-1/2}\right),
    \end{equation}
    \begin{equation}\label{intM-I}
    \begin{aligned}
        &\frac{1}{2\pi\ri} \int_{\partial B_\epsilon(k_0)} \left(\left(M^{(k_0)}(x,t,k)\right)^{-1}-I\right)\rd k=-\frac{Y_1(\zeta)M^{X_1}_1(q(\zeta))Y_1^{-1}(\zeta)}{2\sqrt{t}}+\mathcal{O}\left(t^{-1}\right),\\
        &\frac{1}{2\pi\ri} \int_{\partial B_\epsilon(-k_0)} \left(\left(M^{(-k_0)}(x,t,k)\right)^{-1}-I\right)\rd k=-\frac{Y_2(\zeta)M^{X_2}_1(q(\zeta))Y_2^{-1}(\zeta)}{2\sqrt{t}}+\mathcal{O}\left(t^{-1}\right),
    \end{aligned}
    \end{equation}
     uniformly for $\zeta\in\mathcal{I}_2$ and $l=0,1$. Moreover, \eqref{intM-I} can be differentiated with respect to $x$ without increasing the error term.
    \end{lem}
    \begin{proof}
        According to the definitions of $M^{(k_0)}$ and $M^{(-k_0)}$, we have
        \begin{align*}
            &v^{(2)}(k)-v^{(k_0)}(k)=Y_1(\zeta)\left(\tilde{v}^{(k_0)}(k)-v^{X_1}(z_1(k))\right)Y^{-1}_1(\zeta),\quad &&k\in X^{k_0},\\
            &v^{(2)}(k)-v^{(-k_0)}(k)=Y_2(\zeta)\left(\tilde{v}^{(-k_0)}(k)-v^{X_2}(z_2(k))\right)Y^{-1}_2(\zeta),\quad &&k\in X^{-k_0}.
        \end{align*}
        According to the boundedness of $Y_1(\zeta)$ and $Y_2(\zeta)$ given in Lemma \ref{lem_Y12_bound}, proving that expressions \eqref{Error_v2-vk1} and \eqref{Error_v2-v-k1} are equivalent to proving the following equations hold
        \begin{align*}
            &\left\|\partial_x^l\left(\tilde{v}^{(k_0)}(x,t,\cdot)-v^{X_1}(x,t,\cdot)\right)\right\|_{L^1(X^{k_0}_i)}\leq \frac{C \ln t}{t},\, && \left\|\partial_x^l\left(\tilde{v}^{(k_0)}(x,t,\cdot)-v^{X_1}(x,t,\cdot)\right)\right\|_{L^\infty(X^{k_0}_i)}\leq \frac{C \ln t}{\sqrt{t}},\\
            &\left\|\partial_x^l\left(\tilde{v}^{(-k_0)}(x,t,\cdot)-v^{X_2}(x,t,\cdot)\right)\right\|_{L^1(X^{ -k_0}_j)}\leq \frac{C \ln t}{t},\, && \left\|\partial_x^l\left(\tilde{v}^{(-k_0)}(x,t,\cdot)-v^{X_2}(x,t,\cdot)\right)\right\|_{L^\infty(X^{-k_0}_j)}\leq \frac{C \ln t}{\sqrt{t}},
        \end{align*}
        for $i=1,2,3,4$, $j=5,6,7,8$, and $l=0,1$. We only provide the proof for the case where $i=1$, and the proofs for the other cases are similar.

        For $k\in X^{k_0}_1$, we can obtain
        \begin{equation*}
            \re^{\frac{t}{4}\left|\rre\theta_{23}(\zeta,k)\right|}=\re^{\frac{t}{4}\left|\rre\theta_{23}(\zeta,k)-\rre\theta_{23}(\zeta,k_0)\right|}=\re^{\frac{1}{4}\left|\rre\left(\ri z_1^2/2\right)\right|}\leq \re^{\frac{|z_1|^2}{8}}.
        \end{equation*}
        By using Lemma \ref{lem_r_decomposition} and equation \eqref{Error_d1-1} in the Lemma \ref{lem_Y12_bound}, we can obtain for $k\in X^{k_0}_1$,
        \begin{align*}
            \left|\left(\tilde{v}^{(k_0)}-v^{X_1}\right)_{32}\right|&=\left|\re ^{-2\ri\nu\ln_0(z_1)}\left(d^{(k_0)}_{1}\right)^{-1}\hat{r}_{1,a}(-k)\re^{\frac{\ri z_1^2}{2}}-\hat{r}_{1,a}(-k_0)\re^{-2\ri \nu \ln_0(z_1)}\re^{\frac{\ri z_1^2}{2}}\right|\\
            &=\left|\re ^{-2\ri\nu\ln_0(z_1)}\right|\left|\re^{\frac{\ri z_1^2}{2}}\right|\left|\left(\left(d^{(k_0)}_{1}\right)^{-1}-1\right)\hat{r}_{1,a}(-k)+\left(\hat{r}_{1,a}(-k)-\hat{r}_{1,a}(-k_0)\right)\right|\\
            &\leq C \left(\left|\left(d^{(k_0)}_{1}\right)^{-1}-1\right|\left|\hat{r}_{1,a}(-k)\right|+\left|\hat{r}_{1,a}(-k)-\hat{r}_{1,a}(-k_0)\right|\right)\re^{-|z_1|^2/2}\\
            &\leq C|k-k_0|(1+|\ln|k-k_0||)\re^{-ct|k-k_0|^2}.
        \end{align*}
        Thus, it follows that
        \begin{equation*}
            \left\|\left(\tilde{v}^{(k_0)}-v^{X_1}\right)_{32}\right\|_{L^1\left(X_1^{k_0}\right)}\leq C\int_{0}^{\infty} u (1+|\ln u|)\re^{-ctu^2}\rd u \leq \frac{C \ln t}{t},
        \end{equation*}
        \begin{equation*}
            \left\|\left(\tilde{v}^{(k_0)}-v^{X_1}\right)_{32}\right\|_{L^\infty\left(X_1^{k_0}\right)}\leq C\sup_{u\geq 0} u (1+|\ln u|)\re^{-ctu^2}\leq \frac{C \ln t}{\sqrt{t}}.
        \end{equation*}

        Next, we prove the estimate for the partial derivative with respect to $x$, dividing it into the following three parts:
        \begin{align*}
            \partial_x\left(\tilde{v}^{(k_0)}-v^{X_1}\right)_{32}&=\partial_x\left(\re ^{-2\ri\nu\ln_0(z_1)}\right)\re^{\frac{\ri z_1^2}{2}}\left(\left(d^{(k_0)}_{1}\right)^{-1}-\hat{r}_{1,a}(-k_0)\right)+\re ^{-2\ri\nu\ln_0(z_1)}\partial_x\Bigl(\re^{\frac{\ri z_1^2}{2}}\Bigr)\left(\left(d^{(k_0)}_{1}\right)^{-1}-\hat{r}_{1,a}(-k_0)\right)\\
            &\quad+\re ^{-2\ri\nu\ln_0(z_1)}\re^{\frac{\ri z_1^2}{2}}\partial_x\left(\left(d^{(k_0)}_{1}\right)^{-1}-\hat{r}_{1,a}(-k_0)\right).
        \end{align*}
        According to $\left|\partial_x\left(\re ^{-2\ri\nu\ln_0(z_1)}\right) \right|\leq \frac{C}{t|k-k_0|}$, we have
        \begin{equation*}
            \left\|\partial_x\left(\re ^{-2\ri\nu\ln_0(z_1)}\right)\re^{\frac{\ri z_1^2}{2}}\left(\left(d^{(k_0)}_{1}\right)^{-1}-\hat{r}_{1,a}(-k_0)\right)\right\|_{L^1\left(X_1^{k_0}\right)}\leq Ct^{-1} \int_{0}^{\infty} (1+\ln u)\re^{-ctu^2}\rd u\leq \frac{C\ln t}{t^{3/2}},
        \end{equation*}
        \begin{equation*}
            \left\|\partial_x\left(\re ^{-2\ri\nu\ln_0(z_1)}\right)\re^{\frac{\ri z_1^2}{2}}\left(\left(d^{(k_0)}_{1}\right)^{-1}-\hat{r}_{1,a}(-k_0)\right)\right\|_{L^\infty\left(X_1^{k_0}\right)}\leq Ct^{-1} \sup_{u\geq0} (1+\ln u)\re^{-ctu^2}\leq \frac{C\ln t}{t}.
        \end{equation*}
    The remaining two expressions can also be proved similarly to complete the proof of equations \eqref{Error_v2-vk1} and \eqref{Error_v2-v-k1}.

    According to the definitions of $z_1$ and $z_2$ in equations \eqref{z_1} and \eqref{z_2}, it is known that as $t\to\infty$, the variable $z_1\to\infty$ if $k\in \partial B_\epsilon(k_0)$. Similarly, the variable $z_2\to\infty$ if $k\in \partial B_\epsilon(-k_0)$. From the expansion equations satisfied by \eqref{Mx1_expand} and \eqref{Mx2_expand} of the model RH problems \ref{pcmodel_k1} and \ref{pcmodel_-k1} given in the Appendix \ref{AppendixA}, we can obtain
   \begin{align*}
        &M^{X_1}(q(\zeta),z_1(\zeta,k))=I+\frac{M_1^{X_1}(q(\zeta))}{2\sqrt{t}(k-k_0)}+\mathcal{O}\left(t^{-1}\right), \quad &&k\in \partial
         B_\epsilon(k_0),\,\,t\to\infty,\\
        &M^{X_2}(q(\zeta),z_2(\zeta,k))=I+\frac{M_1^{X_2}(q(\zeta))}{2\sqrt{t}(k+k_0)}+\mathcal{O}\left(t^{-1}\right), \quad &&k\in \partial
         B_\epsilon(-k_0),\,\,t\to\infty,
    \end{align*}
  uniformly for $\zeta\in\mathcal{I}_2$, and the above equations can be differentiated with respect to $x$ without increasing the error term. Combining equations \eqref{Mk_0} and \eqref{M-k_0}, we obtain as $t\to\infty$,
  \begin{align*}
        &\left(M^{(k_0)}\right)^{-1}-I=-\frac{Y_1(\zeta)M_1^{X_1}(q(\zeta))Y_1^{-1}(\zeta)}{2\sqrt{t}(k-k_0)}+\mathcal{O}\left(t^{-1}\right), \\
        &\left(M^{(-k_0)}\right)^{-1}-I=-\frac{Y_2(\zeta)M_1^{X_2}(q(\zeta))Y_2^{-1}(\zeta)}{2\sqrt{t}(k+k_0)}+\mathcal{O}\left(t^{-1}\right), 
    \end{align*}
    uniformly for $k\in \partial B_\epsilon(\pm k_0)$. By using the conclusion in Lemma \ref{lem_Y12_bound} and the Cauchy formula, we can complete the proof of equations \eqref{Mk1-I} and \eqref{intM-I} from the above equation.\qedhere

    \end{proof}

     \subsubsection{A small-norm RH problem}
      According to the definitions of functions $M^{(\pm k_0)}$ given in the previous subsection, we construct 
      \begin{equation}\label{hatM}
        \hat{M}=\left\{
        \begin{aligned}
            &M^{(2)}\left(M^{(k_0)}\right)^{-1},\quad &&k\in B_\epsilon(k_0),\\
            &M^{(2)}\left(M^{(-k_0)}\right)^{-1},\quad &&k\in B_\epsilon(-k_0),\\
            &M^{(2)},\quad &&\text{elsewhere}.
        \end{aligned}
        \right.
    \end{equation}
    The contour corresponding to the function $\hat M$ are shown in Figure \ref{fig_hatGamma}, and the jump matrix $\hat{v}$ are
    \begin{equation}\label{hatv}
        \hat{v}=\left\{
        \begin{aligned}
            &M^{(\pm k_0)}_-v^{(2)}\left(M^{(\pm k_0)}_+\right)^{-1},\quad &&k\in \hat{\Gamma} \cap B_\epsilon(\pm k_0),\\
            & \left(M^{(\pm k_0)}\right)^{-1}, && k\in \partial B_\epsilon(\pm k_0),\\
            &v^{(2)},\quad &&k\in \hat{\Gamma} \setminus B_\epsilon.
        \end{aligned}
        \right.
    \end{equation}

    \begin{figure}[htbp]
    	\centering
    	\begin{tikzpicture}[scale=1]
        
        \draw[very thick, black!40!green] (-4,0) -- (4,0);
        \draw[very thick, black!40!green] (0,2) -- (0,3.5);
        \draw[very thick, black!40!green] (0,-2) -- (0,-3.5);

        \draw[very thick, dashed, black!40!green] (0,-2) -- (0,2);
    		
    		\draw[very thick, black!40!green, -latex] (-4,0) -- (-3.2,0);
            \draw[very thick, black!40!green, -latex] (-4,0) -- (-0.5,0);
            \draw[very thick, black!40!green, -latex] (-4,0) -- (0.8,0);
            \draw[very thick, black!40!green, -latex] (-4,0) -- (3.5,0);

            \draw[very thick, black!40!green, -latex] (0,-2) -- (0,-3);
            \draw[very thick, black!40!green, -latex] (0,2) -- (0,3);

            \draw[very thick, black!40!green] (2+2/1.414,2/1.414) -- (0,-2);
            \draw[very thick, black!40!green] (2+2/1.414,-2/1.414) -- (-0,2);

            \draw[very thick, black!40!green] (-2-2/1.414,2/1.414) -- (0,-2);
            \draw[very thick, black!40!green] (-2-2/1.414,-2/1.414) -- (0,2);

            \draw[very thick, black!40!green,-latex] (2,0) --(2+1.3/1.2,1.3/1.2);
            \draw[very thick, black!40!green,-latex] (2,0) -- (-1.3/1.2+2,-1.3/1.2);
            \draw[very thick, black!40!green,-latex] (2,0) -- (-1.3/1.2+2,1.3/1.2);
            \draw[very thick, black!40!green,-latex] (2,0) -- (2+1.3/1.2,-1.3/1.2);

            \draw[very thick,black!40!green,-latex] (-2,0) --(-2+1.3/1.2,1.3/1.2);
            \draw[very thick, black!40!green,-latex] (-2,0) -- (-1.3/1.2-2,-1.3/1.2);
            \draw[very thick, black!40!green,-latex] (-2,0) -- (-1.3/1.2-2,1.3/1.2);
            \draw[very thick, black!40!green,-latex] (-2,0) -- (-2+1.3/1.2,-1.3/1.2);

    		\fill (-2,0) circle (1.5pt);
            \fill (0,0) circle (1.5pt);
	      \fill (2,0) circle (1.5pt);

           \node[below] at (-2,-0.1) {$-k_0$};
	     \node[below] at (0.3,0) {$\text{0}$};
	     \node[below] at (2,-0.1) {$k_0$};

         \draw[very thick,black!20!blue] (-2,0) circle (1cm);
         \draw[very thick,black!20!blue] (2,0) circle (1cm);

         \draw[very thick, black!20!blue,-latex] (-1.9,1) -- (-2.1,1);
         \draw[very thick, black!20!blue,-latex] (2.1,1) -- (1.9,1);
         \node [above] at (-2,1) {$\partial B_\epsilon(-k_0)$};
         \node [above] at (2,1) {$\partial B_\epsilon(k_0)$};

          \node[right] at (4,0) {$\rre k$};
    		
    	\end{tikzpicture}
    	\caption{The jump contour $\hat \Gamma$ in the complex $k$-plane.}
    	\label{fig_hatGamma}
    \end{figure}

    Furthermore, we  refer to the definition of the space $\dot{E}^3(\mathbb{C} \setminus \hat{\Gamma})$ to include analytic functions $f : \mathbb{C} \setminus\hat{\Gamma} \rightarrow \mathbb{C}$ with the property that for each component $D$ of $\mathbb{C} \setminus \hat{\Gamma}$ there exist curves $\{C_n\}_1^\infty$ in $D$ such that the $C_n$ eventually surround each compact subset of $D$ and
   \begin{equation*}
       \sup_{n \geq 1} \int_{C_n} (1 + |k|) |f(k)|^3 |\rd k| < \infty.
   \end{equation*}
    For specific details, refer to Ref. \cite{Lenells_2018}.

    \begin{rhp}\label{rhp_hatM}
        Find a $3\times3$ matrix-valued function $\hat{M}(x,t,\cdot)\in I+\dot E^3(\mathbb{C}\setminus \hat{\Gamma})$ such that $\hat{M}_+(x,t,k)=\hat{M}_-(x,t,k)\hat{v}(x,t,k)$ for $k\in \hat{\Gamma}$, where $\hat{\Gamma}$ is shown in Figure \ref{fig_hatGamma}.
    \end{rhp}
    Define
    \begin{equation*}
        \tilde\Gamma=\hat{\Gamma}\setminus\left(\mathbb{R}\cup X^{\pm k_0}\cup \partial B_\epsilon(\pm k_0)\right).
    \end{equation*}

    \begin{lem}\label{lem_hatw}
    Let $\hat{w}=\hat{v}-I$, the following estimates hold uniformly for $\zeta\in\mathcal{I}_2$, $t\geq2$ and $l=0,1$:
     \begin{align}
            &\left\|(1+|\cdot|)\partial_x^l \hat{w}\right\|_{\left(L^1\cap L^\infty\right)(\mathbb{R})}\leq \frac{C}{t^{N}},\label{Error_hatw_1}\\
            &\left\|(1+|\cdot|)\partial_x^l \hat{w}\right\|_{\left(L^1\cap L^\infty\right)(\tilde{\Gamma})}\leq C\re^{-ct},\label{Error_hatw_2}\\
            &\left\|\partial_x^l \hat{w}\right\|_{\left(L^1\cap L^\infty\right)(\partial B_\epsilon(\pm k_0))}\leq \frac{C}{\sqrt{t}},\label{Error_hatw_3}\\
            &\left\|\partial_x^l \hat{w}\right\|_{L^1(X^{\pm k_0})}\leq \frac{C\ln t}{t},\label{Error_hatw_4}\\
            &\left\|\partial_x^l \hat{w}\right\|_{L^\infty(X^{\pm k_0})}\leq \frac{C\ln t}{\sqrt{t}}.\label{Error_hatw_5}
        \end{align}
    \end{lem}
    \begin{proof}
        In fact, the result of equation \eqref{Error_hatw_1} has already been proven in Lemma \ref{lem_v2_error}. For equation \eqref{Error_hatw_2}, we consider the case where $k\in \Gamma^{(2)}_1\setminus X^{k_0}_1$. Given that $k=k_0+u\re^{\frac{\pi \ri }{4}}$, $u>\frac{k_0}{2}$ and $\hat{w}$ only the (32) element is non-zero. Thus we have
        \begin{equation*}
            \left|\hat{w}_{32}\left( k_0+u\re^{\frac{\pi \ri }{4}}\right) \right|\leq C \left| \hat{r}_{1,a}\left( -k_0-u\re^{\frac{\pi \ri }{4}}\right) \right| \re^{-t\theta_{23}}\leq C \re^{-ct(k_0+u)^2}.
        \end{equation*}
        The proof for $\partial_x\hat{w}$ can also be obtained using Lemmas \ref{lem_r_decomposition} and \ref{lem_delta}, and the proofs for the other cases for $k\in \tilde\Gamma\setminus\left(\Gamma^{(2)}_1\setminus X^{k_0}_1\right)$ are similar. The proof for equation \eqref{Error_hatw_3} can be directly referred to from  equation \eqref{Mk1-I}. 

        For $k\in X^{\pm k_0}$, $\hat{w}=M_-^{(\pm k_0)}\left(v^{(2)}-v^{(\pm k_0)}\right)\left(M_+^{(\pm k_0)}\right)^{-1}$, so equations \eqref{Error_hatw_4} and \eqref{Error_hatw_5} follow from \eqref{Error_v2-vk1} and \eqref{Error_v2-v-k1} combining $\partial _\zeta^lM^{(\pm k_0)}_\pm $, $l=0,1$, and its inverse function, which are uniformly bounded with respect to $k\in X^{\pm k_0}$.
    \end{proof}

    The Cauchy transform is defined in the following form 
    \begin{equation*}
        \left(\mathcal{C}f\right)(z)=\frac{1}{2\pi\ri} \int_{\hat{\Gamma}}\frac{f(z')}{z'-z}\rd z',\quad z\in \mathbb{C}\setminus\hat{\Gamma},
    \end{equation*}
   where the function $f$ is defined on $\hat{\Gamma}$. The Cauchy transform $\left(\mathcal{C}f\right)(z)$ is analytic from $\mathbb{C} \setminus \hat{\Gamma}$ to $\mathbb{C} $, 
  provided that $(1 + |z|)^{1/3} f(z) \in L^3(\hat{\Gamma})$ (i.e., $f\in \dot{L}^3(\hat{\Gamma})$). Moreover, $\mathcal{C}_{\pm} f $ exist almost everywhere for $z \in \hat{\Gamma} $, $\mathcal{C}_+-\mathcal{C}_-=I$ and $\mathcal{C}_{\pm} \in \mathcal{B}(\dot L^3(\hat{\Gamma})) $, where $\mathcal{B}(\dot L^3(\hat{\Gamma})) $ denotes the space of bounded linear operators on $\dot L^3(\hat{\Gamma})$.

  Lemma \ref{lem_hatw} shows that
  \begin{equation*}
  \left\{
   \begin{aligned}
            &\left\|(1+|\cdot|)\partial_x^l \hat{w}\right\|_{L^1(\hat{\Gamma})}\leq \frac{C}{\sqrt{t}},\\
            &\left\|(1+|\cdot|)\partial_x^l \hat{w}\right\|_{L^\infty(\hat{\Gamma})}\leq \frac{C\ln t}{\sqrt{t}},
        \end{aligned}
      \right.\qquad \zeta\in\mathcal{I}_2,\,\,t\geq2,\,\,l=0,1,
  \end{equation*}
  and the calculation yields
  \begin{equation}\label{hatw_Lp}
      \left\|(1+|\cdot|)\partial_x^l \hat{w}\right\|_{L^p(\hat{\Gamma})}\leq \frac{C\ln t^{\frac{p-1}{p}}}{\sqrt{t}},\quad 1\leq p\leq \infty.
  \end{equation}
  The above estimates show that $\hat{w}\in (\dot{L}^3\cap L^\infty)(\hat{\Gamma})$, and thus we can define:
  \begin{equation*}
      \mathcal{C}_{\hat{w}(x,t,\cdot)}f=\mathcal{C}_-(f\hat
      w),\qquad \mathcal{C}_{\hat{w}}: (\dot{L}^3\cap L^\infty)(\hat{\Gamma})\to\dot{L}^3(\hat{\Gamma}).
  \end{equation*}
  \begin{lem}\label{lem_I-C_inverse}
      For any sufficiently large $t$ and $\zeta\in\mathcal{I}_2$, $I-\mathcal{C}_{\hat{w}(x,t,\cdot)}\in \mathcal{B}(\dot L^3(\hat{\Gamma})) $ is invertible.
  \end{lem}
  \begin{proof}
      For each $f\in \dot L^3(\hat{\Gamma})$, we have
      \begin{equation*}
          \left\|\mathcal{C}_{\hat w}f\right\|_{\dot L^3(\hat{\Gamma})}\leq \left\|\mathcal{C}_-\right\|_{\mathcal{B}(\dot L^3(\hat{\Gamma}))}\left\|\hat{w}\right\|_{L^\infty(\hat{\Gamma})}\left\|f\right\|_{\dot L^3(\hat{\Gamma})}.
      \end{equation*}
      Thus $\left\|\mathcal{C}_{\hat w}\right\|_{\mathcal{B}(\dot L^3(\hat{\Gamma}))}\leq \left\|\mathcal{C}_-\right\|_{\mathcal{B}(\dot L^3(\hat{\Gamma}))}\left\|\hat{w}\right\|_{L^\infty(\hat{\Gamma})}\leq \frac{C\ln t}{\sqrt{t}}\left\|\mathcal{C}_-\right\|_{\mathcal{B}(\dot L^3(\hat{\Gamma}))}$. That is to say, we can find a sufficiently large $t$ such that $\left\|\hat{w}\right\|_{L^\infty(\hat{\Gamma})}\leq \left(\left\|\mathcal{C}_-\right\|_{\mathcal{B}(\dot L^3(\hat{\Gamma}))}\right)^{-1}$, the operator $I-\mathcal{C}_{\hat{w}(x,t,\cdot)}\in \mathcal{B}(\dot L^3(\hat{\Gamma}))$  becomes invertible.
  \end{proof}

  According to Lemma \ref{lem_I-C_inverse} mentioned above, we can define the function $\hat{\mu}(x,t,k)$ for sufficiently large $t$, $\zeta\in\mathcal{I}_2$ and $k\in \hat{\Gamma}$ using the operator $\left(I-\mathcal{C}_{\hat{w}}\right)^{-1}$:
  \begin{equation*}
      \hat{\mu}=I+\left(I-\mathcal{C}_{\hat{w}}\right)^{-1}\mathcal{C}_{\hat{w}}I\in I+\dot L^3(\hat{\Gamma}).
  \end{equation*}

  \begin{lem}\label{lem_hatM_solution}
      For sufficiently large time $t$, the RH problem \ref{rhp_hatM} has a unique solution $\hat{M}\in I+\dot{E}^3(\mathbb{C} \setminus \hat{\Gamma})$, which can be written in the following form
      \begin{equation*}
          \hat{M}(x,t,k)=I+\mathcal{C}(\hat\mu \hat w)=I+\frac{1}{2\pi\ri}\int_{\hat{\Gamma}}\frac{\hat{\mu}(x,t,k')\hat{w}(x,t,k')}{k'-k}\rd k'.
      \end{equation*}
  \end{lem}
  \begin{lem}\label{Error_mu-I}
      For sufficiently large $t$ and $\zeta\in\mathcal{I}_2$, we have
      \begin{equation*}
          \left\|\partial_x^l(\hat{\mu}-I)\right\|_{L^p(\hat{\Gamma})}\leq \frac{C\ln t^{\frac{p-1}{p}}}{\sqrt{t}}, \quad 1<p<\infty, \,\, l=0,1.
      \end{equation*}
  \end{lem}
  \begin{proof}
      Let $\|\mathcal{C}_-\|_p:=\|\mathcal{C}_-\|_{\mathcal{B}(L^p(\hat{\Gamma}))}<\infty$ and assume the time $t$ is sufficiently large such that $\|\hat{w}\|_{L^\infty(\hat{\Gamma})}<\|\mathcal{C}_-\|_p^{-1}$. Using the definition of $\hat{\mu}$ and Neumann series:
      \begin{equation*}
          \|\hat{\mu}-I\|_{L^p(\hat{\Gamma})}\leq \sum_{j=0}^{\infty}\|\mathcal{C}_{\hat{w}}\|_{\mathcal{B}(L^p(\hat{\Gamma}))}^{j}\|\mathcal{C}_{\hat{w}}I\|_{L^p(\hat{\Gamma})}\leq \sum_{j=0}^{\infty}\|\mathcal{C}_-\|_p^{j+1}\|\hat{w}\|_{L^\infty(\hat{\Gamma})}^{j}\|\hat{w}\|_{L^p(\hat{\Gamma})}=\frac{\|\mathcal{C}_-\|_p\|\hat{w}\|_{L^p(\hat{\Gamma})}}{1-\|\mathcal{C}_-\|_p\|\hat{w}\|_{L^\infty(\hat{\Gamma})}}.
      \end{equation*}
      Thus, based on the estimates in equation \eqref{hatw_Lp}, we can complete the proof for the case $l=0$ in the lemma. For $l=1$,
      \begin{align*}
          \|\partial_x(\hat{\mu}-I)\|_{L^p(\hat{\Gamma})}&\leq \sum_{j=1}^{\infty}j\|\mathcal{C}_{\hat{w}}\|_{\mathcal{B}(L^p(\hat{\Gamma}))}^{j-1}\|\partial_x\mathcal{C}_{\hat{w}}\|_{\mathcal{B}(L^p(\hat{\Gamma}))}\|\mathcal{C}_{\hat{w}}I\|_{L^p(\hat{\Gamma})}+\sum_{j=0}^{\infty}\|\mathcal{C}_{\hat{w}}\|_{\mathcal{B}(L^p(\hat{\Gamma}))}^{j}\|\partial_x\mathcal{C}_{\hat{w}}I\|_{L^p(\hat{\Gamma})}\\
          &\leq C\frac{\|\partial_x\hat{w}\|_{L^\infty(\hat{\Gamma})}\|\hat{w}\|_{L^p(\hat{\Gamma})}+\|\partial_x\hat{w}\|_{L^p(\hat{\Gamma})}}{1-\|\mathcal{C}_-\|_p\|\hat{w}\|_{L^\infty(\hat{\Gamma})}}.
      \end{align*}
      By the same reasoning, the lemma is proven.
  \end{proof}

   \subsubsection{Asymptotics of $\phi(x,t)$ and $n(x,t)$}
   In the previous subsection, we used the Cauchy operator to present the form of the solution to the RH problem \ref{rhp_hatM}. This subsection will utilize the aforementioned results to provide the asymptotic forms of the solutions $\phi(x,t)$ and $n(x,t)$ of the YO equation \eqref{YOE} in region $\zeta\in\mathcal{I}_2$ as $t\to\infty$. First, define the nontangential limits
   \begin{equation*}
       Q(x,t)=\lim_{k\to\infty}k\left(\hat{M}(x,t,k)-I\right)=-\frac{1}{2\pi\ri}\int_{\hat{\Gamma}}\hat{\mu}(x,t,k)\hat{w}(x,t,k)\rd k.
   \end{equation*}
   \begin{lem}\label{lem_Q}
       As $t\to\infty$, we have
       \begin{equation*}\label{Q_leading}
           Q(x,t)=-\frac{1}{2\pi\ri}\int_{\partial B_\epsilon }\hat{w}(x,t,k)\rd k+\mathcal{O}\left(\frac{\ln t}{t}\right).
       \end{equation*}
       Furthermore, the above equation can be differentiated term by term with respect to $x$ without increasing its error term.
   \end{lem}
   \begin{proof}
       To complete the proof, we decompose $Q(x,t)$ into the following three parts
       \begin{align*}
           Q(x,t)=-\frac{1}{2\pi\ri}\int_{\partial B_\epsilon}\hat{w}(x,t,k)\rd k-\frac{1}{2\pi\ri}\int_{\hat{\Gamma}\setminus\partial B_\epsilon}\hat{w}(x,t,k)\rd k-\frac{1}{2\pi\ri}\int_{\hat{\Gamma}}(\hat{\mu}(x,t,k)-I)\hat{w}(x,t,k)\rd k,
       \end{align*}
       The remaining results can be directly obtained from Lemmas \ref{lem_hatw} and \ref{Error_mu-I}.
   \end{proof}

   Combining the definition \eqref{hatM} of function $\hat{M}$ and its associated jump matrix $\hat{v}$ in \eqref{hatv}, the following relationship holds for $t\to\infty$:
   \begin{align*}
       Q(x,t)&=-\frac{1}{2\pi\ri}\int_{\partial B_\epsilon }\left(\left(M^{(\pm k_0)}\right)^{-1}-I\right)\rd k+\mathcal{O}\left(\frac{\ln t}{t}\right)=\frac{Y_1(\zeta)M^{X_1}_1(q(\zeta))Y_1^{-1}(\zeta)}{2\sqrt{t}}+\frac{Y_2(\zeta)M^{X_2}_1(q(\zeta))Y_2^{-1}(\zeta)}{2\sqrt{t}}+\mathcal{O}\left(\frac{\ln t}{t}\right)\\
       &=\frac{1}{2\sqrt{t}}\begin{pmatrix}
            0 & d_{0}^{(-k_0)}\re^{t\theta_{12}(-k_0)}\alpha_{12} & 0\\
             \left(d_{0}^{(-k_0)}\right)^{-1}\re^{-t\theta_{12}(-k_0)}\alpha_{21} & 0 & d_{0}^{(k_0)}\re^{t\theta_{23}(k_0)}\alpha_{23}\\
             0 & \left(d_{0}^{(k_0)}\right)^{-1}\re^{-t\theta_{23}(k_0)}\alpha_{32} & 0\\
        \end{pmatrix}+\mathcal{O}\left(\frac{\ln t}{t}\right),
   \end{align*}
   where $\alpha_{ij}$ are given in the Appendix \ref{AppendixA}.

   By tallying all the transformations made to $M(x,t,k)$ in this section, we obtain
   \begin{equation*}
       M(x,t,k)=\hat{M}(x,t,k)T^{-1}(x,t,k) \Delta^{-1}(k), \quad k\in \mathbb{C}\setminus\overline{B_\epsilon},
   \end{equation*}
   where $\Delta$ and $T$ are defined in equations \eqref{Delta} and \eqref{T}. Using the reconstruction formula \eqref{reconstruct}, we obtain
    \begin{align*}
        n(x,t)&=2\ri \frac{\partial}{\partial x}\lim\limits_{k\to\infty}k\left[\begin{pmatrix}
            1 & 0 & 1\\
        \end{pmatrix}\left(\hat{M}T^{-1} \Delta^{-1}-I\right)\right]_{1}=\mathcal{O}\left(\frac{\ln t}{t}\right),\\
    \phi(x,t)&=\ri \lim\limits_{k\to\infty}\left[\begin{pmatrix}
            0 & 1 & 0\\
        \end{pmatrix}\left(\hat{M}T^{-1} \Delta^{-1}-I\right)\right]_{1}\\
    &=\frac{\sqrt{2\pi}\left(2\sqrt{t}\right)^{-2\ri\nu}{\rm{exp}}\left(\ri\nu \ln(2k_0)+\frac{3\pi \ri}{4}+2\ri tk_0^2-\frac{\pi\nu}{2}+s_1\right)}{2\sqrt{t}r_1(-k_0)\Gamma(-\ri\nu)}+\mathcal{O}\left(\frac{\ln t}{t}\right),
    \end{align*}
    where
    \begin{align*}
    &s_1=\frac{1}{2\pi\ri}\int_{k_0}^{\infty}\ln\left({(s+k_0)}{(s-k_0)^2}\right){\rm{d}}\ln(1-8s\left|r_1(-s)\right|^2 ).
    \end{align*}
    Obviously, as a complex component in the YO equation \eqref{YO}, we provide the modulus of the variable $\phi(x,t)$:
    \begin{align*}
        |\phi(x,t)|\simeq\frac{\sqrt{2k_0\nu}}{\sqrt{t}}.
    \end{align*}

    \subsection{Long-time asymptotics in Region {\rm{I}}} 
    When $(x,t)$ in Region \text{I}, i.e., the parameter $M\leq\zeta\leq\infty$ for sufficiently large positive constant $M$. In this region, we provide different parameter $\tau$ and $0\leq\tau\leq M^{-1}$. Consequently, the dispersion relations satisfy:
    \begin{equation*}
    	t\theta_{ij}(\zeta,k)=x\tilde\theta_{ij}(\tau,k),\quad  1\leq i<j\leq3,
    \end{equation*} 
     where $\tilde\theta_{ij}(\tau,k)$ takes the following form:
    \begin{equation}\label{theta_tilde}
    	\begin{aligned}
    		&\tilde \theta_{12}(\tau,k):=2\ri k\left((k-1)\tau +1\right),\quad\tilde \theta_{13}(\tau,k):=4\ri k\left(1-\tau\right),\quad\tilde \theta_{23}(\tau,k):=-2\ri k\left((k+1)\tau-1\right).
    	\end{aligned}
    \end{equation}
    And it can be denoted as $k_0=\frac{1-\tau}{2\tau}$. The signatures of functions $\tilde \theta_{12}(\tau,k)$, $\tilde \theta_{13}(\tau,k)$, and $\tilde \theta_{23}(\tau,k)$ for $\tau\in\mathcal{I}_1$ are identical to those in Figure \ref{fig_theta_sign1}, and are therefore not repeated here. The following two decompositions hold:
    \begin{equation*}
    	\begin{aligned}
    	&x\tilde\theta_{12}(\tau,k)=x\tilde\theta_{12}(\tau,-k_0)+\frac{\ri}{2}\left( 2\sqrt{x\tau}(k+k_0)\right) ^2,\\
    	&x\tilde\theta_{23}(\tau,k)=x\tilde\theta_{12}(\tau,k_0)-\frac{\ri}{2}\left( 2\sqrt{x\tau}(k-k_0)\right) ^2.
    	\end{aligned}
    \end{equation*}
    In fact, we find that the analysis of the asymptotic behavior for the solution to the initial-value problem \eqref{YOE} of the YO equation in Region \text{I} is very similar to the analysis for Region \text{II} in the previous subsection. However, in this region, we utilize the variable $\tau$ and consider the spatial variable $x\to+\infty$; the error estimates for the asymptotics are also in terms of $x$. Here, we will not repeat the same operations, and can directly obtain that the asymptotic behavior of the solution to initial-value problem \eqref{YOE} in Region \text{I} can be described by \eqref{region_1}.

    \begin{remark}
    	Since as $\tau \to 0$, $k_0\to+\infty$, and according to Lemma \ref{lem_r_decomposition}, $r_1(k)$ vanishes to all orders at $k=k_0$, it follows that the functions $\nu$ and $s_1$ vanish to all orders as $\tau \to 0$. Consequently, equation \eqref{region_1} implies that as $x\to+\infty$:
    	\begin{equation*}
    		\begin{aligned}
    			n(x,t)&=\mathcal{O}\left(\frac{1}{x^{N}}+\frac{C_N(\tau)}{x}\right),\qquad
    			\phi(x,t)=\mathcal{O}\left(\frac{1}{x^{N}}+\frac{C_N(\tau)}{x}\right),
    		\end{aligned}
    	\end{equation*}
    	uniformly for $\tau\in\mathcal{I}_1$. In particular, for any fixed $t\geq 0$, the above expression can be reduced to $\mathcal{O}\left({x^{-N}}\right) $ as $x \to +\infty$.
    \end{remark}

\subsection{Long-time asymptotics in Region {\rm{III}}}\label{subsec_region2}
Unlike in region $\zeta\in\mathcal{I}_2$ discussed in the previous section, when $\zeta\in\mathcal{I}_3$ the signature of the dispersion relations \eqref{theta} take the opposite forms shown in Figure \ref{fig_theta_sign2}. It should be noted that in this region $k_0<0$. Moreover, the proofs of most lemmas in this subsection are analogous to those in Subsection \ref{subsec_region1}; for details, refer to the preceding subsection and hence are omitted here.

      \begin{figure}[htbp]
    \centering
    \begin{subfigure}[t]{0.28\textwidth}
        \centering
     	\begin{tikzpicture} [scale=0.7]
        
         \fill[black!20!blue!20] (-3,0) -- (1.5,0) -- (1.5,3) -- (-3,3) -- cycle;
        \fill[black!20!blue!20] (1.5,0) -- (3,0) -- (3,-3) -- (1.5,-3) -- cycle;
	 		
	  \draw [very thick,black!20!blue](-3,0) -- (3,0);
        \draw [very thick,black!20!blue](1.5,-3) -- (1.5,3);

        \fill (-1.5,0) circle (1.5pt);
        \fill (0,0) circle (1.5pt);
	  \fill (1.5,0) circle (1.5pt);

        \node[below] at (1.8,0) {$-k_0$};
	  \node[below] at (0,0) {$\text{0}$};
	  \node[below] at (-1.5,0) {$k_0$};

       \node at (2.4,2) {$V_1^{(12)}$};
      \node at (-1,2) {$V_2^{(12)}$};
      \node at (-1,-2) {$V_3^{(12)}$};
      \node at (2.4,-2) {$V_4^{(12)}$};

        \node[right] at (3,0) {$\mathbb{R}$};
     	\end{tikzpicture}
     		\caption{The signature of $\theta_{12}$.}
    \end{subfigure}
    \quad
    \begin{subfigure}[t]{0.28\textwidth}
        \centering
     	\begin{tikzpicture} [scale=0.7]
        
         \fill[black!20!blue!20] (-3,0) -- (3,0) -- (3,3) -- (-3,3) -- cycle;
         \fill[white] (-3,0) -- (3,0) -- (3,-3) -- (-3,-3) -- cycle;
	 		
	  \draw [very thick,black!20!blue](-3,0) -- (3,0);

        \fill (-1.5,0) circle (1.5pt);
        \fill (0,0) circle (1.5pt);
	  \fill (1.5,0) circle (1.5pt);

        \node[below] at (-1.5,0) {$k_0$};
	  \node[below] at (0,0) {$\text{0}$};
	  \node[below] at (1.5,0) {$-k_0$};

       \node at (0,2) {$V_1^{(13)}$};
      \node at (0,-2) {$V_2^{(13)}$};

        \node[right] at (3,0) {$\mathbb{R}$};
     	\end{tikzpicture}
     		\caption{The signature of $\theta_{13}$.\label{subfig_theta13_2}}
    \end{subfigure}
    \quad
    \begin{subfigure}[t]{0.28\textwidth}
        \centering
     	\begin{tikzpicture} [scale=0.7]
        
         \fill[black!20!blue!20] (3,0) -- (-1.5,0) -- (-1.5,3) -- (3,3) -- cycle;
        \fill[black!20!blue!20] (-1.5,0) -- (-3,0) -- (-3,-3) -- (-1.5,-3) -- cycle;
	 		
	  \draw [very thick,black!20!blue](-3,0) -- (3,0);
        \draw [very thick,black!20!blue](-1.5,-3) -- (-1.5,3);

        \fill (-1.5,0) circle (1.5pt);
        \fill (0,0) circle (1.5pt);
	  \fill (1.5,0) circle (1.5pt);

        \node[below] at (-1.2,0) {$k_0$};
	  \node[below] at (0,0) {$\text{0}$};
	  \node[below] at (1.5,0) {$-k_0$};

      \node at (-2.4,2) {$V_2^{(23)}$};
      \node at (1,2) {$V_1^{(23)}$};
      \node at (1,-2) {$V_4^{(23)}$};
      \node at (-2.4,-2) {$V_3^{(23)}$};

        \node[right] at (3,0) {$\mathbb{R}$};
     	\end{tikzpicture}
     		\caption{The signature of $\theta_{23}$.}
             \end{subfigure}
             \caption{ Open sets in the complex $k$-plane for Region \rm{III}: $\rre \theta_{ij}>0$ (shaded) and $\rre \theta_{ij}<0$ (white).}
             \label{fig_theta_sign2}
\end{figure}

\subsubsection{First transformation}
  As $\zeta\in\mathcal{I}_3$ and $t>0$, the analytic approximations for the reflection coefficients $r_1(k)$, $r_2(k)$ and $\hat{r}_1(k)$ are analogous to Lemma \ref{lem_r_decomposition} in Subsection \ref{subsec_611}, except that the region $U_n^{(ij)}$ in Lemma \ref{lem_r_decomposition} must be replaced by region $V_n^{(ij)}$, $n=1,2,3,4$, $1\leq i<j\leq 3$, in Figure \ref{fig_theta_sign2}. Next, we introduce the following notations and provide the analytic approximations of the reflection coefficients as follows:
\begin{align*}
    &\tilde{r}_2(k):=\frac{r_2(k)}{1-8k|r_1(-k)|^2-|r_2(k)|^2},\quad\tilde{\alpha}(k):=\frac{\alpha(k)}{1-8k|r_1(-k)|^2-|r_2(k)|^2},\quad\hat{\alpha}(k):=\frac{\alpha(k)}{1+8k|r_1(k)|^2}.
\end{align*}

\begin{lem}\label{alpha_decomposition}
For each $\zeta\in\mathcal{I}_3$ and $t>0$, there are several decompositions
    \begin{align*}
        &\tilde{r}_2(k)=\tilde{r}_{2,a}(x,t,k)+\tilde{r}_{2,r}(x,t,k),\quad && k\in\mathbb{R}_+,\\
        &\tilde{\alpha}(k)=\tilde{\alpha}_a(x,t,k)+\tilde{\alpha}_r(x,t,k), &&k\in \left[k_0,\infty\right),\\
        &\hat{\alpha}(k)=\hat{\alpha}_a(x,t,k)+\hat{\alpha}_r(x,t,k), && k\in\left(-\infty,k_0\right],
    \end{align*}
    where the above functions satisfy the following properties:
    \begin{itemize}
     \item $\tilde{r}_{2,a}(x,t,k)$ is defined and continuous for $k\in \{k\,|\,\rre k\geq0,\rim k\leq0\}$, and is analytic for $k\in \{k\,|\,\rre k>0,\rim k<0\}$. $\tilde{\alpha}_a(k)$ is defined and continuous for $k\in \overline{V_4^{(23)}}$, and is analytic for $k\in V_4^{(23)}$. $\hat{\alpha}_a(k)$ is defined and continuous for $k\in \overline{V_2^{(23)}}$, and is analytic for $k\in V_2^{(23)}$.
    \item The functions $\tilde{r}_{2,a}(k)$, $\tilde{\alpha}_a(k)$ and $\hat{\alpha}_a(k)$ satisfy
        \begin{align*}
            &\left|\partial_x^l\left(\tilde{r}_{2,a}(x,t,k)-\sum_{j=0}^{N}\frac{\tilde{r}^{(j)}_{2,a}(0)}{j!}k^j\right)\right|\leq C \left|k\right|^{N+1}\re^{\frac{t}{4}\left|\rre\theta_{13}(\zeta,k)\right|},\quad &&k\in  \{k\,|\,\rre k\geq0,\rim k\leq0\},\\
            &\left|\partial_x^l\left(\tilde{\alpha}_a(x,t,k)-\sum_{j=0}^{N}\frac{\tilde{\alpha}^{(j)}_a(k_0)}{j!}(k-k_0)^{j}\right)\right|\leq C \left|k-k_0\right|^{N+1}\re^{\frac{t}{4}\left|\rre\theta_{23}(\zeta,k)\right|},\quad  &&k\in \overline{V_4^{(23)}},\\
            &\left|\partial_x^l\left(\hat{\alpha}_a(x,t,k)-\sum_{j=0}^{N}\frac{\hat{\alpha}^{(j)}_a(k_0)}{j!}(k-k_0)^{j}\right)\right|\leq C \left|k-k_0\right|^{N+1}\re^{\frac{t}{4}\left|\rre\theta_{23}(\zeta,k)\right|},\quad  &&k\in \overline{V_2^{(23)}},
            \end{align*}
            and
        \begin{align*}    
            &\left|\partial_x^l\left(\tilde{r}_{2,a}(x,t,k)\right)\right|\leq  \frac{C}{1+\left|k\right|}\re^{\frac{t}{4}\left|\rre\theta_{13}(\zeta,k)\right|},\quad  &&k\in \{k\,|\,\rre k\geq0,\rim k\leq0\},\\
            &\left|\partial_x^l\left(\tilde{\alpha}_a(x,t,k)\right)\right|\leq  \frac{C}{1+\left|k\right|}\re^{\frac{t}{4}\left|\rre\theta_{23}(\zeta,k)\right|},\quad  &&k\in \overline{V_4^{(23)}},\\
            &\left|\partial_x^l\left(\hat{\alpha}_a(x,t,k)\right)\right|\leq  \frac{C}{1+\left|k\right|}\re^{\frac{t}{4}\left|\rre\theta_{23}(\zeta,k)\right|},\quad  &&k\in \overline{V_2^{(23)}},
        \end{align*}
        where $l=0,1$, and the constant $C$ is independent of $\zeta,t,k$.
    \item For each $1\leq p\leq \infty$ and $l=0,1$,
    \begin{align*}
        &\left\|(1+|\cdot|)\partial_x^l\tilde{r}_{2,a}(x,t,\cdot)\right\|_{L^p(\mathbb{R}_+)}=\mathcal{O}\left(t^{-N}\right),\\
        &\left\|(1+|\cdot|)\partial_x^l\tilde{\alpha}_a(x,t,\cdot)\right\|_{L^p\left(k_0,\infty\right)}=\left\|\frac{\tilde{\alpha}_a(x,t,\cdot)}{\cdot-k_0}\right\|_{L^p\left(k_0,\infty\right)}=\mathcal{O}\left(t^{-N}\right),\\
         &\left\|(1+|\cdot|)\partial_x^l\hat{\alpha}_a(x,t,\cdot)\right\|_{L^p\left(-\infty,k_0\right)}=\left\|\frac{\hat{\alpha}_a(x,t,\cdot)}{\cdot-k_0}\right\|_{L^p\left(-\infty,k_0\right)}=\mathcal{O}\left(t^{-N}\right),
    \end{align*}
    uniformly for $\zeta\in\mathcal{I}_3$ as $t\to\infty$.
        \end{itemize}
\end{lem}

It is precisely because the signature of the dispersion relation $\theta_{13}(k)$ in Figure~\ref{subfig_theta13_2} is opposite to that in Figure~\ref{subfig_theta13_1} that the treatment in Region III becomes much more complicated than that in Region II. To address the jump problem of $\theta_{13}(k)$ along the real axis $\mathbb{R}$, the following three scalar RH problems need to be introduced:
     \begin{equation}\label{delta1}
    	\left\{
    	\begin{aligned}
    		&\begin{aligned}\delta_{1+}(k)
    			=&\delta_{1-}(k)(1+8k|r_1(k)|^2),\quad  &&k\in \left(-\infty,k_0\right],\\
    			=&\delta_{1-}(k), 
                &&k\in\mathbb{C}\setminus\left(-\infty,k_0\right],
    		\end{aligned}\\
            &\delta_1(k)\rightarrow1,\qquad\qquad\quad\qquad\qquad\,\,\,
              k\rightarrow\infty,
    	\end{aligned}
    	\right.
    \end{equation}
     \begin{equation}\label{delta2}
    	\left\{
    	\begin{aligned}
    		&\begin{aligned}\delta_{2+}(k)
    			=&\delta_{2-}(k)(1+8k|r_1(k)|^2-|r_2(-k)|^2),\quad  &&k\in \left(-\infty,k_0\right],\\
    			=&\delta_{2-}(k), 
                &&k\in\mathbb{C}\setminus\left(-\infty,k_0\right],
    		\end{aligned}\\
            &\delta_2(k)\rightarrow1,\qquad\qquad\qquad\qquad\quad\qquad\qquad\,\,\,\,\,\,
              k\rightarrow\infty,
    	\end{aligned}
    	\right.
    \end{equation}
     \begin{equation}\label{delta3}
    	\left\{
    	\begin{aligned}
    		&\begin{aligned}\delta_{3+}(k)
    			=&\delta_{3-}(k)(1-8k|r_1(-k)|^2-|r_2(k)|^2),\quad  &&k\in \left[k_0,-k_0\right],\\
    			=&\delta_{3-}(k), 
                &&k\in\mathbb{C}\setminus\left[k_0,-k_0\right],
    		\end{aligned}\\
            &\delta_3(k)\rightarrow1,\qquad\qquad\quad\qquad\qquad\qquad\qquad\,\,\,\,\,\,
              k\rightarrow\infty.
    	\end{aligned}
    	\right.
    \end{equation}
    The functions $\delta_j(k)$, $j=1,2,3,$ defined in equations \eqref{delta1}-\eqref{delta3} can be written in the form of the following integral equations:
   \begin{equation}\label{vol_delta1}
   \begin{aligned}
        &\delta_1(\zeta,k)={\rm{exp}}\left\{\frac{1}{2 \pi\ri}\int_{-\infty}^{k_0}\frac{\ln \left(1+8s|r_1(s)|^2\right)}{s-k}{\rm{d}}s\right\},\quad &&k\in \mathbb{C}\backslash \left(-\infty,k_0\right],\\
        &\delta_2(\zeta,k)={\rm{exp}}\left\{\frac{1}{2 \pi\ri}\int_{-\infty}^{k_0}\frac{\ln \left({1+8s|r_1(s)|^2-|r_2(-s)|^2}\right)}{s-k}{\rm{d}}s\right\},\quad &&k\in \mathbb{C}\backslash \left(-\infty,k_0\right],\\
        &\delta_3(\zeta,k)={\rm{exp}}\left\{\frac{1}{2 \pi\ri}\int_{k_0}^{-k_0}\frac{\ln \left({1-8s|r_1(-s)|^2-|r_2(s)|^2}\right)}{s-k}{\rm{d}}s\right\},\quad &&k\in \mathbb{C}\backslash [k_0,-k_0].
   \end{aligned}
    \end{equation}

 \begin{lem}\label{lem_delta1}
    	The functions $\delta_j(\zeta,k)$, $j=1,2,3,$ possess the following properties:
    	\begin{enumerate}
     		\item $\delta_1^{\pm1}(k)$ and $\delta_2^{\pm1}(k)$ are analytic on $\mathbb{C}\backslash \left(-\infty,k_0\right]$, $\delta_3^{\pm1}(k)$ are analytic on $\mathbb{C}\backslash [k_0,-k_0]$, and can be written respectively as
			\begin{equation*}\left\lbrace
            \begin{aligned}
                &\delta_1(\zeta,k)=\re ^{\ri \nu_1\ln_\pi (k-k_0)}{\rm{e}}^{-\chi_1(\zeta,k)},\quad &&k\in\mathbb{C}\backslash \left(-\infty,k_0\right], \\
                &\delta_2(\zeta,k)=\re ^{\ri \nu_2\ln_\pi (k-k_0)}{\rm{e}}^{-\chi_2(\zeta,k)},\quad &&k\in\mathbb{C}\backslash \left(-\infty,k_0\right],\\
                &\delta_3(\zeta,k)=\re ^{\ri \nu_2\ln_\theta (k+k_0)-\ri\nu_3\ln_\theta (k-k_0)}{\rm{e}}^{-\chi_3(\zeta,k)},\quad &&k\in\mathbb{C}\backslash \left[k_0,-k_0\right],
            \end{aligned}\right.
				\qquad \zeta\in\mathcal{I}_3,
			\end{equation*}
			where $\theta=\pi$ as $k$ approaches $-k_0$; and $\theta=0$ as $k$ approaches $k_0$
             for the function $\delta_3$,
		\begin{align*}
			&\nu_1=-\frac{1}{2\pi} \ln \left({1+8k_0|r_1(k_0)|^2}\right),\quad &&\chi_1(\zeta,k)=\frac{1}{2\pi\ri}\int_{-\infty}^{k_0}\ln_\pi(k-s){\rm{d}}\ln \left(1+8s|r_1(s)|^2\right),\\
            &\nu_2=-\frac{1}{2\pi} \ln \left(1+8k_0|r_1(k_0)|^2-|r_2(-k_0)|^2\right),\quad &&\chi_2(\zeta,k)=\frac{1}{2\pi\ri}\int_{-\infty}^{k_0}\ln_\pi(k-s){\rm{d}}\ln \left({1+8s|r_1(s)|^2-|r_2(-s)|^2}\right),\\
            &\nu_3=-\frac{1}{2\pi} \ln \left(1-8k_0|r_1(-k_0)|^2-|r_2(k_0)|^2\right),\quad &&\chi_3(\zeta,k)=\frac{1}{2\pi\ri}\int_{k_0}^{-k_0}\ln_\theta(k-s){\rm{d}}\ln \left({1-8s|r_1(-s)|^2-|r_2(s)|^2}\right).
				\end{align*}
	
    	\item For each $\zeta\in\mathcal{I}_3$, $\delta_1(k)$ and $\delta_2(k)$ are bounded in $\mathbb{C}\backslash \left(-\infty,k_0\right]$, $\delta_3(k)$ is bounded in $\mathbb{C}\backslash \left[k_0,-k_0\right]$, and the symmetry relations hold:
    	\begin{equation*}
    		\delta_j^{-1}(k)= \overline{\delta_j(\overline{k})},\qquad j=1,2,3.
    	\end{equation*}
        \item As $k\to 0$ along the non-tangential direction of $(-\infty,k_0)$, we obtain the following formulas hold for $j = 1, 2$:
        \begin{align*}
            &\left|\chi_j(\zeta,k)-\chi_j({\zeta,k_0})\right|\leq C|k-k_0|\left(1+|\ln|k-k_0||\right),\quad
            &&\left|\partial_x\chi_j(\zeta,k_0)\right|\leq\frac{C}{t},\\
            &\left|\partial_x\left(\chi_j(\zeta,k)-\chi_j({\zeta,k_0})\right)\right|\leq \frac{C}{t}\left(1+|\ln|k-k_0||\right),\quad &&\partial_x\left(\delta_j(\zeta,k)^{\pm1}\right)=\frac{\pm\ri \nu_j}{2t(k-k_0)}\delta_j(\zeta,k)^{\pm1},
             \end{align*}
             and as $k\to 0$ along the non-tangential direction of $(k_0,-k_0)$, we have
        \begin{align*}
            &\left|\chi_3(\zeta,k)-\chi_3({\zeta,\pm k_0})\right|\leq C|k\mp k_0|\left(1+|\ln|k\mp k_0||\right),\quad
            &&\left|\partial_x\chi_3(\zeta,\pm k_0)\right|\leq\frac{C}{t},\\
            &\left|\partial_x\left(\chi_3(\zeta,k)-\chi_3({\zeta,\pm k_0})\right)\right|\leq \frac{C}{t}\left(1+|\ln|k\mp k_0||\right),\quad &&\partial_x\left(\delta_3(\zeta,k)^{\pm1}\right)=\left(\pm\frac{\ri \nu_2}{2t(k+ k_0)}\mp\frac{\ri \nu_3}{2t(k-k_0)}\right)\delta_3(\zeta,k)^{\pm1},
             \end{align*}    
        where $C$ is independent of $\zeta$.
            
    	\end{enumerate}
    \end{lem}

    \begin{proof}
        The lemma is derived from equation \eqref{vol_delta1}, and can be proved by direct estimation.
    \end{proof}
\begin{remark}
 Naturally, the following equations can be derived from equations \eqref{delta1}, \eqref{delta2} and \eqref{delta3}:
    \begin{align*}
        &\delta_1(\zeta,-k)
        ={\rm{exp}}\left\{-\frac{1}{2 \pi\ri}\int_{-k_0}^{+\infty}\frac{\ln \left(1-8s|r_1(-s)|^2\right)}{s-k}{\rm{d}}s\right\},\quad &&k\in \mathbb{C}\backslash \left[-k_0,+\infty\right),\\
        &\delta_2(\zeta,-k)
        ={\rm{exp}}\left\{-\frac{1}{2 \pi\ri}\int_{-k_0}^{+\infty}\frac{\ln \left({1-8s|r_1(-s)|^2-|r_2(s)|^2}\right)}{s-k}{\rm{d}}s\right\},\quad &&k\in \mathbb{C}\backslash \left[-k_0,+\infty\right),\\
        &\delta_3(\zeta,-k)={\rm{exp}}\left\{-\frac{1}{2 \pi\ri}\int_{k_0}^{-k_0}\frac{\ln \left({1+8s|r_1(s)|^2-|r_2(-s)|^2}\right)}{s-k}{\rm{d}}s\right\},\quad &&k\in \mathbb{C}\backslash [k_0,-k_0],
        \end{align*}
and  they can be written as:
    \begin{equation*}
        \delta_1(\zeta,-k)=\re ^{\ri \nu_1\ln_0 (k+k_0)}{\rm{e}}^{-\chi_1(\zeta,-k)}, \quad \delta_2(\zeta,-k)=\re ^{\ri \nu_2\ln_0 (k+k_0)}{\rm{e}}^{-\chi_2(\zeta,-k)},\quad \delta_3(\zeta,-k)=\re ^{\ri \nu_2\ln_{\theta} (k-k_0)-\ri \nu_3\ln_{\theta} (k+k_0)}{\rm{e}}^{-\chi_3(\zeta,-k)},
    \end{equation*}
    where
   \begin{align*}
			&\chi_1(\zeta,-k)=-\frac{1}{2\pi\ri}\int^{+\infty}_{-k_0}\ln_0(k-s){\rm{d}}\ln \left(1-8s|r_1(-s)|^2\right),\\
            &\chi_2(\zeta,-k)=-\frac{1}{2\pi\ri}\int^{+\infty}_{-k_0}\ln_0(k-s){\rm{d}}\ln \left({1-8s|r_1(-s)|^2-|r_2(s)|^2}\right),\\
            &\chi_3(\zeta,-k)=-\frac{1}{2\pi\ri}\int_{k_0}^{-k_0}\ln_\theta(k-s){\rm{d}}\ln \left({1+8s|r_1(s)|^2-|r_2(-s)|^2}\right).
				\end{align*}
\end{remark}

     Based on the functions defined in \eqref{vol_delta1}, we can construct the following transformation for the eigenfunction $M(x,t,k)$ in the RH problem \ref{rhp_M}:
    \begin{equation*}\label{trans_Delta1}
        M^{(1)}(x,t,k)=M(x,t,k)\Delta_1(k),
    \end{equation*}
    where the $3\times3$ matrix-valued function $\Delta_1$ is defined by 
    \begin{equation}\label{Delta1}
          \Delta_1(k)=\begin{pmatrix}
        	\frac{\delta_1(-k)\delta_2(k)}{\delta_2(-k)\delta_3(-k)} & 0 & 0\\
        	0 & \frac{\delta_3(k)\delta_3(-k)}{{\delta_1(k)}{\delta_1(-k)}} & 0\\
        	0 & 0 & \frac{\delta_1(k)\delta_2(-k)}{\delta_2(k)\delta_3(k)}\\
        \end{pmatrix},
    \end{equation}
    which satisfies the symmetry conditions: 
    \begin{equation*}
        \Delta_1(\zeta,\bar k)=\mathcal{A}(k)\Delta_1^{-1}(\zeta,k)\mathcal{A}^{-1}(k),\qquad \Delta_1(\zeta,k)=\mathcal{B}\Delta_1(\zeta,-k)\mathcal{B}.
    \end{equation*}
    In addition, $\Delta_1(\zeta,k)$  is uniformly bounded with respect to $\zeta\in\mathcal{I}_3$ and $k\in \mathbb{C}\setminus\mathbb{R}$, and
    \begin{equation*}\label{Delta1_bound}
        \Delta_1(\zeta,k)=I+\mathcal{O}\left(k^{-1}\right),\quad k\to\infty.
    \end{equation*}

    There exists the jump relation $M^{(1)}_+(\zeta,k)=M^{(1)}_-(\zeta,k)V^{(1)}(\zeta,k)$ for $k\in \Sigma^{(1)}$ (see Figure \ref{fig_Sigma1}), where $V^{(1)}(k)$, $k\in \Sigma^{(1)}_j$, $j=1,2,3$, take the following forms:

     \begin{figure}[htbp]
    	\centering
    	\begin{tikzpicture}[scale=1]
        
        \draw[very thick, black!20!blue] (-4,0) -- (4,0);
    		
    		\draw[very thick, black!20!blue, -latex] (-4,0) -- (-2.5,0);
            \draw[very thick, black!20!blue, -latex] (-4,0) -- (0.2,0);
            \draw[very thick, black!20!blue, -latex] (-4,0) -- (3,0);

    		\fill (-1.5,0) circle (1.5pt);
            \fill (0,0) circle (1.5pt);
	      \fill (1.5,0) circle (1.5pt);

           \node[below] at (-1.5,0) {$k_0$};
	     \node[below] at (0,-0.2) {$\text{0}$};
	     \node[below] at (1.5,0) {$-k_0$};
         
           \node[red!70!black,above] at (-2.7,0.2) {$1$};
	     \node[red!70!black,above] at (0,0.2) {$2$};
          \node[red!70!black,above] at (2.8,0.2) {$3$};

          \node[right] at (4,0) {$\rre k$};
    		
    	\end{tikzpicture}
    	\caption{The jump contour $\Sigma^{(1)}$ in the complex $k$-plane.}
    	\label{fig_Sigma1}
    \end{figure}

     \begin{align*}
            &V^{(1)}_1(x,t,k)
            =\begin{pmatrix}
					1 &-\hat{r}_1(k)\Delta_{12-}(k)\re^{t\theta_{12}(\zeta,k)} & \tilde{r}_2^*(-k)\Delta_{13-}(k)\re^{t\theta_{13}(\zeta,k)}\\
					0 & 1 & 0\\
					0 & \hat{\alpha}^*(k)\Delta^{-1}_{23-}(k)\re^{-t\theta_{23}(\zeta,k)} & 1\\
				\end{pmatrix}\begin{pmatrix}
					1 & 0 & 0\\
					-8k\hat{r}_1^*(k)\Delta^{-1}_{12+}(k)\re^{-t\theta_{12}(\zeta,k)} & 1  & -8k\hat{\alpha}(k)\Delta_{23+}(k)\re^{t\theta_{23}(\zeta,k)}\\
					-\tilde{r}_2(-k)\Delta_{13+}^{-1}(k)\re^{-t\theta_{13}(\zeta,k)} & 0 & 1\\
				\end{pmatrix}, \\
                &V^{(1)}_2(x,t,k)
                =\begin{pmatrix}
					1 & -\tilde{\alpha}^*(-k)\Delta_{12-}(k)\re^{t\theta_{12}(\zeta,k)}    & \tilde{r}_2(k)\Delta_{13-}(k)\re^{t\theta_{13}(\zeta,k)} \\
					0 & 1 & -8k\tilde{\alpha}(k)\Delta_{23-}(k)\re^{t\theta_{23}(\zeta,k)} \\
					0 & 0 & 1\\
				\end{pmatrix} \begin{pmatrix}
					1 & 0  & 0\\
					-8k\tilde{\alpha}(-k)\Delta_{12+}^{-1}(k)\re^{-t\theta_{12}(\zeta,k)}   & 1  & 0\\
					-\tilde{r}_2^*(k)\Delta_{13+}^{-1}(k)\re^{-t\theta_{13}(\zeta,k)}   & \tilde{\alpha}^*(k)\Delta_{23+}^{-1}(k)\re^{-t\theta_{23}(\zeta,k)}   & 1\\
				\end{pmatrix},\\
               & V^{(1)}_3(x,t,k)
            =\begin{pmatrix}
					1 & 0  & \tilde{r}_2(k)\Delta_{13-}(k)\re^{t\theta_{13}(\zeta,k)} \\
					-8k\hat{\alpha}(-k)\Delta_{12-}^{-1}(k)\re^{-t\theta_{12}(\zeta,k)}   & 1  & -8k\tilde{\alpha}(k)\Delta_{23-}(k)\re^{t\theta_{23}(\zeta,k)} \\
					0  & 0 & 1\\
				\end{pmatrix}\begin{pmatrix}
					1 & -\hat{\alpha}^*(-k)\Delta_{12+}(k)\re^{t\theta_{12}(\zeta,k)}  & 0\\
					0  & 1  & 0\\
					-\tilde{r}_2^*(k)\Delta_{13+}^{-1}(k)\re^{-t\theta_{13}(\zeta,k)}  & \tilde{\alpha}^*(k)\Delta_{23+}^{-1}(k)\re^{-t\theta_{23}(\zeta,k)}  & 1\\
				\end{pmatrix},
            \end{align*}
            where $\Delta_{12}(k)=\frac{\delta_2(-k)\delta_3(k)\delta_3^2(-k)}{\delta_1(k)\delta_1^2(-k)\delta_2(k)}$, $\Delta_{23}(k)=\frac{\delta_1^2(k)\delta_1(-k)\delta_2(-k)}{\delta_2(k)\delta_3^2(k)\delta_3(-k)}$, and $\Delta_{13}(k)=\frac{\delta_1(k)\delta_2^2(-k)\delta_3(-k)}{\delta_1(-k)\delta_2^2(k)\delta_3(k)}.$

   \subsubsection{Second transformation}
   In this subsection, the domains of definition for the transformations are identical to those of Figures \ref{fig_RegionD_1} and \ref{fig_RegionD_3} in Subsection \ref{subsec_right_trans2}, except that the roles of $k_0$ and $-k_0$ are interchanged. We continue to use the notation introduced in Subsection \ref{subsec_right_trans2}. First introduce the piecewise analytic function $R^{(1)}(\zeta,k)=R^{(1)}_j(\zeta,k)$ for $\zeta\in\mathcal{I}_3$, $k\in D^{(1)}_j$, $j=1,2,3,4,$ defined by:
   \begin{equation}\label{trans_R1}
    \begin{aligned}
        &R^{(1)}_1(\zeta,k)=\begin{pmatrix}
					1 & 0 & 0\\
					0  & 1  & 0\\
					\tilde{r}_{2,a}^*(k)\Delta_{13}^{-1}(k)\re^{-t\theta_{13}(\zeta,k)}  & 0 & 1\\
				\end{pmatrix},\quad &&R^{(1)}_2(\zeta,k)=\begin{pmatrix}
					1 & 0 & 0\\
					0 & 1  & 0\\
					\tilde{r}_{2,a}(-k)\Delta_{13}^{-1}(k)\re^{-t\theta_{13}(\zeta,k)}  & 0   & 1\\
				\end{pmatrix},\\
        &R^{(1)}_3(\zeta,k)=\begin{pmatrix}
					1 & 0 & \tilde{r}_{2,a}^*(-k)\Delta_{13}(k)\re^{t\theta_{13}(\zeta,k)}\\
					0 & 1  & 0\\
					0  & 0 & 1\\
				\end{pmatrix},\quad &&R^{(1)}_4(\zeta,k)=\begin{pmatrix}
					1 & 0 & \tilde{r}_{2,a}(k)\Delta_{13}(k)\re^{t\theta_{13}(\zeta,k)}\\
					0 & 1 & 0\\
					0  & 0 & 1\\
				\end{pmatrix}.
    \end{aligned}\end{equation}
    The function $R^{(2)}(\zeta,k)=R^{(2)}_j(\zeta,k)$ for $\zeta\in\mathcal{I}_3$, $k\in D^{(1)}_j$, $j=1,2,3,4,$ defined by:
    \begin{equation}\label{trans_R2}
     \begin{aligned}
         &R^{(2)}_1(\zeta,k)=\begin{pmatrix}
					1 & 0 & 0\\
					0 & 1 & 0\\
					0 & -\tilde{\alpha}_a^*(k)\Delta_{23}^{-1}(k)\re^{-t\theta_{23}(\zeta,k)} & 1\\
				\end{pmatrix},\quad && R^{(2)}_2(\zeta,k)=\begin{pmatrix}
					1 & 0 & 0\\
					8k\tilde{\alpha}_a(-k)\Delta_{12}^{-1}(k)\re^{-t\theta_{12}(\zeta,k)} & 1 & 0\\
					0 & 0 & 1\\
				\end{pmatrix},\\
        &R^{(2)}_3(\zeta,k)=\begin{pmatrix}
					1 & -\tilde{\alpha}_a^*(-k)\Delta_{12}(k)\re^{t\theta_{12}(\zeta,k)} & 0\\
					0 & 1 & 0\\
					0 & 0 & 1\\
				\end{pmatrix},\quad &&R^{(2)}_4(\zeta,k)=\begin{pmatrix}
					1 & 0 & 0\\
					0 & 1 & -8k\tilde{\alpha}_a(k)\Delta_{23}(k)\re^{t\theta_{23}(\zeta,k)}\\
					0  & 0  & 1\\
				\end{pmatrix}.
     \end{aligned}\end{equation}
    The final function $R^{(3)}(\zeta,k)$ is defined as $R^{(3)}_j(\zeta,k)$ for $\zeta\in\mathcal{I}_3$, $k\in D^{(2)}_j$, $j=1,2,\cdots,8,$ in this subsection to handle the jump of the eigenfunction $M^{(1)}(k)$ across $\mathbb{R}$:
       \begin{align}
     &R^{(3)}_1(\zeta,k)=\begin{pmatrix}
					1 & \hat{\alpha}_a^*(-k)\Delta_{12}(k)\re^{t\theta_{12}(\zeta,k)}  & 0\\
					0  & 1 & 0\\
					0 & 0 & 1\\
				\end{pmatrix},\quad && R^{(3)}_2(\zeta,k)=\begin{pmatrix}
        	1 & 0 & 0\\
        	8k\tilde{\alpha}_a(-k)\Delta_{12}^{-1}(k)\re^{-t\theta_{12}(\zeta,k)} & 1 & 0\\
        	0 & 0 & 1\\
        \end{pmatrix}, \nonumber\\
        &R^{(3)}_3(\zeta,k)=\begin{pmatrix}
        	1 & -\tilde{\alpha}_a^*(-k)\Delta_{12}(k)\re^{t\theta_{12}(\zeta,k)} & 0\\
        	0 & 1 & 0\\
        	0 & 0 & 1\\
        \end{pmatrix},\quad && R^{(3)}_4(\zeta,k)=\begin{pmatrix}
					1 & 0  & 0\\
					-8k\hat{\alpha}_a(-k)\Delta_{12}^{-1}(k)\re^{-t\theta_{12}(\zeta,k)}  & 1  & 0\\
					0  & 0  & 1\\
				\end{pmatrix},\nonumber \\
         &R^{(3)}_5(\zeta,k)=\begin{pmatrix}
         	1 & 0 & 0\\
         	0 & 1 & 0\\
         	0 & -\tilde{\alpha}_a^*(k)\Delta_{23}^{-1}(k)\re^{-t\theta_{23}(\zeta,k)} & 1\\
         \end{pmatrix},\quad && R^{(3)}_6(\zeta,k)=\begin{pmatrix}
					1 & 0  & 0\\
					0 & 1 & 8k\hat{\alpha}_a(k)\Delta_{23}(k)\re^{t\theta_{23}(\zeta,k)}\\
					0 & 0  & 1\\
				\end{pmatrix},\label{trans_R3} \\
        &R^{(3)}_7(\zeta,k)=\begin{pmatrix}
					1 & 0 & 0\\
					0 & 1 & 0\\
					0 & \hat{\alpha}_a^*(k)\Delta_{23}^{-1}(k)\re^{-t\theta_{23}(\zeta,k)} & 1\\
				\end{pmatrix},\quad && R^{(3)}_8(\zeta,k)=\begin{pmatrix}
        	1 & 0 & 0\\
        	0 & 1 & -8k\tilde{\alpha}_a(k)\Delta_{23}(k)\re^{t\theta_{23}(\zeta,k)}\\
        	0  & 0  & 1\\
        \end{pmatrix},\nonumber \\
        &R^{(3)}(\zeta,k)=I, \qquad k\in \mathbb{C}\setminus\bigcup_{j=1}^8D^{(2)}_j.\nonumber
     \end{align}

    After constructing the three sub-transformations \eqref{trans_R1}, \eqref{trans_R2}, and \eqref{trans_R3} above, we set $M^{(2)}(x,t,k)=M^{(1)}(x,t,k)R(x,t,k)$, where 
    \begin{equation}\label{trans_R}
        R(x,t,k)=R^{(1)}(x,t,k)R^{(2)}(x,t,k)R^{(3)}(x,t,k).
    \end{equation}
    
    \begin{lem}\label{lem_R_bound}
        $R(x,t,k)$ is uniformly bounded for $k\in \mathbb{C}\setminus\Sigma^{(2)}$, $\zeta\in\mathcal{I}_3$, and $t>0$. More importantly,
        \begin{equation*}\label{R_k_infty}
            R(k)=I+\mathcal{O}(k^{-1}),\quad k\to\infty.
        \end{equation*}
     \end{lem}

     \begin{figure}[htbp]
    	\centering
    	\begin{tikzpicture}[scale=1]
        
        \draw[very thick, black!40!green] (-4,0) -- (4,0);
        \draw[very thick, black!40!green] (0,2) -- (0,3.5);
        \draw[very thick, black!40!green] (0,-2) -- (0,-3.5);

        \draw[very thick, dashed, black!40!green] (0,-2) -- (0,2);
    		
    		\draw[very thick, black!40!green, -latex] (-4,0) -- (-2.8,0);
            \draw[very thick, black!40!green, -latex] (-4,0) -- (-0.8,0);
            \draw[very thick, black!40!green, -latex] (-4,0) -- (1.2,0);
            \draw[very thick, black!40!green, -latex] (-4,0) -- (3.2,0);

            \draw[very thick, black!40!green, -latex] (0,-2) -- (0,-3);
            \draw[very thick, black!40!green, -latex] (0,2) -- (0,3);

            \draw[very thick, blue!20!purple] (2+2/1.414,2/1.414) -- (0,-2);
            \draw[very thick, blue!20!purple] (2+2/1.414,-2/1.414) -- (0,2);

            \draw[very thick, black!20!blue] (-2-2/1.414,2/1.414) -- (0,-2);
            \draw[very thick, black!20!blue] (-2-2/1.414,-2/1.414) -- (0,2);

            \draw[very thick, blue!20!purple,-latex] (2,0) --(2+1/1.2,1/1.2);
            \draw[very thick, blue!20!purple,-latex] (2,0) -- (-1/1.2+2,-1/1.2);
            \draw[very thick, blue!20!purple,-latex] (2,0) -- (-1/1.2+2,1/1.2);
            \draw[very thick, blue!20!purple,-latex] (2,0) -- (2+1/1.2,-1/1.2);

            \draw[very thick, black!20!blue,-latex] (-2,0) --(-2+1/1.2,1/1.2);
            \draw[very thick, black!20!blue,-latex] (-2,0) -- (-1/1.2-2,-1/1.2);
            \draw[very thick, black!20!blue,-latex] (-2,0) -- (-1/1.2-2,1/1.2);
            \draw[very thick, black!20!blue,-latex] (-2,0) -- (-2+1/1.2,-1/1.2);

    		\fill (-2,0) circle (1.5pt);
            \fill (0,0) circle (1.5pt);
	      \fill (2,0) circle (1.5pt);

           \node[below] at (-2,-0.1) {$k_0$};
	     \node[below] at (0.3,0) {$\text{0}$};
	     \node[below] at (2,-0.1) {$-k_0$};

          \node[red!70!black,left] at (1.9+1/1.2,1/1.2) {$1$};
          \node[red!70!black,right] at (2.1-1/1.2,1/1.2) {$2$};
          \node[red!70!black,right] at (2.1-1/1.2,-1/1.2) {$3$};
          \node[red!70!black,left] at (1.9+1/1.2,-1/1.2) {$4$};

          \node[red!70!black,left] at (-2.1+1/1.2,1/1.2) {$5$};
          \node[red!70!black,right] at (-1.9-1/1.2,1/1.2) {$6$};
          \node[red!70!black,right] at (-1.9-1/1.2,-1/1.2) {$7$};
          \node[red!70!black,left] at (-2.1+1/1.2,-1/1.2) {$8$};
         
           \node[red!70!black,above] at (-3,0.1) {$9$};
           \node[red!70!black,above] at (-1,0.1) {$10$};
           \node[red!70!black,above] at (1,0.1) {$11$};
           \node[red!70!black,above] at (3,0.1) {$12$};

            \node[red!70!black,right] at (0,2.8) {$13$};
           \node[red!70!black,right] at (0,-2.8) {$14$};
           \node[red!70!black,right] at (0,0.7) {$15$};

          \node[right] at (4,0) {$\rre k$};
    		
    	\end{tikzpicture}
    	\caption{The jump contour $\Sigma^{(2)}$ in the complex $k$-plane.}
    	\label{fig_Sigma2}
    \end{figure}
    $M^{(2)}(x,t,k)$ is analytic on $\mathbb{C}\setminus\Sigma^{(2)}$, where $\Sigma^{(2)}$ is shown in Figure \ref{fig_Sigma2}, and satisfies the jump condition $M^{(2)}_+(x,t,k)=M^{(2)}_-(x,t,k)V^{(2)}(x,t,k)$ on $\Sigma^{(2)}$. Similarly to the discussion of jumps on the contours $\Gamma^{(2)}_j$, $j=9,10,\cdots, 15$ within region $\zeta\in\mathcal{I}_2$, for the contours $\Sigma^{(2)}$ in Figure \ref{fig_Sigma2}, we need only consider the jumps on $\Sigma^{(2)}_j$, $j=1,2,\cdots,8$, the jumps on all remaining contours will be absorbed into an error term with respect to $t$ when computing the long-time asymptotics.  See Remark \ref{remark_v2_9-15} for the detailed analysis. The forms of $V^{(2)}_j(x,t,k)$ for $k\in\Sigma^{(2)}_j$, $j=1,2,\cdots,8$, are as follows:
     \begin{align*}
         &V^{(2)}_1(x,t,k)=\begin{pmatrix}
					1 & -\hat{\alpha}_a^*(-k)\Delta_{12}(k)\re^{t\theta_{12}(\zeta,k)}  & 0\\
					0  & 1 & 0\\
					0 & 0 & 1\\
				\end{pmatrix},\quad &&V^{(2)}_2(x,t,k)=\begin{pmatrix}
					1 & 0 & 0\\
                    8k\tilde{\alpha}_a(-k)\Delta_{12}^{-1}(k)\re^{-t\theta_{12}(\zeta,k)} & 1 & 0\\
					0 & 0 & 1\\
				\end{pmatrix},\\
        &V^{(2)}_3(x,t,k)=\begin{pmatrix}
					1 & \tilde{\alpha}_a^*(-k)\Delta_{12}(k)\re^{t\theta_{12}(\zeta,k)} & 0\\
					0 & 1 & 0\\
					0 & 0 & 1\\
				\end{pmatrix},\quad &&V^{(2)}_4(x,t,k)=\begin{pmatrix}
					1 & 0  & 0\\
					-8k\hat{\alpha}_a(-k)\Delta_{12}^{-1}(k)\re^{-t\theta_{12}(\zeta,k)}  & 1  & 0\\
					0  & 0  & 1\\
				\end{pmatrix},\\
        &V^{(2)}_5(x,t,k)=\begin{pmatrix}
					1 & 0 & 0\\
					0 & 1 & 0\\
					0 & \tilde{\alpha}_a^*(k)\Delta_{23}^{-1}(k)\re^{-t\theta_{23}(\zeta,k)} & 1\\
				\end{pmatrix},\quad &&V^{(2)}_6(x,t,k)=\begin{pmatrix}
					1 & 0  & 0\\
					0 & 1 & 8k\hat{\alpha}_a(k)\Delta_{23}(k)\re^{t\theta_{23}(\zeta,k)}\\
					0 & 0  & 1\\
				\end{pmatrix},\\
        &V^{(2)}_7(x,t,k)=\begin{pmatrix}
					1 & 0 & 0\\
					0 & 1 & 0\\
					0 & -\hat{\alpha}_a^*(k)\Delta_{23}^{-1}(k)\re^{-t\theta_{23}(\zeta,k)} & 1\\
				\end{pmatrix},\quad &&V^{(2)}_8(x,t,k)=\begin{pmatrix}
					1 & 0 & 0\\
					0 & 1 & -8k\tilde{\alpha}_a(k)\Delta_{23}(k)\re^{t\theta_{23}(\zeta,k)}\\
					0 \!\!\!\!& 0\!\!\!\! & 1\\
				\end{pmatrix}.
     \end{align*}

     \begin{lem}\label{lem_V2_error}
    As $t\to\infty$, the jump matrix $V^{(2)}(\zeta,k)$ converges to the identity matrix $I$  and $\partial_xV^{(2)}(\zeta,k)$ converges to the zero matrix $O$ uniformly for $\zeta\in\mathcal{I}_3$ and $k\in \Sigma^{(2)}$ except near the two critical points $\pm k_0$. More importantly, the jump matrices $V^{(2)}_j(\zeta,k)$, $j=9,10,\cdots,14$, satisfy:
       \begin{align*}
            &\left\|(1+|\cdot|)\partial_x^l\left(V^{(2)}_j(x,t,\cdot)-I\right)\right\|_{(L^1\cap L^\infty)\left(\Sigma^{(2)}_j\right)}\leq Ct^{-N}, \quad  &&j=9,10,11,12,\\
            &\left\|(1+|\cdot|)\partial_x^l\left(V^{(2)}_j(x,t,\cdot)-I\right)\right\|_{(L^1\cap L^\infty)\left(\Sigma^{(2)}_j\right)}\leq C\re^{-ct}, \quad &&j=13,14.
        \end{align*}
    \end{lem}
   The proof follows a similar argument as in Lemma \ref{lem_v2_error}, and is therefore omitted.

  \subsubsection{Local parametrix at $\pm k_0$}
   When examining the long-time asymptotic solution of the eigenfunction $M^{(2)}(\zeta,k)$ as $\zeta\in\mathcal{I}_3$, we only need to consider the neighborhoods of two critical points $\pm k_0$. Since in this region $k_0<0$, we now set $X_j^{-k_0}=(-k_0+X_j)\cap B_\epsilon(-k_0)$, $X_{j+4}^{k_0}=(k_0+X_j)\cap B_\epsilon(k_0)$, and $X_\epsilon=X^{k_0}\cup X^{-k_0}$, $j=1,2,3,4$.

   This subsection continues to use the same scaling transformations \eqref{z_1} and \eqref{z_2} for the variable $k$ as previously defined. The functions $\delta_j(\zeta,k)$, $j=1,2,3,$ appear in the off-diagonal entries of the jump matrices $V^{(2)}_j(\zeta,k)$, $j=1,2,\cdots,8$. Below, for $k$ lying in the $\epsilon$ neighborhoods of $k_0$ and $-k_0$, we expand the combinations of the functions $\delta_j(k)$, $j=1,2,3,$ that occupy these off-diagonal positions. For $k\in B_\epsilon(-k_0)\setminus(-k_0/2,-3k_0/2)$, we have
    \begin{align}
  	\Delta_{12}(k)&=(2\sqrt{t})^{2\ri\hat\nu} \re^{-\ri(2\nu_1-\nu_2)\ln_0(z_2)}\re^{-\ri(2\nu_3-\nu_2)\ln_\pi(z_2)}\re^{-\ri(\nu_3-2\nu_2)\ln_\pi(-2k_0)}\re^{2\chi_1(k_0)-\chi_2(k_0)-2\chi_3(k_0)-\chi_3(-k_0)}\nonumber\\
  	&\quad\re^{-\ri(\nu_3-2\nu_2)\ln_\pi((k-k_0)/(-2k_0))}\re^{2(\chi_1(-k)-\chi_1(k_0))-(\chi_2(-k)-\chi_2(k_0))-2(\chi_3(-k)-\chi_3(k_0))-(\chi_3(k)-\chi_3(-k_0))}\delta_1^{-1}(k)\delta_2^{-1}(k)\label{Delta12_decomposition}\\
  	&:=\re^{\ri(\nu_2-2\nu_1)\ln_0(z_2)}\re^{\ri(\nu_2-2\nu_3)\ln_\pi(z_2)}b^{(-k_0)}_0(\zeta)b^{(-k_0)}_1(\zeta,k),\nonumber
  \end{align}
  where $\hat{\nu}=\nu_1-\nu_2+\nu_3$, and
  \begin{align*}
  	&b^{(-k_0)}_0(\zeta)=(2\sqrt{t})^{2\ri\hat\nu} \re^{\ri(2\nu_2-\nu_3)\ln_\pi(-2k_0)}\re^{2\chi_1(k_0)-\chi_2(k_0)-2\chi_3(k_0)-\chi_3(-k_0)}\delta_1^{-1}(-k_0)\delta_2^{-1}(-k_0),\\
  	&b^{(-k_0)}_1(\zeta,k)=\re^{\ri(2\nu_2-\nu_3)\ln_\pi((k-k_0)/(-2k_0))}\re^{2(\chi_1(-k)-\chi_1(k_0))-(\chi_2(-k)-\chi_2(k_0))-2(\chi_3(-k)-\chi_3(k_0))-(\chi_3(k)-\chi_3(-k_0))}\delta_1(-k_0)\delta_2(-k_0)\delta_1^{-1}(k)\delta_2^{-1}(k).
  \end{align*} 
    
     For $k\in B_\epsilon(k_0)\setminus(3k_0/2,k_0/2)$, the following decomposition holds:
     \begin{align}
 	\Delta_{23}(k)&=(2\sqrt{t})^{-2\ri\hat\nu} \re^{\ri(2\nu_1-\nu_2)\ln_\pi(z_1)}\re^{\ri(2\nu_3-\nu_2)\ln_0(z_1)}\re^{\ri(\nu_3-2\nu_2)\ln_0(2k_0)}\re^{-2\chi_1(k_0)+\chi_2(k_0)+2\chi_3(k_0)-\chi_3(-k_0)}\nonumber\\
 	&\quad\re^{\ri(\nu_3-2\nu_2)\ln_0((k+k_0)/(2k_0))}\re^{-2(\chi_1(k)-\chi_1(k_0))+(\chi_2(k)-\chi_2(k_0))+2(\chi_3(k)-\chi_3(k_0))+(\chi_3(-k)-\chi_3(-k_0))}\delta_1(-k)\delta_2(-k)\label{Delta23_decomposition}\\
 	&:=\re^{\ri(2\nu_1-\nu_2)\ln_\pi(z_1)}\re^{\ri(2\nu_3-\nu_2)\ln_0(z_1)}b^{(k_0)}_0(\zeta)b^{(k_0)}_1(\zeta,k),\nonumber
 \end{align}
 where 
 \begin{align*}
 	&b^{(k_0)}_0(\zeta)=(2\sqrt{t})^{-2\ri\hat\nu}\re^{\ri(\nu_3-2\nu_2)\ln_0(2k_0)}\re^{-2\chi_1(k_0)+\chi_2(k_0)+2\chi_3(k_0)+\chi_3(-k_0)}\delta_1(-k_0)\delta_2(-k_0),\\
 	&b^{(k_0)}_1(\zeta,k)=\re^{\ri(\nu_3-2\nu_2)\ln_0((k+k_0)/(2k_0))}\re^{-2(\chi_1(k)-\chi_1(k_0))+(\chi_2(k)-\chi_2(k_0))+2(\chi_3(k)-\chi_3(k_0))+(\chi_3(-k)-\chi_3(-k_0))}\delta_1^{-1}(-k_0)\delta^{-1}_2(-k_0)\delta_1(-k)\delta_2(-k).
 \end{align*}

    \begin{prop}\label{prop_hatnu}
      The following relation holds:
     \begin{equation}\label{nu=hatnu}
           \re^{-2\pi\hat\nu}=1-8k_0|\alpha(k_0)|^2\re^{2\pi\nu_2}=1-\frac{1}{8k_0}\left|\frac{s_{21}(-k_0)}{s_{11}(-k_0)}\right|^2.
       \end{equation}

     \end{prop}

     \begin{proof}
     Formula \eqref{nu=hatnu} is equivalent to proving that the following formula holds:
      \begin{equation*}
          \frac{(1+8k_0|r_1(k_0)|^2)(1-8k_0|r_1(-k_0)|^2-|r_2(k_0)|^2)}{1+8k_0|r_1(k_0)|^2-|r_2(-k_0)|^2}=1-\frac{8k_0|\alpha(k_0)|^2}{1+8k_0|r_1(k_0)|^2-|r_2(-k_0)|^2}=1-\frac{1}{8k_0}\left|\frac{s_{21}(-k_0)}{s_{11}(-k_0)}\right|^2,
      \end{equation*}
       which can be proved using the symmetry relations \eqref{sym_s} and \eqref{sym_sA} together with the fact that $s^T s^A = I$. Starting from the definition of $\hat\nu$, we compute using the symmetries \eqref{sym_s} and \eqref{sym_sA} of the spectral matrices $s(k)$ and $s^A(k)$:
         \begin{align*}
         &\frac{(1+8k_0|r_1(k_0)|^2)(1-8k_0|r_1(-k_0)|^2-|r_2(k_0)|^2)}{1+8k_0|r_1(k_0)|^2-|r_2(-k_0)|^2}=\frac{\left(1-\frac{s_{12}(k_0)}{s_{11}(k_0)}\frac{m_{12}(k_0)}{m_{11}(k_0)}\right)\left(1-\frac{s_{32}(k_0)}{s_{33}(k_0)}\frac{m_{32}(k_0)}{m_{33}(k_0)}+\frac{s_{31}(k_0)}{s_{33}(k_0)}\frac{m_{31}(k_0)}{m_{33}(k_0)}\right)}{1-\frac{s_{12}(k_0)}{s_{11}(k_0)}\frac{m_{12}(k_0)}{m_{11}(k_0)}+\frac{s_{13}(k_0)}{s_{11}(k_0)}\frac{m_{13}(k_0)}{m_{11}(k_0)}}\\
           &\qquad=1-\frac{s_{23}(k_0)}{s_{33}(k_0)}\frac{m_{23}(k_0)}{m_{33}(k_0)}=1-\frac{1}{8k_0}\frac{s_{21}(-k_0)}{s_{11}(-k_0)}\frac{s_{21}^*(-k_0)}{s^*_{11}(-k_0)}.
         \end{align*}
         On the other hand, 
        the relation concerning $\alpha(k_0)$ can also be obtained by direct calculation.
     \end{proof}
     
 In fact, according to equation \eqref{nu=hatnu} in Proposition \ref{prop_hatnu}, we can replace $\hat{\nu}$ with $\nu$. However, note that in the decompositions \eqref{Delta12_decomposition} and \eqref{Delta23_decomposition}, branch cuts are involved that take values $0$ and $\pi$. Therefore, when it comes to the choice of branch cuts, we still use the notation $\hat{\nu}$; in other cases, we use the notation $\nu$. Moreover, we stipulate 
 \begin{equation*}
     z_1^{2\ri \hat{\nu}}:=\re^{\ri(2\nu_1-\nu_2)\ln_\pi(z_1)}\re^{\ri(2\nu_3-\nu_2)\ln_0(z_1)},\qquad z_2^{2\ri \hat{\nu}}:=\re^{\ri(2\nu_1-\nu_2)\ln_0(z_2)}\re^{\ri(2\nu_3-\nu_2)\ln_\pi(z_2)}.
 \end{equation*}
 
 Based on the expansions above near $k_0$ and $-k_0$, construct the following two matrices for $\zeta\in\mathcal{I}_3$: 
     \begin{align*}
        X_1(\zeta)&=\begin{pmatrix}
            \left({b_{0}^{(-k_0)}(\zeta)}\right)^{1/2}\re^{\frac{t}{2}\theta_{12}(\zeta,-k_0)}  & 0 & 0\\
            0 & \left({b_{0}^{(-k_0)}(\zeta)}\right)^{-1/2}\re^{-\frac{t}{2}\theta_{12}(\zeta,-k_0)} & 0\\
            0 & 0 & 1
        \end{pmatrix}, \\
       X_2(\zeta)&=\begin{pmatrix}
            1 & 0 & 0\\
            0 & \left(b^{(k_0)}_{0}(\zeta)\right)^{1/2}\re^{\frac{t}{2}\theta_{23}(\zeta,k_0)} & 0\\
            0 & 0 & \left(b^{(k_0)}_{0}(\zeta)\right)^{-1/2}\re^{-\frac{t}{2}\theta_{23}(\zeta,k_0)}
        \end{pmatrix}.
    \end{align*}
    Together with the matrix transformations 
    \begin{equation*}\label{trans_X12}
         \begin{aligned}
        &\tilde{M}^{(-k_0)}(\zeta,k)=M^{(2)}(\zeta,k)X_1(\zeta), \quad &&k\in B_\epsilon(-k_0),\\
        &\tilde{M}^{(k_0)}(\zeta,k)=M^{(2)}(\zeta,k)X_2(\zeta), \quad &&k\in B_\epsilon(k_0).
    \end{aligned}
    \end{equation*}
   Hence, the jump matrices $\tilde V^{(-k_0)}_j$ of the function $\tilde{M}^{(-k_0)}(\zeta,k)$ for $k\in X^{-k_0}_j$, $j=1,2,3,4$, are as follows:
    \begin{align*}
         &\tilde V^{(-k_0)}_1(\zeta,z_2(k))=\begin{pmatrix}
					1 & -z_2^{-2\ri \hat{\nu}}b^{(-k_0)}_1\hat{\alpha}_a^*(-k)\re^{\frac{\ri z_2^2}{2}}  & 0\\
					0  & 1 & 0\\
					0 & 0 & 1\\
				\end{pmatrix},\quad
        &&\tilde V^{(-k_0)}_2(\zeta,z_2(k))=\begin{pmatrix}
					1 & 0 & 0\\
					z_2^{2\ri \hat{\nu}}\left(b^{(-k_0)}_1\right)^{-1}8k\tilde{\alpha}_a(-k)\re^{-\frac{\ri z_2^2}{2}} & 1 & 0\\
					0 & 0 & 1\\
				\end{pmatrix},\\
        &\tilde V^{(-k_0)}_3(\zeta,z_2(k))=\begin{pmatrix}
					1 & z_2^{-2\ri \hat{\nu}}b^{(-k_0)}_1\tilde{\alpha}_a^*(-k)\re^{\frac{\ri z_2^2}{2}} & 0\\
					0 & 1 & 0\\
					0 & 0 & 1\\
				\end{pmatrix},\quad
        &&\tilde V^{(-k_0)}_4(\zeta,z_2(k))=\begin{pmatrix}
					1 & 0  & 0\\
					-z_2^{2\ri \hat{\nu}}\left(b^{(-k_0)}_1\right)^{-1}8k\hat{\alpha}_a(-k)\re^{-\frac{\ri z_2^2}{2}}  & 1  & 0\\
					0  & 0  & 1\\
				\end{pmatrix},
     \end{align*}
   and the jump matrices $\tilde V^{(k_0)}_j$ of the function $\tilde{M}^{(k_0)}(\zeta,k)$ for $k\in X^{k_0}_j$, $j=5,6,7,8$, are

      \begin{align*}
        &\tilde V^{(k_0)}_5(\zeta,z_1(k))=\begin{pmatrix}
					1 & 0 & 0\\
					0 & 1 & 0\\
					0 & z_1^{-2\ri \hat{\nu}}\left({b^{(k_0)}_1}\right)^{-1}\tilde{\alpha}_a^*(k)\re^{\frac{\ri z_1^2}{2}} & 1\\
				\end{pmatrix},\quad
        &&\tilde V^{(k_0)}_6(\zeta,z_1(k))=\begin{pmatrix}
					1 & 0  & 0\\
					0 & 1 & z_1^{2\ri \hat{\nu}}{b^{(k_0)}_1}8k\hat{\alpha}_a(k)\re^{-\frac{\ri z_1^2}{2}}\\
					0 & 0  & 1\\
				\end{pmatrix},\\
        &\tilde V^{(k_0)}_7(\zeta,z_1(k))=\begin{pmatrix}
					1 & 0 & 0\\
					0 & 1 & 0\\
					0 & -z_1^{-2\ri \hat{\nu}}\left({b^{(k_0)}_1}\right)^{-1}\hat{\alpha}_a^*(k)\re^{\frac{\ri z_1^2}{2}} & 1\\
				\end{pmatrix},\quad
        &&\tilde V^{(k_0)}_8(\zeta,z_1(k))=\begin{pmatrix}
					1 & 0 & 0\\
					0 & 1 & -z_1^{2\ri \hat{\nu}}{b^{(k_0)}_1}8k\tilde{\alpha}_a(k)\re^{-\frac{\ri z_1^2}{2}}\\
					0 \!\!\!\!& 0\!\!\!\! & 1\\
				\end{pmatrix}.
     \end{align*}

     For any fixed $z_1$ and $z_2$, $b^{(-k_0)}_1\to1$ and ${b^{(k_0)}_1}\to1$ as $t\to\infty$. We can obtain that $\tilde{V}^{(-k_0)}_j(x,t,k)$ tends to the jump matrix $V^{X_3}_j$, $j=1,2,3,4,$ defined in the RH problem \ref{pcmodel1_-k1} in the Appendix \ref{AppendixA}, and $\tilde{V}^{(k_0)}_j(x,t,k)$ tends to the jump matrix $V^{X_4}_j$, $j=5,6,7,8,$ defined in the RH problem \ref{pcmodel1_k1}. This also indicates that as $t\to\infty$, the jumps of $M^{(2)}$ approach the function $M^{X_3}X_1^{-1}$ as $k$ tends to $-k_0$, and approach the function $M^{X_4}X_2^{-1}$ as $k$ tends to $k_0$. Thus, we can use the function $M^{(\pm k_0)}$, defined as follows, to approximate the eigenfunction $M^{(2)}(x,t,k)$ in $B_\epsilon(\pm k_0)$:
  \begin{align}
      &M^{(-k_0)}(\zeta,k)=X_1(\zeta)M^{X_3}(z_2(\zeta,k))X^{-1}_1(\zeta), \quad &&k\in B_\epsilon(-k_0),\label{M1-k_0}\\
      &M^{(k_0)}(\zeta,k)=Y_2(\zeta)M^{X_4}(z_1(\zeta,k))X^{-1}_2(\zeta), \quad &&k\in B_\epsilon(k_0).\label{M1k_0}
  \end{align}
  And $M^{(\pm k_0)}(k)\to I$ on $\partial B_\epsilon(\pm k_0)$ as $t\to \infty$, this ensures $M^{(\pm k_0)}(\zeta,k)$ is a good approximation of $M^{(2)}(\zeta,k)$ in $B_\epsilon(\pm k_0)$ for large $t$ and $\zeta\in\mathcal{I}_3$.

    \begin{lem}\label{lem_X12_bound}
        For $\zeta\in\mathcal{I}_3$ and $t\geq2$, the functions $X_1(\zeta)$ and $X_2(\zeta)$ are uniformly bounded, i.e.,
        \begin{equation*}
            \sup_{\zeta\in\mathcal{I}_3}\sup_{t\geq2}\left|\partial_x^lX_j^{\pm 1}(\zeta)\right|<C,\quad l=0,1,\,\,j=1,2.
        \end{equation*}
        The functions $b_0^{(\pm k_0)}(\zeta)$ and $b_1^{(\pm k_0)}(\zeta,k)$ satisfy
        \begin{align*}\label{Error_b0}
            &\left|b_0^{(k_0)}(\zeta)\right|=\re^{2\pi\nu_3-\pi\nu_2},  \quad &&\left|\partial_xb_0^{(k_0)}(\zeta)\right|\leq\frac{C\ln t}{t},\\
            &\left|b_0^{(-k_0)}(\zeta)\right|=\re^{\pi\nu_2-2\pi\nu_1}, \quad&&\left|\partial_xb_0^{(-k_0)}(\zeta)\right|\leq\frac{C\ln t}{t},
        \end{align*}
        and 
        \begin{equation*}\label{Errorb1-1}
             \begin{aligned}
            &\left|b_1^{(k_0)}(\zeta,k)-1\right|\leq C|k-k_0|(1+|\ln|k-k_0||), \quad &&\left|\partial_x b_1^{(k_0)}(\zeta,k)\right|\leq\frac{C|\ln |k-k_0||}{t},\\
            &\left|b_1^{(-k_0)}(\zeta,k)-1\right|\leq C|k+k_0|(1+|\ln|k+k_0||),\quad &&\left|\partial_x b_1^{(-k_0)}(\zeta,k)\right|\leq\frac{C|\ln |k+k_0||}{t}.
        \end{aligned}
        \end{equation*}
    
    \end{lem}
    \begin{proof}
     Direct calculation yields
        \begin{align*}
            &\rre\chi_1(-k)|_{k=-k_0}
            =-\frac{1}{2\pi\ri}\int^{+\infty}_{-k_0}\ri \pi\,{\rm{d}}\ln \left(1-8s|r_1(-s)|^2\right)=-\pi\nu_1,\\
            &\rre\chi_2(\zeta,-k)|_{k=-k_0}
            =-\frac{1}{2\pi\ri}\int^{+\infty}_{-k_0}\ri \pi\,{\rm{d}}\ln \left({1-8s|r_1(-s)|^2-|r_2(s)|^2}\right)=-\pi\nu_2,\\
            &\rre\chi_3(\zeta,-k)|_{k=-k_0}
            =-\frac{1}{2\pi\ri}\int_{k_0}^{-k_0}0\,{\rm{d}}\ln \left({1+8s|r_1(s)|^2-|r_2(-s)|^2}\right)=0,\\
            &\rre\chi_3(\zeta,k)|_{k=-k_0}
            =\frac{1}{2\pi\ri}\int_{k_0}^{-k_0}0\,{\rm{d}}\ln \left({1-8s|r_1(-s)|^2-|r_2(s)|^2}\right)=0.
        \end{align*}
        Combining $\ln_\pi(-2k_0)=\ln(-2k_0)$, we can calculate $|b_0^{(-k_0)}|=\re^{2\pi\nu_1-\pi\nu_2}$. Similarly, the following expressions hold:
        \begin{equation*}
            \rre\chi_1(k)|_{k=k_0}=0,\quad \rre\chi_2(k)|_{k=k_0}=0,\quad \rre\chi_3(k)|_{k=k_0}=\rre\chi_3(-k)|_{k=k_0}=\pi\nu_3-\pi\nu_2.
        \end{equation*}
        Hence $|b_0^{(k_0)}|=\re^{3\pi\nu_3-3\pi\nu_2-\pi(\nu_3-2\nu_2)}=\re^{2\pi\nu_3-\pi\nu_2}$.
 The remaining estimates can be established analogously by referring to Lemmas \ref{lem_delta} and \ref{lem_delta1}.
    \end{proof}

    \begin{lem}\label{lem_Error_v2-vk1_2}
    For each $\zeta\in\mathcal{I}_3$ and $t\geq 2$, the function $M^{(-k_0)}$ defined in equation \eqref{M1-k_0} is analytic and bounded function for $k\in B_\epsilon(-k_0)\setminus X^{- k_0}$, and satisfies the jump relation $M^{(-k_0)}_+(x,t,k)=M^{(-k_0)}_-(x,t,k)V_j^{(-k_0)}(x,t,k)$ across $k\in X_j^{-k_0}$, $j=1,2,3,4$.
    The eigenfunction $M^{(k_0)}$ defined in equation \eqref{M1k_0} is analytic and bounded function for $k\in B_\epsilon(k_0)\setminus X^{ k_0}$, and satisfies the jump relation $M^{(k_0)}_+(x,t,k)=M^{(k_0)}_-(x,t,k)V_j^{(k_0)}(x,t,k)$ across $k\in X_j^{k_0}$, $j=5,6,7,8$.  Moreover, the jump matrices $V^{(\pm k_0)}$ satisfy the following estimates:
        \begin{align*}
            &\left\|\partial_x^l\left(V^{(2)}(x,t,\cdot)-V^{(\pm k_0)}(x,t,\cdot)\right)\right\|_{L^1(X^{\pm k_0})}\leq \frac{C \ln t}{t},\\
            &\left\|\partial_x^l\left(V^{(2)}(x,t,\cdot)-V^{(\pm k_0)}(x,t,\cdot)\right)\right\|_{L^\infty(X^{\pm k_0})}\leq \frac{C \ln t}{\sqrt{t}}.
        \end{align*}
        In addition,
    \begin{equation*}\label{M1k1-I}
        \left\|\partial_x^l\left(\left(M^{(\pm k_0)}(x,t,\cdot)\right)^{-1}-I\right)\right\|_{L^\infty(\partial B_\epsilon(\pm k_0))}=\mathcal{O}\left(t^{-1/2}\right),
    \end{equation*}
    \begin{equation}\label{intM1-I}
    \begin{aligned}
        &\frac{1}{2\pi\ri} \int_{\partial B_\epsilon(-k_0)} \left(\left(M^{(-k_0)}(x,t,k)\right)^{-1}-I\right)\rd k=-\frac{X_1(\zeta)M^{X_3}_1(\zeta))X_1^{-1}(\zeta)}{2\sqrt{t}}+\mathcal{O}\left(t^{-1}\right),\\
        &\frac{1}{2\pi\ri} \int_{\partial B_\epsilon(k_0)} \left(\left(M^{(k_0)}(x,t,k)\right)^{-1}-I\right)\rd k=-\frac{X_2(\zeta)M^{X_4}_1(\zeta)X_2^{-1}(\zeta)}{2\sqrt{t}}+\mathcal{O}\left(t^{-1}\right),
    \end{aligned}
    \end{equation}
     uniformly for $\zeta\in\mathcal{I}_3$ and $l=0,1$. Moreover, \eqref{intM1-I} can be differentiated with respect to $x$ without increasing the error term.
    \end{lem}
    The proof follows a similar argument as in Lemma \ref{lem_Error_v2-vk1}, and is therefore omitted.

     \subsubsection{A small-norm RH problem}
     Using the functions $M^{(\pm k_0)}$ defined in equations \eqref{M1-k_0} and \eqref{M1k_0}, we construct
     \begin{equation}\label{hatM1}
     	\hat{M} = \begin{cases}
     		M^{(2)} \big( M^{(k_0)} \big)^{-1}, & k \in B_\epsilon(k_0), \\[6pt]
     		M^{(2)} \big( M^{(-k_0)} \big)^{-1}, & k \in B_\epsilon(-k_0), \\[6pt]
     		M^{(2)}, & \text{elsewhere},
     	\end{cases}
     \end{equation}
     and define the jump matrix $\hat{V}$ as
     \begin{equation}\label{hatV}
     	\hat{V} = \begin{cases}
     		M^{(\pm k_0)}_- V^{(2)} \big( M^{(\pm k_0)}_+ \big)^{-1}, & k \in \hat{\Sigma} \cap B_\epsilon(\pm k_0), \\[6pt]
     		\big( M^{(\pm k_0)} \big)^{-1}, & k \in \partial B_\epsilon(\pm k_0), \\[6pt]
     		V^{(2)}, & k \in \hat{\Sigma} \setminus B_\epsilon.
     	\end{cases}
     \end{equation}
     The contour $\hat{\Sigma}$ corresponding to $\hat{M}$ is shown in Figure \ref{fig_hatSigma}.

    \begin{figure}[htbp]
    	\centering
    	\begin{tikzpicture}[scale=1]
        
        \draw[very thick, black!40!green] (-4,0) -- (4,0);
        \draw[very thick, black!40!green] (0,2) -- (0,3.5);
        \draw[very thick, black!40!green] (0,-2) -- (0,-3.5);

        \draw[very thick, dashed, black!40!green] (0,-2) -- (0,2);
    		
    		\draw[very thick, black!40!green, -latex] (-4,0) -- (-3.2,0);
            \draw[very thick, black!40!green, -latex] (-4,0) -- (-0.5,0);
            \draw[very thick, black!40!green, -latex] (-4,0) -- (0.8,0);
            \draw[very thick, black!40!green, -latex] (-4,0) -- (3.5,0);

            \draw[very thick, black!40!green, -latex] (0,-2) -- (0,-3);
            \draw[very thick, black!40!green, -latex] (0,2) -- (0,3);

            \draw[very thick, black!40!green] (2+2/1.414,2/1.414) -- (0,-2);
            \draw[very thick, black!40!green] (2+2/1.414,-2/1.414) -- (0,2);

            \draw[very thick, black!40!green] (-2-2/1.414,2/1.414) -- (0,-2);
            \draw[very thick, black!40!green] (-2-2/1.414,-2/1.414) -- (0,2);

            \draw[very thick, black!40!green,-latex] (2,0) --(2+1.3/1.2,1.3/1.2);
            \draw[very thick, black!40!green,-latex] (2,0) -- (-1.3/1.2+2,-1.3/1.2);
            \draw[very thick, black!40!green,-latex] (2,0) -- (-1.3/1.2+2,1.3/1.2);
            \draw[very thick, black!40!green,-latex] (2,0) -- (2+1.3/1.2,-1.3/1.2);

            \draw[very thick,black!40!green,-latex] (-2,0) --(-2+1.3/1.2,1.3/1.2);
            \draw[very thick, black!40!green,-latex] (-2,0) -- (-1.3/1.2-2,-1.3/1.2);
            \draw[very thick, black!40!green,-latex] (-2,0) -- (-1.3/1.2-2,1.3/1.2);
            \draw[very thick, black!40!green,-latex] (-2,0) -- (-2+1.3/1.2,-1.3/1.2);

    		\fill (-2,0) circle (1.5pt);
            \fill (0,0) circle (1.5pt);
	      \fill (2,0) circle (1.5pt);

           \node[below] at (-2,-0.1) {$k_0$};
	     \node[below] at (0.3,0) {$\text{0}$};
	     \node[below] at (2,-0.1) {$-k_0$};

         \draw[very thick,black!20!blue] (-2,0) circle (1cm);
         \draw[very thick,black!20!blue] (2,0) circle (1cm);

         \draw[very thick, black!20!blue,-latex] (-1.9,1) -- (-2.1,1);
         \draw[very thick, black!20!blue,-latex] (2.1,1) -- (1.9,1);
         \node [above] at (-2,1) {$\partial B_\epsilon(k_0)$};
         \node [above] at (2,1) {$\partial B_\epsilon(-k_0)$};

          \node[right] at (4,0) {$\rre k$};
    		
    	\end{tikzpicture}
    	\caption{The jump contour $\hat \Sigma$ in the complex $k$-plane.}
    	\label{fig_hatSigma}
    \end{figure}

    \begin{rhp}\label{rhp_hatM1}
        Find a $3\times3$ matrix-valued function $\hat{M}(x,t,\cdot)\in I+\dot E^3(\mathbb{C}\setminus \hat{\Sigma})$ such that $\hat{M}_+=\hat{M}_-\hat{V}$ for $k\in \hat{\Sigma}$, where $\hat{\Sigma}$ is shown in Figure \ref{fig_hatSigma}.
    \end{rhp}
    
    With the definition $\tilde\Sigma=\hat{\Sigma}\setminus\left(\mathbb{R}\cup X^{\pm k_0}\cup \partial B_\epsilon(\pm k_0)\right)$, we present the following lemma.
    \begin{lem}\label{lem_hatW}
    Let $\hat{W}=\hat{V}-I$, the following estimates hold uniformly for $\zeta\in\mathcal{I}_3$, $t\geq2$ and $l=0,1$:
     \begin{align*}
            &\left\|(1+|\cdot|)\partial_x^l \hat{W}\right\|_{\left(L^1\cap L^\infty\right)(\mathbb{R})}\leq \frac{C}{t^{N}},\\
            &\left\|(1+|\cdot|)\partial_x^l \hat{W}\right\|_{\left(L^1\cap L^\infty\right)(\tilde{\Sigma})}\leq C\re^{-ct},\\
            &\left\|\partial_x^l \hat{W}\right\|_{\left(L^1\cap L^\infty\right)(\partial B_\epsilon(\pm k_0))}\leq \frac{C}{\sqrt{t}},\\
            &\left\|\partial_x^l \hat{W}\right\|_{L^1(X^{\pm k_0})}\leq \frac{C\ln t}{t},\\
            &\left\|\partial_x^l \hat{W}\right\|_{L^\infty(X^{\pm k_0})}\leq \frac{C\ln t}{\sqrt{t}}.
        \end{align*}
    \end{lem}
   The proof follows a similar argument as in Lemma \ref{lem_hatw}, and is therefore omitted.  The Lemma \eqref{lem_hatW} shows that
  \begin{equation*}
  \left\{
   \begin{aligned}
            &\left\|(1+|\cdot|)\partial_x^l \hat{W}\right\|_{L^1(\hat{\Sigma})}\leq \frac{C}{\sqrt{t}},\\
            &\left\|(1+|\cdot|)\partial_x^l \hat{W}\right\|_{L^\infty(\hat{\Sigma})}\leq \frac{C\ln t}{\sqrt{t}},
        \end{aligned}
      \right.\quad \zeta\in\mathcal{I}_3,\,\,t\geq2,\,\,l=0,1,
  \end{equation*}
  and we can obtain
  \begin{equation*}\label{hatW_Lp}
      \left\|(1+|\cdot|)\partial_x^l \hat{W}\right\|_{L^p(\hat{\Sigma})}\leq \frac{C\ln t^{\frac{p-1}{p}}}{\sqrt{t}},\quad 1\leq p\leq \infty.
  \end{equation*}
  The above estimates show that $\hat{W}\in (\dot{L}^3\cap L^\infty)(\hat{\Sigma})$, and thus define:
  \begin{equation*}
      \mathcal{C}_{\hat{W}(x,t,\cdot)}f=\mathcal{C}_-(f\hat
      W),\qquad \mathcal{C}_{\hat{W}}: (\dot{L}^3\cap L^\infty)(\hat{\Sigma})\to\dot{L}^3(\hat{\Sigma}).
  \end{equation*}
  \begin{lem}\label{lem_I-C_inverse1}
      For any sufficiently large $t$ and $\zeta\in\mathcal{I}_3$, $I-\mathcal{C}_{\hat{W}(x,t,\cdot)}\in \mathcal{B}(\dot L^3(\hat{\Sigma})) $ is invertible.
  \end{lem}

  According to the Lemma \ref{lem_I-C_inverse1} mentioned above, we define $\hat{\mu}_1(x,t,k)$ for sufficiently large $t$, $\zeta\in\mathcal{I}_3$ and $k\in \hat{\Sigma}$:
  \begin{equation*}
      \hat{\mu}_1=I+\left(I-\mathcal{C}_{\hat{W}}\right)^{-1}\mathcal{C}_{\hat{W}}I\in I+\dot L^3(\hat{\Sigma}).
  \end{equation*}

  \begin{lem}\label{lem_hatM_solution1}
      For sufficiently large time $t$, the RH problem \ref{rhp_hatM1} has a unique solution $\hat{M}\in I+\dot{E}^3(\mathbb{C} \setminus \hat{\Sigma})$, which can be written in the following form
      \begin{equation*}
          \hat{M}(x,t,k)=I+\mathcal{C}(\hat\mu_1 \hat W)=I+\frac{1}{2\pi\ri}\int_{\hat{\Sigma}}\frac{\hat{\mu}_1(x,t,k')\hat{W}(x,t,k')}{k'-k}\rd k'.
      \end{equation*}
  \end{lem}
  \begin{lem}\label{Error_mu1-I}
      For sufficiently large $t$ and $\zeta\in\mathcal{I}_3$, we have
      \begin{equation*}
          \left\|\partial_x^l(\hat{\mu}_1-I)\right\|_{L^p(\hat{\Sigma})}\leq \frac{C\ln t^{\frac{p-1}{p}}}{\sqrt{t}}, \quad 1<p<\infty, \,\, l=0,1.
      \end{equation*}
  \end{lem}

   \subsubsection{Asymptotics of $\phi(x,t)$ and $n(x,t)$}
  This subsection will provide the asymptotic forms of the solutions $\phi(x,t)$ and $n(x,t)$ of the YO equation \eqref{YOE} in region $\zeta\in\mathcal{I}_3$ as $t\to\infty$. First, define the nontangential limits
   \begin{equation*}
       Q_1(x,t)=\lim_{k\to\infty}k\left(\hat{M}(x,t,k)-I\right)=-\frac{1}{2\pi\ri}\int_{\hat{\Sigma}}\hat{\mu}_1(x,t,k)\hat{W}(x,t,k)\rd k.
   \end{equation*}
   \begin{lem}\label{lem_Q1}
       As $t\to\infty$, the following holds:
       \begin{equation*}\label{Q_leading2}
           Q_1(x,t)=-\frac{1}{2\pi\ri}\int_{\partial B_\epsilon }\hat{W}(x,t,k)\rd k+\mathcal{O}\left(\frac{\ln t}{t}\right).
       \end{equation*}
       Furthermore, the above equation can be differentiated term by term with respect to $x$ without increasing its error term.
   \end{lem}
  
   Combining the definition \eqref{hatM1} of function $\hat{M}$ and its associated jump matrix $\hat{V}$ in \eqref{hatV}, the following relationship holds for $t\to\infty$:
   \begin{align*}
       Q_1(x,t)&=-\frac{1}{2\pi\ri}\int_{\partial B_\epsilon }\left(\left(M^{(\pm k_0)}\right)^{-1}-I\right)\rd k+\mathcal{O}\left(\frac{\ln t}{t}\right)=\frac{X_1(\zeta)M^{X_3}_1(\zeta)X_1^{-1}(\zeta)}{2\sqrt{t}}+\frac{X_2(\zeta)M^{X_4}_1(\zeta)X_2^{-1}(\zeta)}{2\sqrt{t}}+\mathcal{O}\left(\frac{\ln t}{t}\right).
   \end{align*}

   By tallying all the transformations made to $M(x,t,k)$ in this section, we obtain
   \begin{equation*}
       M(x,t,k)=\hat{M}(x,t,k)R^{-1}(x,t,k) \Delta_1^{-1}(k), \qquad k\in \mathbb{C}\setminus\overline{B_\epsilon},
   \end{equation*}
   where $\Delta_1$ and $R$ are defined in equations \eqref{Delta1} and \eqref{trans_R}. Using the reconstruction formula \eqref{reconstruct}, we obtain
    \begin{align*}
        n(x,t)&=2\ri \frac{\partial}{\partial x}\lim\limits_{k\to\infty}k\left[\begin{pmatrix}
            1 & 0 & 1\\
        \end{pmatrix}\left(\hat{M}R^{-1} \Delta_1^{-1}-I\right)\right]_{1}=\mathcal{O}\left(\frac{\ln t}{t}\right),\\
    \phi(x,t)&=\ri \lim\limits_{k\to\infty}\left[\begin{pmatrix}
            0 & 1 & 0\\
        \end{pmatrix}\left(\hat{M}R^{-1} \Delta_1^{-1}-I\right)\right]_{1}\\
    &=\frac{\sqrt{2\pi}\left(2\sqrt{t}\right)^{-2\ri\hat\nu}{\rm{exp}}\left(\ri\hat\nu \ln(-2k_0)+\frac{3\pi \ri}{4}+2\ri tk_0^2-\frac{\pi\hat\nu}{2}-\pi\nu_2+ s_2\right)}{2\sqrt{t} {\alpha}^*(k_0)\Gamma(-\ri\nu)}+\mathcal{O}\left(\frac{\ln t}{t}\right),
    \end{align*}
    where $\hat{\nu}$ is given by equation \eqref{hatnu} and $s_2$ is given by equation \eqref{s2}.
    
It is then evident that the modulus of the complex component $\phi(x,t)$ is given by the following expression:
   \begin{align*}
        |\phi(x,t)|\simeq\frac{\sqrt{2k_0\hat\nu}}{\sqrt{t}}.
    \end{align*}

    \subsection{Long-time asymptotics in Region {\rm{IV}}}\label{subsec_region4}
	 When $(x,t)$ belongs to Region \text{IV}, we have $-\infty\leq \zeta\leq-M$, in other words, $\tau\in\mathcal{I}_4$. Still consider the dispersion relations $\tilde \theta_{ij}(\tau,k)$, $1\leq i<j\leq3$, given by equation \eqref{theta_tilde} defined by the parameter $\tau$.

	Similarly, here we consider the asymptotic behavior of the solution to the initial-value problem \eqref{YOE} as $x\to-\infty$, which is similar to the treatment in Region \text{I}. We refer to the approach adopted in the previous Region \text{III} and we do not repeat the details; instead, we only need to shift the analysis from the time variable $t$ to the spatial variable $x\to-\infty$, with the main modification being similar to the adjustment of the error terms as in Region \text{I}. Thus, the asymptotic behavior of the solution to the initial-value problem \eqref{YOE} for $\tau\in\mathcal{I}_4$ can be described by equation \eqref{region_4}.
	
	\begin{remark}
		Since as $\tau \to 0$, $k_0\to-\infty$, and according to Lemma \ref{lem_r_decomposition}, $r_1(k)$ and $r_2(k)$ vanish to all orders at $k=k_0$, it follows that the functions $\nu$, $\nu_1$, $\nu_2$, $\nu'$ and $s_2$ vanish to all orders as $\tau \to 0$. Consequently, equation \eqref{region_4} implies that as $x\to-\infty$:
		\begin{equation*}
			\begin{aligned}
				n(x,t)&=\mathcal{O}\left(\frac{1}{|x|^{N}}+\frac{C_N(\tau)}{|x|}\right),\qquad
				\phi(x,t)=\mathcal{O}\left(\frac{1}{|x|^{N}}+\frac{C_N(\tau)}{|x|}\right),
			\end{aligned}
		\end{equation*}
		uniformly for $\tau\in\mathcal{I}_4$. In particular, for any fixed $t\geq 0$, the above expression can be reduced to $\mathcal{O}\left({|x|^{-N}}\right) $ as $x \to -\infty$.
	\end{remark}

\begin{appendices}  

    \section{Proofs for the Direct Scattering Problem}\label{appendix:direct}
    \subsection{Proof of Proposition \ref{prop_Xkto0}}\label{Appendix_prop_kto0}
		Based on the transformation \eqref{diagonalize} applied to the diagonalization of the initial Lax pair \eqref{lax_phi}, we construct the function $\mathcal{X}(x,k):=P(k)X(x,k)$ to satisfy the following integral problem
		\begin{equation}\label{Vol_calX}
			\mathcal{X}(x,k)=P(k)-\int_{x}^{\infty}P(k)\re^{(x-y)\widehat{\mathcal{L}}(k)}\left(P^{-1}(k) \tilde{L}_1(y)\mathcal{X}(y,k)\right) \rd y,
		\end{equation}
		where $x\in\mathbb{R}$, $k\in (\overline{\mathbb{C}_+},\mathbb{R},\overline{\mathbb{C}_-})\setminus\{0\}$, and
		\begin{equation*}
			\tilde{L}_1=PL_1P^{-1}=\begin{pmatrix}
				0 & 2\ri n_0 & \overline{\phi}_0\\
				0 & 0 & 0\\
				0 & 2\phi_0 & 0\\
			\end{pmatrix},
		\end{equation*}
        is independent of $k$. It can be calculated that the kernel function
        \begin{equation*}
        	\mathcal{P}(x,y,k):=P(k)\re^{(x-y)\mathcal{L}(k)}P^{-1}(k),
        \end{equation*}
         of the above integral equation is analytic at $k=0$. It is obvious that the function $P(k)$ is analytic on $k \in \mathbb{C}$. Subsequently, the analysis of the Volterra integral equation \eqref{Vol_calX} shows that $\mathcal{X}$ is analytic in the corresponding region of $k$, and $\mathcal{X}$ and its $k$-derivatives have continuous extensions. The Taylor expansion satisfied by $\mathcal{X}$ is 
         \begin{equation*}
         	\mathcal{X}(x,k)=\mathcal{X}(x,0)+k\,\partial_k\mathcal{X}(x,0)+\frac{k^2}{2}\partial_k^2\mathcal{X}(x,0)+\cdots,\quad k\to 0,\,k\in (\overline{\mathbb{C}_+},\mathbb{R},\overline{\mathbb{C}_-})\setminus\{0\},
         \end{equation*}
		where the coefficients are smooth functions of $x\in \mathbb{R}$. By utilizing the properties of the Volterra integral equation \eqref{Vol_calX}, the partial derivative $\partial_k^j\mathcal{X}(x,k)$ converges rapidly to $\partial_k^jP(k)$ as $x\to\infty$ for any integer $j\geq 0$.
		
		The following analysis examines the behavior of $\partial_k^j\mathcal{X}(x,k)$ as $x\to-\infty$. First, we calculate
		\begin{equation*}
			\mathcal{P}(x,y,0)=\begin{pmatrix}
				1 & 0 & 0\\
				-\ri (x-y) & 1 & 0\\
				0 & 0 & 1\\
			\end{pmatrix},
		\end{equation*}
		and then
		\begin{equation*}
			\mathcal{P}(x,y,0)\tilde{L}_1(y)=\begin{pmatrix}
				0 & 2 \ri n_0(y) & \overline{\phi}_0(y)\\
				0 & 2(x-y)n_0(y) & -\ri(x-y)\overline{\phi}_0(y)\\
				0 & 2\phi_0(y) & 0\\
			\end{pmatrix},
		\end{equation*}
		from which we deduce that
		\begin{equation*}
			\mathcal{X}(x,0)=\begin{pmatrix}
				\mathcal{O}(1) & \mathcal{O}(1) & \mathcal{O}(1)\\
				\mathcal{O}(x) & \mathcal{O}(x) & \mathcal{O}(x)\\
				\mathcal{O}(1) & \mathcal{O}(1) & \mathcal{O}(1)\\
			\end{pmatrix},\quad\text{as}\,\, x\to -\infty.
		\end{equation*}
		Next, we conduct similar estimates for $\partial_k\mathcal{X}(x,0)$ and $\partial_k^2\mathcal{X}(x,0)$. Based on the aforementioned calculations for $\mathcal{P}(x,y,0)$, it can be easily derived that
		\begin{equation*}
			\partial_k\mathcal{P}(x,y,0)=\partial_k^2\mathcal{P}(x,y,0)=\begin{pmatrix}
				0 & 0 & 0\\
				0 & 0 & 0\\
				0 & 0 & 0\\
			\end{pmatrix},
		\end{equation*}
		which implies that
		\begin{equation*}
			\partial_k\mathcal{X}(x,0)=\begin{pmatrix}
				\mathcal{O}(x) & \mathcal{O}(x) & \mathcal{O}(x)\\
				\mathcal{O}(x^2) & \mathcal{O}(x^2) & \mathcal{O}(x^2)\\
				\mathcal{O}(x) & \mathcal{O}(x) & \mathcal{O}(x)\\
			\end{pmatrix},\qquad
				\partial_k^2\mathcal{X}(x,0)=\begin{pmatrix}
				\mathcal{O}(x^2) & \mathcal{O}(x^2) & \mathcal{O}(x^2)\\
				\mathcal{O}(x^3) & \mathcal{O}(x^3) & \mathcal{O}(x^3)\\
				\mathcal{O}(x^2) & \mathcal{O}(x^2) & \mathcal{O}(x^2)\\
			\end{pmatrix}, \quad\text{as}\,\, x\to -\infty.
		\end{equation*}

      According to the definition of $P(k)$ in the equation \eqref{diagonalize}, we can denote that $P^{-1}(k)$ has a simple pole at $k=0$:
		\begin{equation*}
			P^{-1}(k)=\frac{P^{(1)}}{k}+P^{(2)}+P^{(3)},
		\end{equation*}
		where
		\begin{equation*}
			P^{(1)}=\begin{pmatrix}
				-1/4 & 0 & 0\\
				0 & 0 & 0\\
				1/4 & 0 & 0\\
			\end{pmatrix},\quad
			P^{(2)}=\begin{pmatrix}
				0 & 0 & 0\\
				0 & 0 & 0\\
				0 & 1 & 0\\
			\end{pmatrix},\quad
			P^{(3)}=\begin{pmatrix}
				0 & 0 & 0\\
				0 & 0 & 1\\
				0 & 0 & 0\\
			\end{pmatrix}.
		\end{equation*}
	Thus, $X=P^{-1}\mathcal{X}$ has at most a first-order pole at $k=0$, with the expansion:
	\begin{equation*}
		X(x,k)=\frac{X^{(-1)}(x)}{k}+I+X^{(0)}(x)+k X^{(1)}(x)+\cdots,\quad k\to 0,
	\end{equation*}
	with
	\begin{align*}
		&X^{(-1)}(x)=P^{(1)}\mathcal{X}(x,0),\quad X^{(0)}_+(x)=\left( P^{(2)}+P^{(3)}\right) \mathcal{X}(x,0)+P^{(1)}\partial_k\mathcal{X}(x,0)-I,\\
		&X^{(1)}(x)=\frac{P^{(1)}}{2}\partial_k^2\mathcal{X}(x,0)+\left( P^{(2)}+P^{(3)}\right) \partial_k\mathcal{X}(x,0), \quad \text{etc}.
	\end{align*}
	In fact, $\partial_k^j\mathcal{X}(x,k)\to\partial_k^jP(k)$ rapidly as $x\to \infty$ and
	\begin{equation*}
		I=P^{-1} P=\left[\frac{P^{(1)}}{k}+P^{(2)}+P^{(3)} \right] \left[ P(0)+k\,P'(0)\right].
	\end{equation*}
	Therefore, the coefficients $X^{(l)}(x)$ $(l\geq-1)$ vanish rapidly as $x\to+\infty$.
	
	Based on the special forms of $P^{(j)}$, $j=1,2,3$, there exist complex-valued functions $\left\lbrace a_{ij},b_{ij},c_{ij} \right\rbrace_{i,j=1}^3 $ such that
	\begin{equation}\label{X_k0_pole}
		\begin{aligned}
			&X^{(-1)}(x)=P^{(1)}\mathcal{X}(x,0)=\frac{1}{4}\begin{pmatrix}
				-a_{11}(x) & -a_{12}(x) & -a_{13}(x)\\
				0 & 0 & 0\\
				a_{11}(x) & a_{12}(x) & a_{13}(x)\\
			\end{pmatrix},\\
			&X^{(0)}_+(x)=\left( P^{(2)}+ P^{(3)}\right) \mathcal{X}(x,0)+P^{(1)}\partial_k\mathcal{X}(x,0)-I\nonumber=-I+\frac{1}{4}\begin{pmatrix}
				-b_{11}(x) & -b_{12}(x) & -b_{13}(x)\\
				a_{31}(x) & a_{32}(x) & a_{33}(x)\\
				b_{11}(x)+4a_{21}(x) & b_{12}(x)+4a_{22}(x) & b_{13}(x)+4a_{23}(x)\\
			\end{pmatrix},\\
			&X^{(1)}(x)=\frac{P^{(1)}}{2}\partial_k^2\mathcal{X}(x,0)+\left( P^{(2)}+P^{(3)}\right) \partial_k\mathcal{X}(x,0)
			=\frac{1}{8}\begin{pmatrix}
				-c_{11}(x) & -c_{12}(x) & -c_{13}(x)\\
				b_{31}(x) & b_{32}(x) & b_{33}(x)\\
				c_{11}(x)+8b_{21}(x) & c_{12}(x)+8b_{22}(x) & c_{13}(x)+8b_{23}(x)
					\end{pmatrix}.
		\end{aligned}
	\end{equation}
	By exploiting the symmetries in equation \eqref{symmetry}, and comparing the coefficients of powers of $k$, They can be written in the form of equation \eqref{X-kto0}.
    
	Since the first and third rows of $\mathcal{X}(x,0)$ are bounded for all $ x \in \mathbb{R}$, it follows that $X^{(-1)}$ is also bounded. Therefore, there exists a bounded positive function $f_1(x)$ that decays rapidly as $x\to\infty$ and satisfies 
	\begin{equation*}
		\left|\alpha_{11}(x) \right| +\left|\alpha_{12}(x) \right| \leq f_1(x).
	\end{equation*}
	 Similarly, the second row of $\mathcal{X}(x,0)$, as well as the first row of $\partial_k\mathcal{X}(x,0)$, vanish rapidly as $x\to\infty$, and they grow at most linearly as $x\to-\infty$. Hence, we obtain the estimates in \eqref{alpha_i}.
	
	The proof for $Y(x,k)$ can be obtained similarly.
\hfill $\square$
\subsection{Proof of Proposition \ref{prop_s}}\label{Appendix_prop_s}
By carefully analyzing the exponential terms involved in the integral equation \eqref{sk}, we can derive the following analytic properties
		\begin{equation*}
			s_{11}(k):\,{\rm{Im}}\,k>0,\qquad s_{33}(k):\,{\rm{Im}}\,k<0.
		\end{equation*}
		Substituting the expansion of $X(x,k)$ as $k\to\infty$ as shown in Proposition \ref{prop_Xpm_k_infty}, we can obtain
		\begin{equation*}
			s(k)=I-\sum_{j=1}^{N}\frac{1}{k^j}\int_\mathbb{R} {\rm{e}}^{-y\widehat{\mathcal{L}}( k )}\left( L_1(x,k)X_j(x)\right)\rd y+\mathcal{O}\left( k^{-N-1}\right) .
		\end{equation*}
		Proposition \ref{prop_Xpm_k_infty} has already proven that the functions $X_j(x)$ and their derivatives are bounded, and that $ L_1\in \mathcal{S}(\mathbb{R}) $. Next, the off-diagonal elements of $s(k)$ as $k\to\infty$ can be obtained by integrating the above equation by parts, based on $\mathcal{L}^{-1}_{ii}(k)=\mathcal{O}(k^{-1})$, $i=1,2,3$. For the diagonal elements of $s(k)$, which do not involve exponential terms, the above equation can be directly expanded based on the form of $L_1$.

		Using the above-defined $\tilde{L}_1$ and $
	\mathcal{X}$, $s(k)$ can be rewritten as
		\begin{equation*}
			s(k)=I-\int_\mathbb{R} {\rm{e}}^{-x\widehat{\mathcal{L}}( k )}\left[ \left(\frac{P^{(1)}}{k}+P^{(2)}+P^{(3)} \right) \tilde{L}_1(x)\left( \mathcal{X}(x,0)+k\partial_k\mathcal{X}(x,0)+\cdots\right) \right] \rd x, \quad k\to 0.
		\end{equation*}
		First, consider the integral function
		\begin{equation*}
			\int_\mathbb{R} {\rm{e}}^{-x\widehat{\mathcal{L}}( k )}\left[ P^{(1)} \tilde{L}_1(x) \mathcal{X}(x,0) \right] \rd x,
		\end{equation*}
		we can obtain the diagonal elements mentioned above do not contain exponential terms. For off-diagonal elements, taking the element at position $(1,2)$ as an example, we calculate the following expression
		\begin{equation*}
			\left|\int_\mathbb{R} \left( {\rm{e}}^{-2\ri kx}-1\right) f(x)\rd x \right| \leq \int_\mathbb{R}\left|2\ri k x f(x) \right| \rd x=\mathcal{O}(k), \quad k\to 0.
		\end{equation*}
		Next, the first term in the expansion of $s(k)$ as $k\to 0$ is
		\begin{equation*}
			s^{(-1)}=-\int_\mathbb{R}  P^{(1)} \tilde{L}_1(x) \mathcal{X}(x,0) \rd x.
		\end{equation*}
		
		Based on the proof of Proposition \ref{prop_Xkto0}, we now determine the specific form of $s^{(-1)}$. In that proposition, we have already obtained a key result required for this derivation:
		\begin{equation*}
			\mathcal{X}(x,0)=\begin{pmatrix}
				4\alpha_{11}(x) & 4\alpha_{12}(x) & 4\alpha_{11}(x)\\
				\alpha_{14}(x) & \frac{8\alpha_{12}(x)\alpha_{14}(x)-\alpha^*_{15}(x)}{8\alpha_{11}(x)} & \alpha_{14}(x)\\
				\alpha_{15}(x) & \frac{\alpha_{12}(x)\alpha_{15}(x)-\alpha^*_{11}(x)}{\alpha_{11}(x) } & \alpha_{15}(x)\\
			\end{pmatrix}.
		\end{equation*}
		Substituting into the expression for $s^{(-1)}$, we can obtain the specific form of the first term in \eqref{s_k0}.

        Assuming $n_0(x)$ and $\phi_0(x)$ have compact support, the integral in \eqref{sk} converges for $k \in \mathbb{C} \setminus \{0\}$. In other words, all entries of $X(x, k), Y(x, k)$, and $s(k)$ are well-defined and analytic for $k \in \mathbb{C} \setminus \{0\}$. Since both $X$ and $Y$ satisfy \eqref{lax_X}, there exists a function $s(k)$, dependent only on $k$, such that \eqref{X_s} holds. We obtain  
      \begin{equation*}
	  	s(k)=\re^{-x\widehat{\mathcal{L}}(k)}\left(Y^{-1}(x,k)X(x,k) \right) .
	  \end{equation*}
      Taking the limit as $x \to -\infty$ yields the definition of $s$. The conclusions about $\det s$ follow directly from the properties of $X$ and $Y$.

		We can observe that $\overline{\mathcal{L}(\bar{k})}=-\mathcal{L}(k)$, and thus proceed to evaluate
		\begin{align*}
			s^\dagger\left(  \bar k \right) =&\left(\re^{-x\mathcal{L}\left(  \bar k \right) } Y^{-1}\left(  x,\bar k \right) X\left(  x,\bar k \right) \re^{x\mathcal{L}\left(  \bar k \right) }\right)^\dagger
			=\re^{-x\mathcal{L}( k)}\mathcal{A}(k) X^{-1}(x,k)Y(x, k)\mathcal{A}^{-1}(k)\re^{x\mathcal{L}(k)}
			=\mathcal{A}(k)s^{-1} (k)\mathcal{A}^{-1}(k).
		\end{align*}
		Moreover, for any $3 \times 3$ matrix $A$, we have 
		\begin{equation*}
			\mathcal{B}\re^{x\mathcal{L}(k)}\mathcal{B}A\mathcal{B}\re^{-x\mathcal{L}(k)}\mathcal{B}=\re^{-x\mathcal{L}(k)}A\re^{x\mathcal{L}(k)},
		\end{equation*} 
		and thus
		\begin{align*}
			s(-k)=&\re^{x\mathcal{L}(k)}Y^{-1}(x,-k)X(x,-k)\re^{-x\mathcal{L}(k)}
			=\mathcal{B}\left( \mathcal{B}\re^{x\mathcal{L}(k)}\mathcal{B}\right) Y^{-1}(x,k)X(x,k)\left( \mathcal{B}\re^{-x\mathcal{L}(k)}\mathcal{B}\right)\mathcal{B} 
			=\mathcal{B}s(k)\mathcal{B}.
		\end{align*}
\hfill
    \subsection{Proof of Proposition \ref{prop_XAkto0}}\label{Appendix_prop_kto02}

    Construct $\mathcal{X}^A(x,k):=P^A(k)X^A(x,k)$  to satisfy the following integral equation
		\begin{equation}\label{cal_XA}
			\mathcal{X}^A(x,k)=P^A(k)+\int_{x}^{\infty}P^A(k)\re^{-(x-y)\widehat{\mathcal{L}}(k)}\left(P^{T}(k) \tilde{L}^T_1(y)\mathcal{X}^A(y,k)\right) \rd y,
		\end{equation}
		where
		\begin{equation*}
			\tilde{L}^T_1=P^AL_1^TP^T=\begin{pmatrix}
				0 & 0 & 0\\
				2\ri n_0 & 0 & 2\phi_0\\
				\overline{\phi}_0 & 0 & 0\\
			\end{pmatrix},
		\end{equation*}
		is independent of $k$. It can be calculated that the kernel function of the above integral equation
		\begin{equation*}
			\tilde{\mathcal{P}}(x,y,k):=P^A(k)\re^{-(x-y) \mathcal{L}(k)} P^{T}(k),
		\end{equation*}
	 is analytic at $k=0$ and
		\begin{equation*}
			\tilde{\mathcal{P}}(x,y,0)=\begin{pmatrix}
				1 & \ri (x-y) & 0\\
				0 & 1 & 0\\
				0 & 0 & 1\\
			\end{pmatrix},\quad \partial_k\tilde{\mathcal{P}}(x,y,0)=\partial_k^2\tilde{\mathcal{P}}(x,y,0)=\begin{pmatrix}
			0 & 0 & 0\\
			0 & 0 & 0\\
			0 & 0 & 0\\
			\end{pmatrix}.
		\end{equation*}
		
	Analyzing the integral equation \eqref{cal_XA}, it can be seen that $\partial_k^j\mathcal{X}^A(x,k)\to\partial_k^jP^A(k)$ rapidly as $x\to\infty$  for any integer $j\geq0$, and $P^A(k)$ has a simple pole at $k = 0$. Next, multiply equation \eqref{cal_XA} by $k$, and regard it as a new Volterra integral equation for $k\mathcal{X}^A$. The similar analysis shows that the vector function  $k\mathcal{X}^A_{1}(x,k)$ is well-defined and analytic for $x\in\mathbb{R}$ and $\rim k<0$, and the vector function $k\mathcal{X}^A_{3}(x,k)$ is well-defined and analytic for $x\in\mathbb{R}$ and $\rim k>0$. Additionally, $k\mathcal{X}^A_{2}(x,k)$ is not analytic in any domain. 
	
	On the one hand, the derivative of $\partial_k^j\left( k\mathcal{X}^A\right) $ converges rapidly to $\partial_k^j\left( kP^A\right) $ as $x\to\infty$, for any integer $j\geq0$. On the other hand, the behavior of $k\mathcal{X}^A$ as $x\to-\infty$ is analyzed. Evaluate
		\begin{equation*}
			\tilde{\mathcal{P}}(x,y,0)\tilde{L}_1^T(y)=\begin{pmatrix}
				-2(x-y)n_0(y) & 0 & 2\ri (x-y)\phi_0(y)\\
				2\ri n_0(y) & 0 & 2\phi_0(y)\\
				\overline{\phi}_0(y) & 0 & 0\\
			\end{pmatrix},
		\end{equation*}
		from which we deduce that
		\begin{equation*}
			\partial_k^j\left( k\mathcal{X}^A\right) (x,0)=\begin{pmatrix}
				\mathcal{O}\left( x^{j+1}\right)  & \mathcal{O}\left( x^{j+1}\right) & \mathcal{O}\left( x^{j+1}\right)\\
				\mathcal{O}\left( x^{j}\right)  & \mathcal{O}\left( x^{j}\right) & \mathcal{O}\left( x^{j}\right)\\
				\mathcal{O}\left( x^{j}\right)  & \mathcal{O}\left( x^{j}\right) & \mathcal{O}\left( x^{j}\right)\\
			\end{pmatrix},\quad\text{as}\,\, x\to -\infty.
		\end{equation*}
		Since $k\mathcal{X}^A$ is bounded as $k\to 0$, $\mathcal{X}^A$ has at most a simple pole at $k=0$. Therefore, $X^A=P^T\mathcal{X}^A$ has at most a simple pole at $k = 0$,  the function $P^T$ is analytic at $k=0$ and can be decomposed into the following form
		\begin{equation*} 
			P^{T}(k)=k P_1+P_2+P_3,
		\end{equation*}
		where
		\begin{equation*}
			P_1=\begin{pmatrix}
				-4 & 0 & 0\\
				0 & 0 & 0\\
				0 & 0 & 0\\
			\end{pmatrix},\quad
			P_2=\begin{pmatrix}
				0 & 1 & 0\\
				0 & 0 & 0\\
				0 & 1 & 0\\
			\end{pmatrix},\quad
			P_3=\begin{pmatrix}
				0 & 0 & 0\\
				0 & 0 & 1\\
				0 & 0 & 0\\
			\end{pmatrix}.
		\end{equation*}

		Then the first three terms of the expansion of $X^A$ as $k\to0$ can be given by
		\begin{align*}
			&\left(X^A \right) ^{(-1)}(x)=\left( P_2+P_3\right) \left( k\mathcal{X}^A\right) (x,0), \\
			&\left(X^A \right) ^{(0)}(x)=P_1 \left(k\mathcal{X}^A \right)(x,0) +\left( P_2+P_3\right) \partial_k\left( k\mathcal{X}^A\right) (x,0)-I, \\
			&\left(X^A \right) ^{(1)}(x)=P_1\partial_k\left( k\mathcal{X}^A\right) (x,0)+\frac{P_2+P_3}{2} \partial_k^2\left( k\mathcal{X}^A\right) (x,0), \quad \text{etc}.
		\end{align*}

		Based on the special forms of $P^{(j)}$, $j=1,2,3$, there exist complex-valued functions $\left\lbrace d_{ij},e_{ij},f_{ij} \right\rbrace_{i,j=1}^3 $  such that
		\begin{equation*}\label{XA_k0_pole}
			\begin{aligned}
				&\left(X^A \right)^{(-1)}(x)=\left( P_2+P_3\right) \left( k\mathcal{X}^A\right) (x,0)=\begin{pmatrix}
					d_{21}(x) & d_{22}(x) & d_{23}(x)\\
					d_{31}(x) & d_{32}(x) & d_{33}(x)\\
					d_{21}(x) & d_{22}(x) & d_{23}(x)\\
				\end{pmatrix},\\
			&\left(X^A \right)^{(0)}(x)=P_1 \left(k\mathcal{X}^A \right)(x,0) +\left( P_2+P_3\right) \partial_k\left( k\mathcal{X}^A\right) (x,0)-I
			=-I+ \setlength{\arraycolsep}{2pt}\begin{pmatrix}
				-4b_{11}(x)+e_{21}(x) & -4b_{12}(x)+e_{22}(x) & -4b_{13}(x)+e_{23}(x) \\
				e_{31}(x) & e_{32}(x) & e_{33}(x)\\
				e_{21}(x) & e_{22}(x) & e_{23}(x)\\
			\end{pmatrix},\\
			&\left(X^A \right)^{(1)}(x)=P_1\partial_k\left( k\mathcal{X}^A\right) (x,0)+\frac{P_2+P_3}{2} \partial_k^2\left( k\mathcal{X}^A\right) (x,0)
			= \setlength{\arraycolsep}{2pt}\begin{pmatrix}
				-4e_{11}(x)+f_{21}(x)/2 & -4e_{12}(x)+f_{22}(x)/2 & -4e_{13}(x)+f_{23}(x)/2 \\
				f_{31}(x) & f_{32}(x) & f_{33}(x)\\
				f_{21}(x) & f_{22}(x) & f_{23}(x)\\
			\end{pmatrix}.
			\end{aligned}
		\end{equation*}
		By exploiting the symmetries in equation \eqref{symmetry}, and comparing the coefficients of powers of $k$, the proof can be completed.\hfill $\square$

\section{$N$-soliton solutions}\label{sec_Nsoliton}

 \subsection{$N$-soliton solution for $k_j \in \textbf{\textbf{Z}} \setminus \ri\mathbb{R}_+$}\label{subsec_N-soliton1}
   Similar to the solution process for the one-soliton case above, when $s_{11}(k)$ has $N$ simple zeros at $k_1,k_2,\cdots,k_N\in \textbf{Z}\setminus\ri\mathbb{R}_+$, $\tilde{M}$ can be expressed by:
   
		\begin{align*}
			\begin{pmatrix}
				 \tilde{M}_{11} & \tilde{M}_{12} & \tilde{M}_{13}\\
				 \tilde{M}_{21} & \tilde{M}_{22} & \tilde{M}_{23}\\
				 \tilde{M}_{31} & \tilde{M}_{32} & \tilde{M}_{33}\\
			\end{pmatrix}(x,t,k)&= I-\sum_{j=1}^N\frac{8\overline{k_j}\,\,\overline{C_{k_j}}\re^{-2\ri \overline{k_j}((\overline{k_j}-1)t+x)}}{k-\overline{k_j}}\begin{pmatrix}
				 \tilde{M}_{12}(x,t,\overline{k_j}) & 0 & 0\\
				 \tilde{M}_{22}(x,t,\overline{k_j}) & 0 & 0\\
				 \tilde{M}_{32}(x,t,\overline{k_j}) & 0 & 0\\
			\end{pmatrix}+\sum_{j=1}^N\frac{C_{k_j} \re^{2\ri k_j  \left( \left(k_j -1 \right)t+x\right)  }}{k-k_j } \begin{pmatrix}
			0 & \tilde{M}_{11}(x,t,k_j ) & 0\\
			0 & \tilde{M}_{21}(x,t,k_j ) & 0\\
			0 & \tilde{M}_{31}(x,t,k_j ) & 0\\
			\end{pmatrix}\\
            &\quad-\sum_{j=1}^N\frac{C_{k_j} \re^{2\ri k_j \left( \left(k_j -1 \right)t+x\right)  }}{k+k_j } \begin{pmatrix}
			0 & \tilde{M}_{13}(x,t,-k_j ) & 0\\
			0 & \tilde{M}_{23}(x,t,-k_j ) & 0\\
			0 & \tilde{M}_{33}(x,t,-k_j ) & 0\\
			\end{pmatrix}+\sum_{j=1}^N\frac{8\overline{k_j}\,\,\overline{C_{k_j}}\re^{-2\ri \overline{k_j}((\overline{k_j}-1)t+x)}}{k+\overline{k_j}}\begin{pmatrix}
				0 & 0 & \tilde{M}_{12}(x,t,-\overline{k_j})\\
				0 & 0 & \tilde{M}_{22}(x,t,-\overline{k_j})\\
				0 & 0 & \tilde{M}_{32}(x,t,-\overline{k_j})\\
			\end{pmatrix}.
		\end{align*}
    Then, the reconstructed formula \eqref{reconstruct} can be rewritten as
	\begin{equation}\label{reconstruct_soliton_N1}
		\left\lbrace 
		\begin{aligned}
			&n(x,t)=-2\ri\frac{\partial}{\partial x}\left[\sum_{j=1}^N8\overline{k_j}\,\,\overline{C_{k_j}}\re^{-2\ri \overline{k_j}((\overline{k_j}-1)t+x)} \left( \tilde{M}_{12}(x,t,\overline{k_j})+\tilde{M}_{32}(x,t,\overline{k_j})\right)\right]  ,\\
			&\phi(x,t)=-\sum_{j=1}^N8\ri \overline{k_j}\,\,\overline{C_{k_j}}\re^{-2\ri \overline{k_j}((\overline{k_j}-1)t+x)}\tilde{M}_{22}(x,t,\overline{k_j}) . 
		\end{aligned}\right.
	\end{equation}
	The algebraic equation is constructed in the following matrix form:
    \begin{equation}\label{matrix_3N1}
       \begin{pmatrix}
       	1+A_{11} & \cdots & A_{1N} & 0 & \cdots & 0 & -B_{11} & \cdots & -B_{1N} \\
       	\vdots & \ddots & \vdots & \vdots & \ddots & \vdots & \vdots & \ddots & \vdots &\\
       	A_{N1} & \cdots & 1+A_{NN} & 0 & \cdots & 0 & -B_{N1} & \cdots & -B_{NN} \\
       	0 & \cdots & 0 & 1+A_{11}-B_{11} & \cdots & A_{1N}-B_{1N} & 0 & \cdots & 0\\
       	\vdots & \ddots & \vdots & \vdots & \ddots & \vdots & \vdots & \ddots & \vdots &\\
       	0 & \cdots & 0 & A_{N1}-B_{N1} & \cdots & 1+A_{NN}-B_{NN} & 0 & \cdots & 0\\
       	-B_{11} & \cdots & -B_{1N} & 0 & \cdots & 0 & 1+A_{11} & \cdots & A_{1N} \\
       	\vdots & \ddots & \vdots & \vdots & \ddots & \vdots & \vdots & \ddots & \vdots &\\
       	-B_{N1} & \cdots & -B_{NN} & 0 & \cdots & 0 & A_{N1} & \cdots & 1+A_{NN} \\
       \end{pmatrix}\begin{pmatrix}
            \tilde{M}_{12}(x,t,\overline{k_1})\\ \vdots \\ \tilde{M}_{12}(x,t,\overline{k_N})\\
            \tilde{M}_{22}(x,t,\overline{k_1})\\ \vdots \\ \tilde{M}_{22}(x,t,\overline{k_N})\\
            \tilde{M}_{32}(x,t,\overline{k_1})\\ \vdots \\ \tilde{M}_{32}(x,t,\overline{k_N})\\
        \end{pmatrix}=\begin{pmatrix}
            a_1 \\ \vdots \\a_N \\1 \\ \vdots \\1 \\ b_1 \\ \vdots \\b_N \\
        \end{pmatrix},
    \end{equation}
    where the $3N\times 3N$ matrix on the left-hand side of the above equation is denoted as $\mathbb{A}$ and
    \begin{equation*}
    \begin{aligned}
      & A_{jn}:=\sum_{m=1}^N\frac{8\overline{k_n}\,\,\overline{C_{k_n}}C_{k_m}\re^{2\ri k_m  \left( \left(k_m -1 \right)t+x\right)-2\ri \overline{k_n}((\overline{k_n}-1)t+x)  }}{(\overline{k_j}-k_m)(k_m-\overline{k_n})},\quad &&a_j:=\sum_{m=1}^N\frac{C_{k_m}\re^{2\ri k_m  \left( \left(k_m -1 \right)t+x\right)  }}{\overline{k_j}-k_m } ,\\\\
      & B_{jn}:=\sum_{m=1}^N\frac{8\overline{k_n}\,\,\overline{C_{k_n}}C_{k_m} \re^{2\ri k_m \left( \left(k_m -1 \right)t+x\right)-2\ri \overline{k_n}((\overline{k_n}-1)t+x)  }}{(\overline{k_j}+k_m)(k_m-\overline{k_n}) },\quad && b_j:= -\sum_{m=1}^N\frac{C_{k_m} \re^{2\ri k_m \left( \left(k_m -1 \right)t+x\right)  }}{\overline{k_j}+k_m }.
    \end{aligned}
    \end{equation*}

    The matrix obtained by replacing the first column to the $3N$-th column of matrix $\mathbb{A}$ with column vector $$\begin{pmatrix}
        a_1 & \cdots & a_N & 1 & \cdots & 1 & b_1 & \cdots & b_N
    \end{pmatrix}^T$$
     are denoted as $\mathbb{A}^{(1)}_1$, $\cdots$, $\mathbb{A}^{(1)}_N$, $\mathbb{A}^{(2)}_1$, $\cdots$, $\mathbb{A}^{(2)}_N$, $\mathbb{A}^{(3)}_1$, $\cdots$, $\mathbb{A}^{(3)}_N$ respectively.  Using Cramer’s rule to solve the above algebraic problem \eqref{matrix_3N1}, we obtain the following: 
     \begin{equation*}
            \tilde{M}_{12}(x,t,\overline{k_j})=\frac{\det \mathbb{A}^{(1)}_j}{\det \mathbb{A}},\qquad \tilde{M}_{22}(x,t,\overline{k_j})=\frac{\det \mathbb{A}_{j}^{(2)}}{\det \mathbb{A}},\qquad \tilde{M}_{32}(x,t,\overline{k_j})=\frac{\det \mathbb{A}_{j}^{(3)}}{\det \mathbb{A}}.
        \end{equation*}
   Then the form of the $N$-soliton solution obtained from equation \eqref{reconstruct_soliton_N1} is
   \begin{equation*}
		\left\lbrace 
		\begin{aligned}
			&n(x,t)=-2\ri\frac{\partial}{\partial x}\left[\sum_{j=1}^N8\overline{k_j}\,\,\overline{C_{k_j}}\re^{-2\ri \overline{k_j}((\overline{k_j}-1)t+x)} \left( \frac{\det \mathbb{A}^{(1)}_j+\det \mathbb{A}^{(3)}_j}{\det \mathbb{A}}\right)\right]  ,\\
			&\phi(x,t)=-\sum_{j=1}^N8\ri \overline{k_j}\,\,\overline{C_{k_j}}\re^{-2\ri \overline{k_j}((\overline{k_j}-1)t+x)}\frac{\det \mathbb{A}_{j}^{(2)}}{\det \mathbb{A}} . 
		\end{aligned}\right.
	\end{equation*}
  \subsection{$N$-soliton solution for $k_j \in \textbf{Z} \cap\ri\mathbb{R}_+$}\label{subsec_N-soliton2} 
    Similar to the $N$-soliton solution in another scenario discussed in Subsection \ref{subsec_N-soliton1},  we present the expression for $\tilde{M}$:
    \begin{equation*}
    	\begin{aligned}
    		\tilde{M}(x,t,k)&=I+\frac{C_{\ri\eta_j} \re^{-4\eta_j(x-t)}}{k-\ri\eta_j } \begin{pmatrix}
    			0 & 0 & \tilde{M}_{11}(x,t,\ri\eta_j )\\
    			0 & 0 & \tilde{M}_{21}(x,t,\ri\eta_j )\\
    			0 & 0 & \tilde{M}_{31}(x,t,\ri\eta_j )\\
    		\end{pmatrix}-\frac{C_{\ri\eta_j} \re^{-4\eta_j(x-t)}}{k+\ri\eta_j } \begin{pmatrix}
    			\tilde{M}_{31}(x,t,\ri\eta_j ) & 0 & 0\\
    			\tilde{M}_{21}(x,t,\ri\eta_j ) & 0 & 0\\
    			\tilde{M}_{11}(x,t,\ri\eta_j ) & 0 & 0\\
    		\end{pmatrix}.
    	\end{aligned}
    \end{equation*}
     The above system of equations is equivalent to the following algebraic equation in matrix form:
    \begin{equation*}\label{matrix_3N}
    	\begin{pmatrix}
    		1 & \cdots & 0 & 0 & \cdots & 0 & D_{11} & \cdots & D_{1N} \\
    		\vdots & \ddots & \vdots & \vdots & \ddots & \vdots & \vdots & \ddots & \vdots &\\
    		0 & \cdots & 1 & 0 & \cdots & 0 & D_{N1} & \cdots & D_{NN} \\
    		0 & \cdots & 0 & 1+D_{11} & \cdots & D_{1N} & 0 & \cdots & 0\\
    		\vdots & \ddots & \vdots & \vdots & \ddots & \vdots & \vdots & \ddots & \vdots &\\
    		0 & \cdots & 0 & D_{N1} & \cdots & 1+D_{NN} & 0 & \cdots & 0\\
    		D_{11} & \cdots & D_{1N} & 0 & \cdots & 0 & 1 & \cdots & 0 \\
    		\vdots & \ddots & \vdots & \vdots & \ddots & \vdots & \vdots & \ddots & \vdots &\\
    		D_{N1} & \cdots & D_{NN} & 0 & \cdots & 0 & 0 & \cdots & 1 \\
    	\end{pmatrix}\begin{pmatrix}
    		\tilde{M}_{11}(x,t,\ri\eta_1)\\ \vdots \\ \tilde{M}_{11}(x,t,\ri\eta_N)\\
    		\tilde{M}_{21}(x,t,\ri\eta_1)\\ \vdots \\ \tilde{M}_{21}(x,t,\ri\eta_N)\\
    		\tilde{M}_{31}(x,t,\ri\eta_1)\\ \vdots \\ \tilde{M}_{31}(x,t,\ri\eta_N)\\
    	\end{pmatrix}=\begin{pmatrix}
    		1 \\ \vdots \\1 \\0 \\ \vdots \\0 \\ 0 \\ \vdots \\0 \\
    	\end{pmatrix},
    \end{equation*}
    where the $3N\times 3N$ matrix on the left-hand side of the above equation is defined as $\mathbb{B}$ and $
    	 D_{jn}:=\frac{C_{\ri\eta_n}\re^{-4\eta_n(x-t)}}{\ri(\eta_j+\eta_n)}.$
    	 
    	 Similarly, let $\mathbb{B}^{(1)}_1$, $\cdots$, $\mathbb{B}^{(1)}_N$, $\mathbb{B}^{(2)}_1$, $\cdots$, $\mathbb{B}^{(2)}_N$, $\mathbb{B}^{(3)}_1$, $\cdots$, $\mathbb{B}^{(3)}_N$ be matrices defined as follows: replace the first column, second column, and up to the $3N$-th column of matrix $\mathbb{B}$ respectively with the vector $\begin{pmatrix}
    	 	1 & \cdots & 1 & 0 & \cdots & 0 & 0 & \cdots & 0
    	 \end{pmatrix}^T$. Thus we have
    	 \begin{equation*}
    	 	\tilde{M}_{11}(x,t,\ri\eta_j)=\frac{\det \mathbb{B}^{(1)}_j}{\det \mathbb{B}},\qquad \tilde{M}_{21}(x,t,\ri\eta_j)=\frac{\det \mathbb{B}_{j}^{(2)}}{\det \mathbb{B}},\qquad \tilde{M}_{31}(x,t,\ri\eta_j)=\frac{\det \mathbb{B}_{j}^{(3)}}{\det \mathbb{B}}.
    	 \end{equation*}
    	 Then the form of the $N$-soliton solution obtained from equation \eqref{reconstruct_soliton_N1} is
    	 \begin{equation*}
    	 	\left\lbrace 
    	 	\begin{aligned}
    	 		&n(x,t)=-2\ri\frac{\partial}{\partial x}\left[\sum_{j=1}^NC_{\ri\eta_j} \re^{-4\eta_j(x-t)} \left( \frac{\det \mathbb{B}^{(1)}_j+\det \mathbb{B}^{(3)}_j}{\det \mathbb{B}}\right)\right]  ,\\
    	 		&\phi(x,t)=\sum_{j=1}^NC_{\ri\eta_j} \re^{-4\eta_j(x-t)}\frac{\det \mathbb{B}_{j}^{(2)}}{\det \mathbb{B}}=0 . 
    	 	\end{aligned}\right.
    	 \end{equation*}

 \section{The model problems}\label{AppendixA}
   Let $X_j=\left\{z\in \mathbb{C}\,|\,z=s\re^{\frac{(2j-1)\pi\ri}{4}}, 0\leq s\leq\infty\right\}$, $j=1,2,3,4$, and $X=\bigcup_{j=1}^4X_j$ as shown in Figure \eqref{fig_modelX}. Define the function $\nu=-\frac{1}{2\pi}\ln(1-8k_0\left|r_1(-k_0) \right|^2 )$, and provide the following four model problems concerning $M^{X_j}$, $j=1,2,3,4$.

    \begin{figure}[htbp]
    	\centering
    	\begin{tikzpicture}[scale=1.2]
    		
            \draw[very thick, black!20!blue] (2/1.414,2/1.414) -- (-2/1.414,-2/1.414);
            \draw[very thick, black!20!blue] (-2/1.414,2/1.414) -- (2/1.414,-2/1.414);

            \draw[very thick, black!20!blue,-latex] (0,0)--(-1.3/1.414,-1.3/1.414);
            \draw[very thick, black!20!blue,-latex] (0,0)--(-1.3/1.414,1.3/1.414);
            \draw[very thick, black!20!blue,-latex] (0,0)--(1.3/1.414,-1.3/1.414);
            \draw[very thick, black!20!blue,-latex] (0,0)--(1.3/1.414,1.3/1.414);

            \node[red!70!black,above,left] at (-1.3/1.414,-1.3/1.414) {$X_3$};
            \node[red!70!black,above,left] at (-1.3/1.414,1.3/1.414) {$X_2$};
            \node[red!70!black,above,right] at (1.3/1.414,-1.3/1.414) {$X_4$};
            \node[red!70!black,above,right] at (1.3/1.414,1.3/1.414) {$X_1$};

            \fill (0,0) circle (1.2pt);
         
          \node[below] at (0,-0.1) {$0$};
    		
    	\end{tikzpicture}
    	\caption{The jump contour $X$ in the complex $z$-plane.}
    	\label{fig_modelX}
    \end{figure}
    
    \begin{rhp}\label{pcmodel_k1}
    Find a $3\times3$ matrix-valued function $M^{X_1}(q,z_1)$ with the following properties:
    \begin{itemize}
		\item The function $M^{X_1}(q,z_1)$ is holomorphic for $z_1\in\mathbb{C}\setminus X$.
		\item The function $M^{X_1}(q,z_1)$ is analytic for $z_1\in\mathbb{C}\setminus X$, and satisfies the following relationship:
		\begin{equation*}
			M^{X_1}_+(q,z_1)=M^{X_1}_-(q,z_1)v^{X_1}(q,z_1), \quad z_1\in X,
		\end{equation*}
    where 
    \begin{align*}
         &v^{X_1}_1(q)=\begin{pmatrix}
            1 & 0 & 0\\
            0 & 1 & 0\\
            0 & \hat{q}z_1 ^{-2\ri\nu(q)}\re^{\frac{\ri z_1^2}{2}} & 1\\
        \end{pmatrix},\, z_1\in X_1, \quad &&v^{X_1}_2(q)=\begin{pmatrix}
            1 & 0 & 0\\
            0 & 1 & 8k_0q^*z_1^{2\ri\nu(q)}\re^{-\frac{\ri z_1^2}{2}}\\
            0 & 0 & 1\\
        \end{pmatrix},\, z_1\in X_2,\\
        &v^{X_1}_3(q)=\begin{pmatrix}
            1 & 0 & 0\\
            0 & 1 & 0\\
            0 & -qz_1^{-2\ri\nu(q)}\re^{\frac{\ri z_1^2}{2}}  & 1\\
        \end{pmatrix},\, z_1\in X_3,\quad
        &&v^{X_1}_4(q)=\begin{pmatrix}
            1 & 0 & 0\\
            0 & 1 & -8k_0\hat{q}^*z_1^{2\ri\nu(q)}\re^{-\frac{\ri z_1^2}{2}}\\
            0 & 0 & 1\\
        \end{pmatrix}, z_1\in X_4,
    \end{align*}
    with the branch cut running along the positive real axis, i.e., $z_1^{2\ri\nu}=\re^{2\ri\nu\ln_0(z_1)}$.
		
		\item  As $z_1\rightarrow\infty$, $M^{X_1}(q,z_1)=I+\mathcal{O}(z_1^{-1})$.
		\item  As $z_1\rightarrow0$, $M^{X_1}(q,z_1)=\mathcal{O}(1)$.
	\end{itemize}
     
    \end{rhp}

       The solution $M^{X_1}(q,z_1)$ of the RH problem \ref{pcmodel_k1} admits the following expansion
    \begin{equation}\label{Mx1_expand}
        M^{X_1}(q,z_1)=I+\frac{M_1^{X_1}(q)}{z_1}+\mathcal{O}\left(\frac{1}{z_1^2}\right), \quad z_1\to\infty,
    \end{equation}
    where
    \begin{equation*}
        M_1^{X_1}(q)=\begin{pmatrix}
            0 & 0 & 0\\
            0 & 0 & \alpha_{23}\\
            0 & \alpha_{32} & 0\\
        \end{pmatrix},
    \end{equation*}
and
     \begin{equation*}
           \alpha_{23}=\frac{\sqrt{2\pi}\re^{-\frac{\pi\ri}{4}}\re^{-\frac{5\pi\nu}{2}}}{q\Gamma(-\ri\nu)},\qquad
          \alpha_{32}=\frac{\sqrt{2\pi}\re^{\frac{\pi\ri}{4}}\re^{\frac{3\pi\nu}{2}}}{8k_0q^*\Gamma(\ri\nu)}.\\ 
       \end{equation*}

        \begin{proof}
    Define the following two matrices
        \begin{equation*}
            z_1^{\ri\nu\tilde{\sigma}_3}=\begin{pmatrix}
                1 & 0 & 0\\
                0 & z_1^{\ri\nu} & 0\\
                0 & 0 & z_1^{-\ri\nu}\\
            \end{pmatrix},\qquad 
            \re^{\frac{\ri z_1^2}{4}\tilde{\sigma}_3}=\begin{pmatrix}
                1 & 0 & 0\\
                0 & \re^{\frac{\ri z_1^2}{4}} & 0\\
                0 & 0 & \re^{-\frac{\ri z_1^2}{4}}\\ 
            \end{pmatrix}.
        \end{equation*}

        \begin{figure}[htbp]
    	\centering
    	\begin{tikzpicture}[scale=1.2]
    		
    		\draw[very thick, black!20!blue] (-2,0) -- (2,0);

            \draw[very thick, black!20!blue,-latex] (0,0) -- (1.3,0);
            \draw[very thick, black!20!blue,-latex] (-2,0) -- (-1,0);

            \draw[very thick, black!20!blue,-latex,dashed] (0,0) -- (-1.2/1.414,-1.2/1.414);
            \draw[very thick, black!20!blue,-latex,dashed] (0,0) -- (1.2/1.414,-1.2/1.414);
            \draw[very thick, black!20!blue,-latex,dashed] (0,0) -- (-1.2/1.414,1.2/1.414);
            \draw[very thick, black!20!blue,-latex,dashed] (0,0) -- (1.2/1.414,1.2/1.414);

            \draw[very thick, black!20!blue,dashed] (2/1.414,2/1.414) -- (-2/1.414,-2/1.414);
            \draw[very thick, black!20!blue,dashed] (-2/1.414,2/1.414) -- (2/1.414,-2/1.414);

            \node[red!70!black,above] at (1.3,0.3) {$\Omega_1$};
            \node[red!70!black,above] at (0,0.9) {$\Omega_2$};
            \node[red!70!black,above] at (-1.3,0.3) {$\Omega_3$};
            \node[red!70!black,below] at (-1.3,-0.5) {$\Omega_4$};
            \node[red!70!black,below] at (0,-0.9) {$\Omega_5$};
            \node[red!70!black,below] at (1.3,-0.5) {$\Omega_6$};

            \node[above,right] at (-1.1/1.414,-1.1/1.414) {$X_3$};
            \node[above,right] at (-1.1/1.414,1.1/1.414) {$X_2$};
            \node[above,left] at (1.1/1.414,-1.1/1.414) {$X_4$};
            \node[above,left] at (1.1/1.414,1.1/1.414) {$X_1$};
            \node[below] at (-1.1,0) {$X_5$};
            \node[below] at (1.2,0) {$X_6$};
            
            \fill (0,0) circle (1.2pt);
         
          \node[right] at (2,0) {$\mathbb{R}$};
          \node[below] at (0,-0.1) {$0$};
    		
    	\end{tikzpicture}
    	\caption{The domain $\Omega$ in the complex $k$-plane.}
    	\label{fig_Omega}
    \end{figure}

        Introduce transformation 
        \begin{equation*}
            \psi(q,z_1)=M^{X_1}(q,z_1)\mathcal{P}(q,z_1),
        \end{equation*}
        where
        \begin{equation*}
            \mathcal{P}(q,z_1)=\left\lbrace
            \begin{aligned}
            &v^{X_1}_1z_1^{\ri\nu\tilde{\sigma}_3},\quad && z_1\in \Omega_1,\\
            &z_1^{\ri\nu\tilde{\sigma}_3},\quad && z_1\in \Omega_2,\\
            &\left(v^{X_1}_2\right)^{-1} z_1^{\ri\nu\tilde{\sigma}_3},  \quad && z_1\in \Omega_3,\\
            &v^{X_1}_3z_1^{\ri\nu\tilde{\sigma}_3},\quad && z_1\in \Omega_4,\\
             &z_1^{\ri\nu\tilde{\sigma}_3},\quad && z_1\in \Omega_5,\\
             &\left(v^{X_1}_4\right)^{-1} z_1^{\ri\nu\tilde{\sigma}_3},  \quad && z_1\in \Omega_6,
                \end{aligned}
            \right.
        \end{equation*}
        and $\Omega_j$, $j=1,2,\cdots,6$, are defined as shown in Figure \ref{fig_Omega}.
        The direct calculation yields $\psi_+(q,z_1)=\psi_-(q,z_1)v^{\psi}(q,z_1)$, $z_1\in\mathbb{R}$. For $z_1\in X_5$, we have $z_{1+}^{\ri\nu}=z_{1-}^{\ri\nu}$, and
        \begin{align*}
            v^{\psi}_5(q,z_1)=\left(v^{X_1}_3z_1^{\ri\nu\tilde{\sigma}_3}\right)_-^{-1}\left(v^{X_1}_1z_1^{\ri\nu\tilde{\sigma}_3}\right)_+=
                \re^{-\frac{\ri z_1^2}{4}\tilde{\sigma}_3}\begin{pmatrix}
                    1 & 0 & 0\\
                    0 & 1 & -8k_0q^*\\
                    0 & q & 1-8k_0|q|^2\\
                \end{pmatrix}\re^{\frac{\ri z_1^2}{4}\tilde{\sigma}_3}:= \re^{-\frac{\ri z_1^2}{4}\tilde{\sigma}_3}\mathcal{V}(q)\re^{\frac{\ri z_1^2}{4}\tilde{\sigma}_3},
        \end{align*}
        and for $z_1\in X_6$, we have $z_{1-}^{\ri\nu}z_{1+}^{-\ri\nu}=1-8k_0|q|^2$,
        \begin{align*}
            v^{\psi}_6(q,z_1)&=\left(\left(v^{X_1}_4\right)^{-1} z_1^{\ri\nu\tilde{\sigma}_3}\right)_-^{-1}\left(v^{X_1}_1z_1^{\ri\nu\tilde{\sigma}_3}\right)_+=
                \re^{-\frac{\ri z_1^2}{4}\tilde{\sigma}_3}\begin{pmatrix}
                    1 & 0 & 0\\
                    0 & \left(1-8k_0|\hat{q}|^2z_{1-}^{2\ri\nu}z_{1+}^{-2\ri\nu}\right)z_{1-}^{-\ri\nu}z_{1+}^{\ri\nu} & -8k_0\hat q^*z_{1-}^{\ri\nu}z_{1+}^{-\ri\nu}\\
                    0 & \hat q z_{1-}^{\ri\nu}z_{1+}^{-\ri\nu} & z_{1-}^{\ri\nu}z_{1+}^{-\ri\nu}\\
                \end{pmatrix}\re^{\frac{\ri z_1^2}{4}\tilde{\sigma}_3}\\
             &=\re^{-\frac{\ri z_1^2}{4}\tilde{\sigma}_3}\begin{pmatrix}
                    1 & 0 & 0\\
                    0 & 1 & -8k_0q^*\\
                    0 & q & 1-8k_0|q|^2\\
                \end{pmatrix}\re^{\frac{\ri z_1^2}{4}\tilde{\sigma}_3}. 
            \end{align*}
        Thus, $v^{\psi}(q,z_1)=\re^{-\frac{\ri z_1^2}{4}\tilde{\sigma}_3}\mathcal{V}(q)\re^{\frac{\ri z_1^2}{4}\tilde{\sigma}_3}$, $z_1\in\mathbb{R}$. Define $
            \Psi=\psi \re^{-\frac{\ri z_1^2}{4}\tilde{\sigma}_3}=\hat{\Psi} z_1^{\ri\nu\tilde{\sigma}_3}\re^{-\frac{\ri z_1^2}{4}\tilde{\sigma}_3}.$
        As $z_1\to\infty$, the following two expansions hold:
        \begin{align*}
            &\psi=\left(I+\frac{M_1^{X_1}(q)}{z_1}+\mathcal{O}\left(\frac{1}{z_1^2}\right)\right)z_1^{\ri\nu\tilde{\sigma}_3},\qquad \hat{\Psi}=I+\frac{M_1^{X_1}(q)}{z_1}+\mathcal{O}\left(\frac{1}{z_1^2}\right).
        \end{align*}

        Moreover, the calculation shows that $\Psi$ satisfies the following jump relations
        \begin{equation*}
            \Psi_+(z_1,q)=\Psi_-(z_1,q)\mathcal{V}(q),\quad z_1\in\mathbb{R},
        \end{equation*}
            and
        \begin{equation*}
            \left(\partial_{z_1}\Psi+\frac{\ri z_1 \tilde{\sigma}_3}{2}\Psi\right)_+=\left(\partial_{z_1}\Psi+\frac{\ri z_1 \tilde{\sigma}_3}{2}\Psi\right)_-\mathcal{V}(q).
        \end{equation*}
        The above two relations imply that the function $\left(\partial_{z_1}\Psi+\frac{\ri z_1 \tilde{\sigma}_3}{2}\Psi\right)\Psi^{-1}$ has no jump on $\mathbb{R}$. According to Liouville's theorem, it can be concluded that $\left(\partial_{z_1}\Psi+\frac{\ri z_1 \tilde{\sigma}_3}{2}\Psi\right)\Psi^{-1}$ is an entire function on $\mathbb{C}$. Then we have
        \begin{align}\label{entire}
            \left(\partial_{z_1}\Psi+\frac{\ri z_1 \tilde{\sigma}_3}{2}\Psi\right)\Psi^{-1}&=\frac{\ri z_1}{2}\left[\tilde{\sigma}_3,\hat{\Psi}\right]\hat{\Psi}^{-1}+\mathcal{O}\left(z_1^{-1}\right)=\frac{\ri }{2}\left[\tilde{\sigma}_3,M_1^{X_1}\right]+\mathcal{O}\left(z_1^{-1}\right).
        \end{align}
    Assume the following matrix form:
    \begin{equation*}
        \frac{\ri }{2}\left[\tilde{\sigma}_3,M_1^{X_1}\right]:=\begin{pmatrix}
            0 & 0 & 0\\
            0 & 0 & \tilde{\alpha}_{23}\\
            0 & \tilde{\alpha}_{32} & 0\\
        \end{pmatrix},
    \end{equation*}
    where $\alpha_{23}=-\ri \tilde{\alpha}_{23}$, $\alpha_{32}=\ri \tilde{\alpha}_{32}$. Given that $\left(\partial_{z_1}\Psi+\frac{\ri z_1 \tilde{\sigma}_3}{2}\Psi\right)\Psi^{-1}$ is an entire function and equation \ref{entire}, we have
    \begin{equation*}
        \partial_{z_1}\Psi+\frac{\ri z_1 \tilde{\sigma}_3}{2}\Psi=\frac{\ri }{2}\left[\tilde{\sigma}_3,M_1^{X_1}\right]\Psi.
    \end{equation*}
   More importantly,
   \begin{equation}\label{relation_Psi}
   \begin{aligned}
       &\partial_{z_1}\Psi_{22}+\frac{\ri z_1 }{2}\Psi_{22}=\tilde{\alpha}_{23} \Psi_{32},\quad  \partial_{z_1}\Psi_{32}-\frac{\ri z_1 }{2}\Psi_{32}=\tilde{\alpha}_{32} \Psi_{22},\\
       &\partial_{z_1}\Psi_{23}+\frac{\ri z_1 }{2}\Psi_{23}=\tilde{\alpha}_{23} \Psi_{33},\quad  \partial_{z_1}\Psi_{33}-\frac{\ri z_1 }{2}\Psi_{33}=\tilde{\alpha}_{32} \Psi_{23}.
   \end{aligned}
   \end{equation}
   Therefore,
   \begin{align*}
       &\partial_{z_1}^2\Psi_{22}+\left(\frac{\ri}{2}+\frac{z_1^2}{4}-\tilde{\alpha}_{23}\tilde{\alpha}_{32}\right)\Psi_{22}=0,\quad \partial_{z_1}^2\Psi_{33}+\left(-\frac{\ri}{2}+\frac{z_1^2}{4}-\tilde{\alpha}_{23}\tilde{\alpha}_{32}\right)\Psi_{33}=0.
   \end{align*}
       For $\rim z_1>0$, we first denote $z_1=\re^{\frac{3\pi\ri}{4}}\xi$ and $\Psi_{22}^+(z_1)=\Psi_{22}^+(\re^{\frac{3\pi\ri}{4}}\xi)=g(\xi)$, $-\frac{3\pi}{4}<\arg (\xi)<\frac{\pi}{4}$, $\frac{\pi}{4}<\arg (-\xi)<\frac{5\pi}{4}$. Then we have
       \begin{equation*}
           \frac{\rd^2g(\xi)}{\rd \xi^2}+\left(\frac{1}{2}-\frac{\xi^2}{4}+a\right)g(\xi)=0,
       \end{equation*}
       where $a=\ri \tilde{\alpha}_{23}\tilde{\alpha}_{32}$. The function $ g(\xi)=c_1D_a(\xi)+c_2D_a(-\xi)$ satisfies the Weber equation 
       \begin{equation*}
           D_a(\xi)=\left\{
           \begin{aligned}
               &\xi^a\re^{-\frac{\xi^2}{4}}\left(1+\mathcal{O}\left(\frac{1}{\xi^2}\right)\right) ,\quad && |\arg \xi|<\frac{3\pi}{4},      \\
               &\xi^a\re^{-\frac{\xi^2}{4}}\left(1+\mathcal{O}\left(\frac{1}{\xi^2}\right)\right) -\frac{\sqrt{2\pi}}{\Gamma(-a)}\re^{a\pi \ri}\xi^{-a-1}\re^{\frac{\xi^2}{4}}\left(1+\mathcal{O}\left(\frac{1}{\xi^2}\right)\right),\quad && \frac{\pi}{4}<\arg \xi<\frac{5\pi}{4},      \\
               &\xi^a\re^{-\frac{\xi^2}{4}}\left(1+\mathcal{O}\left(\frac{1}{\xi^2}\right)\right) -\frac{\sqrt{2\pi}}{\Gamma(-a)}\re^{-a\pi \ri}\xi^{-a-1}\re^{\frac{\xi^2}{4}}\left(1+\mathcal{O}\left(\frac{1}{\xi^2}\right)\right),\quad && -\frac{5\pi}{4}<\arg \xi<-\frac{\pi}{4},   
  \end{aligned}\right.
       \end{equation*}
       where $D_a(z)$ represents the parabolic cylinder function.
       
       Moreover, based on the asymptotic properties of the Weber equation and combined with the asymptotic behavior of $\psi$ as $z_1\to\infty$, we can obtain
       \begin{equation*}
           g(\xi)\to z_1^{\ri\nu}\re^{-\frac{\ri z_1^2}{4}}, \quad a=\ri\nu,\quad c_1=\re^{-\frac{3\pi\nu}{4}},\quad c_2=0.
       \end{equation*}
       In other words,
       \begin{equation*}
           \Psi^+_{22}(z_1)=\re^{-\frac{3\pi\nu}{4}}D_a(\re^{-\frac{3\pi\ri}{4}}z_1),\qquad a=\ri\nu.
       \end{equation*}

       For $\rim z_1<0$, let $z_1=\re^{-\frac{\pi\ri}{4}}\xi$ and $\Psi_{22}^-(z_1)=\Psi_{22}^-(\re^{-\frac{\pi\ri}{4}}\xi)=g(\xi)$, $-\frac{3\pi}{4}<\arg \xi<\frac{\pi}{4}$. Similarly, we obtain
       \begin{align*}
           &\Psi^-_{22}(z_1)=\re^{-\frac{7\pi\nu}{4}}D_a(\re^{\frac{\pi\ri}{4}}z_1),\quad a=\ri\nu.
       \end{align*}
       Thus
       \begin{equation*}
           \Psi_{22}(q,z_1)=\left\{
           \begin{aligned}
               &\re^{-\frac{3\pi\nu}{4}}D_{\ri\nu}(\re^{-\frac{3\pi\ri}{4}}z_1),\quad &&\rim z_1>0,\\
               &\re^{-\frac{7\pi\nu}{4}}D_{\ri\nu}(\re^{\frac{\pi\ri}{4}}z_1),\quad &&\rim z_1<0.
           \end{aligned}
           \right.
       \end{equation*}
       Similarly, for $\rim z_1>0$, $z_1=e^{\frac{\pi\ri}{4}}\xi$ and for $\rim z_1<0$, $z_1=e^{-\frac{3\pi\ri}{4}}\xi$,
       \begin{equation*}
           \Psi_{33}(q,z_1)=\left\{
           \begin{aligned}
               &\re^{\frac{\pi\nu}{4}}D_{-\ri\nu}(\re^{-\frac{\pi\ri}{4}}z_1),\quad &&\rim z_1>0,\\
               &\re^{\frac{5\pi\nu}{4}}D_{-\ri\nu}(\re^{\frac{3\pi\ri}{4}}z_1),\quad &&\rim z_1<0.
           \end{aligned}
           \right.
       \end{equation*}
       Based on the relationships \eqref{relation_Psi}, we can further obtain
       \begin{equation*}
           \Psi_{23}(q,z_1)=\left\{
           \begin{aligned}
               &\re^{\frac{\pi\nu}{4}}\tilde{\alpha}_{32}^{-1}\left(\partial_{z_1}D_{-\ri\nu}(\re^{-\frac{\pi\ri}{4}}z_1)-\frac{\ri z_1}{2}D_{-\ri\nu}(\re^{-\frac{\pi\ri}{4}}z_1)\right),\quad &&\rim z_1>0,\\
               &\re^{\frac{5\pi\nu}{4}}\tilde{\alpha}_{32}^{-1}\left(\partial_{z_1}D_{-\ri\nu}(\re^{\frac{3\pi\ri}{4}}z_1)-\frac{\ri z_1}{2}D_{-\ri\nu}(\re^{\frac{3\pi\ri}{4}}z_1)\right),\quad &&\rim z_1<0,
           \end{aligned}
           \right.
       \end{equation*}
       \begin{equation*}
           \Psi_{32}(q,z_1)=\left\{
           \begin{aligned}
               &\re^{-\frac{3\pi\nu}{4}}\tilde{\alpha}_{23}^{-1}\left(\partial_{z_1}D_{\ri\nu}(\re^{-\frac{3\pi\ri}{4}}z_1)+\frac{\ri z_1}{2}D_{\ri\nu}(\re^{-\frac{3\pi\ri}{4}}z_1)\right),\quad &&\rim z_1>0,\\
               &\re^{-\frac{7\pi\nu}{4}}\tilde{\alpha}_{23}^{-1}\left(\partial_{z_1}D_{\ri\nu}(\re^{\frac{\pi\ri}{4}}z_1)+\frac{\ri z_1}{2}D_{\ri\nu}(\re^{\frac{\pi\ri}{4}}z_1)\right),\quad &&\rim z_1<0.
           \end{aligned}
           \right.
       \end{equation*}
       Finally, based on $\Psi_-^{-1}\Psi_+=V_1(q)$ and $\det\Psi=1$, we have
       \begin{align*}
           q&=\Psi_{22}^-\Psi_{32}^+-\Psi_{32}^-\Psi_{22}^+=\frac{\re^{-\frac{5\pi\nu}{2}}}{\tilde{\alpha}_{23}}W\left(D_{\ri\nu}(\re^{\frac{\pi\ri}{4}}z_1),D_{\ri\nu}(\re^{-\frac{3\pi\ri}{4}}z_1)\right)=\frac{\sqrt{2\pi}\re^{\frac{\pi\ri}{4}}\re^{-\frac{5\pi\nu}{2}}}{\tilde{\alpha}_{23}\Gamma(-\ri\nu)},\\
           q^*&=-\frac{1}{8k_0}\left(\Psi_{33}^-\Psi_{23}^+-\Psi_{23}^-\Psi_{33}^+\right)=-\frac{\re^{\frac{3\pi\nu}{2}}}{8k_0\tilde{\alpha}_{32}}W\left(D_{-\ri\nu}(\re^{\frac{3\pi\ri}{4}}z_1),D_{-\ri\nu}(\re^{-\frac{\pi\ri}{4}}z_1)\right)=\frac{\sqrt{2\pi}\re^{-\frac{\pi\ri}{4}}\re^{\frac{3\pi\nu}{2}}}{8k_0\tilde{\alpha}_{32}\Gamma(\ri\nu)}.
       \end{align*}
       Then,
       \begin{align*}
           &\alpha_{23}=-\ri \tilde{\alpha}_{23}=\frac{\sqrt{2\pi}\re^{-\frac{\pi\ri}{4}}\re^{-\frac{5\pi\nu}{2}}}{q\Gamma(-\ri\nu)},\quad \alpha_{32}=\ri \tilde{\alpha}_{32}=\frac{\sqrt{2\pi}\re^{\frac{\pi\ri}{4}}\re^{\frac{3\pi\nu}{2}}}{8k_0q^*\Gamma(\ri\nu)}. \qedhere
       \end{align*} 

    \end{proof}

     \begin{rhp}\label{pcmodel_-k1}
    Find a $3\times3$ matrix-valued function $M^{X_2}(q,z_2)$ with the following properties:
    \begin{itemize}
		\item The function $M^{X_2}(q,z_2)$ is holomorphic for $z_2\in\mathbb{C}\setminus X$.
		\item The function $M^{X_2}(q,z_2)$ is analytic for $z_2\in\mathbb{C}\setminus X$, and satisfies the following relationship:
		\begin{equation*}
			M^{X_2}_+(q,z_2)=M^{X_2}_-(q,z_2)v^{X_2}(q,z_2), \quad z_2\in X,
		\end{equation*}
    where 
    \begin{align*}
         &v^{X_2}_5(q)=\begin{pmatrix}
            1 & -qz_2^{-2\ri\nu(q)}\re^{\frac{\ri z_2^2}{2}} & 0\\
            0 & 1 & 0\\
            0 & 0 & 1\\
        \end{pmatrix}, z_2\in X_1,\quad
        &&v^{X_2}_6(q)=\begin{pmatrix}
            1 & 0 & 0\\
            -8k_0\hat{q}^*z_2^{2\ri\nu(q)}\re^{-\frac{\ri z_2^2}{2}} & 1 & 0\\
            0 & 0 & 1\\
        \end{pmatrix},z_2\in X_2,\\
        &v^{X_2}_7(q)=\begin{pmatrix}
            1 & \hat{q}z_2^{-2\ri\nu_2(q)}\re^{\frac{\ri z_2^2}{2}} & 0\\
            0 & 1 & 0\\
            0 & 0 & 1\\
        \end{pmatrix},z_2\in X_3,\quad
        &&v^{X_2}_8(q)=\begin{pmatrix}
            1 & 0 & 0\\
            8k_0q^*z_2^{2\ri\nu(q)}\re^{-\frac{\ri z_2^2}{2}} & 1 & 0\\
            0 & 0 & 1\\
        \end{pmatrix},z_2\in X_4,
    \end{align*}
    with the branch cut running along the negative real axis, i.e., $z_2^{2\ri\nu}=\re^{2\ri\nu\ln_\pi(z_2)}$.
		
		\item  As $z_2\rightarrow\infty$, $M^{X_2}(q,z_2)=I+\mathcal{O}(z_2^{-1})$.
		\item  As $z_2\rightarrow0$, $M^{X_2}(q,z_2)=\mathcal{O}(1)$.
	\end{itemize}
     
    \end{rhp}

      The solution $M^{X_2}(q,z_2)$ of the RH problem \ref{pcmodel_-k1} admits the following expansion
    \begin{equation}\label{Mx2_expand}
        M^{X_2}(q,z_2)=I+\frac{M_1^{X_2}(q)}{z_2}+\mathcal{O}\left(\frac{1}{z_2^2}\right), \quad z_2\to\infty,
    \end{equation}
    where
    \begin{equation*}
        M_1^{X_2}(q)=\begin{pmatrix}
            0 & \alpha_{12} & 0\\
            \alpha_{21} & 0 & 0\\
            0 & 0 & 0\\
        \end{pmatrix},
    \end{equation*}
    and
     \begin{equation*}
           \alpha_{12}=\frac{\sqrt{2\pi}\re^{-\frac{\pi\ri}{4}}\re^{-\frac{\pi\nu}{2}}}{8k_0q^*\Gamma(\ri\nu)},\qquad \alpha_{21}=\frac{\sqrt{2\pi}\re^{\frac{\pi\ri}{4}}\re^{-\frac{\pi\nu}{2}}}{q\Gamma(-\ri\nu)}.\\ 
       \end{equation*}

    The proof of this result is similar to the procedure of RH problem \ref{pcmodel_k1}.

     \begin{rhp}\label{pcmodel1_-k1}
    Find a $3\times3$ matrix-valued function $M^{X_3}(\zeta,z_2)$ with the following properties:
    \begin{itemize}
		\item The function $M^{X_3}(\zeta,z_2)$ is holomorphic for $z_2\in\mathbb{C}\setminus X$.
		\item The function $M^{X_3}(\zeta,z_2)$ is analytic for $z_2\in\mathbb{C}\setminus X$, and satisfies the following relationship:
		\begin{equation*}
			M^{X_3}_+(\zeta,z_2)=M^{X_3}_-(\zeta,z_2)V^{X_3}(\zeta,z_2), \quad z_2\in X,
		\end{equation*}
    where 
\begin{align*}
         &V^{X_3}_1(\zeta,z_2)=\begin{pmatrix}
					1 & -\hat{\alpha}^*(k_0)z_2^{-2\ri\hat\nu}\re^{\frac{\ri z_2^2}{2}}  & 0\\
					0  & 1 & 0\\
					0 & 0 & 1\\
				\end{pmatrix},\, z_2\in X_1,\quad
        &&V^{X_3}_2(\zeta,z_2)=\begin{pmatrix}
					1 & 0 & 0\\
					-8k_0\tilde{\alpha}(k_0)z_2^{2\ri\hat\nu}\re^{-\frac{\ri z_2^2}{2}} & 1 & 0\\
					0 & 0 & 1\\
				\end{pmatrix},z_2\in X_2,\\
        &V^{X_3}_3(\zeta,z_2)=\begin{pmatrix}
					1 & \tilde{\alpha}^*(k_0)z_2^{-2\ri\hat\nu}\re^{\frac{\ri z_2^2}{2}} & 0\\
					0 & 1 & 0\\
					0 & 0 & 1\\
				\end{pmatrix},z_2\in X_3,\quad
        &&V^{X_3}_4(\zeta,z_2)=\begin{pmatrix}
					1 & 0  & 0\\
					8k_0\hat{\alpha}(k_0)z_2^{2\ri\hat\nu}\re^{-\frac{\ri z_2^2}{2}}  & 1  & 0\\
					0  & 0  & 1\\
				\end{pmatrix},\, z_2\in X_4,
     \end{align*}
     with the branch cut running along the negative real axis, i.e., $z_2^{2\ri\nu}=\re^{2\ri\nu\ln_\pi(z_2)}$.
		
		\item  As $z_2\rightarrow\infty$, $M^{X_3}(\zeta,z_2)=I+\mathcal{O}(z_2^{-1})$.
		\item  As $z_2\rightarrow0$, $M^{X_3}(\zeta,z_2)=\mathcal{O}(1)$.
	\end{itemize}
     
    \end{rhp}

       The solution $M^{X_3}(\zeta,z_2)$ of the RH problem \ref{pcmodel1_-k1} admits the following expansion
    \begin{equation*}\label{Mx3_expand}
        M^{X_3}(\zeta,z_2)=I+\frac{M_1^{X_3}(\zeta)}{z_2}+\mathcal{O}\left(\frac{1}{z_2^2}\right), \quad z_2\to\infty,
    \end{equation*}
    where
    \begin{equation*}
        M_1^{X_3}(\zeta)=\begin{pmatrix}
            0 & \bar\alpha_{12} & 0\\
            \bar\alpha_{21} & 0 & 0\\
            0 & 0 & 0\\
        \end{pmatrix},
    \end{equation*}
and
      \begin{equation*}
           \bar\alpha_{12}=\frac{\sqrt{2\pi}\re^{-\frac{\pi\ri}{4}}\re^{-\frac{\pi\hat\nu}{2}}}{8k_0{\alpha}(k_0)\re^{2\pi(\nu_2-\nu_1)}\Gamma(\ri\hat\nu)},\qquad \bar\alpha_{21}=\frac{\sqrt{2\pi}\re^{\frac{\pi\ri}{4}}\re^{-\frac{\pi\hat\nu}{2}}}{ {\alpha}^*(k_0)\re^{2\pi\nu_1}\Gamma(-\ri\hat\nu)}.\\ 
       \end{equation*}

       \begin{proof}
           Similar to the proof of RH problem \ref{pcmodel_k1}, only the differences in the proof process are introduced below. First, define the jump $z_{2}^{2\ri\tilde{v}}=\re^{2\ri\tilde{\nu}\ln_\pi(z_2)}$ in $z_2\in X_5$, where $X_5$ is given in Figure \ref{fig_Omega}. Moreover, as a real-valued function of $\zeta$, $\tilde{\nu}=\tilde{\nu}(\zeta)$ satisfies $\tilde{\nu}=\hat\nu=-\frac{1}{2\pi }\ln\left(1-8k_0\alpha(k_0)|^2\re^{2\pi\nu_2}\right)$.
 Define the matrix functions:
           \begin{equation*}
            z_2^{\ri\tilde\nu\bar{\sigma}_3}=\begin{pmatrix}
                z_2^{\ri\tilde\nu} & 0 & 0\\
                0 & z_2^{-\ri\tilde\nu} & 0\\
                0 & 0 & 1\\
            \end{pmatrix},\qquad 
            \re^{\frac{\ri z_2^2}{4}\bar{\sigma}_3}=\begin{pmatrix}
                \re^{\frac{\ri z_2^2}{4}} & 0 & 0\\
                0 & \re^{-\frac{\ri z_2^2}{4}} & 0\\
                0 & 0 & 1\\ 
            \end{pmatrix}.
        \end{equation*}
        and
        \begin{equation*}
           \tilde{\mathcal{P}}(\zeta,z_2)=\left\lbrace
            \begin{aligned}
            &V^{X_3}_1z_2^{-\ri\tilde\nu\bar{\sigma}_3},\quad && z_2\in \Omega_1,\\
            &z_2^{-\ri\tilde\nu\bar{\sigma}_3},\quad && z_2\in \Omega_2,\\
            &\left(V^{X_3}_2\right)^{-1} z_2^{-\ri\tilde\nu\bar{\sigma}_3},  \quad && z_2\in \Omega_3,\\
            &V^{X_3}_3z_2^{-\ri\tilde\nu\bar{\sigma}_3},\quad && z_2\in \Omega_4,\\
             &z_2^{-\ri\tilde\nu\bar{\sigma}_3},\quad && z_2\in \Omega_5,\\
             &\left(V^{X_3}_4\right)^{-1} z_2^{-\ri\tilde\nu\bar{\sigma}_3},  \quad && z_2\in \Omega_6.
                \end{aligned}
            \right.
        \end{equation*}
        The form of the jump matrix $v^{\psi}$ corresponding to the transformation $\psi(\zeta,z_2)=M^{X_3}(\zeta,z_2) \tilde{\mathcal{P}}(\zeta,z_2)$ for $z_2\in X_5$ is:
        \begin{align*}
            V^{\psi}_5(\zeta,z_2)&=\left(V^{X_3}_3z_2^{-\ri\tilde\nu\bar{\sigma}_3}\right)_-^{-1}\left(\left(V^{X_3}_2\right)^{-1} z_2^{-\ri\tilde\nu\bar{\sigma}_3}\right)_+=
                \re^{\frac{\ri z_2^2}{4}\bar{\sigma}_3}z_{2-}^{\ri\tilde\nu\bar{\sigma}_3}
                \begin{pmatrix}
                1-8k_0|\tilde{\alpha}(k_0)|^2z_{2-}^{-2\ri\hat\nu}z_{2+}^{2\ri\hat\nu} & -\tilde{\alpha}^*(k_0)z_{2-}^{-2\ri\hat\nu} & 0\\
                8k_0\tilde{\alpha}(k_0)z_{2+}^{2\ri\hat\nu} & 1 & 0\\
                0 & 0 & 1\\
            \end{pmatrix}z_{2+}^{-\ri\tilde\nu\bar{\sigma}_3}\re^{-\frac{\ri z_2^2}{4}\bar{\sigma}_3}\\
            &= \re^{\frac{\ri z_2^2}{4}\bar{\sigma}_3}
            \begin{pmatrix}
                \left(1-8k_0|{\alpha}(k_0)|^2\re^{2\pi\nu_2}\right)z_{2-}^{\ri\tilde{\nu}}z_{2+}^{-\ri\tilde{\nu}} & -\tilde{\alpha}^*(k_0)z_{2-}^{-2\ri\hat\nu}z_{2-}^{\ri\tilde{\nu}}z_{2+}^{\ri\tilde{\nu}} & 0\\
                8k_0\tilde{\alpha}(k_0)z_{2+}^{2\ri\hat\nu}z_{2-}^{-\ri\tilde{\nu}}z_{2+}^{-\ri\tilde{\nu}} & z_{2-}^{-\ri\tilde{\nu}}z_{2+}^{\ri\tilde{\nu}} & 0\\
                0 & 0 & 1\\
            \end{pmatrix}\re^{-\frac{\ri z_2^2}{4}\bar{\sigma}_3}\\
           &= \re^{\frac{\ri z_2^2}{4}\bar{\sigma}_3}
           \begin{pmatrix}
                1 & -{\alpha}^*(k_0)\re^{2\pi\nu_1} & 0\\
                8k_0{\alpha}(k_0)\re^{2\pi(\nu_2-\nu_1)} & 1-8k_0|{\alpha}(k_0)|^2\re^{2\pi\nu_2} & 0\\
                0 & 0 & 1\\
            \end{pmatrix}\re^{-\frac{\ri z_2^2}{4}\bar{\sigma}_3}:=\re^{\frac{\ri z_2^2}{4}\bar{\sigma}_3}\tilde{\mathcal{V}}(\zeta)\re^{-\frac{\ri z_2^2}{4}\bar{\sigma}_3}.
        \end{align*}
        The above calculation is worth explaining: from $z_{2-}^{-2\ri\hat\nu}z_{2+}^{2\ri\hat\nu}=\re^{\ri(2\nu_1-\nu_2)(\ln_0(z_{2+})-\ln_0(z_{2-}))}\re^{\ri(2\nu_3-\nu_2)(\ln_\pi(z_{z+})-\ln_\pi(z_{z-}))}=\re^{-2\pi(2\nu_3-\nu_2)}$ and $z_{2-}^{-\ri\tilde{\nu}}z_{2+}^{\ri\tilde{\nu}}=\re^{-2\pi\hat\nu}=1-8k_0\alpha(k_0)|^2\re^{2\pi\nu_2}$ for $z_2\in X_5$, we have 
        \begin{equation*}
            \left(1-8k_0|\tilde{\alpha}(k_0)|^2z_{2-}^{-2\ri\hat\nu}z_{2+}^{2\ri\hat\nu}\right)z_{2-}^{\ri\tilde{\nu}}z_{2+}^{-\ri\tilde{\nu}}=
                \left(1-8k_0|{\alpha}(k_0)|^2\re^{2\pi\nu_2}\right)z_{2-}^{\ri\tilde{\nu}}z_{2+}^{-\ri\tilde{\nu}}=1,
                \end{equation*}
        \begin{equation*}
            \tilde{\alpha}^*(k_0)z_{2-}^{-2\ri\hat\nu}z_{2-}^{\ri\tilde{\nu}}z_{2+}^{\ri\tilde{\nu}}=\tilde{\alpha}^*(k_0)z_{2-}^{-2\ri\hat\nu}z_{2+}^{2\ri\tilde\nu}z_{2-}^{\ri\tilde{\nu}}z_{2+}^{-\ri\tilde{\nu}}=\tilde{\alpha}^*(k_0)z_{2-}^{-2\ri\hat\nu}z_{2+}^{2\ri\hat\nu}z_{2-}^{\ri\tilde{\nu}}z_{2+}^{-\ri\tilde{\nu}}=\alpha^*(k_0)\re^{2\pi\nu_3}\re^{-2\pi(2\nu_3-\nu_2)}\re^{2\pi\hat\nu}={\alpha}^*(k_0)\re^{2\pi\nu_1},
        \end{equation*}
        and 
        \begin{equation*}
           8k_0\tilde{\alpha}(k_0)z_{2+}^{2\ri\hat\nu}z_{2-}^{-\ri\tilde{\nu}}z_{2+}^{-\ri\tilde{\nu}}=8k_0\tilde{\alpha}(k_0)z_{2+}^{2\ri\hat\nu}z_{2+}^{-2\ri\tilde\nu}z_{2-}^{-\ri\tilde{\nu}}z_{2+}^{\ri\tilde{\nu}}= 8k_0\tilde{\alpha}(k_0)z_{2-}^{-\ri\tilde{\nu}}z_{2+}^{\ri\tilde{\nu}}=8k_0\alpha(k_0)\re^{2\pi\nu_3}\re^{-2\pi\hat\nu}=8k_0{\alpha}(k_0)\re^{2\pi(\nu_2-\nu_1)}.
        \end{equation*}
      On the other hand, for $z_2\in X_6$, a similar calculation yields:
      \begin{align*}
            V^{\psi}_6(\zeta,z_2)&=\left(\left(V^{X_3}_4\right)^{-1} z_2^{-\ri\tilde\nu\bar{\sigma}_3}\right)_-^{-1}\left(V^{X_3}_1z_2^{-\ri\tilde\nu\bar{\sigma}_3}\right)_+=
                \re^{\frac{\ri z_2^2}{4}\bar{\sigma}_3}z_{2}^{\ri\tilde\nu\bar{\sigma}_3}
                \begin{pmatrix}
                1 & -\hat{\alpha}^*(k_0)z_{2+}^{-2\ri\hat\nu} & 0\\
                8k_0\hat{\alpha}(k_0)z_{2-}^{2\ri\hat\nu} & 1-8k_0|\hat{\alpha}(k_0)|^2z_{2-}^{2\ri\hat\nu}z_{2+}^{-2\ri\hat\nu} & 0\\
                0 & 0 & 1\\
            \end{pmatrix}z_{2}^{-\ri\tilde\nu\bar{\sigma}_3}\re^{-\frac{\ri z_2^2}{4}\bar{\sigma}_3}\\
           &= \re^{\frac{\ri z_2^2}{4}\bar{\sigma}_3}
           \begin{pmatrix}
                1 & -\hat{\alpha}^*(k_0)z_{2+}^{-2\ri\hat\nu}z_2^{2\ri\tilde{\nu}} & 0\\
                8k_0\hat{\alpha}(k_0)z_{2-}^{2\ri\hat\nu}z_2^{-2\ri\tilde{\nu}} & 1-8k_0|\hat{\alpha}(k_0)|^2z_{2-}^{2\ri\hat\nu}z_{2+}^{-2\ri\hat\nu} & 0\\
                0 & 0 & 1\\
            \end{pmatrix}\re^{-\frac{\ri z_2^2}{4}\bar{\sigma}_3}=\re^{\frac{\ri z_2^2}{4}\bar{\sigma}_3}\tilde{\mathcal{V}}(\zeta)\re^{-\frac{\ri z_2^2}{4}\bar{\sigma}_3},
        \end{align*}
        because $|\hat{\alpha}(k_0)|^2z_{2-}^{2\ri\hat\nu}z_{2+}^{-2\ri\hat\nu}=|{\alpha}(k_0)|^2\re^{2\pi\nu_2}$, $\hat{\alpha}^*(k_0)z_{2+}^{-2\ri\hat\nu}z_2^{2\ri\tilde{\nu}}=\hat{\alpha}^*(k_0)z_{2+}^{-2\ri\hat\nu}z_{2+}^{2\ri\tilde{\nu}}=\alpha^*(k_0)\re^{2\pi\nu_1}$, and $$\hat{\alpha}(k_0)z_{2-}^{2\ri\hat\nu}z_2^{2\ri\tilde{\nu}}=\hat{\alpha}(k_0)z_{2-}^{2\ri\hat\nu}z_{2+}^{2\ri\hat{\nu}}z_{2+}^{-2\ri\hat{\nu}}z_{2+}^{2\ri\tilde{\nu}}=\hat{\alpha}(k_0)z_{2-}^{2\ri\hat\nu}z_{2+}^{2\ri\hat{\nu}}=\alpha(k_0)\re^{2\pi\nu_1}\re^{2\pi\nu_2-4\pi\nu_1}={\alpha}(k_0)\re^{2\pi(\nu_2-\nu_1)}.$$
        It should also be noted here that although $z_{2+}^{\ri\tilde\nu}=z_{2-}^{\ri\tilde\nu}$, in the computation of the matrix $V^\psi(\zeta,z_2)$ for $z_2\in X_6$, we still need to choose a uniform direction for $z_2$ approaching $X_6$. Here we choose the approach from the upper half-plane.

        The subsequent solving procedure is the same as the standard procedure for the RH problem \ref{pcmodel_k1}, and will not be repeated here. Moreover, since $\hat\nu=\tilde{\nu}$ holds as a real-valued function, in the result of RH problems \ref{pcmodel1_-k1} and \ref{pcmodel1_k1}, we still use the notation $\hat\nu$.

       \end{proof}

        \begin{rhp}\label{pcmodel1_k1}
    Find a $3\times3$ matrix-valued function $M^{X_4}(\zeta,z_1)$ with the following properties:
    \begin{itemize}
		\item The function $M^{X_4}(\zeta,z_1)$ is holomorphic for $z_1\in\mathbb{C}\setminus X$.
		\item The function $M^{X_4}(\zeta,z_1)$ is analytic for $z_1\in\mathbb{C}\setminus X$, and satisfy the following relationship:
		\begin{equation*}
			M^{X_4}_+(\zeta,z_1)=M^{X_4}_-(\zeta,z_1)V^{X_4}(\zeta,z_1), \quad z_1\in X,
		\end{equation*}
    where 
     \begin{align*}
        & V^{X_4}_5(\zeta,z_1)=\begin{pmatrix}
					1 & 0 & 0\\
					0 & 1 & 0\\
					0 & \tilde{\alpha}^*(k_0)z_1^{-2\ri\hat\nu}\re^{\frac{\ri z_1^2}{2}} & 1\\
				\end{pmatrix},\, z_1\in X_1,\quad 
      &&V^{X_4}_6(\zeta,z_1)=\begin{pmatrix}
					1 & 0  & 0\\
					0 & 1 & 8k_0\hat{\alpha}(k_0)z_1^{-2\ri\hat\nu}\re^{\frac{\ri z_1^2}{2}}\\
					0 & 0  & 1\\
				\end{pmatrix},\, z_1\in X_2,\\
        & V^{X_4}_7(\zeta,z_1)=\begin{pmatrix}
					1 & 0 & 0\\
					0 & 1 & 0\\
					0 & -\hat{\alpha}^*(k_0)z_1^{-2\ri\hat\nu}\re^{\frac{\ri z_1^2}{2}} & 1\\
				\end{pmatrix},\, z_1\in X_3,\quad 
        &&v V^{X_4}_8(\zeta,z_1)=\begin{pmatrix}
					1 & 0 & 0\\
					0 & 1 & -8k_0\tilde{\alpha}(k_0)z_1^{2\ri\hat\nu}\re^{-\frac{\ri z_1^2}{2}}\\
					0  & 0  & 1\\
				\end{pmatrix},\, z_1\in X_4,
     \end{align*}
     with the branch cut running along the positive real axis, i.e., $z_1^{2\ri\nu}=\re^{2\ri\nu\ln_0(z_1)}$.
		
		\item  As $z_1\rightarrow\infty$, $M^{X_4}(\zeta,z_1)=I+\mathcal{O}(z_1^{-1})$.
		\item  As $z_1\rightarrow0$, $M^{X_4}(\zeta,z_1)=\mathcal{O}(1)$.
	\end{itemize}
     
    \end{rhp}

       The solution $M^{X_4}(\zeta,z_1)$ of the RH problem \ref{pcmodel1_k1} admits the following expansion
    \begin{equation*}\label{Mx41_expand}
        M^{X_4}(\zeta,z_1)=I+\frac{M_1^{X_4}(\zeta)}{z_1}+\mathcal{O}\left(\frac{1}{z_1^2}\right), \quad z_1\to\infty,
    \end{equation*}
    where
    \begin{equation*}
        M_1^{X_4}(\zeta)=\begin{pmatrix}
            0 & 0 & 0\\
            0 & 0 & \bar\alpha_{23}\\
            0 & \bar\alpha_{32} & 0\\
        \end{pmatrix},
    \end{equation*}
and
      \begin{equation*}
           \bar\alpha_{23}=\frac{\sqrt{2\pi}\re^{-\frac{\pi\ri}{4}}\re^{-\frac{5\pi\hat\nu}{2}}}{{\alpha}^*(k_0)\re^{2\pi\nu_1}\Gamma(-\ri\hat\nu)},\qquad
          \bar\alpha_{32}=\frac{\sqrt{2\pi}\re^{\frac{\pi\ri}{4}}\re^{\frac{3\pi\hat\nu}{2}}}{8k_0{\alpha}(k_0)\re^{2\pi(\nu_2-\nu_1)}\Gamma(\ri\hat\nu)}.\\ 
       \end{equation*}
       
 The proof of this result is similar to the procedures of RH problems \ref{pcmodel_k1} and \ref{pcmodel1_-k1}. 
\end{appendices}\\
\par
   
\noindent{\bf Declaration of competing interest.}  

The authors have no conflicts 
of interest to declare that are relevant to the content of this paper.
	
\noindent{\bf Data availability.} 
    
    This manuscript has no associated data.

\subsection*{\bf Acknowledgements}
	This work is supported by the National Natural Science Foundation of China, Grant No. 12371247 and No. 12431008, Beijing Natural Science Foundation Grant No. 1262012 and No. JQ26004 and Key Project of the Natural Science Foundation of Inner Mongolia Autonomous Region Grant No. 2026ZD036.
	
	\bibliographystyle{amsplain}

\begin{thebibliography}{99}

\bibitem{YO_1976}
N. Yajima, M. Oikawa,
Formation and interaction of sonic-Langmuir solitons: inverse scattering method,
Prog. Theor. Phys. \textbf{56} (1976), 1719-1739.

\bibitem{YO_1975}
N. Yajima, M. Oikawa,
A class of exactly solvable nonlinear evolution equations,
Prog. Theor. Phys. \textbf{54} (1975), 1576-1577.

\bibitem{Zakharov_1972}
V. E. Zakharov,
Collapse of Langmuir waves,
Sov. Phys. JETP \textbf{35} (1972), 908-914.

\bibitem{Benney_1977}
D. J. Benney,
A general theory for interactions between short and long waves,
Stud. Appl. Math. \textbf{56} (1977), 81-94.

\bibitem{Newell-1978}
A. C. Newell,
Long waves-short waves: A solvable model,
SIAM J. Appl. Math. \textbf{35} (1978), 650-664.

\bibitem{MaYC_1978}
Y. C. Ma,
The complete solution of the long-wave-short-wave resonance equations,
Stud. Appl. Math. \textbf{59} (1978), 201-221.

\bibitem{Ma-Redekopp-1979}
Y. C. Ma and L. G. Redekopp,
Some solutions pertaining to the resonant interaction of long and short waves,
Phys. Fluids \textbf{22} (1979), 1872-1876.

\bibitem{Wright_2006}
O. C. Wright III,
Homoclinic connections of unstable plane waves of the long-wave-short-wave equations,
Stud. Appl. Math. \textbf{117} (2006), 71-93.

\bibitem{Nistazakis}
H. E. Nistazakis, D. J. Frantzeskakis, P. G. Kevrekidis,
B. A. Malomed, and R. R. Carretero-González,
Bright-dark soliton complexes in spinor Bose-Einstein condensates,
Phys. Rev. A \textbf{77} (2008), 033612.

\bibitem{Zabolotskii_2009}
A. A. Zabolotskii,
Inverse scattering transform for the Yajima-Oikawa equations with nonvanishing boundary conditions,
Phys. Rev. A \textbf{80} (2009), 063616.

\bibitem{Degasperis-2021}
M. Caso-Huerta, A. Degasperis, S. Lombardo and M. Sommacal,
A new integrable model of long wave-short wave interaction and linear stability spectra,
Proc. R. Soc. A \textbf{477} (2021), 20210408.

\bibitem{Degasperis-2022}
M. Caso-Huerta, A. Degasperis, P. Leal da Silva, S. Lombardo and M. Sommacal,
Periodic and solitary wave solutions of the long wave-short wave Yajima-Oikawa-Newell model,
Fluids \textbf{7} (2022), 227.

\bibitem{Li-Geng-2019}
R. M. Li, X. G. Geng,
On a vector long wave-short wave-type model,
Stud. Appl. Math. \textbf{144} (2020), 164-184.

\bibitem{Zakharov-Manakov-1976}
V. E. Zakharov and S. V. Manakov,
Asymptotic behavior of nonlinear wave systems integrated by the inverse scattering method,
Sov. Phys. JETP \textbf{44} (1976), 106-112.

\bibitem{DZ_1993}
P. Deift, X. Zhou,
A steepest descent method for oscillatory Riemann-Hilbert problems: asymptotics for the MKdV equation,
Ann. of Math. \textbf{137} (1993), 295-368.

\bibitem{DVZ_1994}
P. Deift, S. Venakides, X. Zhou,
The collisionless shock region for the long-time behavior of solutions of the KdV equation,
Comm. Pure Appl. Math. \textbf{47} (1994), 199-206.

\bibitem{DZ_1994-Tokyo}
P. Deift, X. Zhou,
Long-time behavior of the non-focusing nonlinear Schr\"odinger equation—a case study,
Lectures in Mathematical Sciences, New Series, Vol. 5,
University of Tokyo, Tokyo, 1994.

\bibitem{Kamvissis_2996}
S. Kamvissis,
Long-time behavior for the focusing nonlinear Schr\"odinger equation with real spectral singularities,
Comm. Math. Phys. \textbf{180} (1996), 325-341.

\bibitem{Kamvissis_2993}
S. Kamvissis,
On the long-time behavior of the doubly infinite Toda lattice under initial data decaying at infinity,
Comm. Math. Phys. \textbf{153} (1993), 479-519.

\bibitem{DKKZ_1996}
P. Deift, S. Kamvissis, T. Kriecherbauer, X. Zhou,
The Toda rarefaction problem,
Comm. Pure Appl. Math. \textbf{49} (1996), 35-83.

\bibitem{CVZ_1999}
P. J. Cheng, S. Venakides, X. Zhou,
Long-time asymptotics for the pure radiation solution of the sine-Gordon equation,
Commun. Part. Diff. Eq. \textbf{24} (1999), 1195-1262.

\bibitem{KV_1999}
A. V. Kitaev, A. H. Vartanian,
Asymptotics of solutions to the modified nonlinear Schr\"odinger equation: solitons on a nonvanishing continuous background,
SIAM J. Math. Anal. \textbf{30} (1999), 787-832.

\bibitem{BKST_2009}
A. Boutet de Monvel, A. Kostenko, D. Shepelsky, G. Teschl,
Long-time asymptotics for the Camassa-Holm equation,
SIAM J. Math. Anal. \textbf{41} (2009), 1559-1588.

\bibitem{Yamane_2014}
H. Yamane,
Long-time asymptotics for the defocusing integrable discrete nonlinear Schr\"odinger equation,
J. Math. Soc. Japan \textbf{66} (2014), 765-803.

\bibitem{BLS_2019}
A. Boutet de Monvel, J. Lenells, D. Shepelsky,
Long-time asymptotics for the Degasperis-Procesi equation on the half-line,
Ann. Inst. Fourier \textbf{69} (2019), 171-230.

\bibitem{Fan-JDE-2025}
X. Zhou, Z. Y. Wang, E. Fan,
Soliton resolution and asymptotic stability of $N$-solitons to the Degasperis-Procesi equation on the line,
J. Differential Equations \textbf{447} (2025), 113685.

\bibitem{Geng-Liu-2018}
X. G. Geng, H. Liu,
The nonlinear steepest descent method to long-time asymptotics of the coupled nonlinear Schrödinger equation,
J. Nonlinear Sci. \textbf{28} (2018), 739-763.

\bibitem{Analysis_2023}
C. Charlier, J. Lenells, D. S. Wang,
The ``good'' Boussinesq equation: long-time asymptotics,
Anal. PDE \textbf{16} (2023), 1351-1388.

\bibitem{CL-JMPA_2023}
C. Charlier, J. Lenells,
The soliton resolution conjecture for the Boussinesq equation,
J. Math. Pure. Appl. \textbf{191} (2024), 103621.

\bibitem{Wang-Zhu_2023}
D. S. Wang, X. D. Zhu,
Long-time asymptotics of the Sawada-Kotera equation on the line,
SIAM J. Math. Anal. \textbf{58} (2026), 3325-3388.

\bibitem{Huang-Wang-Zhu_2024}
L. Huang, D. S. Wang, X. D. Zhu,
Long-time asymptotics of the Tzitz\'eica equation on the line,
Math. Ann. \textbf{395} (2026), Article 19.

\bibitem{Lenells-Indiana}
C. Charlier, J. Lenells,
The ``good'' Boussinesq equation: a Riemann-Hilbert approach,
Indiana Univ. Math. J. \textbf{71} (2022), 1505-1562.

\bibitem{Zhou_SIAM_1989_vanishing}
X. Zhou,
The Riemann-Hilbert problem and inverse scattering,
SIAM J. Math. Anal. \textbf{20} (1989), 966-986.

\bibitem{Lenells_2017}
J. Lenells,
The nonlinear steepest descent method for Riemann-Hilbert problems of low regularity,
Indiana Univ. Math. J. \textbf{66} (2017), 1287-1332.

\bibitem{Lenells_2018}
J. Lenells,
Matrix Riemann-Hilbert problems with jumps across Carleson contours,
Monatsh. Math. \textbf{186} (2018), 111-152.




\end{thebibliography}

\end{document}